\documentclass[letterpaper,12pt]{article}

\usepackage[top=.9in, left=.9in, right=.9in, bottom=.9in]{geometry}
\usepackage{setspace} 
\usepackage{latexsym}
\usepackage{xcolor}
\usepackage{soul}
\usepackage{enumerate}  
\usepackage{amssymb,amsmath,bm}
\usepackage{mathtools}
\usepackage{dsfont}
\usepackage{bbm}

\usepackage{graphicx}
\usepackage{caption}
\usepackage{subcaption}
\usepackage{booktabs}
\usepackage{threeparttable}
\usepackage{dcolumn}
\newcolumntype{.}{D{.}{.}{-1}}
\newcolumntype{d}[1]{D{.}{.}{#1}}

\usepackage{titlesec}
\titlespacing*{\section}{0pt}{.5ex}{.5ex}
\titlespacing*{\subsection}{0pt}{.5ex}{.5ex}
\titlespacing*{\subsubsection}{0pt}{.5ex}{.5ex}
\titlespacing*{\paragraph}{0pt}{.5ex}{.5ex}

\usepackage{natbib}
\usepackage[colorlinks,urlcolor=blue,linkcolor=blue,citecolor=blue,bookmarks=false]{hyperref}

\usepackage[title, titletoc]{appendix} 

\usepackage{pgf,tikz}
\usepackage{pgfplots} 
\pgfplotsset{compat=1.18}
\usetikzlibrary{positioning, shapes, arrows, shadows, trees, matrix}
\usepackage{amsthm}  
\usepackage{float} 

\usepackage{fancyhdr}
\fancypagestyle{runningheads}{%
  \fancyhf{}
  \fancyhead[L]{One-Step Functional GEE}%
  \fancyhead[R]{\thepage}%
}
\usepackage{algorithm}
\usepackage[noend]{algpseudocode} 

\usepackage{etoolbox}
\AtBeginEnvironment{algorithmic}{\setstretch{1}\small}

\newtheorem{theorem}{Theorem}[section]
\newtheorem{lemma}[theorem]{Lemma}

\newtheorem{remark}[theorem]{Remark}

\newcommand*\email[1]{\href{mailto:#1}{\nolinkurl{#1}}}

\usepackage[shortlabels]{enumitem}
\def\spacingset#1{\renewcommand{\baselinestretch}{#1}\small\normalsize}

\AtBeginEnvironment{algorithm}{\spacingset{1}}
\AtBeginEnvironment{algorithmic}{%
  \spacingset{1}%
  \setlength{\itemsep}{0pt}%
  \setlength{\parskip}{0pt}%
  \setlength{\parsep}{0pt}%
  \setlength{\topsep}{0pt}%
  \setlength{\partopsep}{0pt}%
}

\spacingset{1}
\setlist[itemize]{nosep,leftmargin=*}
\setlist[enumerate]{nosep,leftmargin=*}

\begin{document}


\newcommand{\blind}{0}
\newcommand{\tit}{\bf Fast Penalized Generalized Estimating Equations for Large Longitudinal Functional Datasets}

\if0\blind

{\title{\tit}
\author{Gabriel Loewinger$^1$, Alexander W. Levis$^2$, Erjia Cui$^3$, Francisco Pereira$^1$ \\  \\
  $^1$Machine Learning Core\\
  National Institute of Mental Health \\ 3D41, Building 10, Bethesda, MD 20892; \email{gloewinger@gmail.com}\\ \\
      $^2$Department of Biostatistics, Epidemiology and Informatics, \\
  University of Pennsylvania\\
  \\
  $^3$Division of Biostatistics and Health Data Science, \\
    University of Minnesota }

\date{}

\maketitle
}\fi

\if1\blind
\title{\bf \tit}
\maketitle
\fi
  \thispagestyle{empty}

\vspace{-0.8cm} 

\begin{abstract}
Longitudinal binary or count functional data are common in neuroscience, but are often too large to analyze with existing functional regression 
methods.
We propose a one-step 
penalized generalized estimating equations framework that supports {generalized functional outcomes (e.g., count, binary, proportion, continuous-valued)} and
is fast even when datasets have a large number of clusters and large cluster sizes. The method applies to functional and scalar covariates and the
one-step estimation framework enables 
efficient smoothing parameter selection, bootstrapping, and joint confidence interval construction. 
Importantly, this semi-parametric
approach yields coefficient confidence intervals that are provably valid asymptotically even under working correlation misspecification.
By developing a general theory for adaptive one-step M-estimation, we prove that the coefficient estimates are asymptotically normal and as efficient as the fully-iterated estimator; we verify these theoretical properties in simulations.
We illustrate the benefits of our approach for analyzing large-scale neural recordings by applying it to a recent calcium imaging dataset published in \textit{Nature}. We show that our method reveals important timing effects obscured in non-functional analyses. In doing so, we also demonstrate scaling to common neuroscience dataset sizes: the one-step estimator fits to a dataset with 150,000 (binary) functional outcomes, each observed at 120 functional domain points, in only {$\sim6.5$ minutes on a laptop without parallelization. We release our methods in the $\texttt{R}$ package $\texttt{fastFGEE}$, which supports a wide range of link functions and working covariances. } 
\end{abstract}


\thispagestyle{empty}

\clearpage
\spacingset{2}
\setcounter{page}{1}
\newpage
\section{Introduction} \label{sec:intro}

Neuroscience studies in animal models provide an invaluable tool to identify the neural mechanisms underpinning 
psychiatric disorders. Researchers can estimate moment-by-moment associations between experimental covariates (e.g. behavior) and the activity of hundreds of neurons per animal, with widely used in vivo brain recording techniques like calcium imaging \citep{calImaging_review} and \textit{Neuropixels} \citep{STEINMETZ201892}.
A neuroscientist might study brain-behavior associations on, for example, a learning task in which an animal learns to press a lever for a food reward. 
These tasks are often performed over hundreds of experimental replicates called ``trials'' (longitudinal observations akin to ``patient visits''). Each trial might be defined as a five second interval starting at extension of the lever and ending with delivery of the food reward. To test whether, for example, mean neural activity is higher on trials when animals press the lever, a common strategy is to 
analyze scalar summaries of each trial's neuronal firing activity. For instance, analysts might calculate a firing rate  of neuron $i$ on trial $j$ by averaging the response, $Y_{i,j}(s)\in \{0,1\}$, across within-trial timepoints indexed by $s$: $\bar{Y}_{i,j} = \frac{1}{|\mathcal{S}|} \sum_{s \in \mathcal{S}} Y_{i,j}(s)$, where $\mathcal{S} \subset [0,5]$ denotes a grid of timepoints at which the outcome is observed. One might then test whether $\mathbb{E}(\bar{Y}_{i,j} \mid X_{i,j}=1) - \mathbb{E}(\bar{Y}_{i,j} \mid X_{i,j}=0) \neq 0$, where $X_{i,j}$ is an indicator that the animal on which neuron $i$ was recorded pressed the lever on trial $j$. 
This discards important temporal information{, however,} by summarizing across timepoints.

Alternatively, the neural response of each five second trial can be conceptualized as a functional outcome, with within-trial timepoints, $s$, representing locations along the functional domain. This allows one to apply functional data analysis 
(FDA) techniques to test how brain–behavior relationships evolve within and across trials \citep{loewinger2025} as both the responses and the covariates (e.g., behavior) {can be modeled as functional variables}. 
%
The size and complexity of our dataset, however, require specialized FDA methods. First, analyses must account for the longitudinal {correlation}, since each neuron's activity is collected across many trials.
Second, the large number of clusters ({500} neurons), large cluster sizes ({300 trials/neuron), and densely sampled functional outcomes ($|\mathcal{S}| = 120$ timepoints/trial) render} many longitudinal methods for {generalized} functional outcomes computationally impractical. {
Our semiparametric approach models correlation in functional and longitudinal directions to improve statistical efficiency, and adopts a variance estimator robust to covariance misspecification to ensure valid inference.}

To conduct inference in longitudinal FDA with large datasets, we propose a one-step estimator for functional generalized estimating equations (fGEE). Procedurally, we first fit a function-on-scalar regression with a working independence correlation structure to obtain a consistent but potentially inefficient initial estimate of the functional coefficients. We then 
update the initial estimate 
with one Newton-Raphson update step, derived from an estimating equation that models intra-cluster correlation. This approach can scale to large datasets and has desirable statistical properties. The initial estimate can be formed quickly because it ignores correlation; using only `one step' in the update is fast because it greatly reduces the number of times potentially large working covariance matrices are inverted. Importantly, our approach still captures much of the statistical efficiency afforded by modeling intra-cluster correlation in longitudinal and/or functional directions with a fully-iterated fGEE. 
In fact, we prove the one-step fGEE is asymptotically as efficient as the fully-iterated version.

Our implementation supports functional data observed on regular, irregular, dense and sparse grids with functional and scalar covariates (see Appendix~\ref{app:comp}).
We propose fast strategies for smoothing parameter tuning, cluster bootstrapping, and joint confidence interval construction. Our general theory for adaptive one-step M-estimation may be of independent interest. 
{Our methods build on the rich literature for conditional and marginal methods for longitudinal function-on-scalar regression \citep{zhu2019fmem, pffr, annurev_fda, morris2006wavelet, ramsay2005functional, guo2002functional}.}
%
Functional mixed models are a versatile conditional strategy for longitudinal FDA. For non-Gaussian functional outcomes, many existing approaches do not scale well to large cluster sizes or cluster numbers. {For example, Functional Additive Mixed Models, implemented in $\texttt{refund::pffr()}$ \citep{pffr}, can become computationally intensive when the number of clusters is in the low hundreds, even when specifying random functional intercept-only models: each participant adds $k_b$ columns to the design matrix, where $k_b$ is the number of spline bases used for the random intercept \citep{cui2022}.}
\cite{cui2022, loewinger2025} proposed a fast functional mixed models approach based on univariate mixed models fits at each functional domain point. For non-Gaussian outcomes this relies, however, on a cluster bootstrap for inference, which 
can be slow for large datasets.
Moreover, for non-Gaussian outcomes these approaches yield coefficient estimates that are only interpretable as conditional on the random effects. In many applications, estimates with marginal interpretations are desirable.

Functional GEE and Quadratic Inference Functions (QIF) are marginal methods for longitudinal function-on-scalar regression. \cite{fQIF} and \cite{fQIF2025} proposed QIF-based methods applicable to FDA, but, to the best of our understanding, these works focused on a single observation of a functional outcome per subject. \cite{gee_functional} proposed a penalized GEE for longitudinal FDA that serves as part of the inspiration for our work.
The method, however, requires inverting an $n_i L \times n_i L$ matrix at each step of model fitting, where $n_i$ is the size of cluster $i$, and $L$ is the number of points in the functional domain. \cite{li_2022} proposed a marginal estimator for continuous data, but it has not been extended to {generalized} outcomes. Taken together, marginal approaches for longitudinal functional regression with {generalized} outcomes do not scale well, limiting their widespread adoption.

{Our work builds on existing one-step approaches to scale GEE to large problem sizes. 
\cite{bootstrap_consistency} establish that a one-step GEE cluster bootstrap yields asymptotically equivalent inference as a fully-iterated GEE cluster bootstrap (see Section~\ref{sec:onestep}). This provides a basis for the computational steps we use for cross-validation and confidence interval construction. \cite{gee_onestep} use a non-adaptive, unpenalized one-step estimator to scale non-functional GEE to large cluster sizes when adopting an exchangeable working correlation. By contrast, our theory and methods are for adaptive penalized one-step M-estimation and can be used for functional data with many working correlations. 
One-step estimators are also used to enable inference for GEE in high dimensional, sparse settings.
For example, \cite{shojaie_2022} propose a one-step estimator using a projected estimating equation to enable inference of linear functionals of high dimensional GEE coefficients in 
non-functional data settings. In contrast, we focus on functional data with low-dimensional covariates where we do not encourage sparsity.} 

{Our paper is organized as follows. We 
present our methodology} in Section~\ref{sec:method}, provide theoretical results in Section~\ref{sec:theory}, simulations in Section~\ref{sec:simulations}, and a data application in Section~\ref{sec:application}.

\section{Methods} \label{sec:method}
We begin by introducing notation, adopting that used in \cite{li_2022} where possible. We suppose that we observe the functional outcome $Y_{i,j}(s)$ at point $s$ for cluster $i \in [N] \coloneqq \{1,\ldots,N\}$, at longitudinal observation (e.g. trial or visit) $j \in [n_i]$. 
 We express grids as regular (i.e. $n_i = n_i(s)~~ \forall~s \in \{s_1,\ldots,s_L\}$) and evenly spaced for ease of notation, but our methods also apply to irregular and unevenly spaced grids.
 We denote $\boldsymbol{Y}_i(s) \in \mathbb{R}^{n_i}$ as the functional outcome vector at point $s$ for cluster $i$, concatenating all observations $Y_{i,j}(s)$ for $j \in [n_i]$, and write $\boldsymbol{Y}_i = [\boldsymbol{Y}_i(s_1)^T,\ldots,\boldsymbol{Y}_i(s_L)^T]^T \in \mathbb{R}^{n_i L}$. We denote covariate vector $\mathbf{X}_{i,j} \in \mathbb{R}^{q}$ for cluster $i$ on observation $j$, and $\mathbf{X}_i = [\mathbf{X}_{i,1},\ldots, \mathbf{X}_{i,n_i}]^T \in \mathbb{R}^{n_i \times q}$. We write covariates as scalar for ease of notation, though our method and theory applies to functional covariates.

\subsection{Functional Generalized Estimating Equations}
 We consider the marginal function-on-scalar regression with link function $g$,
 \begin{equation}
 g\{\mathbb{E}(Y_{i,j}(s) \mid \mathbf{X}_{i,j})\} = \eta_{i,j}(s),~~  \eta_{i,j}(s) = \beta_0(s) + \sum_{r=1}^q X_{i,j,r} \beta_r(s) , ~~ \text{Cov} \left ( \boldsymbol{Y}_i \mid \mathbf{X}_{i} \right ) = \mathbb{V}_i^*
 \end{equation}
  where $\beta_r(\cdot)$ is a (smooth) coefficient function for covariate $r \in [q]$. We let $\mu_{i,j}(s) = g^{-1}(\eta_{i,j}(s))$ denote the {(assumed)} mean function{; $\mathbb{V}_i^*$ represents the true within-cluster covariance, which is arbitrary and unknown}.
 We now discuss estimation of $\mu_{i,j}(s)$ with spline basis expansions of the $\beta_r(\cdot)$, although our methods can be used for other basis functions. For example, denoting $\mathbf{B}(s) = [B_1(s), \ldots, B_m(s)]^T \in \mathbb{R}^m$ as a set of $m$ B-spline basis functions, we can represent the functional coefficients $\beta_r(s) = \sum_{d=1}^m \theta_{r,d} B_d(s)$. {In practice, the dimension $m$ can differ across $r$.} We denote $\boldsymbol{\theta}_r = [\theta_{r,1}, \ldots, \theta_{r,m}]^T \in \mathbb{R}^m$ as an unknown parameter vector associated with covariate $r$, $\mathbf{B} = [\mathbf{B}(s_1), \ldots, \mathbf{B}(s_L)]^T \in \mathbb{R}^{L \times m}$, and the linear predictor for a full observation of the functional outcome as $\boldsymbol{\eta}_{i,j} = [\eta_{i,j}(s_1), 
\ldots, \eta_{i,j}(s_L)]^T = \mathbf{B}\boldsymbol{\theta}_0 + \sum_{r=1}^q X_{i,j,r} \mathbf{B}\boldsymbol{\theta}_r$. We further define $\mathbb{X}_{i,j} = [\mathbf{B} ,~X_{i,j,1}\mathbf{B}, \ldots, X_{i,j,q}\mathbf{B}] \in ~\mathbb{R}^{L \times p}$, where $p = m (q+1)$. 
We then have that $\boldsymbol{\eta}_{i,j}  =  \mathbb{X}_{i,j} \boldsymbol{\theta}$, where $\boldsymbol{\theta} = [\boldsymbol{\theta}_0^T,~\boldsymbol{\theta}_1^T,~\ldots,~\boldsymbol{\theta}_q^T]^T \in \mathbb{R}^p$. Thus, we can estimate
the functional coefficient vector, $\boldsymbol{\beta}_r= [\beta_r(s_1),~\ldots,~\beta_r(s_L)]^T  \in \mathbb{R}^L$, by estimating $\boldsymbol{\theta}$ and calculating $\widehat{\boldsymbol{\beta}}_r = \mathbf{B} \widehat{\boldsymbol{\theta}}_r$.  
 
 We semi-parametrically estimate the $\boldsymbol{\theta}$ with the penalized spline-based fGEE proposed in \cite{gee_functional}. This assumes no likelihood and, if $\mu_{i,j}(s)$ is correctly specified, yields valid inference for $\{\boldsymbol{\beta}_r(s)\}_{s,r}$ even if $\text{Cov}(\boldsymbol{Y}_i \mid \mathbf{X}_i)$ is misspecified. Specifically, the mean model parameters $\boldsymbol{\theta}$ are estimated as the root of the penalized estimating equation 
\begin{align} \label{eq:gee} 
\sum_{i=1}^N \boldsymbol{U}_{\Lambda}(\mathbf{X}_i, \boldsymbol{Y}_i; {\boldsymbol{\theta}}_{\Lambda}) \coloneqq  
\sum_{i=1}^N  \mathbb{D}_i^T \mathbb{V}_i^{-1} \left( \boldsymbol{Y}_i - \boldsymbol{\mu}_i \right) - \Lambda \mathbb{S} \boldsymbol{\theta}_{\Lambda},
\end{align}
where $\mathbb{V}_i\in \mathbb{R}^{n_i L \times n_i L}$ is the working covariance matrix for cluster $i$ ({whose true covariance matrix is $\mathbb{V}_i^*$}), $\mathbb{D}_i = \frac{\partial \boldsymbol{\mu_i}(\boldsymbol{\theta})}{\partial \boldsymbol{\theta}}$, $\mathbb{X}_{i} = [\mathbb{X}_{i,1}^T, \ldots, \mathbb{X}_{i,n_i}^T]^T \in \mathbb{R}^{n_iL \times p}$, and $\boldsymbol{\mu}_{i} = [\boldsymbol{\mu}_{i,1}^T, 
\ldots, \boldsymbol{\mu}_{i,n_i}^T]^T \in \mathbb{R}^{n_iL}$.
The pre-specified penalty matrix, $\mathbb{S} \in \mathbb{R}^{p \times p}$, is associated with the diagonal matrix of smoothing parameters $\Lambda \in \mathbb{R}^{p \times p}$ {(see below for details).}
Although no likelihood is adopted, the estimating equation~\eqref{eq:gee} can be derived from the score equations from, for example, an exponential dispersion family \citep{liang1986longitudinal}; we add the penalty term for improved estimation in finite samples. Compared to a working independence matrix {(e.g., $\mathbb{V}_i = \sigma^2 I_{n_i L}$)}, estimation of $\boldsymbol{\theta}$ can be made more efficient and accurate by exploiting correlation, in functional and longitudinal directions, by choosing the working $\mathbb{V}_i$ to estimate $\mathbb{V}_i^*$. Although such choices for $\mathbb{V}_i$ in this fGEE model yield desirable statistical properties for longitudinal FDA, estimation is computationally intensive: estimating $\boldsymbol{\theta}$ based on equation~\eqref{eq:gee} requires inversion of the $n_i L \times n_i L$ covariance matrix $\mathbb{V}_i$ for each cluster $i$, at each step in an optimization procedure. 
\subsection{One-step fGEE} \label{sec:onestep}
To scale fGEE to large datasets, we propose a one-step estimator of the form $\widehat{\boldsymbol{\theta}}_{ \Lambda_1}^{(1)} = \widehat{\boldsymbol{\theta}}_{ \Lambda_0}^{(0)} -  \left \{ \widehat{\mathbb{E}} \left ( \nabla_{\boldsymbol{\theta}} \boldsymbol{U}_{\Lambda_1}(\mathbf{X}_i, \boldsymbol{Y}_i; \widehat{\boldsymbol{\theta}}_{ \Lambda_0}^{(0)}) \right ) \right \}^{-1} \frac{1}{N} \sum_{i=1}^N\boldsymbol{U}_{\Lambda_1}\left (\mathbf{X}_i, \boldsymbol{Y}_i; \widehat{\boldsymbol{\theta}}_{ \Lambda_0}^{(0)} \right)$, 
 where 
$\widehat{\boldsymbol{\theta}}_{ \Lambda_0}^{(0)}$ is an initial estimate fit with smoothing parameters $\Lambda_0$. 
{We denote $\mathbb{D}_i(\widehat{\boldsymbol{\theta}}^{(0)}_{\Lambda_0}) = \frac{\partial \boldsymbol{\mu_i}(\boldsymbol{\theta})}{\partial \boldsymbol{\theta}} \big |_{\boldsymbol{\theta} = \widehat{\boldsymbol{\theta}}^{(0)}_{\Lambda_0}}$, $\widehat{\mathbb{V}}_i(\widehat{\boldsymbol{\theta}}_{ \Lambda_0}^{(0)})$ as the working covariance estimated with $\widehat{\boldsymbol{\theta}}_{ \Lambda_0}^{(0)}$, and $\widehat{\boldsymbol{\mu}}_i(\widehat{\boldsymbol{\theta}}_{ \Lambda_0}^{(0)}) = g^{-1}(\mathbb{X}_i \widehat{\boldsymbol{\theta}}_{ \Lambda_0}^{(0)})$ with $g^{-1}$ applied component-wise.
Our proposed estimator is formulated as follows}
\begin{align} \label{emp_IF}
\widehat{\boldsymbol{\theta}}_{ \Lambda_1}^{(1)} = \widehat{\boldsymbol{\theta}}_{ \Lambda_0}^{(0)} + \left [  \frac{1}{N} \sum_{i=1}^N \{\mathbb{D}_i(\widehat{\boldsymbol{\theta}}^{(0)}_{\Lambda_0})\}^T \{\widehat{\mathbb{V}}_i(\widehat{\boldsymbol{\theta}}_{ \Lambda_0}^{(0)})\}^{-1} \{\mathbb{D}_i(\widehat{\boldsymbol{\theta}}^{(0)}_{\Lambda_0})\}  + \Lambda_1 \mathbb{S}  \right ]^{-1} \times ~~~~~~~~~~\notag \\ \frac{1}{N} \sum_{i=1}^N\left [ \{\mathbb{D}_i(\widehat{\boldsymbol{\theta}}^{(0)}_{\Lambda_0})\}^T \{\widehat{\mathbb{V}}_i(\widehat{\boldsymbol{\theta}}_{ \Lambda_0}^{(0)})\}^{-1} \left \{ \boldsymbol{Y}_i - \widehat{\boldsymbol{\mu}}_i(\widehat{\boldsymbol{\theta}}^{(0)}_{\Lambda_0}) \right \}  - \Lambda_1 \mathbb{S} \widehat{\boldsymbol{\theta}}_{ \Lambda_0}^{(0)} \right ],
\end{align}
%
where the component of the gradient $\nabla_{\boldsymbol{\theta}} \boldsymbol{U}_{\Lambda_1}(\mathbf{X}_i, \boldsymbol{Y}_i; \widehat{\boldsymbol{\theta}}_{ \Lambda_0}^{(0)})$ corresponding to the derivative of $\mathbb{D}_i$ is evaluated at the population $\boldsymbol{\theta}$ and can be ignored in expectation.
Any consistent estimator for $\boldsymbol{\theta}$ can be used for $\widehat{\boldsymbol{\theta}}_{ \Lambda_0}^{(0)}$; 
{we use a parametric function-on-scalar regression that ignores within-cluster correlation of outcome values across $s$ and $j$, and uses the same penalty and spline bases as the fGEE (similar to solving~\eqref{eq:gee} with a working independence correlation).} 
In practice, we estimate $\widehat{\boldsymbol{\theta}}_{ \Lambda_0}^{(0)}$ with the $\texttt{refund::pffr()}$ function \citep{pffr} 
(see Appendix~\ref{app:pffr} for details). We formalize the necessary consistency properties of $\widehat{\boldsymbol{\theta}}_{ \Lambda_0}^{(0)}$ for the population parameter, $\boldsymbol{\theta}$ in Section~\ref{sec:theory}.

The one-step can be conceptualized as a de-biasing of, or equivalently a single Newton-Raphson step from, the initial estimate $\widehat{\boldsymbol{\theta}}_{ \Lambda_0}^{(0)}$. It is much faster than the fully-iterated fGEE, because it requires inversion of a working covariance matrix only twice per cluster: 1) $\{\widehat{\mathbb{V}}_i(\widehat{\boldsymbol{\theta}}_{ \Lambda_0}^{(0)})\}^{-1}$ to estimate $\widehat{\boldsymbol{\theta}}_{ \Lambda_1}^{(1)}$, and 2) $\{\widehat{\mathbb{V}}_i(\widehat{\boldsymbol{\theta}}_{ \Lambda_1}^{(1)})\}^{-1}$ to estimate $\widehat{\text{Var}} \left( \widehat{\boldsymbol{\theta}}_{\Lambda_1}^{(1)} \right)$ 
(see Section~\ref{sec:sandwich}). {We summarize the estimation steps in Algorithm~\ref{app:fgee_algo} in Appendix~\ref{app:oneStep}.}

\paragraph{Working Correlations}
Although fGEE yields valid inference regardless of the $\mathbb{V}_i$ structure used (under correct mean model specification), the statistical and computational efficiency of fGEE depends heavily on the $\mathbb{V}_i$ form adopted. 
Our one-step estimator is often far faster than the fully-iterated fGEE in \cite{gee_functional} but, if $n_i L$ is large, it still may not scale. Thus, while our theory applies to a one-step with general $\mathbb{V}_i$, we focus on
forms that can be inverted quickly {and improve
statistical efficiency. To see this, we first rewrite $\mathbb{V}_i = \mathbf{A}_i^{1/2} \mathbf{R}_i\mathbf{A}_i^{1/2}$, where
$\mathbf{A}_i = \text{diag}\left(v_{i,1}(s_1), \ldots, v_{i,n_i}(s_1), \ldots v_{i,n_i}(s_L)\right )\in \mathbb{R}^{n_iL \times n_iL}$, $v_{i,j}(s)$ models $\text{Var}(Y_{i,j}(s) \mid \mathbf{X}_{i,j})$, and the elements of $\mathbf{R}_{i} \in \mathbb{R}^{n_iL \times n_iL}$ model $\text{Cor}(Y_{i,j}(s), Y_{i,j'}(s')\mid \mathbf{X}_{i,j})$ for $j,j' \in [n_i], s,s' \in \mathcal{S}$.
%

{For scalability, we focus on two classes of working covariance structures: block diagonal and Kronecker product. Block diagonal structures model correlation in \textit{either} the longitudinal direction $\mathbb{V}_i = \text{bdiag}[\mathbb{V}_i(s_1),\ldots,\mathbb{V}_i(s_L)]$, where $\mathbb{V}_i(s)$ models $\text{Cov}(\boldsymbol{Y}_i(s) \mid \mathbf{X}_i)\in \mathbb{R}^{n_i \times n_i}$, or the functional direction $\mathbb{V}_i = \text{bdiag}[\mathbb{V}_{i,1},\ldots,\mathbb{V}_{i,n_i}]$, where $\mathbb{V}_{i,j}$ models $\text{Cov}(\boldsymbol{Y}_{i,j} \mid \mathbf{X}_i)\in \mathbb{R}^{L \times L}$. Each $\mathbb{V}_i(s_l)$ or $\mathbb{V}_{i,j}$ can be parameterized by a unique $\rho_i(s)$ or $\rho_{i,j}$, respectively. If functional outcomes are observed on the same grid for each longitudinal measurement of a cluster, a Kronecker product-based $\mathbb{V}_i$ can model correlation in \textit{both} directions: $\mathbb{V}_i =\mathbf{A}_i^{1/2} (\mathbf{R}_{i,Lon} \otimes \mathbf{R}_{i,Fun}) \mathbf{A}_i^{1/2}$, where $\mathbf{R}_{i,Lon}$ models $\text{Cor}(\boldsymbol{Y}_i(s) \mid \mathbf{X}_i) \in \mathbb{R}^{n_i \times n_i}$ for all $s \in \mathcal{S}$, and $\mathbf{R}_{i,Fun}$ models $\text{Cor}(\boldsymbol{Y}_{i,j} \mid \mathbf{X}_i)\in \mathbb{R}^{L \times L}$ for all $j \in [n_i]$; we suppress cluster indices $i$ in $\mathbf{R}_{Lon}$ and $\mathbf{R}_{Fun}$, when possible, for ease of notation. The $\mathbf{R}_{Lon}$ and $\mathbf{R}_{Fun}$ are parameterized by $\rho_{Lon}$ and $\rho_{Fun}$ correlation parameters.

These block and Kronecker structures are computationally efficient to invert. First, they avoid inversion of the full $\mathbb{V}_i$ by reducing it to inversions of smaller blocks (block diagonal) or smaller Kronecker factors (Kronecker). For example, if we adopt $\mathbb{V}_i = \text{bdiag}[\mathbb{V}_i(s_1),\ldots,\mathbb{V}_i(s_L)]$, then $\mathbb{V}_i^{-1} = \text{bdiag}[\mathbb{V}^{-1}_i(s_1),\ldots,\mathbb{V}^{-1}_i(s_L)]$. If we adopt $\mathbb{V}_i = \mathbf{A}_i^{1/2} ({\mathbf{R}}_{Lon} \otimes \mathbf{R}_{Fun}) \mathbf{A}_i^{1/2}$, then $\mathbb{V}^{-1}_i = \mathbf{A}_i^{-1/2} ({\mathbf{R}}_{Lon}^{-1} \otimes \mathbf{R}_{Fun}^{-1}) \mathbf{A}_i^{-1/2}$. This can be made more computationally efficient by using Kronecker product identities. For example, denoting
$(\widehat{\mathbf{R}}^{-1}_{Lon} \otimes \mathbf{R}^{-1}_{Fun}) \widehat{\boldsymbol{r}}_i \coloneqq (\widehat{\mathbf{R}}^{-1}_{Lon} \otimes \mathbf{R}^{-1}_{Fun}) \widehat{\mathbf{A}}_i^{-1/2}\left \{ \boldsymbol{Y}_i - \widehat{\boldsymbol{\mu}}_i \right \}$, we calculate $\widehat{\mathbb{V}}_{i}^{-1} \widehat{\boldsymbol{r}}_i$ in \eqref{emp_IF}, 
with the identity $(\widehat{\mathbf{R}}_{Lon}^{-1} \otimes \widehat{\mathbf{R}}_{Fun}^{-1}) \widehat{\boldsymbol{r}}_i = \text{vec}(\widehat{\mathbf{R}}^{-1}_{Lon} \text{mat}(\widehat{\boldsymbol{r}}_i) \widehat{\mathbf{R}}^{-1}_{Fun} )$, where $\text{mat}(\widehat{\boldsymbol{r}}_i) \in \mathbb{R}^{n_i \times L}$ denotes forming $\widehat{\boldsymbol{r}}_i$ into a matrix. Writing this as $\text{vec}(\{ \widehat{\mathbf{R}}_{Fun}^{-1} \{ \widehat{\mathbf{R}}^{-1}_{Lon} \text{mat}(\widehat{\boldsymbol{r}}_i) \}^T \}^T)$ emphasizes how one avoids forming the $n_iL \times n_i L$ matrix $\widehat{\mathbf{R}}_i^{-1}$ by calculating $\tilde{\boldsymbol{r}}_i=\mathbf{R}^{-1}_{Lon} \text{mat}(\widehat{\boldsymbol{r}}_i)$ and then $\widehat{\mathbf{R}}_{Fun}^{-1}\tilde{\boldsymbol{r}}_i$. 
Thus, if $\widehat{\mathbf{R}}_{Fun}$ and $\widehat{\mathbf{R}}_{Lon}$ can be inverted efficiently, $\mathbb{V}_i$ can be inverted efficiently for large $L$ and $n_i$. In fact, if one adopts a working correlation such that the sub-matrices (e.g., ${\mathbf{R}}_{Fun}$, ${\mathbf{R}}_{Lon}$, ${\mathbf{R}}_i(s)$) have Toeplitz structure (e.g., exchangeable, AR1 with $\rho \geq 0$, moving average MA(q), Matern/RBF kernel Gaussian Process covariances), then one can calculate, for example, 
$\widehat{\mathbb{V}}_i^{-1}\widehat{\boldsymbol{r}}_i$ without having to ever construct (or store in memory) the sub-matrices $\widehat{\mathbf{R}}_{Fun}$, ${\mathbf{R}}_{Lon}$, $\mathbf{R}_i(s)$ (or their inverses); the corresponding systems of linear equations can be efficiently solved for large $L$ and $n_i$ with, for example, the generalized Schur algorithm \citep{gschur}
(discussed in Appendices~\ref{app:corr_inv}-\ref{app:corr_param}). Similarly, our implementation also allows for FPCA-based estimation of ${\mathbf{R}}_{Fun}$ which can be inverted efficiently (see Appendix~\ref{sec:fpca}). }} 
\paragraph{Tuning $\Lambda$}
To calculate an initial estimate of $\widehat{\boldsymbol{\theta}}^{(0)}_{\Lambda_0}$, we select the smoothing parameters, denoted as $\Lambda_0$, with fast restricted maximum likelihood \citep{wood2011fast}. 
We found, however, that calculating the one-step estimate with the same $\Lambda_0$ values (i.e. $\boldsymbol{\theta}^{(1)}_{\Lambda_0}$) tends to produce inaccurate coefficient estimates {as the inclusion of $\mathbb{V}_i$ in the estimation equation changes the relative smoothing parameter magnitudes}.
Therefore, we tune the smoothing parameters for the one-step, denoted as $\Lambda_1$, based on the cross-validated prediction performance of the one-step estimator. We develop a fast {cluster} CV for large datasets, by exploiting the fact that the form of the one-step estimator enables calculation of fold-specific estimates with pre-computable terms (see Appendix~\ref{app:cv} for further details {and a summary of the steps as an algorithm}):
rewriting the update as
    $\widehat{\boldsymbol{\theta}}_{ \Lambda_1}^{(1)}~=~\widehat{\boldsymbol{\theta}}_{ \Lambda_0}^{(0)}~+~\left \{  \frac{1}{N} \sum_{i=1}^N  \mathbb{W}_i (\widehat{\boldsymbol{\theta}}^{(0)}_{\Lambda_0}) + \Lambda_1 \mathbb{S}  \right \}^{-1} \times \frac{1}{N} \sum_{i=1}^N \left \{ \mathbf{b}_i(\widehat{\boldsymbol{\theta}}^{(0)}_{\Lambda_0})  - \Lambda_1 \mathbb{S} \widehat{\boldsymbol{\theta}}^{(0)}_{\Lambda_0} \right \}$, 
illustrates that we can pre-compute each cluster's 
$\mathbb{W}_i(\widehat{\boldsymbol{\theta}}^{(0)}_{\Lambda_0}) = \{\mathbb{D}_i(\widehat{\boldsymbol{\theta}}^{(0)}_{\Lambda_0})\}^T \{\widehat{\mathbb{V}}_i(\widehat{\boldsymbol{\theta}}_{ \Lambda_0}^{(0)})\}^{-1} \{\mathbb{D}_i(\widehat{\boldsymbol{\theta}}^{(0)}_{\Lambda_0})\} \in \mathbb{R}^{p \times p}$, and $\mathbf{b}_i(\widehat{\boldsymbol{\theta}}^{(0)}_{\Lambda_0}) = \{\mathbb{D}_i(\widehat{\boldsymbol{\theta}}^{(0)}_{\Lambda_0})\} ^T \{\widehat{\mathbb{V}}_i(\widehat{\boldsymbol{\theta}}_{ \Lambda_0}^{(0)})\}^{-1} \left \{ \boldsymbol{Y}_i - \widehat{\boldsymbol{\mu}}_i(\widehat{\boldsymbol{\theta}}^{(0)}_{\Lambda_0}) \right \} \in \mathbb{R}^{p}$. Moreover, we only need to estimate $\widehat{\boldsymbol{\theta}}^{(0)}_{\Lambda_0}$ once. We can then use that $\widehat{\boldsymbol{\theta}}^{(0)}_{\Lambda_0}$, calculated on the full sample, as the initial estimate for all folds and $\Lambda_1$ values. This is because any consistent initial estimate, $\widehat{\boldsymbol{\theta}}^{(0)}_{\Lambda_0}$, is sufficient to ensure that the one-step estimator of a given fold is consistent for the population $\boldsymbol{\theta}$. 
Finally, assuming $\frac{1}{N} \sum_{i=1}^N  \mathbb{W}_i (\widehat{\boldsymbol{\theta}}^{(0)}_{\Lambda_0}) + \Lambda_1 \mathbb{S} \overset{\mathbb{P}}{\to}  \mathbb{E} \left \{ \nabla_{\boldsymbol{\theta}} \boldsymbol{U}_{\Lambda_1}(\mathbf{X}_i, \boldsymbol{Y}_i; {\boldsymbol{\theta}}_{\Lambda_0})  \right \}$,
we can (heuristically, by Slutsky's theorem) calculate consistent one-step estimates in fold $k$ as
\begin{align}
\widehat{\boldsymbol{\theta}}^{k}_{{\Lambda_1}} = \widehat{\boldsymbol{\theta}}^{(0)}_{{\Lambda_0}} + \left \{  \frac{1}{N} \sum_{i=1}^N  \mathbb{W}_i (\widehat{\boldsymbol{\theta}}^{(0)}_{\Lambda_0}) + \Lambda_1 \mathbb{S} \right \}^{-1} \frac{1}{N} \sum_{i \not\in \mathcal{K}_k} \left \{ \tilde{n}_k \mathbf{b}_i(\widehat{\boldsymbol{\theta}}^{(0)}_{\Lambda_0})  - \Lambda_1\mathbb{S} \widehat{\boldsymbol{\theta}}^{(0)}_{\Lambda_0} \right \},
\end{align}
where $\tilde{n}_k = \frac{\sum_{i=1}^N n_i}{\sum_{i \not\in \mathcal{K}_k}n_i}$ and {$\mathcal{K}_k \subset [N]$} denotes the set of cluster index sets for fold $k$. By using the full sample estimate $\left \{  \frac{1}{N} \sum_{i=1}^N  \mathbb{W}_i (\widehat{\boldsymbol{\theta}}^{(0)}_{\Lambda_0}) + \Lambda_1 \mathbb{S}  \right \}^{-1}$, we only need to invert this $p \times p$ matrix once for each value of $\Lambda_1$, instead of inverting a fold-specific $p \times p$ matrix for each unique $\{k, \Lambda_1\}$ pair. The strategy of keeping $\widehat{\boldsymbol{\theta}}^{(0)}_{{\Lambda_0}}$ and $\left \{  \frac{1}{N} \sum_{i=1}^N  \mathbb{W}_i (\widehat{\boldsymbol{\theta}}^{(0)}_{\Lambda_0}) + \Lambda_1 \mathbb{S}  \right \}^{-1}$ fixed across folds is motivated by an analogous strategy for cluster bootstrapping of unpenalized one-step GEE (see Remark and Theorem 3.3 in \cite{bootstrap_consistency}). Specifically, \cite{bootstrap_consistency} showed that a cluster bootstrap that fixes these two quantities (at the full-sample estimates) across replicates enjoys the same theoretical guarantees asymptotically as an approach that re-estimates these quantities in each replicate-specific sample. In our simulations, our adaptation of this strategy for cluster CV was often dramatically faster than, and performed nearly identically to, a CV strategy that calculates $\widehat{\boldsymbol{\theta}}^{k}_{{\Lambda_1}}$ using the fold-specific estimate 
$\left \{  \frac{1}{N-|\mathcal{K}_k|} \sum_{i \not\in \mathcal{K}_k}  \mathbb{W}_i (\widehat{\boldsymbol{\theta}}^{(0)}_{\Lambda_0}) + \left(\frac{\sum_{i \not\in \mathcal{K}_k}n_i}{\sum_{i=1}^N n_i} \right) \Lambda_1 \mathbb{S} \right \}^{-1}$. 

\paragraph{Uncertainty Quantification} \label{sec:sandwich}
The one-step variance can be estimated with sandwich estimator $\widehat{\boldsymbol{V}}_{\Lambda}(\boldsymbol{\theta})= \frac{1}{N}\{\widehat{\boldsymbol{H}}_\Lambda(\boldsymbol{\theta})\}^{-1} \widehat{\boldsymbol{M}}_\Lambda (\boldsymbol{\theta})\{\widehat{\boldsymbol{H}}_\Lambda(\boldsymbol{\theta})\}^{-1}$,
where $\widehat{\boldsymbol{M}}_\Lambda(\boldsymbol{\theta}) = \frac{1}{N}\sum_{i = 1}^N\{ \boldsymbol{U}_\Lambda(\mathbf{X}_i, \boldsymbol{Y}_i; \boldsymbol{\theta}) \boldsymbol{U}_\Lambda(\mathbf{X}_i, \boldsymbol{Y}_i; \boldsymbol{\theta}) ^T \}$, 
and $\widehat{\boldsymbol{H}}_\Lambda(\boldsymbol{\theta}) = \frac{1}{N}\sum_{i = 1}^N\{\mathbb{D}_i(\boldsymbol{\theta})^T \widehat{\mathbb{V}}_i(\boldsymbol{\theta})^{-1} \mathbb{D}_i(\boldsymbol{\theta})\} + \Lambda \mathbb{S}$.
We estimate $\widehat{\text{Var}}(\widehat{\boldsymbol{\theta}}_{\Lambda_1}^{(1)})$ via $\widehat{\boldsymbol{V}}_{\Lambda_1}(\widehat{\boldsymbol{\theta}}_{\Lambda_1}^{(1)})$, i.e., by plugging in $\widehat{\boldsymbol{\theta}}^{(1)}_{\Lambda_1}$, $\widehat{\mathbb{V}}_i(\widehat{\boldsymbol{\theta}}^{(1)}_{\Lambda_1})$, $\widehat{\mathbf{A}}_i(\widehat{\boldsymbol{\theta}}^{(1)}_{\Lambda_1})$, and $\Lambda_1$. 
For fixed basis $\mathbf{B}$,
$\widehat{\text{Var}}(\widehat{\boldsymbol{\beta}}) =  \text{bdiag}(\widehat{\Sigma}^{(\beta)}_1, \ldots, \widehat{\Sigma}^{(\beta)}_q ) = \mathbb{B} {\widehat{\text{Var}}} \left( \widehat{\boldsymbol{\theta}}_{\Lambda_1}^{(1)} \right) \mathbb{B} ^T$, where $\mathbb{B} = (I_{q+1} \otimes \mathbf{B})$ is a block diagonal matrix, with $\mathbf{B}$ in each block. An asymptotically valid ($1-\alpha)$-level \textit{pointwise} CI 
is given by $\widehat{\beta}_r(s) \pm z_{1-\alpha/2} ~\hat{\sigma}_r^{(\beta)}(s)$, where $\hat{\sigma}_r^{(\beta)}(s)=\sqrt{\widehat{\Sigma}_{r}^{(\beta)}(s)}$, and $\widehat{\Sigma}_{r}^{(\beta)}(s) \in \mathbb{R}$ is diagonal entry $s$ of $\widehat{\Sigma}^{(\beta)}_r$. {We found, however, that smoothing and small-sample bias sometimes led to undercoverage when constructing CIs with Gaussian quantiles, $z_{1-\alpha/2}$. We therefore constructed pointwise and joint CIs as $\widehat{\beta}_r(s) \pm {c}^{pt}_{r} ~\hat{\sigma}_r^{(\beta)}(s)$ and $\widehat{\beta}_r(s) \pm {c}^{jt}_{r} ~\hat{\sigma}_r^{(\beta)}(s)$, respectively, where ${c}^{pt}_{r}= (t_{1-\alpha/2, \text{df}_r} / z_{1-\alpha/2})\tilde{c}^{pt}_{r}$ and ${c}^{jt}_{r}= (t_{1-\alpha/2, \text{df}_r} / z_{1-\alpha/2})\tilde{c}^{jt}_{r}$, $\tilde{c}^{pt}_{r}$ and $\tilde{c}^{jt}_{r}$ are quantiles obtained from a fast score-based wild cluster bootstrap, and $(t_{1-\alpha/2, \tilde{\text{df}_r}} / z_{1-\alpha/2})$ is a small-$N$ effective degree of freedom inflation factor (see Appendix~\ref{app:wild} 
for the algorithm).} 

\section{Theory} \label{sec:theory}
For a fixed sequence $\Lambda_N = \mathrm{diag}(\lambda_{N,1}, \ldots, \lambda_{N,p})$ with $\frac{1}{N} \max_{1 \leq j \leq p} \lambda_{N,j} \to 0$, let $\boldsymbol{\theta}_N$ be the solution to the population estimating equation $\mathbb{E}\{ \boldsymbol{U}_N(\mathbf{X}_i, \boldsymbol{Y}_i; \boldsymbol{\theta}) \} = \boldsymbol{0}$, where $\boldsymbol{U}_N(\mathbf{X}_i, \boldsymbol{Y}_i; \boldsymbol{\theta}) = \mathbb{D}_i^T(\boldsymbol{\theta}) \mathbb{V}_i^{-1}(\boldsymbol{Y}_i - \boldsymbol{\mu}_i(\boldsymbol{\theta})) - \frac{1}{N}\Lambda_N\mathbb{S}\boldsymbol{\theta}$ has components denoted $U_{N,j}(\mathbf{X}_i, \boldsymbol{Y}_i; \boldsymbol{\theta})$ for $j \in [p]$, and let $\boldsymbol{\beta}_{N} = \mathbb{B} \boldsymbol{\theta}_N$, where $\mathbb{B} = (I_{q+1} \otimes \mathbf{B})$. We define 
$\boldsymbol{H}_N(\boldsymbol{\theta}) =  \mathbb{E}\{\mathbb{D}_i^T(\boldsymbol{\theta}) \mathbb{V}_i^{-1} \mathbb{D}_i(\boldsymbol{\theta}) + \frac{1}{N}\Lambda_N \mathbb{S}\}$, and 
$\boldsymbol{M}_N(\boldsymbol{\theta}) = \mathbb{E} \{ \boldsymbol{U}_N(\mathbf{X}_i, \boldsymbol{Y}_i; \boldsymbol{\theta}) \boldsymbol{U}_N(\mathbf{X}_i, \boldsymbol{Y}_i; \boldsymbol{\theta})^T \}$. Note that implicit to these definitions, and to the ensuing theory, is that we treat the working covariance matrices $\mathbb{V}_i$ as fixed or computed with the ``true'' limiting correlation parameters $\rho(s)$ throughout. In practice, we can replace these with their estimated counterparts $\widehat{\mathbb{V}}_i$ under the assumption that $\widehat{\rho}(s) = \rho(s) + o_{\mathbb{P}}({N}^{-1/2})$ for all $s$, for some limiting parameters $\rho(s)$---see Corollary 1 of~\citet{gee_functional}.
\begin{theorem}\label{thm:asymp}
     Suppose that one-step estimators, $\widehat{\boldsymbol{\theta}}^{(1)}_{\Lambda_N}$ and $\widehat{\boldsymbol{\beta}}^{(1)}_{\Lambda_N}$, are constructed using initial estimate, $\widehat{\boldsymbol{\theta}}^{(0)}_{N}$, and that conditions (i)-(vi) in Appendix~\ref{app:theory_conditions} hold:
    then the one-step estimator satisfies 
    $\boldsymbol{W}_N(\widehat{\boldsymbol{\theta}}^{(1)}_{\Lambda_N} - \boldsymbol{\theta}_N) \overset{d}{\to} \mathcal{N}(\boldsymbol{0}, {I}_p)$,
    where $\boldsymbol{W}_N \coloneqq \sqrt{N} \left\{\boldsymbol{M}_N(\boldsymbol{\theta}_N)\right\}^{-1/2} \boldsymbol{H}_N(\boldsymbol{\theta}_N)$. Consequently, $\widehat{\boldsymbol{\beta}}^{(1)}_{\Lambda_N}$ is asymptotically normal, centered around the sequence $\boldsymbol{\beta}_N$, with large sample variance given by
    $\boldsymbol{V}_N = \frac{1}{N}\mathbb{B}\{\boldsymbol{H}_N(\boldsymbol{\theta}_N)\}^{-1} \boldsymbol{M}_N(\boldsymbol{\theta}_N)  \{\boldsymbol{H}_N^T(\boldsymbol{\theta}_N)\}^{-1} \mathbb{B}^T$.
\end{theorem} 

We develop a more general result for adaptive one-step M-estimation that may be of independent interest in Appendix~\ref{app:theory}.
{Theorem~\ref{thm:asymp} provides a basis for asymptotically valid inference, akin to seminal results in non-functional settings \citep{gee_theory}. The result assumes that the true mean function lies in the span of the basis functions, and holds for arbitrary $n_i$ and fixed $L$, $q$ and $p$. Technical conditions (i)-(vi) are standard and hold if typical link functions are used, and the sample, true, and limiting $\boldsymbol{M}_N$ and $\boldsymbol{H}_N$ matrices are full rank, and if the initial estimator converges with sufficient rates. These are discussed further in Appendix~\ref{app:theory_conditions}.  }
\vspace{-7mm}
%
%
\begin{remark}
    Our result shows that the one-step is asymptotically equivalent to the fully-iterated fGEE. Moreover, our result extends to non-linear link functions, while existing fGEE theory is restricted to the linear case \citep{gee_functional}. In the linear case, the one-step shares the same properties as those characterized in \cite{gee_functional}, such as the convergence rates in small knot and large knot regimes in terms of the smoothing parameter, $\Lambda$. We discuss the convergence rates of the coefficient estimates in terms of the parameters $\Lambda$ further in Appendix~\ref{app:theory}.
\end{remark}



\section{Simulations} \label{sec:simulations}
We conducted simulations to assess 95\% CI coverage {and width}, coefficient estimate {error}, and algorithm timing. {We ran $T=300$ simulation replicates. We used $\texttt{pffr}$, a maximum likelihood estimator, for an initial $\widehat{\boldsymbol{\theta}}^{(0)}_{\Lambda_0}$ and benchmark, using penalized P-splines with 15 knots per functional coefficient (see Appendix~\ref{app:pffr} for details). We set $s \in \mathcal{S} \subset [0,1]$ with
$L =50$ evenly-spaced points,  $N\in \{25, 50, 500\}$, and $n_i \in \{5, 100\}$.} 
Denoting $\widehat{\beta}_{r}^t(s)$ as {the} functional coefficient $r$ for simulation replicate $t$ at $s$,
we report 
$\text{RMSE} = \frac{1}{T}\sum_{t=1}^T \left\{\frac{1}{(q+1) | \mathcal{S}|} \sum_{r=0}^{q} \sum_{s \in \mathcal{S}}  \left(\beta_{r}(s) - \widehat{\beta}_{r}^t(s) \right)^2 \right \}^{1/2}$. Denoting $\text{pCI}_r^t(s)$ as the pointwise CI for replicate $t$ for functional coefficient $r$ at point $s$, we report average empirical pointwise coverage as: $\frac{1}{T(q+1) | \mathcal{S}|} \sum_{t=1}^T \sum_{r=0}^{q} \sum_{s \in \mathcal{S}}  \mathds{1}  \left\{\beta_{r}(s) \in \text{pCI}_r^t(s)   \right \}$. We report empirical joint {(uniform)} coverage as: $\frac{1}{T(q+1)} \sum_{t=1}^T \sum_{r=0}^{q}  \mathds{1}  \left\{\beta_{r}(s) \in \text{jCI}_r^t(s)   ~\forall s \in \mathcal{S} \right \}$, where $\text{jCI}_r^t(s)$ denotes the joint CI {at $s$}. 

{We simulated outcomes with mean model $\text{g} \left \{ \mathbb{E}(Y_{i,j}(s) \mid \mathbf{X}_{i,j}) \right \} = \beta_0(s) + X_{1,i}\beta_1(s) + X_{2,i,j}\beta_2(s)$,
where $\beta_0(s) = 1 + \frac{1}{3}\text{sin}(\pi s) + \frac{\sqrt{2}}{3} \text{cos}(3\pi s)$, $\beta_1(s) = 1 + \frac{1}{3}\text{cos}(2\pi s) + \frac{\sqrt{2}}{3} \text{cos}(3\pi s)$, $\beta_2(s) = \frac{5}{3} \phi(\frac{s-0.35}{0.1}) -\frac{5}{3}\phi(\frac{s-0.65}{0.2})$, and $\phi(\cdot)$ denotes the standard normal density function. We drew Gaussian (identity), binary (logit), Poisson (log) and gamma (log) outcomes (link functions, $g$) with a copula method \citep{nelsen2006introduction}. 
Briefly, we drew $\boldsymbol{Z}_i \sim N_{n_iL}( \boldsymbol{0}, \mathbb{V}_i^*)$, transformed it into a vector of uniform random variables with the probability integral transform $\boldsymbol{Q}_i \sim \Phi(\boldsymbol{Z}_i)$, and set $\boldsymbol{Y}_i = F^{-1}(\boldsymbol{Q}_i)$ where we applied the inverse CDF $F^{-1}(\cdot)$ of the target distribution, and the standard normal CDF $\Phi(\cdot)$, element-wise. We used $\mathbb{V}_i^* = (\mathbf{A}_i^*)^{1/2} \mathbf{R}^*_i(\mathbf{A}_i^*)^{1/2}$, with $v_{i,j}(s)=10$ for all $i,j,s$ for Gaussian data and $v_{i,j}(s)=1$ for other families. We set $n_i=n$ for all $i \in [N]$.
We drew $X_{1,i} \sim N(0,1)$, and $X_{2,i,j} = j + e_{i,j}$, where $e_{i,j} \sim N(0.7 e_{i, j-1}, 1)$, and $e_{i,0}=0$. We plot the true functional coefficients and average fGEE estimates in simulations in Appendix~\ref{sec:sims_app}.

We compared performance of the one-step and a fully-iterated fGEE in a setting where the outcome was simulated to be correlated in both longitudinal and functional directions. Namely, we set $\mathbf{R}_i^* \coloneqq\text{Cor}(\boldsymbol{Y}_i \mid \mathbf{X}_i) = \mathbf{R}^*_{Lon} \otimes \mathbf{R}^*_{Fun}$, where the $j,j'$ entry of $\mathbf{R}^*_{Lon} \in \mathbb{R}^{n_i \times n_i}$ contains $\text{Cor}({Y}_{i,j}(s), {Y}_{i,j'}(s) \mid \mathbf{X}_i)$ for $s \in \mathcal{S}, i \in [N]$, and the $s,s'$ entry of $\mathbf{R}^*_{Fun} \in \mathbb{R}^{L \times L}$ contains $\text{Cor}({Y}_{i,j}(s), {Y}_{i,j}(s') \mid \mathbf{X}_i)$ for $i \in [N], j \in [n_i]$. We set $\mathbf{R}^*_{Fun}$ to have an AR1 structure $\text{Cor}(Y_{i,j}(s), Y_{i,j}(s') \mid \mathbf{X}_{i,j}) = (\rho_{Fun}^*)^{|s-s'|}$. We simulated $\mathbf{R}^*_{Lon}$ with both an AR1 $\text{Cor}(Y_{i,j}(s), Y_{i,j'}(s) \mid \mathbf{X}_{i,j}) = (\rho_{Lon}^*)^{|j-j'|}$, and exchangeable structure $\text{Cor}(Y_{i,j}(s), Y_{i,j'}(s) \mid \mathbf{X}_{i,j}) = \rho_{Lon}^*$. We set $\rho^*_{Fun} = \rho^*_{Lon} = 0.75$. 
We then measured how the fGEE performs under varying degrees of working correlation misspecification: working correlations that are 1) correctly specified; specified correctly in the 2) functional or 3) longitudinal direction; or 4) a working independence structure: $\mathbf{R}_i \in \{\mathbf{R}_{Lon} \otimes \mathbf{R}_{Fun}, \mathbf{R}_{Lon} \otimes {I}_L, {I}_{n_i} \otimes \mathbf{R}_{Fun}, {I}_{n_i} \otimes {I}_L \}$. 

We compare performance of two fully-iterated fGEEs: ``Full-1Step'' which selects $\Lambda$ with our fast cluster CV, and ``Full-Full'' which instead selects $\Lambda$ with a cluster CV that fits each fold-specific coefficient estimate, $\widehat{\boldsymbol{\theta}}^{k}_{\Lambda}$, with a fully-iterated fGEE. 
The Full-Full approach does not use one-step methods and thus shows how slow fGEE is without our one-step framework. In fact, Full-Full simulations finished only with small sample sizes even on a high performance cluster (see Appendix~\ref{app:full_sims} for Full-Full results). We describe implementation details in Appendix~\ref{app:full_itr}. 

Table~\ref{tab:rmse_rel} shows that, relative to a $\texttt{pffr}$ fit, the one-step with correctly specified working correlation ($\text{1-Step } {\mathbf{R}}_{\text{Lon}} \otimes {\mathbf{R}}_{\text{Fun}}$) reduces functional coefficient estimation RMSE by as much as 35\% for certain combinations of $N,n_i$ and outcome distributions. Interestingly, 
an fGEE that adopts a working correlation that is correctly specified in just the longitudinal direction 
(e.g., $\text{1-Step } {\mathbf{R}}_{\text{Lon}} \otimes {I}_L$) often performs comparably to an fGEE that correctly specifies correlation in both directions (e.g., $\text{1-Step } {\mathbf{R}}_{\text{Lon}} \otimes {\mathbf{R}}_{\text{Fun}}$). 
Table~\ref{tab:rmse_rel} also shows how the one-step performs nearly identically to a fully-iterated fGEE that uses the same smoothing parameter values: Full-1Step and one-step coefficients with a given working correlation are estimated with the same $\Lambda$ value because they apply the same cluster CV procedure. 
Table~\ref{tab:rmse_rel} also shows how the fast cluster CV reduces estimation error compared to $\texttt{pffr}$. The $\text{1-Step } {I}_{n_i} \otimes {I}_L$ and $\texttt{pffr}$ are fit with a full working independence correlation structure, but the one-step is fit with $\Lambda$ values selected with fast cluster CV, whereas $\texttt{pffr}$ selects $\Lambda$ with a fast REML strategy. 
Interestingly, when $N$ is small, the fully-iterated fGEE sometimes performs far worse than the one-step, potentially because re-estimating correlation parameters at each iteration causes instability. 

Table~\ref{tab:jci} shows that the one-step achieves reasonable joint (uniform) CI coverage in large samples, but sometimes exhibits over/under-coverage. Under-coverage was more common in small $N$ settings. Table~\ref{tab:pci} shows that the one-step achieves pointwise CI coverage at roughly nominal levels for all settings. Importantly, these results show how the ``fast/simple'' approach of constructing CIs with a sandwich estimator based on a $\texttt{pffr}$ fit can yield coverage far below the nominal levels. Thus, we also show CI coverage from $\texttt{pffr}$ CIs constructed with our wild cluster bootstrap approach, which achieve the nominal coverage but can be conservative. In fact, compared to these $\texttt{pffr}$ (wild) CIs, the one-step substantially reduces pointwise CI width (Table~\ref{tab:pt_width}) in most simulation settings. This shows how the fGEE can substantially improve statistical power. Finally,  Table~\ref{tab:time} shows that the one-step is relatively fast for large datasets. 

In Appendix~\ref{sec:sims_app} we show results for an fGEE that uses FPCA to estimate ${\mathbf{R}}_{\text{Fun}}$, and for simulations with gamma-distributed data. In Appendix~\ref{app:Li} we show results from additional simulations to compare our method with the approach of \cite{li_2022} for continuous outcomes. The outcomes are generated to be correlated in both longitudinal and functional directions with a simulation scheme that differs from the Kronecker product approach above. These results show that our one-step achieves comparable coefficient estimation RMSE, yields far better pointwise coverage with large cluster sizes, and scales to bigger datasets. 
}

\begin{table}[!h]
\centering
\caption{\label{tab:rmse_rel} \footnotesize Functional Coefficient Estimation Performance (RMSE) of each method relative to the $\texttt{pffr}$ fit ($\text{RMSE}/\text{RMSE}_{\text{pffr}}$).
    Cells contain the average of 300 replicates $\pm$ SE (SE$=0.00$ indicates a value $<0.01$). Outcomes were simulated with an $\mathbf{R}^* = \mathbf{R}^*_{\text{Lon}} \otimes \mathbf{R}^*_{\text{Fun}}$, where $\mathbf{R}^*_{\text{Fun}}$ had an AR1 structure and the table columns indicate results where $\mathbf{R}^*_{\text{Lon}}$ had exchangeable or AR1 correlation. The ``1-Step'' indicates a one-step was used for tuning and final coefficient estimation, and ``Full-1Step'' indicates one-step tuning and a fully-iterated fGEE for final coefficient estimation, with the indicated working correlation. $*$ indicates that values $\geq100$ were removed from that cell to avoid skewing the mean.}
\centering
\resizebox{\ifdim\width>\linewidth\linewidth\else\width\fi}{!}{
\fontsize{9}{11}\selectfont
\begin{tabular}[t]{>{\raggedright\arraybackslash}p{4.8cm}>{\raggedright\arraybackslash}p{2.1cm}>{\raggedright\arraybackslash}p{2.1cm}>{\raggedright\arraybackslash}p{2.1cm}>{\raggedright\arraybackslash}p{2.1cm}>{\raggedright\arraybackslash}p{2.1cm}>{\raggedright\arraybackslash}p{2.1cm}}
\toprule
\multicolumn{1}{c}{ } & \multicolumn{3}{c}{Exchangeable} & \multicolumn{3}{c}{AR(1)} \\
\cmidrule(l{3pt}r{3pt}){2-4} \cmidrule(l{3pt}r{3pt}){5-7}
Method & Gaussian & Poisson & Binomial & Gaussian & Poisson & Binomial\\
\midrule
\addlinespace[0em]
\multicolumn{7}{l}{\textbf{$N = 25,\; n_i = 5$}}\\
\hspace{1em}$\text{1-Step } {\mathbf{R}}_{\text{Lon}} \otimes {\mathbf{R}}_{\text{Fun}}$ & 0.82 $\pm$ 0.01 & 0.83 $\pm$ 0.01 & 0.88 $\pm$ 0.01 & 0.85 $\pm$ 0.01 & 0.85 $\pm$ 0.01 & 0.92 $\pm$ 0.01\\
\hspace{1em}$\text{1-Step } {\mathbf{R}}_{\text{Lon}} \otimes I_L$ & 0.82 $\pm$ 0.01 & 0.85 $\pm$ 0.01 & 0.88 $\pm$ 0.01 & 0.85 $\pm$ 0.01 & 0.83 $\pm$ 0.01 & 0.92 $\pm$ 0.01\\
\hspace{1em}$\text{1-Step } I_{n_i} \otimes {\mathbf{R}}_{\text{Fun}}$ & 0.86 $\pm$ 0.01 & 0.91 $\pm$ 0.01 & 0.93 $\pm$ 0.01 & 0.90 $\pm$ 0.01 & 0.92 $\pm$ 0.01 & 0.95 $\pm$ 0.01\\
\hspace{1em}$\text{1-Step } I_{n_i} \otimes I_L$ & 0.87 $\pm$ 0.01 & 0.88 $\pm$ 0.00 & 0.93 $\pm$ 0.01 & 0.90 $\pm$ 0.01 & 0.90 $\pm$ 0.00 & 0.95 $\pm$ 0.01\\
\hspace{1em}$\text{Full-1Step } {\mathbf{R}}_{\text{Lon}} \otimes {\mathbf{R}}_{\text{Fun}}$ & 0.82 $\pm$ 0.01 & 0.83 $\pm$ 0.01 & 1.00 $\pm$ 0.09$^*$ & 0.86 $\pm$ 0.01 & 0.85 $\pm$ 0.01 & 0.94 $\pm$ 0.01\\
\hspace{1em}$\text{Full-1Step } {\mathbf{R}}_{\text{Lon}} \otimes I_L$ & 0.81 $\pm$ 0.01 & 1.02 $\pm$ 0.09$^*$ & 1.29 $\pm$ 0.21$^*$ & 0.85 $\pm$ 0.01 & 0.83 $\pm$ 0.01 & 0.93 $\pm$ 0.01\\
\hspace{1em}$\text{Full-1Step } I_{n_i} \otimes {\mathbf{R}}_{\text{Fun}}$ & 0.87 $\pm$ 0.01 & 0.91 $\pm$ 0.01 & 0.96 $\pm$ 0.02 & 0.90 $\pm$ 0.01 & 0.92 $\pm$ 0.01 & 0.97 $\pm$ 0.01\\
\hspace{1em}$\text{Full-1Step } I_{n_i} \otimes I_L$ & 0.87 $\pm$ 0.01 & 0.89 $\pm$ 0.01 & 0.95 $\pm$ 0.01 & 0.90 $\pm$ 0.01 & 0.90 $\pm$ 0.00 & 0.96 $\pm$ 0.01\\
\addlinespace[0.6em]
\multicolumn{7}{l}{\textbf{$N = 25,\; n_i = 100$}}\\
\hspace{1em}$\text{1-Step } {\mathbf{R}}_{\text{Lon}} \otimes {\mathbf{R}}_{\text{Fun}}$ & 0.78 $\pm$ 0.01 & 0.65 $\pm$ 0.01 & 0.82 $\pm$ 0.01 & 0.87 $\pm$ 0.01 & 0.87 $\pm$ 0.01 & 0.89 $\pm$ 0.01\\
\hspace{1em}$\text{1-Step } {\mathbf{R}}_{\text{Lon}} \otimes I_L$ & 0.77 $\pm$ 0.01 & 0.65 $\pm$ 0.01 & 0.81 $\pm$ 0.01 & 0.85 $\pm$ 0.01 & 0.86 $\pm$ 0.01 & 0.88 $\pm$ 0.01\\
\hspace{1em}$\text{1-Step } I_{n_i} \otimes {\mathbf{R}}_{\text{Fun}}$ & 0.80 $\pm$ 0.01 & 0.87 $\pm$ 0.01 & 0.79 $\pm$ 0.01 & 0.93 $\pm$ 0.01 & 0.95 $\pm$ 0.00 & 0.93 $\pm$ 0.01\\
\hspace{1em}$\text{1-Step } I_{n_i} \otimes I_L$ & 0.79 $\pm$ 0.01 & 0.85 $\pm$ 0.01 & 0.79 $\pm$ 0.01 & 0.90 $\pm$ 0.01 & 0.94 $\pm$ 0.00 & 0.92 $\pm$ 0.01\\
\hspace{1em}$\text{Full-1Step } {\mathbf{R}}_{\text{Lon}} \otimes {\mathbf{R}}_{\text{Fun}}$ & 0.78 $\pm$ 0.01 & 0.66 $\pm$ 0.01 & 4.75 $\pm$ 0.69$^*$ & 0.87 $\pm$ 0.01 & 0.87 $\pm$ 0.01 & 0.90 $\pm$ 0.01\\
\hspace{1em}$\text{Full-1Step } {\mathbf{R}}_{\text{Lon}} \otimes I_L$ & 0.77 $\pm$ 0.01 & 0.86 $\pm$ 0.07$^*$ & 7.42 $\pm$ 0.95$^*$ & 0.85 $\pm$ 0.01 & 0.86 $\pm$ 0.01 & 0.88 $\pm$ 0.01\\
\hspace{1em}$\text{Full-1Step } I_{n_i} \otimes {\mathbf{R}}_{\text{Fun}}$ & 0.81 $\pm$ 0.01 & 0.87 $\pm$ 0.01 & 0.83 $\pm$ 0.01 & 0.93 $\pm$ 0.01 & 0.95 $\pm$ 0.00 & 0.94 $\pm$ 0.01\\
\hspace{1em}$\text{Full-1Step } I_{n_i} \otimes I_L$ & 0.79 $\pm$ 0.01 & 0.85 $\pm$ 0.01 & 0.90 $\pm$ 0.09 & 0.90 $\pm$ 0.01 & 0.94 $\pm$ 0.00 & 0.92 $\pm$ 0.01\\
\addlinespace[0.6em]
\multicolumn{7}{l}{\textbf{$N = 50,\; n_i = 5$}}\\
\hspace{1em}$\text{1-Step } {\mathbf{R}}_{\text{Lon}} \otimes {\mathbf{R}}_{\text{Fun}}$ & 0.88 $\pm$ 0.01 & 0.84 $\pm$ 0.01 & 0.88 $\pm$ 0.01 & 0.89 $\pm$ 0.01 & 0.87 $\pm$ 0.01 & 0.89 $\pm$ 0.01\\
\hspace{1em}$\text{1-Step } {\mathbf{R}}_{\text{Lon}} \otimes I_L$ & 0.88 $\pm$ 0.01 & 0.84 $\pm$ 0.01 & 0.87 $\pm$ 0.01 & 0.88 $\pm$ 0.01 & 0.85 $\pm$ 0.01 & 0.89 $\pm$ 0.01\\
\hspace{1em}$\text{1-Step } I_{n_i} \otimes {\mathbf{R}}_{\text{Fun}}$ & 0.93 $\pm$ 0.01 & 0.93 $\pm$ 0.00 & 0.92 $\pm$ 0.01 & 0.94 $\pm$ 0.01 & 0.93 $\pm$ 0.00 & 0.93 $\pm$ 0.01\\
\hspace{1em}$\text{1-Step } I_{n_i} \otimes I_L$ & 0.93 $\pm$ 0.01 & 0.91 $\pm$ 0.01 & 0.92 $\pm$ 0.01 & 0.94 $\pm$ 0.01 & 0.91 $\pm$ 0.00 & 0.92 $\pm$ 0.01\\
\hspace{1em}$\text{Full-1Step } {\mathbf{R}}_{\text{Lon}} \otimes {\mathbf{R}}_{\text{Fun}}$ & 0.88 $\pm$ 0.01 & 0.84 $\pm$ 0.01 & 0.90 $\pm$ 0.01 & 0.89 $\pm$ 0.01 & 0.87 $\pm$ 0.01 & 0.91 $\pm$ 0.01\\
\hspace{1em}$\text{Full-1Step } {\mathbf{R}}_{\text{Lon}} \otimes I_L$ & 0.88 $\pm$ 0.01 & 1.01 $\pm$ 0.07 & 0.96 $\pm$ 0.08$^*$ & 0.88 $\pm$ 0.01 & 0.85 $\pm$ 0.01 & 0.90 $\pm$ 0.01\\
\hspace{1em}$\text{Full-1Step } I_{n_i} \otimes {\mathbf{R}}_{\text{Fun}}$ & 0.94 $\pm$ 0.01 & 0.93 $\pm$ 0.01 & 0.94 $\pm$ 0.01 & 0.95 $\pm$ 0.01 & 0.93 $\pm$ 0.00 & 0.95 $\pm$ 0.01\\
\hspace{1em}$\text{Full-1Step } I_{n_i} \otimes I_L$ & 0.93 $\pm$ 0.01 & 0.91 $\pm$ 0.01 & 0.93 $\pm$ 0.01 & 0.94 $\pm$ 0.01 & 0.91 $\pm$ 0.00 & 0.93 $\pm$ 0.01\\
\addlinespace[0.6em]
\multicolumn{7}{l}{\textbf{$N = 50,\; n_i = 100$}}\\
\hspace{1em}$\text{1-Step } {\mathbf{R}}_{\text{Lon}} \otimes {\mathbf{R}}_{\text{Fun}}$ & 0.85 $\pm$ 0.01 & 0.65 $\pm$ 0.01 & 0.83 $\pm$ 0.01 & 0.90 $\pm$ 0.01 & 0.87 $\pm$ 0.00 & 0.89 $\pm$ 0.01\\
\hspace{1em}$\text{1-Step } {\mathbf{R}}_{\text{Lon}} \otimes I_L$ & 0.85 $\pm$ 0.01 & 0.65 $\pm$ 0.01 & 0.81 $\pm$ 0.01 & 0.86 $\pm$ 0.01 & 0.87 $\pm$ 0.00 & 0.88 $\pm$ 0.01\\
\hspace{1em}$\text{1-Step } I_{n_i} \otimes {\mathbf{R}}_{\text{Fun}}$ & 0.85 $\pm$ 0.01 & 0.89 $\pm$ 0.01 & 0.82 $\pm$ 0.01 & 0.95 $\pm$ 0.01 & 0.97 $\pm$ 0.00 & 0.93 $\pm$ 0.01\\
\hspace{1em}$\text{1-Step } I_{n_i} \otimes I_L$ & 0.84 $\pm$ 0.01 & 0.87 $\pm$ 0.01 & 0.81 $\pm$ 0.01 & 0.92 $\pm$ 0.00 & 0.96 $\pm$ 0.00 & 0.92 $\pm$ 0.01\\
\hspace{1em}$\text{Full-1Step } {\mathbf{R}}_{\text{Lon}} \otimes {\mathbf{R}}_{\text{Fun}}$ & 0.86 $\pm$ 0.01 & 0.66 $\pm$ 0.01 & 2.63 $\pm$ 0.52$^*$ & 0.90 $\pm$ 0.01 & 0.87 $\pm$ 0.00 & 0.89 $\pm$ 0.01\\
\hspace{1em}$\text{Full-1Step } {\mathbf{R}}_{\text{Lon}} \otimes I_L$ & 0.85 $\pm$ 0.01 & 0.76 $\pm$ 0.03 & 4.33 $\pm$ 0.85$^*$ & 0.86 $\pm$ 0.01 & 0.87 $\pm$ 0.00 & 0.88 $\pm$ 0.01\\
\hspace{1em}$\text{Full-1Step } I_{n_i} \otimes {\mathbf{R}}_{\text{Fun}}$ & 0.86 $\pm$ 0.01 & 0.90 $\pm$ 0.01 & 0.85 $\pm$ 0.01 & 0.95 $\pm$ 0.01 & 0.97 $\pm$ 0.00 & 0.93 $\pm$ 0.01\\
\hspace{1em}$\text{Full-1Step } I_{n_i} \otimes I_L$ & 0.84 $\pm$ 0.01 & 0.87 $\pm$ 0.01 & 0.82 $\pm$ 0.01 & 0.92 $\pm$ 0.00 & 0.96 $\pm$ 0.00 & 0.92 $\pm$ 0.01\\
\addlinespace[0.6em]
\multicolumn{7}{l}{\textbf{$N = 500,\; n_i = 5$}}\\
\hspace{1em}$\text{1-Step } {\mathbf{R}}_{\text{Lon}} \otimes {\mathbf{R}}_{\text{Fun}}$ & 0.89 $\pm$ 0.01 & 0.87 $\pm$ 0.01 & 0.91 $\pm$ 0.01 & 0.89 $\pm$ 0.01 & 0.89 $\pm$ 0.01 & 0.92 $\pm$ 0.01\\
\hspace{1em}$\text{1-Step } {\mathbf{R}}_{\text{Lon}} \otimes I_L$ & 0.85 $\pm$ 0.01 & 0.88 $\pm$ 0.01 & 0.89 $\pm$ 0.01 & 0.86 $\pm$ 0.01 & 0.89 $\pm$ 0.01 & 0.90 $\pm$ 0.01\\
\hspace{1em}$\text{1-Step } I_{n_i} \otimes {\mathbf{R}}_{\text{Fun}}$ & 0.95 $\pm$ 0.00 & 0.96 $\pm$ 0.00 & 0.95 $\pm$ 0.01 & 0.96 $\pm$ 0.00 & 0.97 $\pm$ 0.00 & 0.96 $\pm$ 0.01\\
\hspace{1em}$\text{1-Step } I_{n_i} \otimes I_L$ & 0.91 $\pm$ 0.00 & 0.95 $\pm$ 0.00 & 0.93 $\pm$ 0.01 & 0.92 $\pm$ 0.00 & 0.96 $\pm$ 0.00 & 0.94 $\pm$ 0.01\\
\hspace{1em}$\text{Full-1Step } {\mathbf{R}}_{\text{Lon}} \otimes {\mathbf{R}}_{\text{Fun}}$ & 0.89 $\pm$ 0.01 & 0.87 $\pm$ 0.01 & 0.92 $\pm$ 0.01 & 0.89 $\pm$ 0.01 & 0.89 $\pm$ 0.01 & 0.92 $\pm$ 0.01\\
\hspace{1em}$\text{Full-1Step } {\mathbf{R}}_{\text{Lon}} \otimes I_L$ & 0.85 $\pm$ 0.01 & 0.88 $\pm$ 0.01 & 0.89 $\pm$ 0.01 & 0.86 $\pm$ 0.01 & 0.89 $\pm$ 0.01 & 0.90 $\pm$ 0.01\\
\hspace{1em}$\text{Full-1Step } I_{n_i} \otimes {\mathbf{R}}_{\text{Fun}}$ & 0.95 $\pm$ 0.00 & 0.96 $\pm$ 0.00 & 0.95 $\pm$ 0.01 & 0.96 $\pm$ 0.00 & 0.97 $\pm$ 0.00 & 0.96 $\pm$ 0.01\\
\hspace{1em}$\text{Full-1Step } I_{n_i} \otimes I_L$ & 0.91 $\pm$ 0.00 & 0.95 $\pm$ 0.00 & 0.94 $\pm$ 0.01 & 0.92 $\pm$ 0.00 & 0.96 $\pm$ 0.00 & 0.94 $\pm$ 0.01\\
\addlinespace[0.6em]
\multicolumn{7}{l}{\textbf{$N = 500,\; n_i = 100$}}\\
\hspace{1em}$\text{1-Step } {\mathbf{R}}_{\text{Lon}} \otimes {\mathbf{R}}_{\text{Fun}}$ & 0.90 $\pm$ 0.00 & 0.67 $\pm$ 0.01 & 0.96 $\pm$ 0.01 & 0.93 $\pm$ 0.00 & 0.88 $\pm$ 0.01 & 0.94 $\pm$ 0.01\\
\hspace{1em}$\text{1-Step } {\mathbf{R}}_{\text{Lon}} \otimes I_L$ & 0.87 $\pm$ 0.00 & 0.67 $\pm$ 0.01 & 0.94 $\pm$ 0.01 & 0.90 $\pm$ 0.00 & 0.89 $\pm$ 0.00 & 0.93 $\pm$ 0.01\\
\hspace{1em}$\text{1-Step } I_{n_i} \otimes {\mathbf{R}}_{\text{Fun}}$ & 0.92 $\pm$ 0.00 & 0.96 $\pm$ 0.00 & 0.90 $\pm$ 0.01 & 0.99 $\pm$ 0.00 & 0.98 $\pm$ 0.00 & 0.98 $\pm$ 0.00\\
\hspace{1em}$\text{1-Step } I_{n_i} \otimes I_L$ & 0.88 $\pm$ 0.00 & 0.95 $\pm$ 0.00 & 0.89 $\pm$ 0.01 & 0.96 $\pm$ 0.00 & 0.99 $\pm$ 0.00 & 0.97 $\pm$ 0.00\\
\hspace{1em}$\text{Full-1Step } {\mathbf{R}}_{\text{Lon}} \otimes {\mathbf{R}}_{\text{Fun}}$ & 0.90 $\pm$ 0.00 & 0.67 $\pm$ 0.01 & 0.98 $\pm$ 0.01 & 0.93 $\pm$ 0.00 & 0.88 $\pm$ 0.01 & 0.94 $\pm$ 0.01\\
\hspace{1em}$\text{Full-1Step } {\mathbf{R}}_{\text{Lon}} \otimes I_L$ & 0.87 $\pm$ 0.00 & 0.67 $\pm$ 0.01 & 0.95 $\pm$ 0.01 & 0.90 $\pm$ 0.00 & 0.89 $\pm$ 0.00 & 0.93 $\pm$ 0.01\\
\hspace{1em}$\text{Full-1Step } I_{n_i} \otimes {\mathbf{R}}_{\text{Fun}}$ & 0.92 $\pm$ 0.00 & 0.96 $\pm$ 0.00 & 0.91 $\pm$ 0.01 & 0.99 $\pm$ 0.00 & 0.98 $\pm$ 0.00 & 0.98 $\pm$ 0.00\\
\hspace{1em}$\text{Full-1Step } I_{n_i} \otimes I_L$ & 0.88 $\pm$ 0.00 & 0.95 $\pm$ 0.00 & 0.89 $\pm$ 0.01 & 0.96 $\pm$ 0.00 & 0.99 $\pm$ 0.00 & 0.97 $\pm$ 0.00\\
\bottomrule
\end{tabular}}
\end{table}

\begin{table}[!h] 
\centering
\caption{\label{tab:jci} \footnotesize
Joint 95\% CI coverage from 300 replicates $\pm$ SE (SE$=0.00$ indicates a value $<0.01$). 
Table columns indicate if $\mathbf{R}^*_{\text{Lon}}$ had exchangeable or AR1 correlation. 
pffr (Wild) indicates CIs constructed for $\texttt{pffr}$ coefficient estimates with quantiles obtained from a wild cluster bootstrap. pffr ($z_{1-\alpha/2}$) are standard Wald CIs constructed with Gaussian quantiles.
}
\centering
\resizebox{\ifdim\width>\linewidth\linewidth\else\width\fi}{!}{
\fontsize{9}{11}\selectfont
\begin{tabular}[t]{>{\raggedright\arraybackslash}p{4.8cm}>{\raggedright\arraybackslash}p{2.1cm}>{\raggedright\arraybackslash}p{2.1cm}>{\raggedright\arraybackslash}p{2.1cm}>{\raggedright\arraybackslash}p{2.1cm}>{\raggedright\arraybackslash}p{2.1cm}>{\raggedright\arraybackslash}p{2.1cm}}
\toprule
\multicolumn{1}{c}{ } & \multicolumn{3}{c}{Exchangeable} & \multicolumn{3}{c}{AR(1)} \\
\cmidrule(l{3pt}r{3pt}){2-4} \cmidrule(l{3pt}r{3pt}){5-7}
Method & Gaussian & Poisson & Binomial & Gaussian & Poisson & Binomial\\
\midrule
\addlinespace[0em]
\multicolumn{7}{l}{\textbf{$N = 25,\; n_i = 5$}}\\
\hspace{1em}$\text{1-Step } {\mathbf{R}}_{\text{Lon}} \otimes {\mathbf{R}}_{\text{Fun}}$ & 0.90 $\pm$ 0.02 & 0.90 $\pm$ 0.02 & 0.95 $\pm$ 0.01 & 0.90 $\pm$ 0.02 & 0.88 $\pm$ 0.02 & 0.95 $\pm$ 0.01\\
\hspace{1em}$\text{1-Step } {\mathbf{R}}_{\text{Lon}} \otimes I_L$ & 0.92 $\pm$ 0.02 & 0.92 $\pm$ 0.02 & 0.95 $\pm$ 0.01 & 0.91 $\pm$ 0.02 & 0.90 $\pm$ 0.02 & 0.94 $\pm$ 0.01\\
\hspace{1em}$\text{1-Step } I_{n_i} \otimes {\mathbf{R}}_{\text{Fun}}$ & 0.89 $\pm$ 0.02 & 0.84 $\pm$ 0.02 & 0.94 $\pm$ 0.01 & 0.87 $\pm$ 0.02 & 0.85 $\pm$ 0.02 & 0.93 $\pm$ 0.01\\
\hspace{1em}$\text{1-Step } I_{n_i} \otimes I_L$ & 0.91 $\pm$ 0.02 & 0.88 $\pm$ 0.02 & 0.95 $\pm$ 0.01 & 0.90 $\pm$ 0.02 & 0.86 $\pm$ 0.02 & 0.94 $\pm$ 0.01\\
\hspace{1em}$\text{Full-1Step } {\mathbf{R}}_{\text{Lon}} \otimes {\mathbf{R}}_{\text{Fun}}$ & 0.90 $\pm$ 0.02 & 0.90 $\pm$ 0.02 & 0.94 $\pm$ 0.01 & 0.89 $\pm$ 0.02 & 0.88 $\pm$ 0.02 & 0.93 $\pm$ 0.01\\
\hspace{1em}$\text{Full-1Step } {\mathbf{R}}_{\text{Lon}} \otimes I_L$ & 0.93 $\pm$ 0.02 & 0.91 $\pm$ 0.02 & 0.90 $\pm$ 0.02 & 0.91 $\pm$ 0.02 & 0.90 $\pm$ 0.02 & 0.93 $\pm$ 0.01\\
\hspace{1em}$\text{Full-1Step } I_{n_i} \otimes {\mathbf{R}}_{\text{Fun}}$ & 0.89 $\pm$ 0.02 & 0.84 $\pm$ 0.02 & 0.93 $\pm$ 0.01 & 0.86 $\pm$ 0.02 & 0.84 $\pm$ 0.02 & 0.92 $\pm$ 0.02\\
\hspace{1em}$\text{Full-1Step } I_{n_i} \otimes I_L$ & 0.91 $\pm$ 0.02 & 0.87 $\pm$ 0.02 & 0.94 $\pm$ 0.01 & 0.90 $\pm$ 0.02 & 0.86 $\pm$ 0.02 & 0.93 $\pm$ 0.02\\
\hspace{1em}$\text{pffr (Wild)}$ & 0.92 $\pm$ 0.02 & 0.85 $\pm$ 0.02 & 0.95 $\pm$ 0.01 & 0.93 $\pm$ 0.01 & 0.86 $\pm$ 0.02 & 0.96 $\pm$ 0.01\\
\hspace{1em}$\text{pffr ($z_{1-\alpha/2}$)}$ & 0.06 $\pm$ 0.01 & 0.55 $\pm$ 0.03 & 0.32 $\pm$ 0.03 & 0.06 $\pm$ 0.01 & 0.56 $\pm$ 0.03 & 0.32 $\pm$ 0.03\\
\addlinespace[0.6em]
\multicolumn{7}{l}{\textbf{$N = 25,\; n_i = 100$}}\\
\hspace{1em}$\text{1-Step } {\mathbf{R}}_{\text{Lon}} \otimes {\mathbf{R}}_{\text{Fun}}$ & 0.96 $\pm$ 0.01 & 0.98 $\pm$ 0.01 & 0.99 $\pm$ 0.01 & 0.89 $\pm$ 0.02 & 0.88 $\pm$ 0.02 & 0.93 $\pm$ 0.02\\
\hspace{1em}$\text{1-Step } {\mathbf{R}}_{\text{Lon}} \otimes I_L$ & 0.97 $\pm$ 0.01 & 0.98 $\pm$ 0.01 & 0.99 $\pm$ 0.01 & 0.92 $\pm$ 0.02 & 0.88 $\pm$ 0.02 & 0.93 $\pm$ 0.01\\
\hspace{1em}$\text{1-Step } I_{n_i} \otimes {\mathbf{R}}_{\text{Fun}}$ & 0.97 $\pm$ 0.01 & 0.91 $\pm$ 0.02 & 0.99 $\pm$ 0.01 & 0.88 $\pm$ 0.02 & 0.82 $\pm$ 0.02 & 0.93 $\pm$ 0.01\\
\hspace{1em}$\text{1-Step } I_{n_i} \otimes I_L$ & 0.97 $\pm$ 0.01 & 0.91 $\pm$ 0.02 & 0.98 $\pm$ 0.01 & 0.92 $\pm$ 0.02 & 0.83 $\pm$ 0.02 & 0.93 $\pm$ 0.01\\
\hspace{1em}$\text{Full-1Step } {\mathbf{R}}_{\text{Lon}} \otimes {\mathbf{R}}_{\text{Fun}}$ & 0.96 $\pm$ 0.01 & 0.97 $\pm$ 0.01 & 0.81 $\pm$ 0.02 & 0.89 $\pm$ 0.02 & 0.88 $\pm$ 0.02 & 0.93 $\pm$ 0.01\\
\hspace{1em}$\text{Full-1Step } {\mathbf{R}}_{\text{Lon}} \otimes I_L$ & 0.98 $\pm$ 0.01 & 0.98 $\pm$ 0.01 & 0.61 $\pm$ 0.03 & 0.92 $\pm$ 0.02 & 0.88 $\pm$ 0.02 & 0.93 $\pm$ 0.01\\
\hspace{1em}$\text{Full-1Step } I_{n_i} \otimes {\mathbf{R}}_{\text{Fun}}$ & 0.96 $\pm$ 0.01 & 0.91 $\pm$ 0.02 & 0.99 $\pm$ 0.01 & 0.88 $\pm$ 0.02 & 0.82 $\pm$ 0.02 & 0.93 $\pm$ 0.01\\
\hspace{1em}$\text{Full-1Step } I_{n_i} \otimes I_L$ & 0.97 $\pm$ 0.01 & 0.91 $\pm$ 0.02 & 0.98 $\pm$ 0.01 & 0.92 $\pm$ 0.02 & 0.83 $\pm$ 0.02 & 0.93 $\pm$ 0.01\\
\hspace{1em}$\text{pffr (Wild)}$ & 0.79 $\pm$ 0.02 & 0.78 $\pm$ 0.02 & 0.75 $\pm$ 0.03 & 0.99 $\pm$ 0.01 & 0.78 $\pm$ 0.02 & 0.93 $\pm$ 0.01\\
\hspace{1em}$\text{pffr ($z_{1-\alpha/2}$)}$ & 0.49 $\pm$ 0.03 & 0.81 $\pm$ 0.02 & 0.64 $\pm$ 0.03 & 0.23 $\pm$ 0.02 & 0.79 $\pm$ 0.02 & 0.62 $\pm$ 0.03\\
\addlinespace[0.6em]
\multicolumn{7}{l}{\textbf{$N = 50,\; n_i = 5$}}\\
\hspace{1em}$\text{1-Step } {\mathbf{R}}_{\text{Lon}} \otimes {\mathbf{R}}_{\text{Fun}}$ & 0.90 $\pm$ 0.02 & 0.90 $\pm$ 0.02 & 0.95 $\pm$ 0.01 & 0.89 $\pm$ 0.02 & 0.87 $\pm$ 0.02 & 0.94 $\pm$ 0.01\\
\hspace{1em}$\text{1-Step } {\mathbf{R}}_{\text{Lon}} \otimes I_L$ & 0.92 $\pm$ 0.02 & 0.90 $\pm$ 0.02 & 0.95 $\pm$ 0.01 & 0.92 $\pm$ 0.02 & 0.90 $\pm$ 0.02 & 0.95 $\pm$ 0.01\\
\hspace{1em}$\text{1-Step } I_{n_i} \otimes {\mathbf{R}}_{\text{Fun}}$ & 0.89 $\pm$ 0.02 & 0.85 $\pm$ 0.02 & 0.95 $\pm$ 0.01 & 0.89 $\pm$ 0.02 & 0.85 $\pm$ 0.02 & 0.94 $\pm$ 0.01\\
\hspace{1em}$\text{1-Step } I_{n_i} \otimes I_L$ & 0.92 $\pm$ 0.02 & 0.88 $\pm$ 0.02 & 0.95 $\pm$ 0.01 & 0.90 $\pm$ 0.02 & 0.86 $\pm$ 0.02 & 0.94 $\pm$ 0.01\\
\hspace{1em}$\text{Full-1Step } {\mathbf{R}}_{\text{Lon}} \otimes {\mathbf{R}}_{\text{Fun}}$ & 0.90 $\pm$ 0.02 & 0.90 $\pm$ 0.02 & 0.94 $\pm$ 0.01 & 0.89 $\pm$ 0.02 & 0.87 $\pm$ 0.02 & 0.94 $\pm$ 0.01\\
\hspace{1em}$\text{Full-1Step } {\mathbf{R}}_{\text{Lon}} \otimes I_L$ & 0.92 $\pm$ 0.02 & 0.90 $\pm$ 0.02 & 0.94 $\pm$ 0.01 & 0.92 $\pm$ 0.02 & 0.90 $\pm$ 0.02 & 0.94 $\pm$ 0.01\\
\hspace{1em}$\text{Full-1Step } I_{n_i} \otimes {\mathbf{R}}_{\text{Fun}}$ & 0.89 $\pm$ 0.02 & 0.85 $\pm$ 0.02 & 0.94 $\pm$ 0.01 & 0.88 $\pm$ 0.02 & 0.85 $\pm$ 0.02 & 0.93 $\pm$ 0.01\\
\hspace{1em}$\text{Full-1Step } I_{n_i} \otimes I_L$ & 0.92 $\pm$ 0.02 & 0.87 $\pm$ 0.02 & 0.94 $\pm$ 0.01 & 0.90 $\pm$ 0.02 & 0.86 $\pm$ 0.02 & 0.94 $\pm$ 0.01\\
\hspace{1em}$\text{pffr (Wild)}$ & 0.97 $\pm$ 0.01 & 0.91 $\pm$ 0.02 & 0.99 $\pm$ 0.01 & 0.97 $\pm$ 0.01 & 0.91 $\pm$ 0.02 & 0.99 $\pm$ 0.01\\
\hspace{1em}$\text{pffr ($z_{1-\alpha/2}$)}$ & 0.04 $\pm$ 0.01 & 0.59 $\pm$ 0.03 & 0.23 $\pm$ 0.02 & 0.04 $\pm$ 0.01 & 0.57 $\pm$ 0.03 & 0.26 $\pm$ 0.03\\
\addlinespace[0.6em]
\multicolumn{7}{l}{\textbf{$N = 50,\; n_i = 100$}}\\
\hspace{1em}$\text{1-Step } {\mathbf{R}}_{\text{Lon}} \otimes {\mathbf{R}}_{\text{Fun}}$ & 0.96 $\pm$ 0.01 & 0.97 $\pm$ 0.01 & 0.98 $\pm$ 0.01 & 0.93 $\pm$ 0.02 & 0.90 $\pm$ 0.02 & 0.95 $\pm$ 0.01\\
\hspace{1em}$\text{1-Step } {\mathbf{R}}_{\text{Lon}} \otimes I_L$ & 0.97 $\pm$ 0.01 & 0.98 $\pm$ 0.01 & 0.98 $\pm$ 0.01 & 0.95 $\pm$ 0.01 & 0.91 $\pm$ 0.02 & 0.95 $\pm$ 0.01\\
\hspace{1em}$\text{1-Step } I_{n_i} \otimes {\mathbf{R}}_{\text{Fun}}$ & 0.98 $\pm$ 0.01 & 0.91 $\pm$ 0.02 & 0.99 $\pm$ 0.01 & 0.92 $\pm$ 0.02 & 0.84 $\pm$ 0.02 & 0.94 $\pm$ 0.01\\
\hspace{1em}$\text{1-Step } I_{n_i} \otimes I_L$ & 0.98 $\pm$ 0.01 & 0.94 $\pm$ 0.01 & 0.99 $\pm$ 0.01 & 0.95 $\pm$ 0.01 & 0.85 $\pm$ 0.02 & 0.95 $\pm$ 0.01\\
\hspace{1em}$\text{Full-1Step } {\mathbf{R}}_{\text{Lon}} \otimes {\mathbf{R}}_{\text{Fun}}$ & 0.96 $\pm$ 0.01 & 0.98 $\pm$ 0.01 & 0.89 $\pm$ 0.02 & 0.93 $\pm$ 0.02 & 0.90 $\pm$ 0.02 & 0.95 $\pm$ 0.01\\
\hspace{1em}$\text{Full-1Step } {\mathbf{R}}_{\text{Lon}} \otimes I_L$ & 0.97 $\pm$ 0.01 & 0.98 $\pm$ 0.01 & 0.76 $\pm$ 0.03 & 0.95 $\pm$ 0.01 & 0.91 $\pm$ 0.02 & 0.95 $\pm$ 0.01\\
\hspace{1em}$\text{Full-1Step } I_{n_i} \otimes {\mathbf{R}}_{\text{Fun}}$ & 0.98 $\pm$ 0.01 & 0.91 $\pm$ 0.02 & 0.98 $\pm$ 0.01 & 0.92 $\pm$ 0.02 & 0.84 $\pm$ 0.02 & 0.94 $\pm$ 0.01\\
\hspace{1em}$\text{Full-1Step } I_{n_i} \otimes I_L$ & 0.98 $\pm$ 0.01 & 0.94 $\pm$ 0.01 & 0.98 $\pm$ 0.01 & 0.95 $\pm$ 0.01 & 0.85 $\pm$ 0.02 & 0.95 $\pm$ 0.01\\
\hspace{1em}$\text{pffr (Wild)}$ & 0.89 $\pm$ 0.02 & 0.82 $\pm$ 0.02 & 0.85 $\pm$ 0.02 & 1.00 $\pm$ 0.00 & 0.85 $\pm$ 0.02 & 0.99 $\pm$ 0.01\\
\hspace{1em}$\text{pffr ($z_{1-\alpha/2}$)}$ & 0.50 $\pm$ 0.03 & 0.82 $\pm$ 0.02 & 0.70 $\pm$ 0.03 & 0.19 $\pm$ 0.02 & 0.81 $\pm$ 0.02 & 0.65 $\pm$ 0.03\\
\addlinespace[0.6em]
\multicolumn{7}{l}{\textbf{$N = 500,\; n_i = 5$}}\\
\hspace{1em}$\text{1-Step } {\mathbf{R}}_{\text{Lon}} \otimes {\mathbf{R}}_{\text{Fun}}$ & 0.94 $\pm$ 0.01 & 0.93 $\pm$ 0.01 & 0.96 $\pm$ 0.01 & 0.95 $\pm$ 0.01 & 0.93 $\pm$ 0.02 & 0.96 $\pm$ 0.01\\
\hspace{1em}$\text{1-Step } {\mathbf{R}}_{\text{Lon}} \otimes I_L$ & 0.97 $\pm$ 0.01 & 0.94 $\pm$ 0.01 & 0.96 $\pm$ 0.01 & 0.96 $\pm$ 0.01 & 0.93 $\pm$ 0.01 & 0.97 $\pm$ 0.01\\
\hspace{1em}$\text{1-Step } I_{n_i} \otimes {\mathbf{R}}_{\text{Fun}}$ & 0.95 $\pm$ 0.01 & 0.93 $\pm$ 0.02 & 0.96 $\pm$ 0.01 & 0.95 $\pm$ 0.01 & 0.93 $\pm$ 0.02 & 0.96 $\pm$ 0.01\\
\hspace{1em}$\text{1-Step } I_{n_i} \otimes I_L$ & 0.97 $\pm$ 0.01 & 0.94 $\pm$ 0.01 & 0.96 $\pm$ 0.01 & 0.96 $\pm$ 0.01 & 0.93 $\pm$ 0.01 & 0.97 $\pm$ 0.01\\
\hspace{1em}$\text{Full-1Step } {\mathbf{R}}_{\text{Lon}} \otimes {\mathbf{R}}_{\text{Fun}}$ & 0.94 $\pm$ 0.01 & 0.93 $\pm$ 0.01 & 0.96 $\pm$ 0.01 & 0.95 $\pm$ 0.01 & 0.92 $\pm$ 0.02 & 0.96 $\pm$ 0.01\\
\hspace{1em}$\text{Full-1Step } {\mathbf{R}}_{\text{Lon}} \otimes I_L$ & 0.97 $\pm$ 0.01 & 0.94 $\pm$ 0.01 & 0.96 $\pm$ 0.01 & 0.96 $\pm$ 0.01 & 0.93 $\pm$ 0.01 & 0.96 $\pm$ 0.01\\
\hspace{1em}$\text{Full-1Step } I_{n_i} \otimes {\mathbf{R}}_{\text{Fun}}$ & 0.95 $\pm$ 0.01 & 0.93 $\pm$ 0.01 & 0.96 $\pm$ 0.01 & 0.95 $\pm$ 0.01 & 0.93 $\pm$ 0.02 & 0.96 $\pm$ 0.01\\
\hspace{1em}$\text{Full-1Step } I_{n_i} \otimes I_L$ & 0.97 $\pm$ 0.01 & 0.94 $\pm$ 0.01 & 0.96 $\pm$ 0.01 & 0.96 $\pm$ 0.01 & 0.93 $\pm$ 0.02 & 0.97 $\pm$ 0.01\\
\hspace{1em}$\text{pffr (Wild)}$ & 1.00 $\pm$ 0.00 & 1.00 $\pm$ 0.00 & 1.00 $\pm$ 0.00 & 1.00 $\pm$ 0.00 & 1.00 $\pm$ 0.00 & 1.00 $\pm$ 0.00\\
\hspace{1em}$\text{pffr ($z_{1-\alpha/2}$)}$ & 0.01 $\pm$ 0.00 & 0.63 $\pm$ 0.03 & 0.09 $\pm$ 0.02 & 0.01 $\pm$ 0.01 & 0.61 $\pm$ 0.03 & 0.07 $\pm$ 0.01\\
\addlinespace[0.6em]
\multicolumn{7}{l}{\textbf{$N = 500,\; n_i = 100$}}\\
\hspace{1em}$\text{1-Step } {\mathbf{R}}_{\text{Lon}} \otimes {\mathbf{R}}_{\text{Fun}}$ & 0.94 $\pm$ 0.01 & 0.92 $\pm$ 0.02 & 0.97 $\pm$ 0.01 & 0.94 $\pm$ 0.01 & 0.89 $\pm$ 0.02 & 0.95 $\pm$ 0.01\\
\hspace{1em}$\text{1-Step } {\mathbf{R}}_{\text{Lon}} \otimes I_L$ & 0.96 $\pm$ 0.01 & 0.92 $\pm$ 0.02 & 0.97 $\pm$ 0.01 & 0.95 $\pm$ 0.01 & 0.91 $\pm$ 0.02 & 0.96 $\pm$ 0.01\\
\hspace{1em}$\text{1-Step } I_{n_i} \otimes {\mathbf{R}}_{\text{Fun}}$ & 0.95 $\pm$ 0.01 & 0.92 $\pm$ 0.02 & 0.98 $\pm$ 0.01 & 0.95 $\pm$ 0.01 & 0.91 $\pm$ 0.02 & 0.95 $\pm$ 0.01\\
\hspace{1em}$\text{1-Step } I_{n_i} \otimes I_L$ & 0.97 $\pm$ 0.01 & 0.93 $\pm$ 0.01 & 0.98 $\pm$ 0.01 & 0.96 $\pm$ 0.01 & 0.91 $\pm$ 0.02 & 0.95 $\pm$ 0.01\\
\hspace{1em}$\text{Full-1Step } {\mathbf{R}}_{\text{Lon}} \otimes {\mathbf{R}}_{\text{Fun}}$ & 0.94 $\pm$ 0.01 & 0.92 $\pm$ 0.02 & 0.97 $\pm$ 0.01 & 0.94 $\pm$ 0.01 & 0.89 $\pm$ 0.02 & 0.95 $\pm$ 0.01\\
\hspace{1em}$\text{Full-1Step } {\mathbf{R}}_{\text{Lon}} \otimes I_L$ & 0.96 $\pm$ 0.01 & 0.92 $\pm$ 0.02 & 0.98 $\pm$ 0.01 & 0.95 $\pm$ 0.01 & 0.90 $\pm$ 0.02 & 0.96 $\pm$ 0.01\\
\hspace{1em}$\text{Full-1Step } I_{n_i} \otimes {\mathbf{R}}_{\text{Fun}}$ & 0.95 $\pm$ 0.01 & 0.92 $\pm$ 0.02 & 0.98 $\pm$ 0.01 & 0.95 $\pm$ 0.01 & 0.91 $\pm$ 0.02 & 0.95 $\pm$ 0.01\\
\hspace{1em}$\text{Full-1Step } I_{n_i} \otimes I_L$ & 0.97 $\pm$ 0.01 & 0.93 $\pm$ 0.01 & 0.98 $\pm$ 0.01 & 0.96 $\pm$ 0.01 & 0.91 $\pm$ 0.02 & 0.95 $\pm$ 0.01\\
\hspace{1em}$\text{pffr (Wild)}$ & 1.00 $\pm$ 0.00 & 0.93 $\pm$ 0.01 & 0.98 $\pm$ 0.01 & 1.00 $\pm$ 0.00 & 0.96 $\pm$ 0.01 & 1.00 $\pm$ 0.00\\
\hspace{1em}$\text{pffr ($z_{1-\alpha/2}$)}$ & 0.26 $\pm$ 0.03 & 0.84 $\pm$ 0.02 & 0.62 $\pm$ 0.03 & 0.11 $\pm$ 0.02 & 0.83 $\pm$ 0.02 & 0.56 $\pm$ 0.03\\
\bottomrule
\end{tabular}}
\end{table}

\begin{table}[!h]

\centering
\caption{\label{tab:pci} \footnotesize
Pointwise 95\% CI coverage from 300 replicates $\pm$ SE (SE$=0.00$ indicates a value $<0.01$). 
Table columns indicate if $\mathbf{R}^*_{\text{Lon}}$ had exchangeable or AR1 correlation. 
pffr (Wild) indicates CIs constructed for $\texttt{pffr}$ coefficient estimates with quantiles obtained from a wild cluster bootstrap. pffr ($z_{1-\alpha/2}$) are standard Wald CIs constructed with Gaussian quantiles.}
\centering
\resizebox{\ifdim\width>\linewidth\linewidth\else\width\fi}{!}{
\fontsize{9}{11}\selectfont
\begin{tabular}[t]{>{\raggedright\arraybackslash}p{4.8cm}>{\raggedright\arraybackslash}p{2.1cm}>{\raggedright\arraybackslash}p{2.1cm}>{\raggedright\arraybackslash}p{2.1cm}>{\raggedright\arraybackslash}p{2.1cm}>{\raggedright\arraybackslash}p{2.1cm}>{\raggedright\arraybackslash}p{2.1cm}}
\toprule
\multicolumn{1}{c}{ } & \multicolumn{3}{c}{Exchangeable} & \multicolumn{3}{c}{AR(1)} \\
\cmidrule(l{3pt}r{3pt}){2-4} \cmidrule(l{3pt}r{3pt}){5-7}
Method & Gaussian & Poisson & Binomial & Gaussian & Poisson & Binomial\\
\midrule
\addlinespace[0em]
\multicolumn{7}{l}{\textbf{$N = 25,\; n_i = 5$}}\\
\hspace{1em}$\text{1-Step } {\mathbf{R}}_{\text{Lon}} \otimes {\mathbf{R}}_{\text{Fun}}$ & 0.95 $\pm$ 0.01 & 0.95 $\pm$ 0.01 & 0.97 $\pm$ 0.01 & 0.94 $\pm$ 0.01 & 0.95 $\pm$ 0.01 & 0.96 $\pm$ 0.01\\
\hspace{1em}$\text{1-Step } {\mathbf{R}}_{\text{Lon}} \otimes I_L$ & 0.95 $\pm$ 0.01 & 0.95 $\pm$ 0.01 & 0.97 $\pm$ 0.01 & 0.95 $\pm$ 0.01 & 0.95 $\pm$ 0.01 & 0.96 $\pm$ 0.01\\
\hspace{1em}$\text{1-Step } I_{n_i} \otimes {\mathbf{R}}_{\text{Fun}}$ & 0.94 $\pm$ 0.01 & 0.93 $\pm$ 0.01 & 0.96 $\pm$ 0.01 & 0.93 $\pm$ 0.01 & 0.93 $\pm$ 0.01 & 0.96 $\pm$ 0.01\\
\hspace{1em}$\text{1-Step } I_{n_i} \otimes I_L$ & 0.95 $\pm$ 0.01 & 0.94 $\pm$ 0.01 & 0.96 $\pm$ 0.01 & 0.94 $\pm$ 0.01 & 0.94 $\pm$ 0.01 & 0.96 $\pm$ 0.01\\
\hspace{1em}$\text{Full-1Step } {\mathbf{R}}_{\text{Lon}} \otimes {\mathbf{R}}_{\text{Fun}}$ & 0.95 $\pm$ 0.01 & 0.95 $\pm$ 0.01 & 0.96 $\pm$ 0.01 & 0.94 $\pm$ 0.01 & 0.95 $\pm$ 0.01 & 0.96 $\pm$ 0.01\\
\hspace{1em}$\text{Full-1Step } {\mathbf{R}}_{\text{Lon}} \otimes I_L$ & 0.95 $\pm$ 0.01 & 0.95 $\pm$ 0.01 & 0.93 $\pm$ 0.02 & 0.95 $\pm$ 0.01 & 0.95 $\pm$ 0.01 & 0.96 $\pm$ 0.01\\
\hspace{1em}$\text{Full-1Step } I_{n_i} \otimes {\mathbf{R}}_{\text{Fun}}$ & 0.94 $\pm$ 0.01 & 0.93 $\pm$ 0.01 & 0.96 $\pm$ 0.01 & 0.93 $\pm$ 0.01 & 0.93 $\pm$ 0.01 & 0.96 $\pm$ 0.01\\
\hspace{1em}$\text{Full-1Step } I_{n_i} \otimes I_L$ & 0.95 $\pm$ 0.01 & 0.94 $\pm$ 0.01 & 0.96 $\pm$ 0.01 & 0.94 $\pm$ 0.01 & 0.94 $\pm$ 0.01 & 0.96 $\pm$ 0.01\\
\hspace{1em}$\text{pffr (Wild)}$ & 0.96 $\pm$ 0.01 & 0.93 $\pm$ 0.01 & 0.97 $\pm$ 0.01 & 0.97 $\pm$ 0.01 & 0.94 $\pm$ 0.01 & 0.97 $\pm$ 0.01\\
\hspace{1em}$\text{pffr ($z_{1-\alpha/2}$)}$ & 0.67 $\pm$ 0.03 & 0.87 $\pm$ 0.02 & 0.80 $\pm$ 0.02 & 0.66 $\pm$ 0.03 & 0.87 $\pm$ 0.02 & 0.80 $\pm$ 0.02\\
\addlinespace[0.6em]
\multicolumn{7}{l}{\textbf{$N = 25,\; n_i = 100$}}\\
\hspace{1em}$\text{1-Step } {\mathbf{R}}_{\text{Lon}} \otimes {\mathbf{R}}_{\text{Fun}}$ & 0.97 $\pm$ 0.01 & 0.98 $\pm$ 0.01 & 0.98 $\pm$ 0.01 & 0.95 $\pm$ 0.01 & 0.95 $\pm$ 0.01 & 0.96 $\pm$ 0.01\\
\hspace{1em}$\text{1-Step } {\mathbf{R}}_{\text{Lon}} \otimes I_L$ & 0.97 $\pm$ 0.01 & 0.98 $\pm$ 0.01 & 0.99 $\pm$ 0.01 & 0.95 $\pm$ 0.01 & 0.95 $\pm$ 0.01 & 0.96 $\pm$ 0.01\\
\hspace{1em}$\text{1-Step } I_{n_i} \otimes {\mathbf{R}}_{\text{Fun}}$ & 0.97 $\pm$ 0.01 & 0.95 $\pm$ 0.01 & 0.99 $\pm$ 0.01 & 0.94 $\pm$ 0.01 & 0.94 $\pm$ 0.01 & 0.96 $\pm$ 0.01\\
\hspace{1em}$\text{1-Step } I_{n_i} \otimes I_L$ & 0.97 $\pm$ 0.01 & 0.95 $\pm$ 0.01 & 0.98 $\pm$ 0.01 & 0.95 $\pm$ 0.01 & 0.94 $\pm$ 0.01 & 0.96 $\pm$ \vphantom{1} 0.01\\
\hspace{1em}$\text{Full-1Step } {\mathbf{R}}_{\text{Lon}} \otimes {\mathbf{R}}_{\text{Fun}}$ & 0.97 $\pm$ 0.01 & 0.98 $\pm$ 0.01 & 0.84 $\pm$ 0.02 & 0.95 $\pm$ 0.01 & 0.95 $\pm$ 0.01 & 0.96 $\pm$ 0.01\\
\hspace{1em}$\text{Full-1Step } {\mathbf{R}}_{\text{Lon}} \otimes I_L$ & 0.98 $\pm$ 0.01 & 0.98 $\pm$ 0.01 & 0.67 $\pm$ 0.03 & 0.95 $\pm$ 0.01 & 0.95 $\pm$ 0.01 & 0.96 $\pm$ 0.01\\
\hspace{1em}$\text{Full-1Step } I_{n_i} \otimes {\mathbf{R}}_{\text{Fun}}$ & 0.97 $\pm$ 0.01 & 0.95 $\pm$ 0.01 & 0.98 $\pm$ 0.01 & 0.94 $\pm$ 0.01 & 0.94 $\pm$ 0.01 & 0.96 $\pm$ 0.01\\
\hspace{1em}$\text{Full-1Step } I_{n_i} \otimes I_L$ & 0.97 $\pm$ 0.01 & 0.95 $\pm$ 0.01 & 0.98 $\pm$ 0.01 & 0.95 $\pm$ 0.01 & 0.94 $\pm$ 0.01 & 0.96 $\pm$ \vphantom{1} 0.01\\
\hspace{1em}$\text{pffr (Wild)}$ & 0.94 $\pm$ 0.01 & 0.93 $\pm$ 0.01 & 0.92 $\pm$ 0.02 & 0.99 $\pm$ 0.00 & 0.93 $\pm$ 0.01 & 0.96 $\pm$ 0.01\\
\hspace{1em}$\text{pffr ($z_{1-\alpha/2}$)}$ & 0.86 $\pm$ 0.02 & 0.94 $\pm$ 0.01 & 0.90 $\pm$ 0.02 & 0.78 $\pm$ 0.02 & 0.93 $\pm$ 0.01 & 0.89 $\pm$ 0.02\\
\addlinespace[0.6em]
\multicolumn{7}{l}{\textbf{$N = 50,\; n_i = 5$}}\\
\hspace{1em}$\text{1-Step } {\mathbf{R}}_{\text{Lon}} \otimes {\mathbf{R}}_{\text{Fun}}$ & 0.94 $\pm$ 0.01 & 0.94 $\pm$ 0.01 & 0.96 $\pm$ 0.01 & 0.93 $\pm$ 0.01 & 0.94 $\pm$ 0.01 & 0.95 $\pm$ 0.01\\
\hspace{1em}$\text{1-Step } {\mathbf{R}}_{\text{Lon}} \otimes I_L$ & 0.95 $\pm$ 0.01 & 0.95 $\pm$ 0.01 & 0.96 $\pm$ 0.01 & 0.94 $\pm$ 0.01 & 0.94 $\pm$ 0.01 & 0.95 $\pm$ 0.01\\
\hspace{1em}$\text{1-Step } I_{n_i} \otimes {\mathbf{R}}_{\text{Fun}}$ & 0.93 $\pm$ 0.01 & 0.93 $\pm$ 0.01 & 0.95 $\pm$ 0.01 & 0.93 $\pm$ 0.01 & 0.93 $\pm$ 0.01 & 0.95 $\pm$ 0.01\\
\hspace{1em}$\text{1-Step } I_{n_i} \otimes I_L$ & 0.94 $\pm$ 0.01 & 0.93 $\pm$ 0.01 & 0.95 $\pm$ 0.01 & 0.94 $\pm$ 0.01 & 0.93 $\pm$ 0.01 & 0.95 $\pm$ 0.01\\
\hspace{1em}$\text{Full-1Step } {\mathbf{R}}_{\text{Lon}} \otimes {\mathbf{R}}_{\text{Fun}}$ & 0.94 $\pm$ 0.01 & 0.94 $\pm$ 0.01 & 0.95 $\pm$ 0.01 & 0.93 $\pm$ 0.01 & 0.94 $\pm$ 0.01 & 0.95 $\pm$ 0.01\\
\hspace{1em}$\text{Full-1Step } {\mathbf{R}}_{\text{Lon}} \otimes I_L$ & 0.95 $\pm$ 0.01 & 0.94 $\pm$ 0.01 & 0.95 $\pm$ 0.01 & 0.94 $\pm$ 0.01 & 0.94 $\pm$ 0.01 & 0.95 $\pm$ 0.01\\
\hspace{1em}$\text{Full-1Step } I_{n_i} \otimes {\mathbf{R}}_{\text{Fun}}$ & 0.93 $\pm$ 0.01 & 0.93 $\pm$ 0.01 & 0.95 $\pm$ 0.01 & 0.93 $\pm$ 0.01 & 0.93 $\pm$ 0.01 & 0.95 $\pm$ 0.01\\
\hspace{1em}$\text{Full-1Step } I_{n_i} \otimes I_L$ & 0.94 $\pm$ 0.01 & 0.93 $\pm$ 0.01 & 0.95 $\pm$ 0.01 & 0.94 $\pm$ 0.01 & 0.93 $\pm$ 0.01 & 0.95 $\pm$ 0.01\\
\hspace{1em}$\text{pffr (Wild)}$ & 0.98 $\pm$ 0.01 & 0.95 $\pm$ 0.01 & 0.99 $\pm$ 0.01 & 0.99 $\pm$ 0.01 & 0.96 $\pm$ 0.01 & 0.99 $\pm$ 0.01\\
\hspace{1em}$\text{pffr ($z_{1-\alpha/2}$)}$ & 0.65 $\pm$ 0.03 & 0.87 $\pm$ 0.02 & 0.77 $\pm$ 0.02 & 0.65 $\pm$ 0.03 & 0.86 $\pm$ 0.02 & 0.77 $\pm$ 0.02\\
\addlinespace[0.6em]
\multicolumn{7}{l}{\textbf{$N = 50,\; n_i = 100$}}\\
\hspace{1em}$\text{1-Step } {\mathbf{R}}_{\text{Lon}} \otimes {\mathbf{R}}_{\text{Fun}}$ & 0.96 $\pm$ 0.01 & 0.98 $\pm$ 0.01 & 0.98 $\pm$ 0.01 & 0.95 $\pm$ 0.01 & 0.95 $\pm$ 0.01 & 0.96 $\pm$ 0.01\\
\hspace{1em}$\text{1-Step } {\mathbf{R}}_{\text{Lon}} \otimes I_L$ & 0.97 $\pm$ 0.01 & 0.98 $\pm$ 0.01 & 0.98 $\pm$ 0.01 & 0.96 $\pm$ 0.01 & 0.95 $\pm$ 0.01 & 0.96 $\pm$ 0.01\\
\hspace{1em}$\text{1-Step } I_{n_i} \otimes {\mathbf{R}}_{\text{Fun}}$ & 0.97 $\pm$ 0.01 & 0.95 $\pm$ 0.01 & 0.98 $\pm$ 0.01 & 0.95 $\pm$ 0.01 & 0.94 $\pm$ 0.01 & 0.96 $\pm$ 0.01\\
\hspace{1em}$\text{1-Step } I_{n_i} \otimes I_L$ & 0.97 $\pm$ 0.01 & 0.95 $\pm$ 0.01 & 0.98 $\pm$ 0.01 & 0.95 $\pm$ 0.01 & 0.94 $\pm$ 0.01 & 0.96 $\pm$ 0.01\\
\hspace{1em}$\text{Full-1Step } {\mathbf{R}}_{\text{Lon}} \otimes {\mathbf{R}}_{\text{Fun}}$ & 0.97 $\pm$ 0.01 & 0.97 $\pm$ 0.01 & 0.90 $\pm$ 0.02 & 0.95 $\pm$ 0.01 & 0.95 $\pm$ 0.01 & 0.96 $\pm$ 0.01\\
\hspace{1em}$\text{Full-1Step } {\mathbf{R}}_{\text{Lon}} \otimes I_L$ & 0.97 $\pm$ 0.01 & 0.98 $\pm$ 0.01 & 0.79 $\pm$ 0.02 & 0.96 $\pm$ 0.01 & 0.95 $\pm$ 0.01 & 0.96 $\pm$ 0.01\\
\hspace{1em}$\text{Full-1Step } I_{n_i} \otimes {\mathbf{R}}_{\text{Fun}}$ & 0.96 $\pm$ 0.01 & 0.95 $\pm$ 0.01 & 0.98 $\pm$ 0.01 & 0.95 $\pm$ 0.01 & 0.94 $\pm$ 0.01 & 0.96 $\pm$ 0.01\\
\hspace{1em}$\text{Full-1Step } I_{n_i} \otimes I_L$ & 0.97 $\pm$ 0.01 & 0.95 $\pm$ 0.01 & 0.98 $\pm$ 0.01 & 0.95 $\pm$ 0.01 & 0.94 $\pm$ 0.01 & 0.96 $\pm$ 0.01\\
\hspace{1em}$\text{pffr (Wild)}$ & 0.95 $\pm$ 0.01 & 0.93 $\pm$ 0.01 & 0.94 $\pm$ 0.01 & 1.00 $\pm$ 0.00 & 0.93 $\pm$ 0.01 & 0.98 $\pm$ 0.01\\
\hspace{1em}$\text{pffr ($z_{1-\alpha/2}$)}$ & 0.85 $\pm$ 0.02 & 0.93 $\pm$ 0.01 & 0.90 $\pm$ 0.02 & 0.77 $\pm$ 0.02 & 0.93 $\pm$ 0.02 & 0.88 $\pm$ 0.02\\
\addlinespace[0.6em]
\multicolumn{7}{l}{\textbf{$N = 500,\; n_i = 5$}}\\
\hspace{1em}$\text{1-Step } {\mathbf{R}}_{\text{Lon}} \otimes {\mathbf{R}}_{\text{Fun}}$ & 0.95 $\pm$ 0.01 & 0.95 $\pm$ 0.01 & 0.96 $\pm$ 0.01 & 0.95 $\pm$ 0.01 & 0.95 $\pm$ 0.01 & 0.95 $\pm$ 0.01\\
\hspace{1em}$\text{1-Step } {\mathbf{R}}_{\text{Lon}} \otimes I_L$ & 0.96 $\pm$ 0.01 & 0.95 $\pm$ 0.01 & 0.96 $\pm$ 0.01 & 0.96 $\pm$ 0.01 & 0.95 $\pm$ 0.01 & 0.96 $\pm$ 0.01\\
\hspace{1em}$\text{1-Step } I_{n_i} \otimes {\mathbf{R}}_{\text{Fun}}$ & 0.95 $\pm$ 0.01 & 0.95 $\pm$ 0.01 & 0.96 $\pm$ 0.01 & 0.95 $\pm$ 0.01 & 0.95 $\pm$ 0.01 & 0.95 $\pm$ 0.01\\
\hspace{1em}$\text{1-Step } I_{n_i} \otimes I_L$ & 0.95 $\pm$ 0.01 & 0.95 $\pm$ 0.01 & 0.96 $\pm$ 0.01 & 0.96 $\pm$ 0.01 & 0.95 $\pm$ 0.01 & 0.96 $\pm$ 0.01\\
\hspace{1em}$\text{Full-1Step } {\mathbf{R}}_{\text{Lon}} \otimes {\mathbf{R}}_{\text{Fun}}$ & 0.95 $\pm$ 0.01 & 0.95 $\pm$ 0.01 & 0.95 $\pm$ 0.01 & 0.95 $\pm$ 0.01 & 0.95 $\pm$ 0.01 & 0.95 $\pm$ 0.01\\
\hspace{1em}$\text{Full-1Step } {\mathbf{R}}_{\text{Lon}} \otimes I_L$ & 0.96 $\pm$ 0.01 & 0.95 $\pm$ 0.01 & 0.96 $\pm$ 0.01 & 0.96 $\pm$ 0.01 & 0.95 $\pm$ 0.01 & 0.96 $\pm$ 0.01\\
\hspace{1em}$\text{Full-1Step } I_{n_i} \otimes {\mathbf{R}}_{\text{Fun}}$ & 0.95 $\pm$ 0.01 & 0.95 $\pm$ 0.01 & 0.95 $\pm$ 0.01 & 0.95 $\pm$ 0.01 & 0.95 $\pm$ 0.01 & 0.95 $\pm$ 0.01\\
\hspace{1em}$\text{Full-1Step } I_{n_i} \otimes I_L$ & 0.95 $\pm$ 0.01 & 0.95 $\pm$ 0.01 & 0.96 $\pm$ 0.01 & 0.96 $\pm$ 0.01 & 0.95 $\pm$ 0.01 & 0.96 $\pm$ 0.01\\
\hspace{1em}$\text{pffr (Wild)}$ & 1.00 $\pm$ 0.00 & 1.00 $\pm$ 0.00 & 1.00 $\pm$ 0.00 & 1.00 $\pm$ 0.00 & 1.00 $\pm$ 0.00 & 1.00 $\pm$ 0.00\\
\hspace{1em}$\text{pffr ($z_{1-\alpha/2}$)}$ & 0.60 $\pm$ 0.03 & 0.88 $\pm$ 0.02 & 0.71 $\pm$ 0.03 & 0.60 $\pm$ 0.03 & 0.88 $\pm$ 0.02 & 0.71 $\pm$ 0.03\\
\addlinespace[0.6em]
\multicolumn{7}{l}{\textbf{$N = 500,\; n_i = 100$}}\\
\hspace{1em}$\text{1-Step } {\mathbf{R}}_{\text{Lon}} \otimes {\mathbf{R}}_{\text{Fun}}$ & 0.95 $\pm$ 0.01 & 0.95 $\pm$ 0.01 & 0.96 $\pm$ 0.01 & 0.95 $\pm$ 0.01 & 0.94 $\pm$ 0.01 & 0.95 $\pm$ 0.01\\
\hspace{1em}$\text{1-Step } {\mathbf{R}}_{\text{Lon}} \otimes I_L$ & 0.96 $\pm$ 0.01 & 0.95 $\pm$ 0.01 & 0.96 $\pm$ 0.01 & 0.95 $\pm$ 0.01 & 0.94 $\pm$ 0.01 & 0.95 $\pm$ 0.01\\
\hspace{1em}$\text{1-Step } I_{n_i} \otimes {\mathbf{R}}_{\text{Fun}}$ & 0.95 $\pm$ 0.01 & 0.94 $\pm$ 0.01 & 0.96 $\pm$ 0.01 & 0.95 $\pm$ 0.01 & 0.94 $\pm$ 0.01 & 0.95 $\pm$ 0.01\\
\hspace{1em}$\text{1-Step } I_{n_i} \otimes I_L$ & 0.96 $\pm$ 0.01 & 0.94 $\pm$ 0.01 & 0.96 $\pm$ 0.01 & 0.95 $\pm$ 0.01 & 0.94 $\pm$ 0.01 & 0.96 $\pm$ 0.01\\
\hspace{1em}$\text{Full-1Step } {\mathbf{R}}_{\text{Lon}} \otimes {\mathbf{R}}_{\text{Fun}}$ & 0.95 $\pm$ 0.01 & 0.95 $\pm$ 0.01 & 0.96 $\pm$ 0.01 & 0.95 $\pm$ 0.01 & 0.94 $\pm$ 0.01 & 0.95 $\pm$ 0.01\\
\hspace{1em}$\text{Full-1Step } {\mathbf{R}}_{\text{Lon}} \otimes I_L$ & 0.96 $\pm$ 0.01 & 0.95 $\pm$ 0.01 & 0.96 $\pm$ 0.01 & 0.95 $\pm$ 0.01 & 0.94 $\pm$ 0.01 & 0.95 $\pm$ 0.01\\
\hspace{1em}$\text{Full-1Step } I_{n_i} \otimes {\mathbf{R}}_{\text{Fun}}$ & 0.95 $\pm$ 0.01 & 0.94 $\pm$ 0.01 & 0.96 $\pm$ 0.01 & 0.95 $\pm$ 0.01 & 0.94 $\pm$ 0.01 & 0.95 $\pm$ 0.01\\
\hspace{1em}$\text{Full-1Step } I_{n_i} \otimes I_L$ & 0.96 $\pm$ 0.01 & 0.94 $\pm$ 0.01 & 0.96 $\pm$ 0.01 & 0.95 $\pm$ 0.01 & 0.94 $\pm$ 0.01 & 0.95 $\pm$ 0.01\\
\hspace{1em}$\text{pffr (Wild)}$ & 0.99 $\pm$ 0.00 & 0.94 $\pm$ 0.01 & 0.97 $\pm$ 0.01 & 1.00 $\pm$ 0.00 & 0.97 $\pm$ 0.01 & 1.00 $\pm$ 0.00\\
\hspace{1em}$\text{pffr ($z_{1-\alpha/2}$)}$ & 0.80 $\pm$ 0.02 & 0.93 $\pm$ 0.02 & 0.88 $\pm$ 0.02 & 0.74 $\pm$ 0.03 & 0.92 $\pm$ 0.02 & 0.87 $\pm$ 0.02\\
\bottomrule
\end{tabular}}
\end{table}

\begin{table}[!h]
\centering
\caption{\label{tab:pt_width} \footnotesize Relative pointwise CI width (mean $\pm$ SE) vs.\ pffr (Wild). SE$=0.00$ indicates a value $<0.01$. We denote $UB^{(r)}(s)$ and $LB^{(r)}(s)$ and $UB_\text{pffr}^{(r)}(s)/LB_\text{pffr}^{(r)}(s)$ as the upper/lower bounds of the CIs (at $s$) of the indicated method and pffr, respectively. Below we report the average ratio $\frac{1}{(q+1)|\mathcal{S}|}\sum_{r=0}^q\sum_{s \in \mathcal{S}}\frac{UB^{(r)}(s) -LB(s)^{(r)}}{UB^{(r)}_\text{pffr}(s) -LB^{(r)}_\text{pffr}(s)}$ across 300 simulation replicates. Values $<1$ indicate narrower 95\% CIs. Values with $*$ indicate extreme outliers (from poor estimates) were removed from the average of that cell to avoid skewing results.}
\centering
\resizebox{\ifdim\width>\linewidth\linewidth\else\width\fi}{!}{
\fontsize{9}{11}\selectfont
\begin{tabular}[t]{>{\raggedright\arraybackslash}p{4.8cm}>{\raggedright\arraybackslash}p{2.1cm}>{\raggedright\arraybackslash}p{2.1cm}>{\raggedright\arraybackslash}p{2.1cm}>{\raggedright\arraybackslash}p{2.1cm}>{\raggedright\arraybackslash}p{2.1cm}>{\raggedright\arraybackslash}p{2.1cm}}
\toprule
\multicolumn{1}{c}{ } & \multicolumn{3}{c}{Exchangeable} & \multicolumn{3}{c}{AR(1)} \\
\cmidrule(l{3pt}r{3pt}){2-4} \cmidrule(l{3pt}r{3pt}){5-7}
Method & Gaussian & Poisson & Binomial & Gaussian & Poisson & Binomial\\
\midrule
\addlinespace[0em]
\multicolumn{7}{l}{\textbf{$N = 25,\; n_i = 5$}}\\
\hspace{1em}$\text{1-Step } {\mathbf{R}}_{\text{Lon}} \otimes {\mathbf{R}}_{\text{Fun}}$ & 0.75 $\pm$ 0.01 & 0.83 $\pm$ 0.00 & 0.85 $\pm$ 0.01 & 0.75 $\pm$ 0.01 & 0.83 $\pm$ 0.00 & 0.82 $\pm$ 0.01\\
\hspace{1em}$\text{1-Step } {\mathbf{R}}_{\text{Lon}} \otimes I_L$ & 0.78 $\pm$ 0.01 & 0.87 $\pm$ 0.01 & 0.84 $\pm$ 0.01 & 0.77 $\pm$ 0.01 & 0.83 $\pm$ 0.00 & 0.80 $\pm$ 0.01\\
\hspace{1em}$\text{1-Step } I_{n_i} \otimes {\mathbf{R}}_{\text{Fun}}$ & 0.83 $\pm$ 0.01 & 0.92 $\pm$ 0.00 & 0.94 $\pm$ 0.01 & 0.80 $\pm$ 0.01 & 0.91 $\pm$ 0.00 & 0.87 $\pm$ 0.01\\
\hspace{1em}$\text{1-Step } I_{n_i} \otimes I_L$ & 0.86 $\pm$ 0.01 & 0.92 $\pm$ 0.00 & 0.92 $\pm$ 0.01 & 0.83 $\pm$ 0.01 & 0.90 $\pm$ 0.00 & 0.86 $\pm$ 0.01\\
\hspace{1em}$\text{Full-1Step } {\mathbf{R}}_{\text{Lon}} \otimes {\mathbf{R}}_{\text{Fun}}$ & 0.75 $\pm$ 0.01 & 0.83 $\pm$ 0.00 & 0.90 $\pm$ 0.04 & 0.75 $\pm$ 0.01 & 0.83 $\pm$ 0.00 & 0.82 $\pm$ 0.01\\
\hspace{1em}$\text{Full-1Step } {\mathbf{R}}_{\text{Lon}} \otimes I_L$ & 0.78 $\pm$ 0.01 & 0.90 $\pm$ 0.02 & 1.00 $\pm$ 0.05 & 0.77 $\pm$ 0.01 & 0.83 $\pm$ 0.00 & 0.81 $\pm$ 0.01\\
\hspace{1em}$\text{Full-1Step } I_{n_i} \otimes {\mathbf{R}}_{\text{Fun}}$ & 0.83 $\pm$ 0.01 & 0.92 $\pm$ 0.00 & 0.94 $\pm$ 0.01 & 0.80 $\pm$ 0.01 & 0.91 $\pm$ 0.00 & 0.87 $\pm$ 0.01\\
\hspace{1em}$\text{Full-1Step } I_{n_i} \otimes I_L$ & 0.86 $\pm$ 0.01 & 0.92 $\pm$ 0.00 & 0.93 $\pm$ 0.01 & 0.83 $\pm$ 0.01 & 0.90 $\pm$ 0.00 & 0.86 $\pm$ 0.01\\
\addlinespace[0.6em]
\multicolumn{7}{l}{\textbf{$N = 25,\; n_i = 100$}}\\
\hspace{1em}$\text{1-Step } {\mathbf{R}}_{\text{Lon}} \otimes {\mathbf{R}}_{\text{Fun}}$ & 0.79 $\pm$ 0.01 & 0.78 $\pm$ 0.01 & 1.18 $\pm$ 0.01 & 0.51 $\pm$ 0.00 & 0.87 $\pm$ 0.00 & 0.82 $\pm$ 0.00\\
\hspace{1em}$\text{1-Step } {\mathbf{R}}_{\text{Lon}} \otimes I_L$ & 0.83 $\pm$ 0.01 & 0.82 $\pm$ 0.01 & 1.18 $\pm$ 0.01 & 0.50 $\pm$ 0.00 & 0.87 $\pm$ 0.00 & 0.81 $\pm$ 0.00\\
\hspace{1em}$\text{1-Step } I_{n_i} \otimes {\mathbf{R}}_{\text{Fun}}$ & 1.00 $\pm$ 0.01 & 1.06 $\pm$ 0.00 & 1.39 $\pm$ 0.01 & 0.57 $\pm$ 0.00 & 1.01 $\pm$ 0.00 & 0.90 $\pm$ 0.00\\
\hspace{1em}$\text{1-Step } I_{n_i} \otimes I_L$ & 1.02 $\pm$ 0.01 & 1.05 $\pm$ 0.00 & 1.35 $\pm$ 0.01 & 0.56 $\pm$ 0.00 & 1.00 $\pm$ 0.00 & 0.89 $\pm$ 0.00\\
\hspace{1em}$\text{Full-1Step } {\mathbf{R}}_{\text{Lon}} \otimes {\mathbf{R}}_{\text{Fun}}$ & 0.79 $\pm$ 0.01 & 0.78 $\pm$ 0.01 & 2.23 $\pm$ 0.21 & 0.51 $\pm$ 0.00 & 0.87 $\pm$ 0.00 & 0.82 $\pm$ 0.00\\
\hspace{1em}$\text{Full-1Step } {\mathbf{R}}_{\text{Lon}} \otimes I_L$ & 0.83 $\pm$ 0.01 & 0.96 $\pm$ 0.11 & 5.29 $\pm$ 0.58$^*$ & 0.50 $\pm$ 0.00 & 0.87 $\pm$ 0.00 & 0.81 $\pm$ 0.00\\
\hspace{1em}$\text{Full-1Step } I_{n_i} \otimes {\mathbf{R}}_{\text{Fun}}$ & 1.00 $\pm$ 0.01 & 1.06 $\pm$ 0.00 & 1.39 $\pm$ 0.01 & 0.57 $\pm$ 0.00 & 1.01 $\pm$ 0.00 & 0.90 $\pm$ 0.00\\
\hspace{1em}$\text{Full-1Step } I_{n_i} \otimes I_L$ & 1.02 $\pm$ 0.01 & 1.06 $\pm$ 0.00 & 1.37 $\pm$ 0.01 & 0.56 $\pm$ 0.00 & 1.00 $\pm$ 0.00 & 0.89 $\pm$ 0.00\\
\addlinespace[0.6em]
\multicolumn{7}{l}{\textbf{$N = 50,\; n_i = 5$}}\\
\hspace{1em}$\text{1-Step } {\mathbf{R}}_{\text{Lon}} \otimes {\mathbf{R}}_{\text{Fun}}$ & 0.62 $\pm$ 0.00 & 0.72 $\pm$ 0.00 & 0.65 $\pm$ 0.00 & 0.60 $\pm$ 0.00 & 0.72 $\pm$ 0.00 & 0.62 $\pm$ 0.00\\
\hspace{1em}$\text{1-Step } {\mathbf{R}}_{\text{Lon}} \otimes I_L$ & 0.63 $\pm$ 0.01 & 0.74 $\pm$ 0.00 & 0.64 $\pm$ 0.00 & 0.61 $\pm$ 0.01 & 0.72 $\pm$ 0.00 & 0.61 $\pm$ 0.00\\
\hspace{1em}$\text{1-Step } I_{n_i} \otimes {\mathbf{R}}_{\text{Fun}}$ & 0.70 $\pm$ 0.00 & 0.83 $\pm$ 0.00 & 0.72 $\pm$ 0.00 & 0.66 $\pm$ 0.00 & 0.80 $\pm$ 0.00 & 0.66 $\pm$ 0.00\\
\hspace{1em}$\text{1-Step } I_{n_i} \otimes I_L$ & 0.71 $\pm$ 0.01 & 0.83 $\pm$ 0.00 & 0.71 $\pm$ 0.00 & 0.68 $\pm$ 0.01 & 0.80 $\pm$ 0.00 & 0.65 $\pm$ 0.00\\
\hspace{1em}$\text{Full-1Step } {\mathbf{R}}_{\text{Lon}} \otimes {\mathbf{R}}_{\text{Fun}}$ & 0.62 $\pm$ 0.00 & 0.72 $\pm$ 0.00 & 0.65 $\pm$ 0.00 & 0.60 $\pm$ 0.00 & 0.72 $\pm$ 0.00 & 0.62 $\pm$ 0.00\\
\hspace{1em}$\text{Full-1Step } {\mathbf{R}}_{\text{Lon}} \otimes I_L$ & 0.63 $\pm$ 0.01 & 0.76 $\pm$ 0.01 & 0.74 $\pm$ 0.07 & 0.61 $\pm$ 0.01 & 0.72 $\pm$ 0.00 & 0.61 $\pm$ 0.00\\
\hspace{1em}$\text{Full-1Step } I_{n_i} \otimes {\mathbf{R}}_{\text{Fun}}$ & 0.70 $\pm$ 0.00 & 0.83 $\pm$ 0.00 & 0.72 $\pm$ 0.00 & 0.66 $\pm$ 0.00 & 0.80 $\pm$ 0.00 & 0.66 $\pm$ 0.00\\
\hspace{1em}$\text{Full-1Step } I_{n_i} \otimes I_L$ & 0.71 $\pm$ 0.01 & 0.83 $\pm$ 0.00 & 0.71 $\pm$ 0.00 & 0.68 $\pm$ 0.01 & 0.80 $\pm$ 0.00 & 0.66 $\pm$ 0.00\\
\addlinespace[0.6em]
\multicolumn{7}{l}{\textbf{$N = 50,\; n_i = 100$}}\\
\hspace{1em}$\text{1-Step } {\mathbf{R}}_{\text{Lon}} \otimes {\mathbf{R}}_{\text{Fun}}$ & 0.75 $\pm$ 0.00 & 0.70 $\pm$ 0.00 & 1.07 $\pm$ 0.01 & 0.37 $\pm$ 0.00 & 0.83 $\pm$ 0.00 & 0.69 $\pm$ 0.00\\
\hspace{1em}$\text{1-Step } {\mathbf{R}}_{\text{Lon}} \otimes I_L$ & 0.77 $\pm$ 0.00 & 0.73 $\pm$ 0.01 & 1.05 $\pm$ 0.01 & 0.37 $\pm$ 0.00 & 0.83 $\pm$ 0.00 & 0.68 $\pm$ 0.00\\
\hspace{1em}$\text{1-Step } I_{n_i} \otimes {\mathbf{R}}_{\text{Fun}}$ & 0.85 $\pm$ 0.00 & 1.02 $\pm$ 0.00 & 1.18 $\pm$ 0.01 & 0.42 $\pm$ 0.00 & 0.99 $\pm$ 0.00 & 0.76 $\pm$ 0.00\\
\hspace{1em}$\text{1-Step } I_{n_i} \otimes I_L$ & 0.86 $\pm$ 0.00 & 1.03 $\pm$ 0.00 & 1.16 $\pm$ 0.01 & 0.42 $\pm$ 0.00 & 0.99 $\pm$ 0.00 & 0.75 $\pm$ 0.00\\
\hspace{1em}$\text{Full-1Step } {\mathbf{R}}_{\text{Lon}} \otimes {\mathbf{R}}_{\text{Fun}}$ & 0.75 $\pm$ 0.00 & 0.70 $\pm$ 0.00 & 1.42 $\pm$ 0.10 & 0.37 $\pm$ 0.00 & 0.83 $\pm$ 0.00 & 0.69 $\pm$ 0.00\\
\hspace{1em}$\text{Full-1Step } {\mathbf{R}}_{\text{Lon}} \otimes I_L$ & 0.77 $\pm$ 0.00 & 0.74 $\pm$ 0.01 & 4.05 $\pm$ 0.55 & 0.37 $\pm$ 0.00 & 0.83 $\pm$ 0.00 & 0.68 $\pm$ 0.00\\
\hspace{1em}$\text{Full-1Step } I_{n_i} \otimes {\mathbf{R}}_{\text{Fun}}$ & 0.85 $\pm$ 0.00 & 1.02 $\pm$ 0.00 & 1.19 $\pm$ 0.01 & 0.42 $\pm$ 0.00 & 0.99 $\pm$ 0.00 & 0.76 $\pm$ 0.00\\
\hspace{1em}$\text{Full-1Step } I_{n_i} \otimes I_L$ & 0.86 $\pm$ 0.00 & 1.03 $\pm$ 0.00 & 1.16 $\pm$ 0.01 & 0.42 $\pm$ 0.00 & 0.99 $\pm$ 0.00 & 0.75 $\pm$ 0.00\\
\addlinespace[0.6em]
\multicolumn{7}{l}{\textbf{$N = 500,\; n_i = 5$}}\\
\hspace{1em}$\text{1-Step } {\mathbf{R}}_{\text{Lon}} \otimes {\mathbf{R}}_{\text{Fun}}$ & 0.25 $\pm$ 0.00 & 0.36 $\pm$ 0.00 & 0.22 $\pm$ 0.00 & 0.24 $\pm$ 0.00 & 0.35 $\pm$ 0.00 & 0.21 $\pm$ 0.00\\
\hspace{1em}$\text{1-Step } {\mathbf{R}}_{\text{Lon}} \otimes I_L$ & 0.25 $\pm$ 0.00 & 0.36 $\pm$ 0.00 & 0.22 $\pm$ 0.00 & 0.24 $\pm$ 0.00 & 0.36 $\pm$ 0.00 & 0.21 $\pm$ 0.00\\
\hspace{1em}$\text{1-Step } I_{n_i} \otimes {\mathbf{R}}_{\text{Fun}}$ & 0.29 $\pm$ 0.00 & 0.42 $\pm$ 0.00 & 0.25 $\pm$ 0.00 & 0.27 $\pm$ 0.00 & 0.40 $\pm$ 0.00 & 0.23 $\pm$ 0.00\\
\hspace{1em}$\text{1-Step } I_{n_i} \otimes I_L$ & 0.28 $\pm$ 0.00 & 0.42 $\pm$ 0.00 & 0.25 $\pm$ 0.00 & 0.26 $\pm$ 0.00 & 0.40 $\pm$ 0.00 & 0.23 $\pm$ 0.00\\
\hspace{1em}$\text{Full-1Step } {\mathbf{R}}_{\text{Lon}} \otimes {\mathbf{R}}_{\text{Fun}}$ & 0.25 $\pm$ 0.00 & 0.36 $\pm$ 0.00 & 0.22 $\pm$ 0.00 & 0.24 $\pm$ 0.00 & 0.35 $\pm$ 0.00 & 0.21 $\pm$ 0.00\\
\hspace{1em}$\text{Full-1Step } {\mathbf{R}}_{\text{Lon}} \otimes I_L$ & 0.25 $\pm$ 0.00 & 0.36 $\pm$ 0.00 & 0.22 $\pm$ 0.00 & 0.24 $\pm$ 0.00 & 0.36 $\pm$ 0.00 & 0.21 $\pm$ 0.00\\
\hspace{1em}$\text{Full-1Step } I_{n_i} \otimes {\mathbf{R}}_{\text{Fun}}$ & 0.29 $\pm$ 0.00 & 0.42 $\pm$ 0.00 & 0.25 $\pm$ 0.00 & 0.27 $\pm$ 0.00 & 0.40 $\pm$ 0.00 & 0.23 $\pm$ 0.00\\
\hspace{1em}$\text{Full-1Step } I_{n_i} \otimes I_L$ & 0.28 $\pm$ 0.00 & 0.42 $\pm$ 0.00 & 0.25 $\pm$ 0.00 & 0.26 $\pm$ 0.00 & 0.40 $\pm$ 0.00 & 0.23 $\pm$ 0.00\\
\addlinespace[0.6em]
\multicolumn{7}{l}{\textbf{$N = 500,\; n_i = 100$}}\\
\hspace{1em}$\text{1-Step } {\mathbf{R}}_{\text{Lon}} \otimes {\mathbf{R}}_{\text{Fun}}$ & 0.43 $\pm$ 0.00 & 0.57 $\pm$ 0.00 & 0.77 $\pm$ 0.00 & 0.12 $\pm$ 0.00 & 0.69 $\pm$ 0.00 & 0.27 $\pm$ 0.00\\
\hspace{1em}$\text{1-Step } {\mathbf{R}}_{\text{Lon}} \otimes I_L$ & 0.43 $\pm$ 0.00 & 0.58 $\pm$ 0.00 & 0.75 $\pm$ 0.00 & 0.12 $\pm$ 0.00 & 0.69 $\pm$ 0.00 & 0.27 $\pm$ 0.00\\
\hspace{1em}$\text{1-Step } I_{n_i} \otimes {\mathbf{R}}_{\text{Fun}}$ & 0.46 $\pm$ 0.00 & 0.96 $\pm$ 0.00 & 0.76 $\pm$ 0.00 & 0.14 $\pm$ 0.00 & 0.84 $\pm$ 0.00 & 0.30 $\pm$ 0.00\\
\hspace{1em}$\text{1-Step } I_{n_i} \otimes I_L$ & 0.46 $\pm$ 0.00 & 0.97 $\pm$ 0.00 & 0.75 $\pm$ 0.00 & 0.13 $\pm$ 0.00 & 0.85 $\pm$ 0.00 & 0.30 $\pm$ 0.00\\
\hspace{1em}$\text{Full-1Step } {\mathbf{R}}_{\text{Lon}} \otimes {\mathbf{R}}_{\text{Fun}}$ & 0.43 $\pm$ 0.00 & 0.57 $\pm$ 0.00 & 0.77 $\pm$ 0.00 & 0.12 $\pm$ 0.00 & 0.69 $\pm$ 0.00 & 0.27 $\pm$ 0.00\\
\hspace{1em}$\text{Full-1Step } {\mathbf{R}}_{\text{Lon}} \otimes I_L$ & 0.43 $\pm$ 0.00 & 0.58 $\pm$ 0.00 & 0.75 $\pm$ 0.00 & 0.12 $\pm$ 0.00 & 0.69 $\pm$ 0.00 & 0.27 $\pm$ 0.00\\
\hspace{1em}$\text{Full-1Step } I_{n_i} \otimes {\mathbf{R}}_{\text{Fun}}$ & 0.46 $\pm$ 0.00 & 0.96 $\pm$ 0.00 & 0.76 $\pm$ 0.00 & 0.14 $\pm$ 0.00 & 0.84 $\pm$ 0.00 & 0.30 $\pm$ 0.00\\
\hspace{1em}$\text{Full-1Step } I_{n_i} \otimes I_L$ & 0.46 $\pm$ 0.00 & 0.97 $\pm$ 0.00 & 0.75 $\pm$ 0.00 & 0.13 $\pm$ 0.00 & 0.85 $\pm$ 0.00 & 0.30 $\pm$ 0.00\\
\bottomrule
\end{tabular}}
\end{table}

\begin{table}[!h]
\centering
\caption{\label{tab:time} \footnotesize Computation time in seconds (mean $\pm$ SE) averaged across 300 simulation replicates. SE$=0.00$ indicates a value $<0.01$.}
\centering
\resizebox{\ifdim\width>\linewidth\linewidth\else\width\fi}{!}{
\fontsize{9}{11}\selectfont
\begin{tabular}[t]{>{\raggedright\arraybackslash}p{4.8cm}>{\raggedright\arraybackslash}p{2.1cm}>{\raggedright\arraybackslash}p{2.1cm}>{\raggedright\arraybackslash}p{2.1cm}>{\raggedright\arraybackslash}p{2.1cm}>{\raggedright\arraybackslash}p{2.1cm}>{\raggedright\arraybackslash}p{2.1cm}}
\toprule
\multicolumn{1}{c}{ } & \multicolumn{3}{c}{Exchangeable} & \multicolumn{3}{c}{AR(1)} \\
\cmidrule(l{3pt}r{3pt}){2-4} \cmidrule(l{3pt}r{3pt}){5-7}
Method & Gaussian & Poisson & Binomial & Gaussian & Poisson & Binomial\\
\midrule
\addlinespace[0em]
\multicolumn{7}{l}{\textbf{$N = 25,\; n_i = 5$}}\\
\hspace{1em}$\text{1-Step } {\mathbf{R}}_{\text{Lon}} \otimes {\mathbf{R}}_{\text{Fun}}$ & 2.22 $\pm$ 0.05 & 2.47 $\pm$ 0.05 & 2.49 $\pm$ 0.05 & 2.12 $\pm$ 0.01 & 2.25 $\pm$ 0.01 & 2.40 $\pm$ 0.02\\
\hspace{1em}$\text{1-Step } {\mathbf{R}}_{\text{Lon}} \otimes I_L$ & 4.57 $\pm$ 0.12 & 4.71 $\pm$ 0.08 & 4.75 $\pm$ 0.08 & 4.32 $\pm$ 0.02 & 4.67 $\pm$ 0.04 & 5.01 $\pm$ 0.06\\
\hspace{1em}$\text{1-Step } I_{n_i} \otimes {\mathbf{R}}_{\text{Fun}}$ & 1.16 $\pm$ 0.02 & 1.37 $\pm$ 0.03 & 1.38 $\pm$ 0.03 & 1.08 $\pm$ 0.00 & 1.24 $\pm$ 0.01 & 1.33 $\pm$ 0.01\\
\hspace{1em}$\text{1-Step } I_{n_i} \otimes I_L$ & 0.66 $\pm$ 0.01 & 0.83 $\pm$ 0.02 & 0.84 $\pm$ 0.02 & 0.64 $\pm$ 0.00 & 0.74 $\pm$ 0.00 & 0.80 $\pm$ 0.01\\
\hspace{1em}$\text{Full-1Step } {\mathbf{R}}_{\text{Lon}} \otimes {\mathbf{R}}_{\text{Fun}}$ & 7.74 $\pm$ 0.17 & 7.94 $\pm$ 0.18 & 12.56 $\pm$ 0.53 & 6.96 $\pm$ 0.05 & 8.07 $\pm$ 0.28 & 11.95 $\pm$ 0.50\\
\hspace{1em}$\text{Full-1Step } {\mathbf{R}}_{\text{Lon}} \otimes I_L$ & 12.98 $\pm$ 0.28 & 35.67 $\pm$ 2.85 & 32.97 $\pm$ 2.17 & 11.99 $\pm$ 0.07 & 17.72 $\pm$ 0.72 & 24.82 $\pm$ 0.83\\
\hspace{1em}$\text{Full-1Step } I_{n_i} \otimes {\mathbf{R}}_{\text{Fun}}$ & 3.37 $\pm$ 0.07 & 3.22 $\pm$ 0.07 & 4.60 $\pm$ 0.20 & 3.08 $\pm$ 0.02 & 4.37 $\pm$ 0.37 & 6.17 $\pm$ 0.52\\
\hspace{1em}$\text{Full-1Step } I_{n_i} \otimes I_L$ & 0.87 $\pm$ 0.02 & 1.08 $\pm$ 0.02 & 1.33 $\pm$ 0.04 & 0.84 $\pm$ 0.00 & 4.08 $\pm$ 0.70 & 5.03 $\pm$ 0.90\\
\hspace{1em}$\text{pffr (Wild)}$ & 2.59 $\pm$ 0.05 & 2.84 $\pm$ 0.06 & 2.78 $\pm$ 0.05 & 2.42 $\pm$ 0.01 & 2.61 $\pm$ 0.02 & 2.63 $\pm$ 0.02\\
\hspace{1em}$\text{pffr ($z_{1-\alpha/2}$)}$ & 0.36 $\pm$ 0.01 & 0.48 $\pm$ 0.01 & 0.56 $\pm$ 0.01 & 0.35 $\pm$ 0.00 & 0.45 $\pm$ 0.01 & 0.56 $\pm$ 0.01\\
\addlinespace[0.6em]
\multicolumn{7}{l}{\textbf{$N = 25,\; n_i = 100$}}\\
\hspace{1em}$\text{1-Step } {\mathbf{R}}_{\text{Lon}} \otimes {\mathbf{R}}_{\text{Fun}}$ & 6.94 $\pm$ 0.02 & 8.50 $\pm$ 0.02 & 8.65 $\pm$ 0.08 & 8.96 $\pm$ 0.17 & 8.29 $\pm$ 0.04 & 8.51 $\pm$ 0.06\\
\hspace{1em}$\text{1-Step } {\mathbf{R}}_{\text{Lon}} \otimes I_L$ & 7.43 $\pm$ 0.02 & 9.15 $\pm$ 0.02 & 9.38 $\pm$ 0.08 & 9.91 $\pm$ 0.19 & 9.03 $\pm$ 0.04 & 9.27 $\pm$ 0.06\\
\hspace{1em}$\text{1-Step } I_{n_i} \otimes {\mathbf{R}}_{\text{Fun}}$ & 10.71 $\pm$ 0.02 & 12.13 $\pm$ 0.02 & 12.50 $\pm$ 0.11 & 14.09 $\pm$ 0.26 & 11.86 $\pm$ 0.04 & 12.27 $\pm$ 0.08\\
\hspace{1em}$\text{1-Step } I_{n_i} \otimes I_L$ & 4.21 $\pm$ 0.02 & 5.40 $\pm$ 0.02 & 5.80 $\pm$ 0.05 & 4.59 $\pm$ 0.08 & 5.15 $\pm$ 0.03 & 5.68 $\pm$ 0.04\\
\hspace{1em}$\text{Full-1Step } {\mathbf{R}}_{\text{Lon}} \otimes {\mathbf{R}}_{\text{Fun}}$ & 20.78 $\pm$ 0.15 & 22.20 $\pm$ 0.20 & 71.49 $\pm$ 3.56 & 17.17 $\pm$ 0.35 & 15.42 $\pm$ 0.08 & 18.26 $\pm$ 0.15\\
\hspace{1em}$\text{Full-1Step } {\mathbf{R}}_{\text{Lon}} \otimes I_L$ & 18.22 $\pm$ 0.12 & 45.69 $\pm$ 3.03 & 117.68 $\pm$ 5.27 & 19.08 $\pm$ 0.72 & 17.58 $\pm$ 0.62 & 20.11 $\pm$ 0.15\\
\hspace{1em}$\text{Full-1Step } I_{n_i} \otimes {\mathbf{R}}_{\text{Fun}}$ & 37.38 $\pm$ 0.28 & 32.62 $\pm$ 0.21 & 46.32 $\pm$ 1.04 & 29.05 $\pm$ 0.56 & 26.49 $\pm$ 1.37 & 30.41 $\pm$ 0.25\\
\hspace{1em}$\text{Full-1Step } I_{n_i} \otimes I_L$ & 4.63 $\pm$ 0.03 & 6.02 $\pm$ 0.03 & 7.94 $\pm$ 0.14 & 4.84 $\pm$ 0.09 & 5.53 $\pm$ 0.03 & 6.79 $\pm$ 0.05\\
\hspace{1em}$\text{pffr (Wild)}$ & 6.69 $\pm$ 0.01 & 7.11 $\pm$ 0.02 & 6.90 $\pm$ 0.06 & 8.24 $\pm$ 0.14 & 6.95 $\pm$ 0.03 & 6.74 $\pm$ 0.04\\
\hspace{1em}$\text{pffr ($z_{1-\alpha/2}$)}$ & 2.66 $\pm$ 0.02 & 3.00 $\pm$ 0.01 & 2.72 $\pm$ 0.02 & 3.03 $\pm$ 0.06 & 2.92 $\pm$ 0.02 & 2.70 $\pm$ 0.02\\
\addlinespace[0.6em]
\multicolumn{7}{l}{\textbf{$N = 50,\; n_i = 5$}}\\
\hspace{1em}$\text{1-Step } {\mathbf{R}}_{\text{Lon}} \otimes {\mathbf{R}}_{\text{Fun}}$ & 3.57 $\pm$ 0.06 & 3.73 $\pm$ 0.03 & 3.50 $\pm$ 0.01 & 3.40 $\pm$ 0.01 & 3.83 $\pm$ 0.02 & 3.83 $\pm$ 0.02\\
\hspace{1em}$\text{1-Step } {\mathbf{R}}_{\text{Lon}} \otimes I_L$ & 7.71 $\pm$ 0.12 & 7.54 $\pm$ 0.10 & 7.55 $\pm$ 0.02 & 7.51 $\pm$ 0.03 & 7.98 $\pm$ 0.06 & 7.87 $\pm$ 0.04\\
\hspace{1em}$\text{1-Step } I_{n_i} \otimes {\mathbf{R}}_{\text{Fun}}$ & 1.80 $\pm$ 0.03 & 1.97 $\pm$ 0.02 & 2.03 $\pm$ 0.01 & 1.71 $\pm$ 0.01 & 1.93 $\pm$ 0.02 & 1.89 $\pm$ 0.01\\
\hspace{1em}$\text{1-Step } I_{n_i} \otimes I_L$ & 0.84 $\pm$ 0.01 & 0.94 $\pm$ 0.01 & 1.00 $\pm$ 0.00 & 0.80 $\pm$ 0.00 & 0.96 $\pm$ 0.00 & 1.20 $\pm$ 0.01\\
\hspace{1em}$\text{Full-1Step } {\mathbf{R}}_{\text{Lon}} \otimes {\mathbf{R}}_{\text{Fun}}$ & 11.94 $\pm$ 0.22 & 11.80 $\pm$ 0.16 & 16.00 $\pm$ 0.28 & 10.61 $\pm$ 0.06 & 11.65 $\pm$ 0.09 & 16.18 $\pm$ 0.30\\
\hspace{1em}$\text{Full-1Step } {\mathbf{R}}_{\text{Lon}} \otimes I_L$ & 20.36 $\pm$ 0.31 & 43.17 $\pm$ 3.27 & 37.27 $\pm$ 1.85 & 19.67 $\pm$ 0.13 & 24.21 $\pm$ 0.86 & 31.17 $\pm$ 0.28\\
\hspace{1em}$\text{Full-1Step } I_{n_i} \otimes {\mathbf{R}}_{\text{Fun}}$ & 5.23 $\pm$ 0.09 & 4.68 $\pm$ 0.06 & 7.16 $\pm$ 0.09 & 4.53 $\pm$ 0.03 & 4.63 $\pm$ 0.03 & 6.99 $\pm$ 0.10\\
\hspace{1em}$\text{Full-1Step } I_{n_i} \otimes I_L$ & 1.09 $\pm$ 0.02 & 1.20 $\pm$ 0.01 & 1.43 $\pm$ 0.01 & 1.02 $\pm$ 0.00 & 1.17 $\pm$ 0.00 & 1.42 $\pm$ 0.01\\
\hspace{1em}$\text{pffr (Wild)}$ & 2.98 $\pm$ 0.05 & 2.91 $\pm$ 0.03 & 2.85 $\pm$ 0.01 & 2.79 $\pm$ 0.01 & 2.99 $\pm$ 0.03 & 2.84 $\pm$ 0.01\\
\hspace{1em}$\text{pffr ($z_{1-\alpha/2}$)}$ & 0.72 $\pm$ 0.01 & 0.66 $\pm$ 0.01 & 0.51 $\pm$ 0.00 & 0.69 $\pm$ 0.00 & 0.69 $\pm$ 0.01 & 0.51 $\pm$ 0.01\\
\addlinespace[0.6em]
\multicolumn{7}{l}{\textbf{$N = 50,\; n_i = 100$}}\\
\hspace{1em}$\text{1-Step } {\mathbf{R}}_{\text{Lon}} \otimes {\mathbf{R}}_{\text{Fun}}$ & 13.07 $\pm$ 0.07 & 16.17 $\pm$ 0.03 & 17.13 $\pm$ 0.24 & 13.65 $\pm$ 0.20 & 19.73 $\pm$ 0.55 & 20.65 $\pm$ 0.40\\
\hspace{1em}$\text{1-Step } {\mathbf{R}}_{\text{Lon}} \otimes I_L$ & 13.93 $\pm$ 0.07 & 17.90 $\pm$ 0.04 & 18.27 $\pm$ 0.25 & 14.44 $\pm$ 0.22 & 21.12 $\pm$ 0.50 & 21.52 $\pm$ 0.40\\
\hspace{1em}$\text{1-Step } I_{n_i} \otimes {\mathbf{R}}_{\text{Fun}}$ & 19.82 $\pm$ 0.06 & 23.36 $\pm$ 0.04 & 24.25 $\pm$ 0.34 & 20.96 $\pm$ 0.30 & 28.06 $\pm$ 0.65 & 29.24 $\pm$ 0.54\\
\hspace{1em}$\text{1-Step } I_{n_i} \otimes I_L$ & 7.23 $\pm$ 0.06 & 10.67 $\pm$ 0.04 & 10.58 $\pm$ 0.14 & 6.65 $\pm$ 0.09 & 11.41 $\pm$ 0.26 & 12.51 $\pm$ 0.21\\
\hspace{1em}$\text{Full-1Step } {\mathbf{R}}_{\text{Lon}} \otimes {\mathbf{R}}_{\text{Fun}}$ & 31.44 $\pm$ 0.17 & 38.99 $\pm$ 0.29 & 94.21 $\pm$ 4.64 & 24.47 $\pm$ 0.40 & 35.64 $\pm$ 0.86 & 41.81 $\pm$ 0.93\\
\hspace{1em}$\text{Full-1Step } {\mathbf{R}}_{\text{Lon}} \otimes I_L$ & 30.35 $\pm$ 0.16 & 68.70 $\pm$ 4.12 & 167.87 $\pm$ 8.92 & 26.07 $\pm$ 0.38 & 39.29 $\pm$ 0.87 & 51.02 $\pm$ 3.07\\
\hspace{1em}$\text{Full-1Step } I_{n_i} \otimes {\mathbf{R}}_{\text{Fun}}$ & 56.71 $\pm$ 0.33 & 59.48 $\pm$ 0.32 & 81.54 $\pm$ 1.35 & 44.86 $\pm$ 2.11 & 55.77 $\pm$ 1.25 & 70.31 $\pm$ 1.53\\
\hspace{1em}$\text{Full-1Step } I_{n_i} \otimes I_L$ & 7.95 $\pm$ 0.05 & 10.88 $\pm$ 0.06 & 14.01 $\pm$ 0.18 & 7.55 $\pm$ 0.09 & 12.54 $\pm$ 0.32 & 14.86 $\pm$ 0.33\\
\hspace{1em}$\text{pffr (Wild)}$ & 9.67 $\pm$ 0.03 & 11.32 $\pm$ 0.02 & 10.10 $\pm$ 0.12 & 9.97 $\pm$ 0.13 & 12.77 $\pm$ 0.36 & 11.94 $\pm$ 0.20\\
\hspace{1em}$\text{pffr ($z_{1-\alpha/2}$)}$ & 4.67 $\pm$ 0.02 & 5.62 $\pm$ 0.02 & 4.56 $\pm$ 0.06 & 4.68 $\pm$ 0.06 & 6.34 $\pm$ 0.16 & 5.47 $\pm$ 0.10\\
\addlinespace[0.6em]
\multicolumn{7}{l}{\textbf{$N = 500,\; n_i = 5$}}\\
\hspace{1em}$\text{1-Step } {\mathbf{R}}_{\text{Lon}} \otimes {\mathbf{R}}_{\text{Fun}}$ & 30.28 $\pm$ 0.27 & 32.41 $\pm$ 0.43 & 31.90 $\pm$ 0.26 & 29.85 $\pm$ 0.31 & 38.94 $\pm$ 0.79 & 32.75 $\pm$ 0.07\\
\hspace{1em}$\text{1-Step } {\mathbf{R}}_{\text{Lon}} \otimes I_L$ & 56.86 $\pm$ 0.46 & 60.19 $\pm$ 0.71 & 59.45 $\pm$ 0.44 & 57.88 $\pm$ 0.50 & 71.17 $\pm$ 1.27 & 61.01 $\pm$ 0.14\\
\hspace{1em}$\text{1-Step } I_{n_i} \otimes {\mathbf{R}}_{\text{Fun}}$ & 12.59 $\pm$ 0.10 & 14.28 $\pm$ 0.17 & 14.55 $\pm$ 0.11 & 12.47 $\pm$ 0.11 & 16.52 $\pm$ 0.30 & 14.49 $\pm$ 0.03\\
\hspace{1em}$\text{1-Step } I_{n_i} \otimes I_L$ & 3.42 $\pm$ 0.03 & 5.11 $\pm$ 0.06 & 5.38 $\pm$ 0.04 & 3.39 $\pm$ 0.03 & 5.79 $\pm$ 0.11 & 5.36 $\pm$ 0.01\\
\hspace{1em}$\text{Full-1Step } {\mathbf{R}}_{\text{Lon}} \otimes {\mathbf{R}}_{\text{Fun}}$ & 70.57 $\pm$ 0.65 & 89.43 $\pm$ 1.20 & 104.10 $\pm$ 1.00 & 70.65 $\pm$ 0.73 & 102.98 $\pm$ 2.20 & 104.37 $\pm$ 0.49\\
\hspace{1em}$\text{Full-1Step } {\mathbf{R}}_{\text{Lon}} \otimes I_L$ & 122.42 $\pm$ 1.26 & 162.84 $\pm$ 2.16 & 200.21 $\pm$ 1.59 & 131.24 $\pm$ 8.02 & 183.09 $\pm$ 3.39 & 193.70 $\pm$ 0.87\\
\hspace{1em}$\text{Full-1Step } I_{n_i} \otimes {\mathbf{R}}_{\text{Fun}}$ & 25.21 $\pm$ 0.21 & 29.30 $\pm$ 0.38 & 37.51 $\pm$ 0.29 & 24.35 $\pm$ 0.23 & 33.84 $\pm$ 0.63 & 36.19 $\pm$ 0.16\\
\hspace{1em}$\text{Full-1Step } I_{n_i} \otimes I_L$ & 4.14 $\pm$ 0.03 & 6.16 $\pm$ 0.07 & 7.24 $\pm$ 0.05 & 4.11 $\pm$ 0.03 & 6.94 $\pm$ 0.13 & 6.98 $\pm$ 0.02\\
\hspace{1em}$\text{pffr (Wild)}$ & 6.58 $\pm$ 0.05 & 7.16 $\pm$ 0.07 & 6.90 $\pm$ 0.04 & 6.31 $\pm$ 0.05 & 8.21 $\pm$ 0.13 & 6.95 $\pm$ 0.02\\
\hspace{1em}$\text{pffr ($z_{1-\alpha/2}$)}$ & 2.69 $\pm$ 0.02 & 2.94 $\pm$ 0.03 & 2.74 $\pm$ 0.02 & 2.64 $\pm$ 0.02 & 3.34 $\pm$ 0.05 & 2.75 $\pm$ 0.01\\
\addlinespace[0.6em]
\multicolumn{7}{l}{\textbf{$N = 500,\; n_i = 100$}}\\
\hspace{1em}$\text{1-Step } {\mathbf{R}}_{\text{Lon}} \otimes {\mathbf{R}}_{\text{Fun}}$ & 111.85 $\pm$ 0.27 & 149.61 $\pm$ 0.84 & 198.96 $\pm$ 3.74 & 113.00 $\pm$ 1.58 & 148.51 $\pm$ 0.36 & 189.01 $\pm$ 3.48\\
\hspace{1em}$\text{1-Step } {\mathbf{R}}_{\text{Lon}} \otimes I_L$ & 118.14 $\pm$ 0.35 & 151.87 $\pm$ 0.98 & 213.95 $\pm$ 4.25 & 123.12 $\pm$ 1.61 & 150.03 $\pm$ 0.34 & 203.80 $\pm$ 3.94\\
\hspace{1em}$\text{1-Step } I_{n_i} \otimes {\mathbf{R}}_{\text{Fun}}$ & 178.16 $\pm$ 0.37 & 215.29 $\pm$ 1.16 & 268.15 $\pm$ 5.33 & 189.46 $\pm$ 2.57 & 210.95 $\pm$ 0.49 & 256.99 $\pm$ 4.94\\
\hspace{1em}$\text{1-Step } I_{n_i} \otimes I_L$ & 54.11 $\pm$ 0.25 & 81.74 $\pm$ 0.57 & 122.54 $\pm$ 2.24 & 59.04 $\pm$ 0.74 & 84.63 $\pm$ 0.46 & 104.11 $\pm$ 1.97\\
\hspace{1em}$\text{Full-1Step } {\mathbf{R}}_{\text{Lon}} \otimes {\mathbf{R}}_{\text{Fun}}$ & 194.44 $\pm$ 0.38 & 264.40 $\pm$ 1.82 & 449.09 $\pm$ 10.68 & 179.31 $\pm$ 2.83 & 223.57 $\pm$ 0.77 & 295.21 $\pm$ 6.39\\
\hspace{1em}$\text{Full-1Step } {\mathbf{R}}_{\text{Lon}} \otimes I_L$ & 210.18 $\pm$ 0.88 & 328.53 $\pm$ 12.12 & 538.08 $\pm$ 16.18 & 213.07 $\pm$ 3.01 & 263.79 $\pm$ 1.25 & 336.11 $\pm$ 6.46\\
\hspace{1em}$\text{Full-1Step } I_{n_i} \otimes {\mathbf{R}}_{\text{Fun}}$ & 371.03 $\pm$ 0.74 & 421.48 $\pm$ 3.08 & 661.56 $\pm$ 14.27 & 324.73 $\pm$ 4.74 & 364.12 $\pm$ 1.93 & 500.39 $\pm$ 10.45\\
\hspace{1em}$\text{Full-1Step } I_{n_i} \otimes I_L$ & 65.97 $\pm$ 0.28 & 92.03 $\pm$ 0.60 & 125.46 $\pm$ 2.23 & 62.81 $\pm$ 1.08 & 88.82 $\pm$ 0.39 & 118.41 $\pm$ 1.93\\
\hspace{1em}$\text{pffr (Wild)}$ & 49.35 $\pm$ 0.22 & 56.25 $\pm$ 0.32 & 61.14 $\pm$ 1.00 & 51.59 $\pm$ 0.53 & 54.28 $\pm$ 0.18 & 59.47 $\pm$ 0.90\\
\hspace{1em}$\text{pffr ($z_{1-\alpha/2}$)}$ & 37.16 $\pm$ 0.10 & 44.27 $\pm$ 0.23 & 49.47 $\pm$ 0.87 & 39.71 $\pm$ 0.44 & 43.32 $\pm$ 0.14 & 48.18 $\pm$ 0.78\\
\bottomrule
\end{tabular}}
\end{table}

\clearpage

\section{Application} \label{sec:application}
We apply our framework to calcium imaging data to illustrate the benefits of longitudinal FDA in analyzing neural recordings. {In Appendix~\ref{app:background_calcimImaging}, we} describe 
what we have observed are common strategies among experimentalists for hypothesis testing of covariate-neural activity associations, as they have similar goals to our proposed method. These approaches seem to vary 
largely in how 1) the target neural population is defined, 2) the longitudinal structure is accounted for, and 3) the trial-level neural time-series are modeled. To address limitations {in these strategies}, we apply our fGEE to a dataset that pools neurons across animals  (akin to a ``neural pseudo-population'' strategy). This accounts for the 1) longitudinal and 2) functional nature of the response in each cluster (i.e., neuron), and 3) does not discard uncertainty in the animal-specific estimates in providing an overall \textit{neural pseudo population-level} estimate. 
This assumes correlation between neurons within-animal is negligible. 
In some cases, it may be preferable to apply a fGEE to the neurons in each animal separately,
but this necessitates an approach to construct a pooled estimate of animal-level fits that propagates uncertainty. 
\paragraph{Application Background}
We apply our method on data from a recent \textit{Nature} paper studying the role of pyramidal neurons in the primary somatosensory cortex (S1) in behavior and sensory input \citep{nature_sl}. This study recorded neuronal activity in five mice, from
$155-262$ (mean$\pm$ SEM: $184.4 \pm 22.61$) neurons per animal. Recording was done in head-fixed animals, running on a ball that tracked their movement speed. The authors
were interested in identifying S1 neurons active during spontaneous movements. 
They 
tested whether each S1 neuron was associated with running speed, whisker movement, and whisker sensory input. 

\paragraph{Identifying neural activity encoding speed information}

The correlation analysis used in the original paper could only test the neural activity–speed association on average within trials. In contrast, an FDA approach can test how this association evolves across trial timepoints. To demonstrate that, we randomly sampled $N=500$ neurons (clusters) from the five animals, and identified five second intervals when animals spontaneously began to run. {Each running burst was considered a ``trial'' (experimental replicate)} as typically analyzed in neuroscience. Thus the functional outcome for neuron $i$, on trial $j$ was a binary timeseries vector across five seconds of neural activity measured at 30 Hz (starting at the onset of the running bout). 
This results in a functional outcome measured at an evenly spaced grid of $|\mathcal{S}|=150$ points. Speed, $X_{i,j}(s) \in \mathbb{R}$ was defined as a functional covariate and took the same value for all neurons recorded from the same animal. The cluster size was $n_i = 29$ for all neurons $i \in [N]$.  We fit the model $\text{logit}\{\mathbb{E}( {Y}_{i,j,l}(s) \mid X_{i,j}(s)) \} =~ \beta_0(s) +  X_{i,j}(s) {\beta}_1(s)$,
with {an AR1 working correlation in both longitudinal and functional directions (see Figure~\ref{fig:neuro_fig}).
The results were consistent across a range of working correlation structures (see Appendix~\ref{app:neuro_analysis_supp}).} 

The most salient finding was that the speed–neural activity association does not become significant until a few seconds after the animals begin to run and that association becomes non-significant fairly quickly (see Figure~\ref{fig:neuro_fig}). The timing of the association suggests that these neurons are not driving the movement of the animal. 
This shows how the timing sensitivity of fGEE can help identify the type of cognition or behavior a brain region encodes, which is much harder to do with analyses of trial summary measures.

\begin{figure*}[!t]
\centering
	\begin{subfigure}[t]{0.45\textwidth}
\includegraphics[width=0.95 \linewidth]{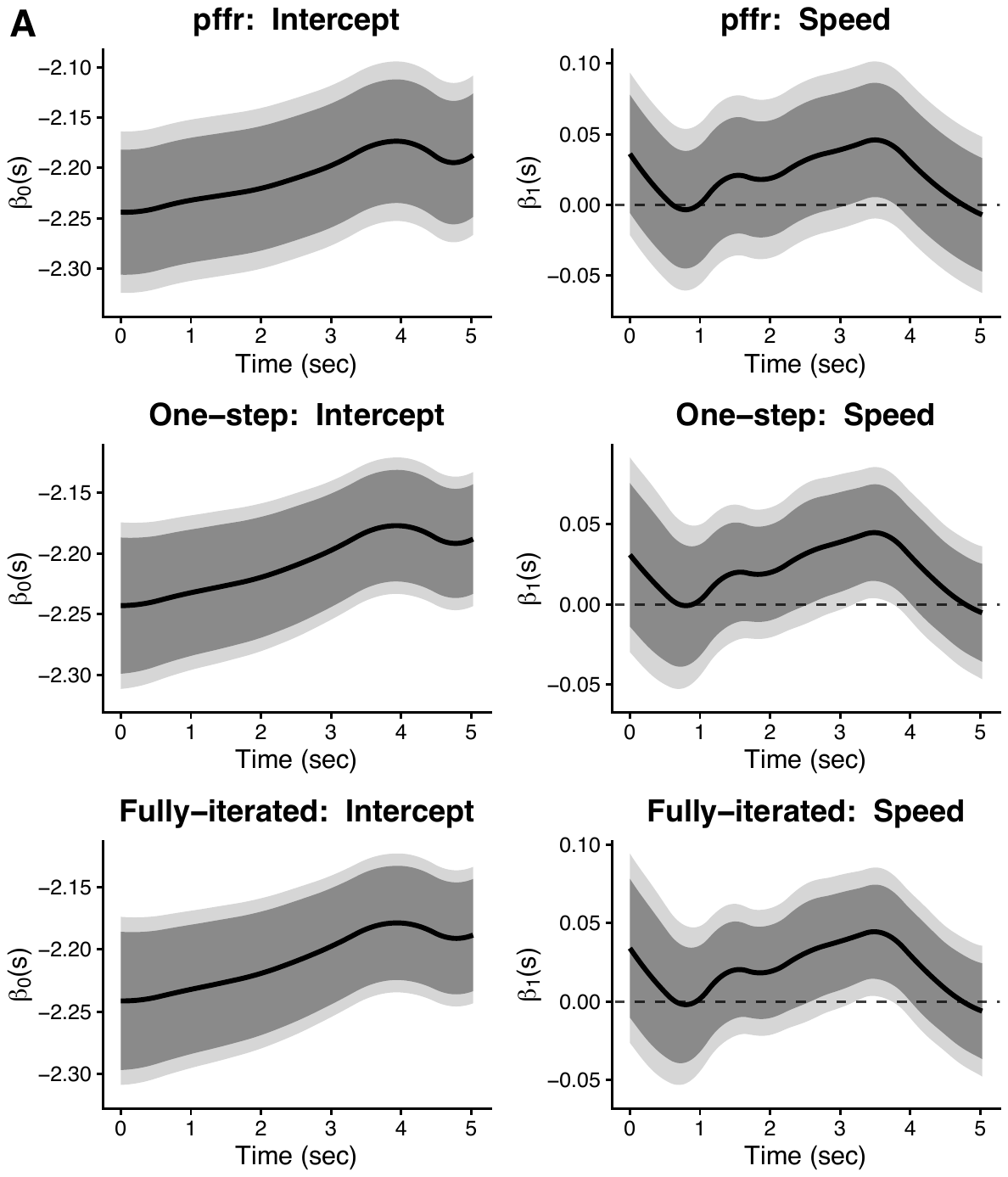}
\caption{\footnotesize \textbf{Speed–neural activity association.}}
\label{fig:neuro_fig}
	\end{subfigure}
    \centering
	\begin{subfigure}[t]{0.45\textwidth}
\includegraphics[width=0.95 \linewidth]{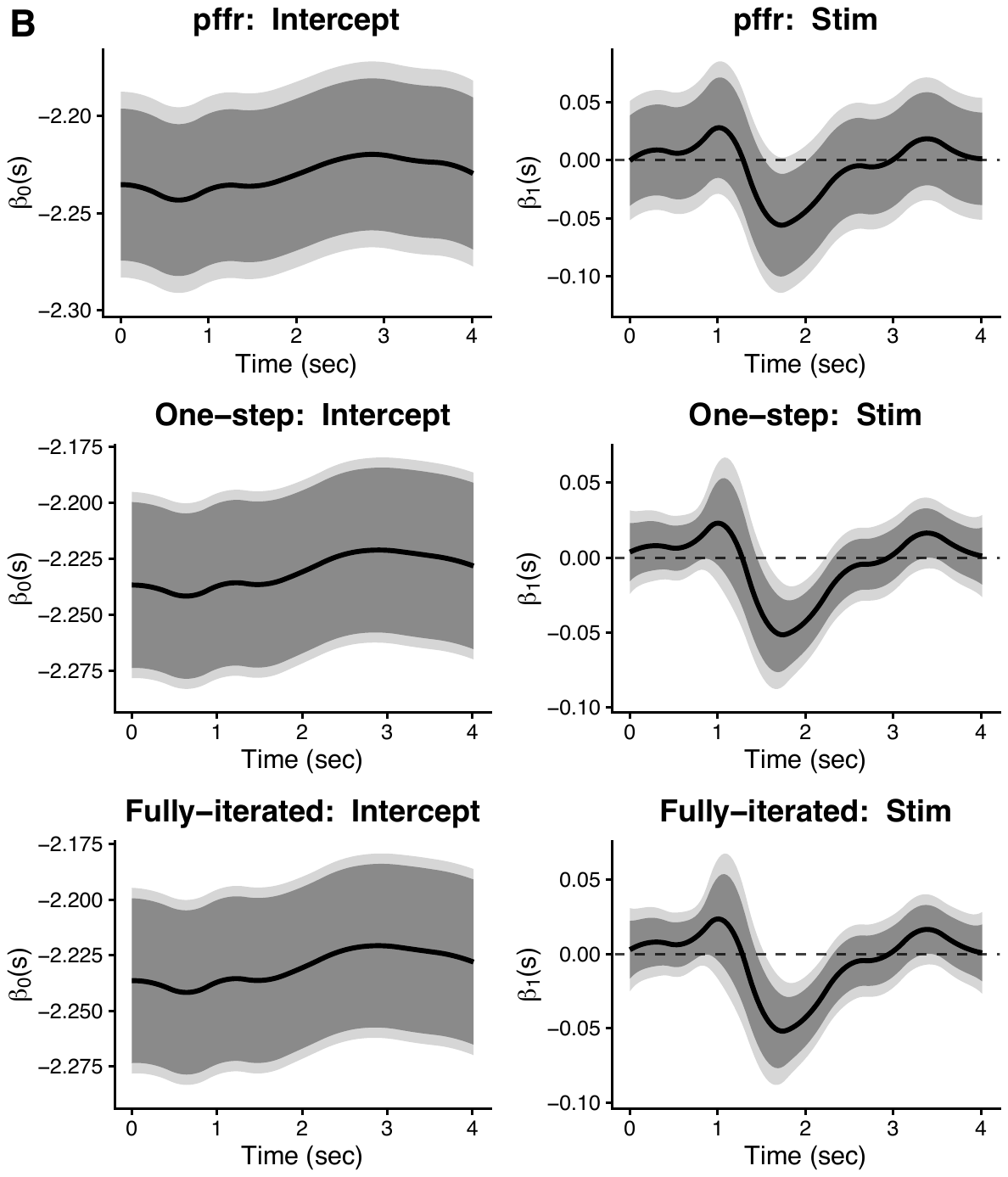}
\caption{\footnotesize 
\textbf{Whisker stimulation effect.}
}
\label{fig:stim_fig}
	\end{subfigure}
\caption{\footnotesize
Functional coefficient one-step estimates. fGEE models (one-step and fully-iterated) were fit with a working covariance $\mathbf{R}_i = \mathbf{R}_{Lon} \otimes \mathbf{R}_{Fun}$ where $\mathbf{R}_{Lon}$ and $\mathbf{R}_{Fun}$ have an AR1 structure each parameterized by separate $\rho_{Lon}$ and $\rho_{Fun}$ AR1 correlation parameters. Dark and light bands indicate pointwise and joint 95\% CIs, respectively. CIs for both $\texttt{pffr}$ and the fGEE models were constructed with a sandwich variance estimator and CI quantiles obtained from a wild cluster bootstrap. The fGEE models (one-step and fully-iterated) show tighter confidence intervals than the $\texttt{pffr}$ fits. In fact, the $\texttt{pffr}$ fits do not show any time-intervals in the $\texttt{Speed}$ or $\texttt{Stim}$ functional coefficient estimate plots that are statistically significant according to the joint CIs. Fit times (seconds) on a laptop without parallelization were (a) $\texttt{pffr}$ 
11.58; one-step 41.84; Full-1Step 255.12; and (b) $\texttt{pffr}$ 
97.47; one-step 349.84; Full-1Step 1131.92. } 
\end{figure*}

\paragraph{Whisker Stimulation}
The authors of \cite{nature_sl} were also interested in how {whisker stimulation affected S1 neural activity.}
The summary analyses they conducted only allowed estimation of the extent to which whisker stimulation changed neural activity on average across within-trial timepoints. To characterize the ``temporal dynamics'' of the neural response to this manipulation, we applied our method to activity from $N=500$ randomly selected neurons (clusters), each with $n_i = 300$ observations of the functional outcome: 
4 sec of neural activity, measured 1 sec before whisker stimulation to 3 sec after ($L=120$). 
We fit the model $\text{logit}\{\mathbb{E}( Y_{i,j,l}(s) \mid X_{i,j}) \} =~ \beta_0(s) +  X_{i,j} \beta_1(s)$,
where $X_{i,j} \in \{0,1\}$ is an indicator that neuron $i$ was recorded from an animal that was stimulated on trial $j$. We adopted {AR1 working correlations in longitudinal and functional directions because longitudinal and functional domains are defined as between- and within-trial timepoints, respectively.}

The one-step estimator, fit to 150,000 functional observations, took {$\sim$6.5 min} to fit on a MacBook Pro with an Apple M1 Max chip with 64 GB of RAM, without parallelization, and maintained a reasonable memory footprint throughout. The coefficient associated with stimulation, $\widehat{\beta}_1(s)$, shows the estimated mean difference in neural activity between stimulated and non-stimulated trials. It appears stimulation leads to a rapid reduction in activity (see Figure~\ref{fig:stim_fig}). {The results were consistent across working correlation structures (see Appendix~\ref{app:neuro_analysis_supp}).}
This analysis {shows} how the one-step makes it possible to identify a clear temporal profile for the effect of interest, and is scalable enough to allow estimation over a large sample of neurons.

{The above analyses also demonstrate how fGEE yields more efficient inference than using a sandwich variance estimator to construct CIs centered at a $\texttt{pffr}$ fit. In contrast to the one-step, the $\texttt{pffr}$ fits do not show any time-intervals in the $\texttt{Speed}$ or $\texttt{Stim}$ analyses that are statistically significant based on the joint CIs. In Figure~\ref{fig:stim_fig} we also show estimates from a fully-iterated fGEE that uses the fast (one-step) cluster CV for smoothing parameter tuning (``Full-1Step''). The fully-iterated and one-step estimates are nearly indistinguishable, highlighting our theoretical result that the one-step is asymptotically equivalent to a fully-iterated fGEE.  }
\vspace{-0.75cm}

\section{Discussion} \label{sec:discussion}
We hope that by developing a method that can scale to dataset sizes common in calcium imaging and electrophysiology research, we will encourage neuroscientists to apply longitudinal FDA in these settings. While FDA methods have been applied extensively in human neuroimaging 
\citep{zhu2019fmem, cui2022, fQIF2025}, little work, to our knowledge, has sought to apply these tools to neural recordings from animal model data. We believe, however, that FDA methods have great potential to help researchers answer a range of scientific questions common in animal model studies. For example, longitudinal FDA may be beneficial for estimating \textit{tuning curves} that show how mean neural responding evolves depending on a tuning variable, like the volume of an auditory cue or the angle a visual stimulus is presented at. 
Since neuroscientists often analyze how tuning curves differ across covariate values (treatment vs. control arm), and data are collected across multiple experimental replicates within-subject (``trials''), longitudinal FDA provides a natural strategy for inference.
%
 

The one-step fGEE can be extended in a range of related settings. {For example, our implementation can be adapted for flexible $\mathbb{V}_i$ forms (e.g., $\mathbb{V}_i(s)$ with unstructured working covariance) and would likely be fast for moderate $n_i$. Some $\mathbb{V}_i$ forms (e.g., Kronecker product) for multivariate functional domains can also be scalable. In Appendix~\ref{app:fast_boot}, we also describe a fast one-step cluster bootstrap as an alternative estimator for $\widehat{\text{Var}}(\widehat{\boldsymbol{\beta}}_{\Lambda_1}^{(1)})$. 
In preliminary simulations the bootstrap and sandwich $\widehat{\text{Var}}(\widehat{\boldsymbol{\beta}}_{\Lambda_1}^{(1)})$ estimates were nearly identical.
This may therefore provide a fast and flexible strategy to construct CIs for functions/transformations $f(\cdot)$ of the coefficients (e.g., derivatives) without having to derive analytic forms of $\widehat{\text{Var}}\{f(\widehat{\boldsymbol{\beta}}_{\Lambda_1}^{(1)})\}$.}
%
{Finally, future work might explore generalized cross-validation for fGEE smoothing parameter selection. 

We release our methods in the $\texttt{R}$ package $\texttt{fastFGEE}$, which supports a wide range of link functions, implements fGEE based on quasi-likelihoods deriving from many distributions, and supports both functional and scalar outcomes and covariates.} 
We hope our theoretical guarantees and efficient implementation encourage analysts to apply FDA methods more widely. 

\section*{Acknowledgments}
{GL and FP were supported by the Intramural Research Program of the National Institute of Mental Health, project ZIC-MH002968. The contributions of the NIH author(s) were made as part of their official duties as NIH federal employees, are in compliance with agency policy requirements, and are considered Works of the United States Government. However, the findings and conclusions presented in this paper are those of the authors and do not necessarily reflect the views of the NIH or the U.S. Department of Health and Human Services. This study utilized the high-performance computational capabilities of the Biowulf Linux cluster at the NIH, Bethesda, MD (http://biowulf.nih.gov). We thank Drs. Soohyun Lee for providing
data for our application, Yuan Zhao for help with data pre-processing, 
and Keith Goldfeld for sharing his package code. We acknowledge the use of AI to extend an initial version of our package, modify existing code to generate 
tables, and generate one Appendix figure (see Appendix~\ref{app:ai})}.

\section*{Supplementary Materials}
{Appendices A–D below provide additional method details, proofs, simulation results, and application analyses. Code and example application data are available at the links in the Data Availability statement.}


\section*{Data Availability}
{For our code and $\texttt{fastFGEE}$ package, see: \url{https://github.com/gloewing/fgee_onestep}. Application datasets are available at: \url{https://osf.io/b93dc/files/osfstorage}.}
 
\clearpage

\newpage
\begin{appendices}
\makeatletter
\renewcommand{\thesection}{\Alph{section}}
\renewcommand{\thesubsection}{\Alph{section}.\arabic{subsection}}
\renewcommand{\thefigure}{S\arabic{figure}}
\renewcommand{\thetable}{S\arabic{table}}
\renewcommand{\thealgorithm}{S\arabic{algorithm}}
\renewcommand{\theHsection}{appendix.\Alph{section}}
\renewcommand{\theHsubsection}{appendix.\Alph{section}.\arabic{subsection}}
\renewcommand{\theHfigure}{S\arabic{figure}}
\renewcommand{\theHtable}{S\arabic{table}}
\renewcommand{\p@subsection}{}
\renewcommand{\p@subsubsection}{}
\makeatother

\singlespacing
\setcounter{figure}{0}
\setcounter{table}{0}
\setcounter{algorithm}{0}
\captionsetup[algorithm]{font={stretch=1}}
\AtBeginEnvironment{algorithmic}{\setstretch{1}\small}
    
\section{Additional Method Details} \label{app:method_details}

{
\subsection{One-Step fGEE Estimation Algorithm} \label{app:oneStep}

We first describe the one-step fGEE estimation procedure here and the fully-iterated fGEE version in the following subsection.

\begin{algorithm}[H]
\caption{One-step update using residual-based correlation and pffr initialization}\label{app:fgee_algo}
\begin{algorithmic}[1] 
\State \textbf{Inputs:}
\begin{itemize}
  \item Functional responses $\{\boldsymbol{Y}_i\}_{i=1}^N$, where $\boldsymbol{Y}_i = [\boldsymbol{Y}_i(s_1)^T,\ldots,\boldsymbol{Y}_i(s_L)^T]^T \in \mathbb{R}^{n_i L}$, $\{s_l\}_{l=1}^L \subset \mathcal{S}$
  \item Covariate vectors $\{\mathbf{X}_{i,j}\}_{i \in [N], j \in n_i}$, where $\mathbf{X}_{i,j} \in \mathbb{R}^{q}$.
  \item Penalty matrix $\mathbb{S}$, and grid of values to tune smoothing parameters $\Lambda_1$ 
  \item Basis representation encoded in matrix $\mathbf{B}$ (e.g., B-splines)
    \item Family/variance function $v(\mu)$ for standardized residuals (e.g., Bernoulli: $v(\mu)=\mu(1-\mu)$; Poisson: $v(\mu)=\mu$; Gaussian: $v(\mu)=\sigma^2$)
  \item Pointwise working covariance structure type (independence, AR1, or exchangeable)
  \item Link function $g()$ and mean model specification for $g\{\mathbb{E}(Y_{i,j}(s) \mid \mathbf{X}_{i,j})\}$
\end{itemize}

\State \textbf{Initialization (function-on-scalar fit):}
\begin{itemize}
  \item Obtain the initial parameter estimate $\widehat{\boldsymbol{\theta}}^{(0)}_{\Lambda_0}$ 
  \begin{itemize}
      \item We use a likelihood-based function-on-scalar regression using $\texttt{refund::pffr()}$
      \item We use fREML for smoothing parameter selection (to select $\Lambda_0$)
  \end{itemize}
\end{itemize}

\State \textbf{Residuals and correlation estimation:}
\begin{itemize}
    \item Calculate fitted values using $\widehat{\boldsymbol{\theta}}^{(0)}_{\Lambda_0}$: $\widehat{\boldsymbol{\mu}}_i\big(\widehat{\boldsymbol{\theta}}^{(0)}_{\Lambda_0}\big)$ for $i=1,\dots,N$.
  \begin{itemize}
      \item $\widehat{\boldsymbol{\mu}}_i\big(\widehat{\boldsymbol{\theta}}^{(0)}_{\Lambda_0}\big)=g^{-1}(\mathbf{B} \widehat{\boldsymbol{\theta}}^{(0)}_{\Lambda_0, 0} + \sum_{r=1}^q X_{i,j,r} \mathbf{B} \widehat{\boldsymbol{\theta}}^{(0)}_{\Lambda_0, r})$, where $g^{-1}(\cdot)$ is applied component-wise and $\widehat{\boldsymbol{\theta}}^{(0)}_{\Lambda_0, r}$ is the subvector of $\widehat{\boldsymbol{\theta}}^{(0)}_{\Lambda_0}$ associated with covariate $r$
  \end{itemize}
  \item Define per-point variance diagonal blocks $\mathbb{A}_i(s_l) = \mathrm{diag}\!\big(v\{\widehat{\boldsymbol{\mu}}_i(s_l)\}\big) \in \mathbb{R}^{n_i \times n_i}$; let $\mathbb{A}_i = \mathrm{bdiag}\big(\mathbb{A}_i(s_1),\ldots,A_i(s_L)\big)$
  \item Compute residuals on the grid: $\widehat{\boldsymbol{e}}_i = \mathbb{A}_i^{-1/2}\{\boldsymbol{Y}_i - \widehat{\boldsymbol{\mu}}_i\big(\widehat{\boldsymbol{\theta}}^{(0)}_{\Lambda_0}\big)\} $ for $i=1,\dots,N$.
  \item For each $s_l \in \{s_1,\dots,s_L\}$, estimate correlation parameters $\widehat{\rho}(s_l)$ using $\{\widehat{\boldsymbol{e}}_i\}_{i=1}^N$.
  \item Construct 
  $\widehat{\mathbb{V}}_i\big(\widehat{\boldsymbol{\theta}}^{(0)}_{\Lambda_0}\big)$ for each $i$ using the estimated $\widehat{\rho}(s)$ (and any variance function) 
  \item Invert $\widehat{\mathbb{V}}_i\big(\widehat{\boldsymbol{\theta}}^{(0)}_{\Lambda_0}\big)$ for each $i$ using specialized algorithms (e.g., see Appendix~\ref{app:corr_inv}).
\end{itemize}

\State \textbf{Derivative (Jacobian) matrices:}
\begin{itemize}
  \item For each $i$, compute
  $\mathbb{D}_i\big(\widehat{\boldsymbol{\theta}}^{(0)}_{\Lambda_0}\big) = \left.\dfrac{\partial \widehat{\boldsymbol{\mu}}_i(\boldsymbol{\theta})}{\partial \boldsymbol{\theta}}\right|_{\boldsymbol{\theta}=\widehat{\boldsymbol{\theta}}^{(0)}_{\Lambda_0}}$, using initial $\widehat{\boldsymbol{\theta}}^{(0)}_{\Lambda_0}$ and $\widehat{\mathbb{V}}^{-1}_i\big(\widehat{\boldsymbol{\theta}}^{(0)}_{\Lambda_0}\big)$
\end{itemize}

\State \textbf{Cluster Cross-Validation to Select $\Lambda_1$:}
\begin{itemize}
    \item Select $\Lambda_1$ values that minimize average cross-validation criteria (see Algorithm~\ref{algo:cv})
\end{itemize}

\State \textbf{One-step estimator:}
\begin{itemize}
  \item Compute the update 
\item[] \noindent\resizebox{\textwidth}{!}{$\widehat{\boldsymbol{\theta}}_{ \Lambda_1}^{(1)} = \widehat{\boldsymbol{\theta}}_{ \Lambda_0}^{(0)} + \left [  \frac{1}{N} \sum_{i=1}^N \{\mathbb{D}_i(\widehat{\boldsymbol{\theta}}^{(0)}_{\Lambda_0})\}^T \{\widehat{\mathbb{V}}_i(\widehat{\boldsymbol{\theta}}_{ \Lambda_0}^{(0)})\}^{-1} \{\mathbb{D}_i(\widehat{\boldsymbol{\theta}}^{(0)}_{\Lambda_0})\}  + \Lambda_1 \mathbb{S}  \right ]^{-1}  \frac{1}{N} \sum_{i=1}^N\left [ \{\mathbb{D}_i(\widehat{\boldsymbol{\theta}}^{(0)}_{\Lambda_0})\}^T \{\widehat{\mathbb{V}}_i(\widehat{\boldsymbol{\theta}}_{ \Lambda_0}^{(0)})\}^{-1} \left \{ \boldsymbol{Y}_i - \widehat{\boldsymbol{\mu}}_i(\widehat{\boldsymbol{\theta}}^{(0)}_{\Lambda_0}) \right \}  - \Lambda_1 \mathbb{S} \widehat{\boldsymbol{\theta}}_{ \Lambda_0}^{(0)} \right ]$}

\end{itemize}

\State \textbf{Calculate $\widehat{\text{Var}}(\widehat{\boldsymbol{\beta}})$ and 95\% Confidence Intervals (CIs):}
\begin{itemize}
    \item Repeat steps 3-4 with $\{\Lambda_1 , \widehat{\boldsymbol{\theta}}_{ \Lambda_1}^{(1)}\}$ instead of $\{\Lambda_0 , \widehat{\boldsymbol{\theta}}_{ \Lambda_0}^{(0)}\}$, get: $\mathbb{D}_i(\widehat{\boldsymbol{\theta}}^{(1)}_{\Lambda_1})$, $\widehat{\boldsymbol{\mu}}_i(\widehat{\boldsymbol{\theta}}^{(1)}_{\Lambda_1})$, $\widehat{\mathbb{V}}^{-1}_i(\widehat{\boldsymbol{\theta}}_{ \Lambda_1}^{(1)})$
    \item Calculate ${\widehat{\text{Var}}} \left( \widehat{\boldsymbol{\theta}}_{\Lambda_1}^{(1)} \right )$ using $\{\mathbb{D}_i(\widehat{\boldsymbol{\theta}}^{(1)}_{\Lambda_1})$, $\widehat{\boldsymbol{\mu}}_i(\widehat{\boldsymbol{\theta}}^{(1)}_{\Lambda_1})$, $\widehat{\mathbb{V}}^{-1}_i(\widehat{\boldsymbol{\theta}}_{ \Lambda_1}^{(1)})\}_{i \in [N]}$ 
    \begin{itemize}
        \item Use sandwich estimator given (Section~\ref{sec:sandwich}) or cluster bootstrap (Appendix~\ref{app:fast_boot})
    \end{itemize}
    \item Calculate $\widehat{\text{Var}}(\widehat{\boldsymbol{\beta}}) =  \text{bdiag}(\widehat{\Sigma}^{(\beta)}_1, \ldots, \widehat{\Sigma}^{(\beta)}_q ) = \mathbb{B} {\widehat{\text{Var}}} \left( \widehat{\boldsymbol{\theta}}_{\Lambda_1}^{(1)} \right) \mathbb{B} ^T$, where $\mathbb{B} = (I_q \otimes \mathbf{B})$ 
    \item Calculate pointwise and joint CIs for $\widehat{\boldsymbol{\beta}}$ (Appendix~\ref{app:wild}) . 
\end{itemize}
\Return Point estimates, pointwise CIs and joint CIs for functional coefficient estimates $\widehat{\boldsymbol{\beta}}$

\end{algorithmic}
\end{algorithm}

\subsection{Fully-Iterated fGEE Estimation} \label{app:full_itr}
We implemented the fully-iterated fGEE procedure described in \cite{gee_functional}. This uses the same steps as described in Algorithm~\ref{app:fgee_algo} in Appendix~\ref{app:oneStep}, except that we repeat steps 3,4,6 iteratively until convergence. Specifically, we stop at the first iteration $k$ in which the maximum element of $|\widehat{\boldsymbol{\theta}}^{(k)} - \widehat{\boldsymbol{\theta}}^{(k-1)} | $ was less than $\epsilon$. We used $\epsilon =10^{-6}$ for the final coefficient estimates, and $\epsilon=10^{-3}$ for fold-specific coefficient estimates during smoothing parameter selection. The fully-iterated fGEE uses the same computationally efficient covariance matrix inversion strategies that we applied for the one-step estimator.

\newpage

\subsection{Efficient inversion of working covariance matrices} \label{app:corr_inv}
When $n_i$ or $L$ are large, even inverting submatrices like $\mathbb{V}_{Lon}$, $\mathbb{V}_{Fun}$, or $\mathbb{V}_i(s)^{-1}$ at one point $s$ is computationally intensive with standard linear algebra routines. Luckily, many common covariances, like exchangeable and AR1 working covariance matrices can be inverted efficiently. To describe this precisely we first rewrite $\mathbb{V}_i(s) = \mathbf{A}_i^{1/2}(s)\mathbf{R}_i(s)\mathbf{A}_i^{1/2}(s)$, where $\mathbf{A}_i(s) = \text{diag}\left(v_{i,1}(s), \ldots, v_{i,n_i}(s) \right )$ and $v_{i,j}(s)$ models $\text{Var}(Y_{i,j}(s) \mid \mathbf{X}_{i,j}) \in \mathbb{R}$.
%

If $\mathbb{V}_i(s)$ has, for example, then AR1 structure  $\text{Cor}(Y_{i,j}(s), Y_{i,j'}(s)\mid \mathbf{X}_{i,j}) = \rho(s)^{|j-j'|}$, $\mathbb{V}_i^{-1}(s)$ can also be efficiently computed because, for $\rho(s) \geq 0$, its decomposition yields a Toeplitz $\mathbf{R}_i(s)$. For an estimated $\widehat{\boldsymbol{e}}_i(s) = [\boldsymbol{Y}_i(s) - \widehat{\boldsymbol{\mu}}_i(s)]\widehat{\mathbf{A}}_i(s)^{-1/2}$, we can quickly calculate $\widehat{\mathbb{V}}_i(s)^{-1}\widehat{\boldsymbol{e}}_i(s)$
by solving the Toeplitz system $\boldsymbol{a} = \widehat{\mathbf{R}}_i(s) \widehat{\boldsymbol{e}}_i(s)$ with, for example, the generalized Schur algorithm \citep{gschur}. This is done without fully constructing the $n_i \times n_i$ matrix $\widehat{\mathbb{V}}_i(s)$ or storing $\widehat{\mathbb{V}}^{-1}_i(s)$ in memory.
When data are observed irregularly, one can use the algorithm proposed in \cite{irreg_ar1}. 

More generally, if $\mathbb{V}_{i,j}$, $\mathbb{V}_{Lon}$ or $\mathbb{V}_{Fun}$ have a number of different common covariance structures (e.g., an exchangeable, AR1, AR(p), moving average MA(q), Gaussian Process covariances when using Matern/RBF kernels) the respective $\mathbf{R}_{i,j}$, $\mathbf{R}_{Lon}$, $\mathbf{R}_{Fun}$ will have Toeplitz structure and can be inverted with the same algorithms. When analyzing a longitudinal functional outcome observed on a multivariate functional domain (e.g., fMRI), the Kronecker identity approach to efficiently invert working covariance matrices described in Section~\ref{sec:onestep} can be applied to efficiently invert covariance matrices of the form $\mathbf{R}_i = \mathbf{R}_1 \otimes \mathbf{R}_2 \otimes \mathbf{R}_3$. 


}

\color{black}
\subsection{Working correlation parameters} \label{app:corr_param}
When adopting the block exchangeable or AR1 correlation structures, each ${\mathbb{V}}_i(s)$ is a function of a nuisance correlation parameter $\rho(s)$. We estimate these at each point $s$ separately and
then optionally smooth over the functional domain to reduce variability.
Defining residuals as $e_{i,j}(s) = \frac{Y_{i,j}(s) - {\mu}_{i,j}(s)}{\sqrt{v(\mu_{i,j}(s))}}$, we use the method of moments estimator for the exchangeable structure \citep{discrete_longitudinal}: $\widehat{\rho}(s) = \frac{1}{N} \sum_{i=1}^N \frac{1}{n_i(n_i-1)} \sum_{j\leq k} \hat{e}_{i,j}(s)\hat{e}_{i,k}(s)$, and truncate the $\widehat{\rho}(s)$ at $1-\epsilon$ or $-1 + \epsilon$ if they fall outside the $(-1,1)$ range.
For an AR1 structure, we estimate each $\rho_i(s)$ with the Yule-Walker equations \citep{yule, walker}
when the longitudinal observations are sampled at regular time intervals. If sampled irregularly, we estimate $\rho_i(s)$ with the MLE estimator proposed in \cite{irreg_ar1}. We then calculate $\widehat{\rho}(s) = \frac{1}{N} \sum_{i=1}^N \widehat{\rho}_i(s)$ and truncate the $\widehat{\rho}(s)$ at $0$ or $1 - \epsilon$ if they fall outside the $[0,1)$ range. {For Kronecker product working covariances, there is only one correlation parameter for $\Sigma_{Lon}$ and one for $\Sigma_{Fun}$. In these cases, the estimator for $\rho_{Lon}$ and $\rho_{Fun}$ is calculated based on a weighted average of the functional domain-specific estimates. For example, for a method of moments estimator for the exchangeable $\rho_{Lon}$, we estimate $\widehat{\rho}_{Lon} =\frac{ \sum_{s \in \mathcal{S}} \sum_{i=1}^N \sum_{j\leq k} \hat{e}_{i,j}(s)\hat{e}_{i,k}(s)}{\sum_{i=1}^N \sum_{s \in \mathcal{S}} \binom{n_i}{2}}$}. 

The correlation parameters (e.g., $\widehat{\rho}(s)$, $\rho_{Lon}$) are calculated twice in our framework: 1) first using $\widehat{\boldsymbol{\theta}}_{ \Lambda_0}^{(0)}$ to calculate $\widehat{\mathbb{V}}_i(\widehat{\boldsymbol{\theta}}_{ \Lambda_0}^{(0)})$ that is plugged into the one-step estimator~\eqref{emp_IF}, and 2) second using $\widehat{\boldsymbol{\theta}}_{ \Lambda_1}^{(1)}$ to calculate $\widehat{\mathbb{V}}_i(\widehat{\boldsymbol{\theta}}_{ \Lambda_1}^{(1)})$ that is plugged into the  $\widehat{\text{Var}} \left( \widehat{\boldsymbol{\beta}}_{\Lambda_1}^{(1)} \right)$ estimator (see main text Section~\ref{sec:sandwich}).


{
\subsection{Functional Principal Component Analysis} \label{sec:fpca}
Here we discuss estimation of $\mathbf{R}_{i,Fun}$ using functional principal component analysis (FPCA).

Procedurally, we use the initial fit $\widehat{\boldsymbol{\theta}}_{ \Lambda_0}^{(0)}$ to calculate residuals, as used for estimation of other correlation parameters above. We then define $\widehat{e}_{i,j}(s) = \frac{Y_{i,j}(s) - \widehat{\mu}_{i,j}(s)}{\sqrt{v(\widehat{\mu}_{i,j}(s))}}$ and concatenate these into the vectors  \[ \widehat{\boldsymbol{e}}_{i,j}= [\widehat{e}_{i,j}(s_1), \ldots , \widehat{e}_{i,j}(s_L) ]^T \in \mathbb{R}^L \] for all $i \in [N], j \in [n_i]$. We pool these across $i \in [N]$ and $j \in [n_i]$, ignoring residual dependence, and fit a functional principal component analysis (FPCA) using $\texttt{refund::fpca.face}$ for dense grids and $\texttt{refund::fpca.sc}$ for sparse grids \citep{fpca.face, xiao2016fast}. 

We use this to estimate $\mathbf{R}_{i,Fun}$. For example, we denote $\mathbf{C}\{E(s),E(s')\}$ as the mean-zero residual process defined along the functional domain (after standardization). We can then use FPCA to approximate the functional covariance kernel \[\mathbf{C}(s,s') = Cov\{E(s), E(s') \} \] via the eigen-expansion \[ \text{Cov}\{E(s), E(s')\} = \sum_{k \geq 1} \lambda_k \phi_k(s) \phi_k(s')\] where $\phi_k(s)$ and $\lambda_k$ are eigenfunction and eigenvalue $k$ at $s$, respectively. After fitting the above FPCA, for the observed functional domain point grid $\mathcal{S}$, we denote the matrix of estimated eigenfunctions as $\widehat{\Phi} \in \mathbb{R}^{L \times K}$ where $K$ is the number of retained eigenfunctions, and column $k$ contains $\widehat{\phi}_k$ evaluated on the $L \times 1$ grid $\{ s_l\}_{l=1}^L$. We write $\widehat{\Lambda} = \text{diag}(\widehat{\lambda}_1, \ldots, \widehat{\lambda}_K)$. Then we estimate our working covariance as \[ \widehat{\mathbf{C}} = \widehat{\Phi} \widehat{\Lambda} \widehat{\Phi}^T + \hat{\sigma}^2 I \] where $\hat{\sigma}^2$ is the residual variance. In practice, we select the number of eigenfunctions to retain as the smallest $K$ such that $\frac{\sum_{k =1}^K \widehat{\lambda}_k}{\sum_{k \geq 1} \widehat{\lambda}_k} \geq \texttt{pve}$ for percent variance explained, $\texttt{pve} = 0.95$. We then take $\widehat{\mathbf{R}}_{i,Fun} = \widehat{\mathbf{C}}_i$, where $\widehat{\mathbf{C}}_i$ is constructed on the grid of functional domain points observed for cluster $i$. Inverting $\widehat{\mathbf{C}}_i$ can be made more computationally efficient by using a Woodbury inversion formula since $\widehat{\mathbf{C}} = \widehat{\Phi} \widehat{\Lambda} \widehat{\Phi}^T + \hat{\sigma}^2 I$ is defined as a low rank (rank $K$) matrix $\widehat{\Phi} \widehat{\Lambda} \widehat{\Phi}^T$, plus a diagonal matrix $\hat{\sigma}^2 I$.

Future extensions could adopt working covariance matrices that define $\widehat{\mathbf{R}}_{i,Fun}$ with cluster-specific estimates of $\widehat{\mathbf{C}}_{i}$ using longitudinal functional PCA methods \citep{yao2005functional, greven2011longitudinal, shou2015structured, cui2023fast} to account for the clustering (e.g., correlation in the observational direction). If one adopts a block diagonal fGEE working covariance matrix that specifies correlation in the functional direction, $\mathbf{R}_i = \text{bdiag}(\widehat{\mathbf{R}}_{i,1}, \ldots, \widehat{\mathbf{R}}_{i,n_i} )$, one could, use longitudinal FPCA estimates to construct cluster- and longitudinal observation-specific covariance estimates with blocks $\widehat{\mathbf{R}}_{i,j} = \widehat{\mathbf{C}}_{i,j}$.

}

\subsection{Cross-Validation Procedure}\label{app:cv}

\subsubsection{Fast Cluster Cross-Validation}
We provide a full description of our {fast} cluster cross-validation (CV) here. We found in simulations that cluster cross-validation produced one-step estimates with better estimation accuracy than one-step estimates based on smoothing parameters selected with 
the bootstrap-based procedure proposed in \cite{gee_functional}.
Moreover, restricted Maximum Likelihood cannot be used to select smoothing parameters for fGEE because there is no likelihood to maximize \citep{gee_functional}.
We use the negative log-likelihood as a CV fit criteria. We propose the following scalable (K-fold) cluster CV for large datasets. 

We define the folds, $\{\mathcal{K}_1, \ldots, \mathcal{K}_K\}$, as a disjoint partition of cluster index sets (i.e. the held-out cluster indices) where, for $K \geq 2$, $\mathcal{K}_k \subset [N]$ for each $k \in [K]$, $\bigcup_k \mathcal{K}_k = [N]$, and $\mathcal{K}_{k_1} \bigcap \mathcal{K}_{k_2} = \emptyset$ for all $k_1 \neq k_2$. To scale CV to large datasets, we exploit four features of the problem structure. First, each fold's one-step estimate is calculated with pre-computable quantities. For example, rewriting the update as
\begin{align*}
    \widehat{\boldsymbol{\theta}}_{ \Lambda_1}^{(1)} = \widehat{\boldsymbol{\theta}}_{ \Lambda_0}^{(0)} +
 \left [  \frac{1}{N} \sum_{i=1}^N  \mathbb{W}_i (\widehat{\boldsymbol{\theta}}^{(0)}_{\Lambda_0}) + \Lambda_1 \mathbb{S}  \right ]^{-1} \frac{1}{N} \sum_{i=1}^N \left \{ \mathbf{b}_i(\widehat{\boldsymbol{\theta}}^{(0)}_{\Lambda_0})  - \Lambda_1 \mathbb{S} \widehat{\boldsymbol{\theta}}^{(0)}_{\Lambda_0} \right \}, 
\end{align*}
illustrates that we can pre-compute each cluster's 
$\mathbb{W}_i(\widehat{\boldsymbol{\theta}}^{(0)}_{\Lambda_0}) = [\mathbb{D}_i(\widehat{\boldsymbol{\theta}}^{(0)}_{\Lambda_0})]^T [\widehat{\mathbb{V}}_i(\widehat{\boldsymbol{\theta}}_{ \Lambda_0}^{(0)})]^{-1} [\mathbb{D}_i(\widehat{\boldsymbol{\theta}}^{(0)}_{\Lambda_0})] \in \mathbb{R}^{p \times p}$, and $\mathbf{b}_i(\widehat{\boldsymbol{\theta}}^{(0)}_{\Lambda_0}) = [\mathbb{D}_i(\widehat{\boldsymbol{\theta}}^{(0)}_{\Lambda_0})] ^T [\widehat{\mathbb{V}}_i(\widehat{\boldsymbol{\theta}}_{ \Lambda_0}^{(0)})]^{-1} \left [ \boldsymbol{Y}_i - \widehat{\boldsymbol{\mu}}_i(\widehat{\boldsymbol{\theta}}^{(0)}_{\Lambda_0}) \right ] \in \mathbb{R}^{p}$. Second, we only need to estimate $\widehat{\boldsymbol{\theta}}^{(0)}_{\Lambda_0}$ once. We can then use that $\widehat{\boldsymbol{\theta}}^{(0)}_{\Lambda_0}$, calculated on the full sample, as the initial estimate for all folds and $\Lambda_1$ values. This is because any consistent initial estimate, $\widehat{\boldsymbol{\theta}}^{(0)}_{\Lambda_0}$, is sufficient to ensure that the one-step estimator of a given fold is consistent for the population $\boldsymbol{\theta}$. This strategy may be unnecessary for datasets where $K$ fold-specific initial estimates can be calculated quickly.
Third, assuming $\frac{1}{N} \sum_{i=1}^N  \mathbb{W}_i (\widehat{\boldsymbol{\theta}}^{(0)}_{\Lambda_0}) + \Lambda_1 \overset{\mathbb{P}}{\to}  \mathbb{E} \left [ \nabla_{\boldsymbol{\theta}} \boldsymbol{U}_{\Lambda_1}(\mathbf{X}_i, \boldsymbol{Y}_i; {\boldsymbol{\theta}}_{\Lambda_0})  \right ]$,
we can (heuristically, by Slutsky's theorem) calculate consistent one-step estimates in fold $k$ as
\begin{align}
\widehat{\boldsymbol{\theta}}^{k}_{{\Lambda_1}} = \widehat{\boldsymbol{\theta}}^{(0)}_{{\Lambda_0}} + \left [  \frac{1}{N} \sum_{i=1}^N  \mathbb{W}_i (\widehat{\boldsymbol{\theta}}^{(0)}_{\Lambda_0}) + \Lambda_1 \mathbb{S} \right ]^{-1} \frac{1}{N} \sum_{i \not\in \mathcal{K}_k} \left \{ \tilde{n}_k \mathbf{b}_i(\widehat{\boldsymbol{\theta}}^{(0)}_{\Lambda_0})  - \Lambda_1\mathbb{S} \widehat{\boldsymbol{\theta}}^{(0)}_{\Lambda_0} \right \},
\end{align}
where $\tilde{n}_k = \frac{\sum_{i=1}^N n_i}{\sum_{i \not\in \mathcal{K}_k}n_i}$. By using the full sample estimate $\left [  \frac{1}{N} \sum_{i=1}^N  \mathbb{W}_i (\widehat{\boldsymbol{\theta}}^{(0)}_{\Lambda_0}) + \Lambda_1 \mathbb{S}  \right ]^{-1}$, we only need to invert this $p \times p$ matrix once for each value of $\Lambda_1$, instead of inverting a fold-specific $p \times p$ matrix for each unique $\{k, \Lambda_1\}$ pair. The strategy of keeping $\widehat{\boldsymbol{\theta}}^{(0)}_{{\Lambda_0}}$ and $\left [  \frac{1}{N} \sum_{i=1}^N  \mathbb{W}_i (\widehat{\boldsymbol{\theta}}^{(0)}_{\Lambda_0}) + \Lambda_1 \mathbb{S}  \right ]^{-1}$ fixed across folds is motivated by an analogous strategy for cluster bootstrapping of unpenalized one-step GEE (see Remark and Theorem 3.3 in \cite{bootstrap_consistency}). Specifically, \cite{bootstrap_consistency} showed that a cluster bootstrap that fixes these two quantities (at the full-sample estimates) across replicates enjoys the same theoretical guarantees asymptotically as an approach that re-estimates these quantities in each replicate-specific sample. In our simulations, our adaptation of this strategy for cluster CV was often dramatically faster than, and performed nearly identically to, a CV strategy that calculates $\widehat{\boldsymbol{\theta}}^{k}_{{\Lambda_1}}$ using the fold-specific estimate 
$\left [  \frac{1}{N-|\mathcal{K}_k|} \sum_{i \not\in \mathcal{K}_k}  \mathbb{W}_i (\widehat{\boldsymbol{\theta}}^{(0)}_{\Lambda_0}) + \left(\frac{\sum_{i \not\in \mathcal{K}_k}n_i}{\sum_{i=1}^N n_i} \right) \Lambda_1 \mathbb{S} \right ]^{-1}$.
%
%
Fourth, we avoid tuning over a large grid of $\Lambda_1$ values by 
using a sequential CV procedure (see Appendix~\ref{app:seq_cv} below for details). 
We found these strategies performed well with $K=10$ in our simulations and data application. 

We describe the fast cluster CV procedure in Algorithm~\ref{algo:cv}:

\begin{algorithm}[htbp] 
\caption{Fast cluster CV for one-step fGEE}  
\label{algo:cv}
\begin{algorithmic}[1]
    \State Inputs: 
    \begin{itemize}
    \item Initial estimate $\widehat{\boldsymbol{\theta}}^{(0)}_{\boldsymbol{\Lambda}_0}$
    \item Penalty matrix $\mathbb{S}$
    \item Cluster folds $\{\mathcal{K}_1,\dots,\mathcal{K}_K\}$ where for $K \geq2$ $\mathcal{K}_k \subset [N], ~\cup_k\mathcal{K}_k=[N], ~\cap_k \mathcal{K}_k = \emptyset$
    \item Candidate grid $\mathcal{G}$ of smoothing parameter matrices $\boldsymbol{\Lambda}$
    \item Precomputed \{$\mathbb{W}_i(\widehat{\boldsymbol{\theta}}^{(0)}_{\Lambda_0}), \mathbf{b}_i(\widehat{\boldsymbol{\theta}}^{(0)}_{\Lambda_0}) \}_{i=1}^N$ 
    \item $\mathrm{FoldLoss}(\cdot)$: cross-validated error criteria (e.g., squared-error, negative log-likelihood)
    \begin{itemize}
        \item $\mathbb{W}_i \equiv \mathbb{W}_i(\widehat{\boldsymbol{\theta}}^{(0)}_{\Lambda_0}) = \{\mathbb{D}_i(\widehat{\boldsymbol{\theta}}^{(0)}_{\Lambda_0})\}^T \{\widehat{\mathbb{V}}_i(\widehat{\boldsymbol{\theta}}_{ \Lambda_0}^{(0)})\}^{-1} \{\mathbb{D}_i(\widehat{\boldsymbol{\theta}}^{(0)}_{\Lambda_0})\} \in \mathbb{R}^{p \times p}$
        \item $\mathbf{b}_i \equiv \mathbf{b}_i(\widehat{\boldsymbol{\theta}}^{(0)}_{\Lambda_0}) = \{\mathbb{D}_i(\widehat{\boldsymbol{\theta}}^{(0)}_{\Lambda_0})\} ^T \{\widehat{\mathbb{V}}_i(\widehat{\boldsymbol{\theta}}_{ \Lambda_0}^{(0)})\}^{-1} \left \{ \boldsymbol{Y}_i - \widehat{\boldsymbol{\mu}}_i(\widehat{\boldsymbol{\theta}}^{(0)}_{\Lambda_0}) \right \} \in \mathbb{R}^{p}$
    \end{itemize}
\end{itemize}

\For{each $\boldsymbol{\Lambda}\in\mathcal{G}$}
\State $\boldsymbol{H}_{\mathrm{glob}}(\boldsymbol{\Lambda}) \gets \Big(\frac{1}{N}\sum_{i=1}^N \mathbb{W}_i\Big) + \boldsymbol{\Lambda}\,\mathbb{S}$.
\State Invert $\boldsymbol{H}_{\mathrm{glob}}(\boldsymbol{\Lambda})$ and store $\boldsymbol{H}_{\mathrm{glob}}(\boldsymbol{\Lambda})^{-1}$
  \For{$k=1$ to $K$}
    \State $\tilde{n}_k \gets \Big(\sum_{i=1}^N n_i\Big)\Big/\Big(\sum_{i\notin \mathcal{K}_k} n_i\Big)$
    \State $\widehat{\boldsymbol{\theta}}^{(1)}_{k}(\boldsymbol{\Lambda}) \gets \widehat{\boldsymbol{\theta}}^{(0)}_{\boldsymbol{\Lambda}_0}
    + \boldsymbol{H}_{\mathrm{glob}}(\boldsymbol{\Lambda})^{-1}\Big[\frac{1}{N}\sum_{i\notin \mathcal{K}_k}\big\{\tilde{n}_k\,\mathbf{b}_i - \boldsymbol{\Lambda}\,\mathbb{S}\,\widehat{\boldsymbol{\theta}}^{(0)}_{\boldsymbol{\Lambda}_0}\big\}\Big]$
    \State $\mathrm{NLL}_k(\boldsymbol{\Lambda}) \gets \mathrm{FoldLoss}\big(\widehat{\boldsymbol{\theta}}^{(1)}_{k}(\boldsymbol{\Lambda}),\,\mathcal{K}_k\big)$
  \EndFor
  \State $\mathrm{CV}(\boldsymbol{\Lambda}) \gets \frac{1}{K}\sum_{k=1}^K \mathrm{NLL}_k(\boldsymbol{\Lambda})$
\EndFor
\State $\boldsymbol{\Lambda}^\star \gets \arg\min_{\boldsymbol{\Lambda}\in\mathcal{G}} \mathrm{CV}(\boldsymbol{\Lambda})$
\State \Return $\boldsymbol{\Lambda}^\star$
\end{algorithmic}
\end{algorithm}

\subsubsection{Full-Sample Cross Validation} \label{app:full_cv}
We define a standard cluster (K-Fold) cross-validation one-step estimator as
\begin{align}
\widehat{\boldsymbol{\theta}}^{k}_{{\Lambda_1}} = \widehat{\boldsymbol{\theta}}^{(0)}_{\Lambda_0} + \frac{1}{N-|\mathcal{K}_k|}\left [  \frac{1}{N-|\mathcal{K}_k|} \sum_{i \not\in \mathcal{K}_k}  \mathbb{W}_i (\widehat{\boldsymbol{\theta}}^{(0)}_{\Lambda_0}) + {n}^*_k \Lambda_1 \mathbb{S} \right ]^{-1} \sum_{i \not\in \mathcal{K}_k} \left \{  \mathbb{B}_i(\widehat{\boldsymbol{\theta}}^{(0)}_{\Lambda_0})  - {n}^*_k \Lambda_1\mathbb{S} \widehat{\boldsymbol{\theta}}^{(0)}_{\Lambda_0} \right \},
\end{align}
where ${n}^*_k  = \frac{\sum_{i \not\in \mathcal{K}_k}n_i}{\sum_{i=1}^N n_i}$. {In preliminary simulations (not shown here), one-step fGEE estimates fit with smoothing parameters selected with the full-sample and fast cluster CVs produced nearly identical functional coefficient RMSE and empirical 95\% CI coverage values. For that reason, all results shown in the main text and Appendix are obtained from fGEEs fit with smoothing parameters selected with the fast cluster CV.}

\subsubsection{Sequential Tuning Procedure} \label{app:seq_cv}
The diagonal smoothing matrix $\Lambda$ contains smoothing parameters $\lambda_1, \ldots, \lambda_q$, each repeated based on the number of knots used for its functional coefficient. To identify the correct range for the $\lambda_1, \ldots, \lambda_q$, we apply an iterative CV strategy that is designed to be fast as $q$ grows. By dividing the tuning into three stages, we avoid tuning over a large $q$ dimensional grid that can be computationally impractical even for $q \geq 3$. In step (1), we tune over a small one-dimensional grid to identify the correct order of magnitude for the smoothing parameters: $\boldsymbol{\Lambda}_{(1)} = \{\alpha_1 \Lambda^{(0)}, \ldots, \alpha_L \Lambda^{(0)}\}$ where, for example, $\{\alpha_1, \ldots, \alpha_{L_1}\} = \{0.001, 0.01, 0.1, 1, 10, 100, 1000\}$. This exploits the fact that the unique diagonal entries $\lambda_1^{(0)},\ldots,\lambda_q^{(0)}$ of the $\Lambda_0$ selected by REML are, in our experience, on a reasonable relative scale. 
In step (2), we tune over a small $q$-dimensional grid constructed around the smoothing parameter values that minimize the cross-validated criteria (e.g. MSE, negative log-likelihood) in step (1). Denoting $\tilde{\lambda}^{(1)}_1, \ldots, \tilde{\lambda}^{(1)}_q$ as these selected values, we tune over a grid with all unique combinations (with the $\texttt{R}$ function $\texttt{expand.grid}()$) of  grid $\boldsymbol{\Lambda}_{(2)} = \left \{ \{\alpha_1 \tilde{\lambda}^{(1)}_1, \ldots, \alpha_{L_2} \tilde{\lambda}^{(1)}_1\}, \ldots, \{\alpha_1 \tilde{\lambda}^{(1)}_q, \ldots, \alpha_L \tilde{\lambda}^{(1)}_q\} \right \}$, where, for example, $\{\alpha_1, \ldots, \alpha_{L_2}\}=\{0.001, 0.01, 0.1, 1, 10, 100, 1000\}$. Step (2) identifies a reasonable order of magnitude for each smoothing parameter with a small grid. Finally, in step (3), we tune each smoothing parameter within a local neighborhood around the selected values from step (2), $\tilde{\lambda}^{(2)}_1,\ldots, \tilde{\lambda}^{(2)}_q$. Specifically, we tune over the unique combinations, $\boldsymbol{\Lambda}_{(3)} = \left \{ \{\alpha_1 \tilde{\lambda}^{(2)}_1, \ldots, \alpha_{L_3} \tilde{\lambda}^{(2)}_1\}, \ldots, \{\alpha_1 \tilde{\lambda}^{(2)}_q, \ldots, \alpha_{L_3} \tilde{\lambda}^{(2)}_q\} \right \}$, where, for example, $\{\alpha_1, \ldots, \alpha_{L_3}\} \subset [0.1, 10]$. We found this was fast and performed well in simulations, across a wide range of $N$ and $n_i$ values. We anticipate the relative speedups of using fast cluster CV, compared to standard K-fold CV, would grow as $p$ increases (e.g. from increasing the number of knots or covariates, $q$). This is because fast K-fold inverts a $p \times p$ matrix only once per $\Lambda_1$ value, whereas standard K-fold inverts a similar $p \times p$ matrix for every unique $\{k, \Lambda_1\}$ pair.

{
\subsection{Fast Wild Cluster Bootstrap} \label{app:wild}

Motivated by classic \citep{cameron2008bootstrap} and score-based wild cluster bootstrapping \citep{kline2010score} approaches, we apply a fast wild cluster bootstrap to calculate quantiles to construct Wald 95\% CIs. This can also be described as a special case of a cluster multiplier bootstrap with Rademacher weights. We describe our approach in Algorithm~\ref{algo:wild} below. Since all bootstrap replicate one-step estimates use the same initial estimate $\widehat{\boldsymbol{\theta}}^{(0)}_{{\Lambda_1}}$ and same inverse $\boldsymbol{H}^{-1}=\left [  \frac{1}{N} \sum_{i=1}^N  \mathbb{W}_i (\widehat{\boldsymbol{\theta}}^{(0)}_{{\Lambda_0}}) + \Lambda_1 \mathbb{S} \right ]^{-1}$, this bootstrapping procedure typically takes less than a second for moderately sized $p$. This approach of keeping the above two quantities fixed is motivated by theory for cluster bootstrapping for non-functional (unpenalized) one-step GEE \citep{bootstrap_consistency} as well as for score-based wild cluster bootstrapping \citep{kline2010score}. 

We also apply a small cluster effective degree of freedom (edf) inflation factor based on the edf definition described in \cite{wood2016smoothing} adapted to the GEE setting. Specifically, denoting $\bar{\mathbb{W}} = \mathbb{W}_i(\widehat{\boldsymbol{\theta}}^{(0)}_{\boldsymbol{\Lambda}_0})$, and penalty matrix $\boldsymbol{P} \coloneqq \boldsymbol{\Lambda}_1\,\mathbb{S}$ we take edf$_r = tr(\{ \bar{\mathbb{W}} + \boldsymbol{P} \}^{-1} \bar{\mathbb{W}})_r$, where the $r$ subscript means we take the diagonal elements corresponding to the basis coefficients, $\widehat{\boldsymbol{\theta}}$, associated with functional coefficient $r$. We found that calculating the df$_r$ for each functional coefficient separately yielded CIs that achieved roughly nominal coverage. This was preferable to calculating an effective degree of freedom for the entire model because such an approach would prevent us from using the correction for small cluster settings where $\sum_r \text{edf}_r > N$ (i.e., our approach allows for corrections as long as $\text{edf}_r < N$ for all $r \in [p]$). 
\begin{algorithm}[htbp] 
\caption{Wild cluster bootstrap critical values for joint/pointwise CIs for one-step fGEE}
\label{algo:wild}
\begin{algorithmic}[1]
\footnotesize 
\State \textbf{Inputs:}
\begin{itemize}
    \item \textbf{Pivot} coefficient estimate $\widehat{\boldsymbol{\theta}}^{(0)}_{\boldsymbol{\Lambda}_0} \in \mathbb{R}^p$ (initial estimate; used for bootstrap centering).
    \item \textbf{Center} (target) coefficient estimate $\widehat{\boldsymbol{\theta}}^{(1)}_{\boldsymbol{\Lambda}_1} \in \mathbb{R}^p$ (one-step or final estimate to center reported bands).
    \item Penalty matrix $\boldsymbol{P} \coloneqq \boldsymbol{\Lambda}_1\,\mathbb{S} \in \mathbb{R}^{p\times p}$.
    \item Cluster-wise quantities evaluated at $\widehat{\boldsymbol{\theta}}^{(0)}_{\boldsymbol{\Lambda}_0}$:
    $\{\mathbb{W}_i(\widehat{\boldsymbol{\theta}}^{(0)}_{\boldsymbol{\Lambda}_0}),\mathbf{b}_i(\widehat{\boldsymbol{\theta}}^{(0)}_{\boldsymbol{\Lambda}_0})\}_{i=1}^N$, where
    \begin{itemize}
        \item $\mathbb{W}_i(\widehat{\boldsymbol{\theta}}^{(0)}_{\boldsymbol{\Lambda}_0})\in\mathbb{R}^{p\times p}$,
        \item $\mathbf{b}_i(\widehat{\boldsymbol{\theta}}^{(0)}_{\boldsymbol{\Lambda}_0})\in\mathbb{R}^{p}$.
    \end{itemize}
    \item Evaluation grid $\{s_\ell\}_{\ell=1}^L\subset\mathcal{S}$ and term-wise linear maps $\{\mathbb{A}_r\}_{r=1}^q$ with $\mathbb{A}_r\in\mathbb{R}^{L\times p}$.
    \item Bootstrap replicates $B$, nominal level $\alpha\in(0,1)$.
    \item Covariance estimate $\widehat{\mathrm{Var}}(\widehat{\boldsymbol{\theta}}^{(1)}_{\boldsymbol{\Lambda}_1})=\widehat{\boldsymbol{\Sigma}}\in\mathbb{R}^{p\times p}$ for studentization.
\end{itemize}

\vspace{0.25em}
\State \textbf{Precomputation:}
\State $\overline{\mathbb{W}} \gets \frac{1}{N}\sum_{i=1}^N \mathbb{W}_i(\widehat{\boldsymbol{\theta}}^{(0)}_{\boldsymbol{\Lambda}_0})$.
\State $\boldsymbol{H} \gets \overline{\mathbb{W}} + \boldsymbol{P}$.
\State Compute and store $\boldsymbol{H}^{-1}$.
\State $\mathbf{p} \gets \boldsymbol{P} \widehat{\boldsymbol{\theta}}^{(0)}_{\boldsymbol{\Lambda}_0}$ \hfill (penalty contribution in one-step score)

\vspace{0.25em}
\State \textbf{Compute pivot curves, center curves, and fixed SE curves (term-wise):}
\For{$r=1$ to $q$}
    \State $\widehat{\boldsymbol{\beta}}^{\mathrm{pivot}}_r \gets \mathbb{A}_r \widehat{\boldsymbol{\theta}}^{(0)}_{\boldsymbol{\Lambda}_0}\in \mathbb{R}^{L}$.
    \State $\widehat{\boldsymbol{\beta}}^{\mathrm{ctr}}_r \gets \mathbb{A}_r \widehat{\boldsymbol{\theta}}^{(1)}_{\boldsymbol{\Lambda}_1}\in \mathbb{R}^{L}$.
    \State $\widehat{\mathbf{s}}_r \gets \sqrt{\mathrm{diag}\!\big(\mathbb{A}_r\,\widehat{\boldsymbol{\Sigma}}\,\mathbb{A}_r^T\big)} \in \mathbb{R}^{L}$.
\EndFor

\vspace{0.25em}
\State \textbf{Wild cluster bootstrap: generate studentized processes}
\For{$b=1$ to $B$}
    \State Draw i.i.d. Rademacher multipliers $\{\xi_i^{(b)}\}_{i=1}^N$ with $\mathbb{P}(\xi_i^{(b)}=\pm 1)=1/2$.
    \State $\mathbf{B}^{(b)} \gets \frac{1}{N}\sum_{i=1}^N \xi_i^{(b)}\,\mathbf{b}_i(\widehat{\boldsymbol{\theta}}^{(0)}_{\boldsymbol{\Lambda}_0}) \in\mathbb{R}^{p}$.
    \State $\widehat{\boldsymbol{\theta}}^{(b)} \gets \widehat{\boldsymbol{\theta}}^{(0)}_{\boldsymbol{\Lambda}_0}
      + \boldsymbol{H}^{-1}\left(\mathbf{B}^{(b)} - \mathbf{p}\right)$ \hfill (one-step bootstrap update; fixed bread)

    \For{$r=1$ to $q$}
        \State $\boldsymbol{\beta}_r^{(b)} \gets \mathbb{A}_r \widehat{\boldsymbol{\theta}}^{(b)} \in \mathbb{R}^{L}$.
        \State $\mathbf{T}_r^{(b)} \gets 
        \left(\boldsymbol{\beta}_r^{(b)}-\widehat{\boldsymbol{\beta}}^{\mathrm{pivot}}_r\right)\oslash \widehat{\mathbf{s}}_r
        \in\mathbb{R}^{L}.$ \hfill (studentized process; centered at pivot)
        \State $M_r^{(b)} \gets \max_{\ell \in [L] }\left|T_{r,\ell}^{(b)}\right|$.
        \State Store pooled pointwise magnitudes $\{|T_{r,\ell}^{(b)}|:\ell \in [L]\}$ for scalar pointwise calibration.
    \EndFor
\EndFor

\vspace{0.25em}
\State \textbf{Critical values (per curve $r$):}
\For{$r=1$ to $q$}
    \State $c^{\mathrm{jt}}_r \gets \mathrm{Quantile}_{1-\alpha}\left(\{M_r^{(b)}\}_{b=1}^B\right)$.
    \State $c^{\mathrm{pt}}_r \gets \mathrm{Quantile}_{1-\alpha}\left(\{|T_{r,\ell}^{(b)}|:\; b \in [B], \ell \in [L]\}\right)$.
\EndFor

\vspace{0.25em}
\State \textbf{(Optional) EDF/\emph{t}-adjustment of critical values:}
\For{$r=1$ to $q$}
    \State $\mathrm{df}_r \gets \max(\mathrm{df}_{\min},\,N-\mathrm{edf}_r)$.
    \State $a_r \gets \dfrac{t_{1-\alpha/2,\mathrm{df}_r}}{z_{1-\alpha/2}}$.
\EndFor

\vspace{0.25em}
\State \textbf{Construct confidence bands (reported bands centered at $\widehat{\boldsymbol{\beta}}^{\mathrm{ctr}}_r$):}
\For{$r=1$ to $q$}
\State $\widehat{\beta}^{\mathrm{ctr}}_r(s_\ell)\ \pm\ a_rc^{\mathrm{jt}}_r\,\widehat{s}_r(s_\ell),
\qquad \ell=1,\ldots,L.$
\State $\widehat{\beta}^{\mathrm{ctr}}_r(s_\ell)\ \pm\ a_rc^{\mathrm{pt}}_r\,\widehat{s}_r(s_\ell),
\qquad \ell=1,\ldots,L.$
\EndFor
\State \Return $\{a_rc^{\mathrm{jt}}_r,~a_rc^{\mathrm{pt}}_r\}_{r=1}^p$ and the corresponding joint/pointwise scalar bands.
\end{algorithmic}
\end{algorithm}
}

\subsection{Fast Cluster Bootstrap Variance Estimator} \label{app:fast_boot}
Motivated by theory for cluster bootstrapping in non-functional one-step GEE \citep{bootstrap_consistency}, we propose a fast cluster bootstrap as an alternative method to estimate $\text{Var}(\widehat{\boldsymbol{\theta}}_{{\Lambda_1}}^{(1)})$, or to construct non-parametric bootstrap-based joint CIs. Namely, for bootstrap replicate $t$
\begin{align}\label{eq:fcbve}
\widehat{\boldsymbol{\theta}}^{t}_{{\Lambda_1}} = \widehat{\boldsymbol{\theta}}^{(0)}_{{\Lambda_1}} + \left [  \frac{1}{N} \sum_{i=1}^N  \mathbb{W}_i (\widehat{\boldsymbol{\theta}}^{(0)}_{{\Lambda_0}}) + \Lambda_1 \mathbb{S} \right ]^{-1} \frac{1}{N} \sum_{i \in \mathcal{R}_t} \left \{ \tilde{n}_t \mathbf{b}_i(\widehat{\boldsymbol{\theta}}^{(0)}_{{\Lambda_0}})  - \Lambda_1\mathbb{S} \widehat{\boldsymbol{\theta}}^{(0)}_{{\Lambda_0}} \right \},
\end{align}
where $\mathcal{R}_t$ is a set of cluster indices of size $N$, sampled with replacement, and $\tilde{n}_t = \frac{\sum_{i=1}^N n_i}{\sum_{i \in \mathcal{R}_t}n_i}$. We estimate $\text{Var}_{\text{boot}} \left( \widehat{\boldsymbol{\theta}}_{\Lambda}^{(1)} \right)$ as the sample covariance matrix of the $T$ bootstrap replicates. Since equation \eqref{eq:fcbve} uses the same initial estimate $\widehat{\boldsymbol{\theta}}^{(0)}_{{\Lambda_1}}$ and keeps the matrix $\left [  \frac{1}{N} \sum_{i=1}^N  \mathbb{W}_i (\widehat{\boldsymbol{\theta}}^{(0)}_{{\Lambda_0}}) + \Lambda_1 \mathbb{S} \right ]^{-1}$ fixed for all $t$, this bootstrapping procedure typically takes less than a second for moderately sized $p$. 
{In preliminary simulations, we found empirical coverage to be comparable when constructing CIs} with sandwich and fast bootstrap variance estimators. 

\subsection{Initial $\texttt{pffr}$ fit} \label{app:pffr}
As an initial fit, the penalized maximum likelihood estimator
\begin{align} \label{eq:pffr_loss}
\widehat{\boldsymbol{\theta}}_{\Lambda_0}^{(0)} = \underset{{\boldsymbol{\theta}} }{\mbox{argmin }}~ -2\sum_{i=1}^N \sum_{j=1}^{n_i} \sum_{s \in \mathcal{S}} {l}({Y}_{i,j}(s), \mathbb{X}_{i,j}; \boldsymbol{\theta}) + \boldsymbol{\theta}^T\Lambda_0 \mathbb{S}\boldsymbol{\theta},
\end{align}
where ${l}({Y}_{i,j}(s), \mathbb{X}_{i,j}; \boldsymbol{\theta})$ is the log-likelihood evaluated on a single observation of the outcome ${Y}_{i,j}(s) \in \mathbb{R}$ from cluster $i$, at longitudinal observation $j$, at functional domain point $s$. We denote the model parameter $\boldsymbol{\theta}$, $\mathbb{S}$ as the penalty matrix, and $\Lambda_0$ as an associated diagonal matrix of smoothing parameters. As discussed in the main text, $\mathbb{X}_{i,j} \in \mathbb{R}^p$ is constructed as a product of pre-defined basis functions (e.g. B-splines) and the original covariates, $\boldsymbol{x}_{i,j} \in \mathbb{R}^q$. Model~\ref{eq:pffr_loss} is equivalent to adopting a correlation structure that assumes all observations $Y_{i,j}(s)$ are mutually independent across $i$, $j$, and $s$. 

We fit model~\ref{eq:pffr_loss} with the $\texttt{refund}$ package in $\texttt{R}$ \citep{refund} with the $\texttt{pffr}$ function \citep{pffr}. This calls the $\texttt{mgcv}$ package in $\texttt{R}$ \citep{mgcv} to fit the model with the $\texttt{gam}$ or $\texttt{bam}$ functions. For example, for a model with two covariates, $\texttt{X1}$ and $\texttt{X2}$, we use the following code:
\begin{verbatim}
initial_fit = refund::pffr(Y ~ X1 + X2, 
                          family = fam,
                          algorithm = "bam",
                          method = "fREML",
                          discrete = TRUE,
                          bs.yindex = list(bs = spline.basis, 
                                           k = knots,
                                           m = m.pffr),
                          data = data_df)
\end{verbatim}

where $\texttt{fam}$ is the exponential dispersion family adopted to construct a pseudo-likelihood based estimating equation in the fGEE, $\texttt{spline.basis}$ is the pre-specified basis spline from the $\texttt{mgcv}$ family (e.g. $\texttt{bs}$, $\texttt{ps}$, $\texttt{tp}$), $\texttt{knots}$ is the number of knots, and $\texttt{m.pffr}$ is the penalty type. For example, as noted in the $\texttt{pffr}$ function, $\texttt{bs.yindex = list(bs="ps", k=5, m=c(2, 1))}$ indicates 5 cubic B-splines bases with a first order difference penalty.

{We use the above fit as both the initial estimate for the one-step fGEE and as benchmark method ``pffr'' in the simulations presented in Section~\ref{sec:simulations} and throughout the simulations presented in the Appendix. We do not use functional random effects in any $\texttt{pffr}$ fits.} 

\subsection{Computational Details}\label{app:comp}
Our implementation uses a number of $\texttt{R}$ packages for estimation of nuisance parameters and to improve computational speed. We use the $\texttt{data.table}$ package extensively to increase computational efficiency \citep{data_table}. Our code structure {was initially} loosely based on the structure from the $\texttt{gee1step}$ package that implements the one-step (non-functional) GEE \citep{gee_onestep} available on the $\texttt{Github}$ \url{https://github.com/kgoldfeld/gee1step} of Professor Keith Goldfeld. We use the $\texttt{Rfast}$ package \citep{rfast} to estimate the $\rho(s)$ for AR1 correlation and to speed up other standard computations. We use the $\texttt{MASS}$ package to draw multivariate normals \citep{MASS_package}. We use \cite{irreg_ar1} to estimate and invert covariance matrices with an AR1 structure when the time intervals are irregular. We use the $\texttt{SuperGauss}$ package \citep{SuperGauss} to invert Toeplitz correlation matrices.  We use the $\texttt{sanic}$ package to quickly invert other positive definite matrices \citep{sanic_package}. As mentioned in the main text and previous Appendix section, we use $\texttt{mgcv}$ \citep{mgcv} and $\texttt{refund}$ \citep{refund} packages to estimate initial fits and to (optionally) smooth the correlation parameters across the functional domain.

\subsection{Description of the use of AI} \label{app:ai}
{
We used Large Language Models (LLMs) to extend and optimize our existing package code (described below), modify our existing code to generate tables summarizing simulation results, and to generate $\texttt{ggplot2}$ code to generate Appendix Figure~\ref{fig:coef_est_grid}, which summarizes simulation results. Specifically, we used LLMs to write a computationally efficient version of the copula method to generate longitudinal functional outcomes given covariates that we used to simulate data. We did not use LLMs for other portions of our simulations. We extensively verified and tested code produced by LLMs with large simulation experiments. We used LLMs only in code added in revisions of an original manuscript, not in any code of the original manuscript. We did not use LLMs for writing or editing the manuscript, any of the mathematical proofs, or core methodological ideas. We did, however, revise some of the Theorem 3.1 wording after LLMs pointed out that the limiting variance of the functional coefficient estimates can be rank deficient depending on how the grid of functional domain points is selected.

The key areas we used LLMs for in extending our original package code are: 1) implementing an efficient $\texttt{data.table}$-based version of an initial Kronecker product working covariance; 2) the FPCA-based working covariance; 3) extension of the fGEE to link functions beyond log, logit and identity; 4) extension to estimating equations based on pseudo-distributions other than those deriving from Poisson, binomial, and Gaussian; 5) calculation of the small-$N$ effective degree of freedom inflation factor; 6) implementation of a computationally efficient version of the wild cluster bootstrap; and 7) implementation of speedups for the fast cluster CV (e.g., using $\texttt{Rcpp}$-based evaluations of the negative log-likelihood). We have acknowledged LLM use in the package documentation. When LLMs were used in generating code outside the package, we acknowledge their use in the R files (see Github repo: \url{https://github.com/gloewing/fgee_onestep}). All LLMs were accessed through the platforms provided to NIH Intramural Research Program by the Department of Health and Human Services (HHS), namely Gemini 2.5 Pro, Claude Sonnet 4.5 (while available to HHS employees), and ChatGPT 5.2 Pro.
}

\section{Theory} \label{app:theory}
\subsection{Main Text Theorem Conditions} \label{app:theory_conditions}
 \begin{enumerate}[(i)]
    \item The inverse link function $g^{-1}$ is three times continuously differentiable.
    \item The covariates and outcomes have bounded support, i.e. $\exists M > 0$ such that $P[\lVert \boldsymbol{Y}_i(s) \rVert \leq M] = 1$ for all $s \in \mathcal{S}$, and $P[|X_{i,j,r}| < M] = 1$, for all $j \in [n_i]$ and $r \in [q]$.
     \item $\exists ~ s,t >0~:~\lambda_{\mathrm{min}}(\boldsymbol{M}_N(\boldsymbol{\theta}_N)) \geq s$, $\mathbb{E}( | U_{N,j}(\mathbf{X}_i, \boldsymbol{Y}_i; \boldsymbol{\theta}) U_{N,k}(\mathbf{X}_i, \boldsymbol{Y}_i; \boldsymbol{\theta}) U_{N,l}(\mathbf{X}_i, \boldsymbol{Y}_i; \boldsymbol{\theta})|) \leq t, \forall n \in \mathbb{N}, \forall~j,k,l$.
     \item $\boldsymbol{H}_N(\boldsymbol{\theta}_N)$ is invertible, $\mathbb{P} \left [  \mathbb{P}_N\left(\nabla_{\boldsymbol{\theta}} \boldsymbol{U}_{N}(\mathbf{X}_i, \boldsymbol{Y}_i;\boldsymbol{\theta}) \bigg |_{\boldsymbol{\theta} = \widehat{\boldsymbol{\theta}}^{(0)}_N}\right)~\text{non-singular} \right ] = 1$, \\ and $\left(\mathbb{P}_N\left[\nabla_{\boldsymbol{\theta}} \boldsymbol{U}_{N}(\mathbf{X}_i, \boldsymbol{Y}_i;\boldsymbol{\theta}) \bigg |_{\boldsymbol{\theta} = \widehat{\boldsymbol{\theta}}^{(0)}_N}\right] \right )^{-1} = O_{\mathbb{P}}(1)$.
     \item $\boldsymbol{M}_N(\boldsymbol{\theta}_N) = O(1)$, $\boldsymbol{H}_N(\boldsymbol{\theta}_N)   = O(1)$ and $\{\boldsymbol{H}_N(\boldsymbol{\theta}_N) \}^{-1}  = O(1)$.
     \item $\sqrt{N} \{\boldsymbol{M}_N(\boldsymbol{\theta}_N) \}^{-1/2} \boldsymbol{H}_N(\boldsymbol{\theta}_N)(\widehat{\boldsymbol{\theta}}^{(0)}_{N} - \boldsymbol{\theta}_N) = O_{\mathbb{P}}(1)$.
 \end{enumerate}

 Condition (i) is a mild smoothness condition that holds for all standard link functions (e.g. logit, log). Condition (ii) is also standard---we expect it holds across essentially all biomedical settings. Note that it could be replaced by weaker moments conditions on the estimating equation, and its derivatives. Condition (iii) is a sufficient condition for the estimating equation to be asymptotically normal and implies that $\boldsymbol{M}_N(\boldsymbol{\theta}_N)$ is invertible for all $N \in \mathbb{N}$. Condition (iv) also states that the $\boldsymbol{H}_N(\boldsymbol{\theta}_N)$, and its sample analogue, are invertible for all $N \in \mathbb{N}$. Condition (v) should hold when the limiting (unpenalized) estimating equation results in full rank limiting $\boldsymbol{M}_N$ and $\boldsymbol{H}_N$. This should hold when the design matrices, $\mathbb{X}_i$, are full rank. Finally condition (vi) is a statement about the rate of convergence of the initial estimator. In practice, when the $\widehat{\boldsymbol{\theta}}^{(0)}_{N}$ is estimated using a penalized unweighted estimating equation, this implies some conditions on the rates of convergence of the smoothing parameter values, $\Lambda_{0,N}$ and $\Lambda_N$. We provide an expanded discussion of this below.
\subsection{General Result} 
Throughout, we fix an arbitrary matrix norm $\lVert \, \cdot \,\rVert$ (e.g., operator, Frobenius), and for a sequence of random matrices $(A_N)_{N = 1}^{\infty}$, we write $A_N = O_{\mathbb{P}}(1)$ if $\lVert A_N\rVert = O_{\mathbb{P}}(1)$, and $A_N = o_{\mathbb{P}}(1)$ if $\lVert A_N\rVert = o_{\mathbb{P}}(1)$. For symmetric matrix $A$, we write $\lambda_{\mathrm{min}}(A)$ and $\lambda_{\mathrm{max}}(A)$ for the smallest and largest eigenvalues of $A$, respectively. We begin by stating and proving a lemma that we use to prove our general theorem.
\begin{lemma}\label{lemma_matrix_slut}
Let $A_N \in \mathbb{R}^{p \times p}$ be a sequence of fixed and invertible matrices such that $A_N = O(1)$ and $A_N^{-1} = O(1)$, let $X_N \in \mathbb{R}^p$ be a sequence of random vectors such that 
for constant scalars $r_N \to \infty$,
$r_N A_N X_N = O_{\mathbb{P}}(1)$. If $B_N \in \mathbb{R}^{p \times p}$ is a sequence of random matrices that satisfy $B_N \overset{\mathbb{P}}{\to} \boldsymbol{0}_{p \times p}$, then $r_N A_N B_NX_N  \overset{\mathbb{P}}{\to} \boldsymbol{0}_p$.
\end{lemma}

\begin{proof}
Noting that $A_N B_N = O(1) o_{\mathbb{P}}(1) = o_{\mathbb{P}}(1)$, and
    \[
     r_N X_N = A_N^{-1} (r_N A_N X_N) = O(1) O_{\mathbb{P}}(1) = O_{\mathbb{P}}(1) \]
     we can immediately conclude that
     \[r_N A_N B_N X_N = A_N B_N (r_N X_N) = o_{\mathbb{P}}(1) O_{\mathbb{P}}(1) = o_{\mathbb{P}}(1),\]
     as claimed.
\end{proof}

We now consider a general adaptive $M$-estimation setting. Suppose we observe an iid sequence of random vectors $(Z_i)_{i = 1}^{\infty}$, with generic observation denoted $Z \sim \mathbb{P}$, and for each fixed sample size $N$ we work with the differentiable (in $\boldsymbol{\theta}$) estimating equation $\boldsymbol{U}_N(Z; \boldsymbol{\theta}) \in \mathbb{R}^p$, for parameters $\boldsymbol{\theta} \in \mathbb{R}^p$. We write $U_{N,\ell}(Z; \boldsymbol{\theta})$ for the $\ell$-th component of $\boldsymbol{U}_N(Z;\boldsymbol{\theta})$. The ``fully iterated'' estimator $\widehat{\boldsymbol{\theta}}_N^*$ would be given by solving
\[\mathbb{P}_N \left\{\boldsymbol{U}_N(Z; \boldsymbol{\theta})\right\}\equiv  \frac{1}{N}\sum_{i=1}^N \boldsymbol{U}_N(Z_i; \boldsymbol{\theta}) = \boldsymbol{0}_p,\]
targeting the population parameter $\boldsymbol{\theta}_N$ that solves $\mathbb{E}\{\boldsymbol{U}_N(Z; \boldsymbol{\theta})\} = \boldsymbol{0}_p$.
In practice, we will apply the general theory to the longitudinal functional setup by taking $Z_i \equiv (\mathbb{X}_i, \boldsymbol{Y}_i)$ and $\boldsymbol{U}_N(Z_i; \boldsymbol{\theta}) = \mathbb{D}_i^T \widetilde{\mathbb{V}}_{i,n}^{-1}(\boldsymbol{Y}_i - g^{-1}(\mathbb{X}_i\boldsymbol{\theta})) - \frac{1}{N}\Lambda_N\mathbb{S}\boldsymbol{\theta}$. In this special case, so long as $\Lambda_N \to \boldsymbol{0}_{p \times p}$ and $\widetilde{\mathbb{V}}_{n}^{-1} \to \mathbb{V}^{-1}$, we have $\boldsymbol{U}_N \to \boldsymbol{U}_{\infty}$ where $\boldsymbol{U}_{\infty}(Z; \boldsymbol{\theta}) = \mathbb{D}^T \mathbb{V}^{-1}(\boldsymbol{Y} - g^{-1}(\mathbb{X}\boldsymbol{\theta}))$, with corresponding parameter $\boldsymbol{\theta}_{\infty}$ solving $\mathbb{E}\{\boldsymbol{U}_{\infty}(Z; \boldsymbol{\theta})\} = \boldsymbol{0}_p$---note that we do not explicitly require such convergence in our general setup.

We require notation for a number of related important quantities. First, define the variance quantities $\boldsymbol{H}_N(\boldsymbol{\theta}) = \mathbb{E}\{\nabla_{\boldsymbol{\theta}}\boldsymbol{U}_N(Z; \boldsymbol{\theta}) \}$, $\boldsymbol{M}_N(\boldsymbol{\theta}) =\mathbb{E}\{\boldsymbol{U}_N(Z; \boldsymbol{\theta}) \boldsymbol{U}_N(Z; \boldsymbol{\theta})^T \}$.
Next, for an initial estimator $\widehat{\boldsymbol{\theta}}_N^{(0)}$, we define
\[\widehat{\boldsymbol{\theta}}^{(1)}_N \coloneqq \widehat{\boldsymbol{\theta}}^{(0)}_N - \left(\mathbb{P}_N \left [  \nabla_{\boldsymbol{\theta}} \boldsymbol{U}_{N}(Z; \boldsymbol{\theta}) \bigg |_{\boldsymbol{\theta} = \widehat{\boldsymbol{\theta}}^{(0)}_N} \right] \right)^{-1} \mathbb{P}_N \left [ \boldsymbol{U}_N(Z; \widehat{\boldsymbol{\theta}}^{(0)}_N) \right ].\] 
Under weak conditions, writing $\boldsymbol{W}_N \coloneqq \sqrt{N} \left\{\boldsymbol{M}_N(\boldsymbol{\theta}_N)\right\}^{-1/2} \boldsymbol{H}_N(\boldsymbol{\theta}_N)$, we typically have the following asymptotic normality result for the fully iterated estimator: \[\boldsymbol{W}_N\left(\widehat{\boldsymbol{\theta}}_N^* - \boldsymbol{\theta}_N\right) \overset{d}{\to} \mathcal{N}(\boldsymbol{0}_p, I_p).\] In the following result, we lay out conditions under which the one-step estimator $\widehat{\boldsymbol{\theta}}_N^{(1)}$ achieves the same convergence properties, i.e., is asymptotically equivalent to $\widehat{\boldsymbol{\theta}}_N^*$. 

\begin{theorem}\label{thm:asymp_app}
     Suppose the following conditions hold:
\begin{enumerate}[(i)]
    \item $\boldsymbol{U}_N(z; \boldsymbol{\theta})$ is twice differentiable in $\boldsymbol{\theta}$ for all $z$, and the second derivative is uniformly bounded:
     $\exists~C_1 > 0$ such that $\mathbb{P}\bigg[\sup_{\boldsymbol{\theta}}\bigg |\frac{\partial^2 U_{N, \ell}(Z;\boldsymbol{\theta})}{\partial \theta_j \partial \theta_k} \bigg | \leq C_1\bigg] = 1$, for all $N, j, k, \ell$.
     \item The second and third moments of $\boldsymbol{U}_N(Z;\boldsymbol{\theta}_N)$ are uniformly bounded below and above, respectively: $\exists ~ s,t >0$ such that $\lambda_{\mathrm{min}}(\boldsymbol{M}_N(\boldsymbol{\theta}_N)) \geq s$, and 
     \[\mathbb{E}( | U_{N,j}(Z; \boldsymbol{\theta}_N) U_{N,k}(Z; \boldsymbol{\theta}_N) U_{N,\ell}(Z; \boldsymbol{\theta}_N)|) \leq t,\] for all $N \in \mathbb{N}$ and all $j,k,\ell$.
     \item $\boldsymbol{H}_N(\boldsymbol{\theta}_N)$ is invertible and $\mathbb{P} \left [  \mathbb{P}_N\left(\nabla_{\boldsymbol{\theta}} \boldsymbol{U}_{N}(Z; \boldsymbol{\theta}) \bigg |_{\boldsymbol{\theta} = \widehat{\boldsymbol{\theta}}^{(0)}_N}\right)~\text{is non-singular} \right ] = 1$, for all $N \in \mathbb{N}$. Moreover, $\{\boldsymbol{H}_N(\boldsymbol{\theta}_N) \}^{-1}  = O(1)$ and $\left ( \mathbb{P}_N\left[\nabla_{\boldsymbol{\theta}} \boldsymbol{U}_{N}(Z;\boldsymbol{\theta}) \bigg |_{\boldsymbol{\theta} = \widehat{\boldsymbol{\theta}}^{(0)}_N} \right]\right )^{-1} = O_{\mathbb{P}}(1)$.
     \item $\boldsymbol{M}_N(\boldsymbol{\theta}_N) = O(1)$ and $\boldsymbol{H}_N(\boldsymbol{\theta}_N) = O(1)$.
     \item $\exists~C_2 > 0: \mathbb{E}( \lVert \left. \nabla_{\boldsymbol{\theta}} \boldsymbol{U}_{N}(Z;\boldsymbol{\theta})\right|_{\boldsymbol{\theta} = \boldsymbol{\theta}_N} - \boldsymbol{H}_N(\boldsymbol{\theta}_N) \rVert^2) \leq C_2$, for all $N \in \mathbb{N}$.
     \item $\sqrt{N} \{\boldsymbol{M}_N(\boldsymbol{\theta}_N) \}^{-1/2} \boldsymbol{H}_N(\boldsymbol{\theta}_N)\left(\widehat{\boldsymbol{\theta}}^{(0)}_{N} - \boldsymbol{\theta}_N\right) = O_{\mathbb{P}}(1)$.\label{app:condition_bias}
    \end{enumerate}
    Then the one-step estimator satisfies
    $\boldsymbol{W}_N \left(\widehat{\boldsymbol{\theta}}^{(1)}_{N} - \boldsymbol{\theta}_N\right) \overset{d}{\to} \mathcal{N}(\boldsymbol{0}_p, {I}_p)$.
\end{theorem}

\begin{proof}[Proof of Theorem~\ref{thm:asymp_app}]
    Writing $\boldsymbol{U}_N^{(N)}(\boldsymbol{\theta}) = \mathbb{P}_N\left[\boldsymbol{U}_N(Z; \boldsymbol{\theta})\right]$, and employing a Taylor expansion of $\boldsymbol{U}_N^{(N)}$ at the initial estimator around $\boldsymbol{\theta}_N$, we have
    \begin{align*}
    \boldsymbol{U}_N^{(N)}(\widehat{\boldsymbol{\theta}}_N^{(0)}) = \boldsymbol{U}_N^{(N)}(\boldsymbol{\theta}_N) + \nabla_{\boldsymbol{\theta}} \boldsymbol{U}_{N}^{(N)}(\boldsymbol{\theta}) \bigg |_{\boldsymbol{\theta} = {\boldsymbol{\theta}}_N} \left(\widehat{\boldsymbol{\theta}}_N^{(0)} - \boldsymbol{\theta}_N\right) + \frac{1}{2}
\begin{bmatrix}
    (\widehat{\boldsymbol{\theta}}_N^{(0)} - \boldsymbol{\theta}_N)^T \boldsymbol{Q}_{N,1}(\widetilde{\boldsymbol{\theta}}_{N,1}) (\widehat{\boldsymbol{\theta}}_N^{(0)} - \boldsymbol{\theta}_N) \\
    \vdots \\
   (\widehat{\boldsymbol{\theta}}_N^{(0)} - \boldsymbol{\theta}_N)^T \boldsymbol{Q}_{N,p}(\widetilde{\boldsymbol{\theta}}_{N,p}) (\widehat{\boldsymbol{\theta}}_N^{(0)} - \boldsymbol{\theta}_N) ,
\end{bmatrix}
\end{align*}
for some $\widetilde{\boldsymbol{\theta}}$'s on the line segment between $\widehat{\boldsymbol{\theta}}_N^{(0)}$ and ${\boldsymbol{\theta}}_N$, and where $\boldsymbol{Q}_{N,j}(\boldsymbol{\theta}) = \nabla_{\boldsymbol{\theta}}^2 \,U_{N,j}^{(N)}({\boldsymbol{\theta}}) \in \mathbb{R}^{p \times p}$ for each $j \in [p]$. By definition of the one-step estimator,
\begin{align*}\widehat{\boldsymbol{\theta}}^{(1)}_N - \boldsymbol{\theta}_N 
&= (\widehat{\boldsymbol{\theta}}^{(0)}_{N} - \boldsymbol{\theta}_N) - \left (\mathbb{P}_N \bigg [  \nabla_{\boldsymbol{\theta}} \boldsymbol{U}_{N}(Z; \boldsymbol{\theta}) \bigg |_{\boldsymbol{\theta} = \widehat{\boldsymbol{\theta}}^{(0)}_N} \bigg ] \right )^{-1} \mathbb{P}_N \left [ \boldsymbol{U}_N(Z; \widehat{\boldsymbol{\theta}}^{(0)}_N) \right ] \\
&= (\widehat{\boldsymbol{\theta}}^{(0)}_{N} - \boldsymbol{\theta}_N) - \left(\nabla_{\boldsymbol{\theta}} \boldsymbol{U}_{N}^{(N)}(\boldsymbol{\theta}) \bigg |_{\boldsymbol{\theta} = \widehat{\boldsymbol{\theta}}^{(0)}_N}\right)^{-1} \boldsymbol{U}_N^{(N)}(\widehat{\boldsymbol{\theta}}^{(0)}_N),
\end{align*}
so the Taylor expansion implies
\begin{align*}
& \widehat{\boldsymbol{\theta}}^{(1)}_N - \boldsymbol{\theta}_N \\
    &= - \left ( \nabla_{\boldsymbol{\theta}} \boldsymbol{U}_{N}^{(N)}(\boldsymbol{\theta}) \bigg |_{\boldsymbol{\theta} = \widehat{\boldsymbol{\theta}}^{(0)}_N} \right )^{-1} \boldsymbol{U}_{N}^{(N)}(\boldsymbol{\theta}_N) + \Bigg \{ I_p - \left ( \nabla_{\boldsymbol{\theta}} \boldsymbol{U}_{N}^{(N)}(\boldsymbol{\theta}) \bigg |_{\boldsymbol{\theta} = \widehat{\boldsymbol{\theta}}^{(0)}_N} \right )^{-1}  \nabla_{\boldsymbol{\theta}} \boldsymbol{U}_{N}^{(N)}(\boldsymbol{\theta}) \Bigg |_{\boldsymbol{\theta} = {\boldsymbol{\theta}}_N} \\
    &\quad \quad -\frac{1}{2}\left ( \nabla_{\boldsymbol{\theta}} \boldsymbol{U}_{N}^{(N)}(\boldsymbol{\theta}) \bigg |_{\boldsymbol{\theta} = \widehat{\boldsymbol{\theta}}^{(0)}_N} \right )^{-1} \begin{bmatrix}
    (\widehat{\boldsymbol{\theta}}_N^{(0)} - \boldsymbol{\theta}_N)^T \boldsymbol{Q}_{N,1}(\widetilde{\boldsymbol{\theta}}_{N,1})  \\
    \vdots \\
    (\widehat{\boldsymbol{\theta}}_N^{(0)} - \boldsymbol{\theta}_N)^T \boldsymbol{Q}_{N,p}(\widetilde{\boldsymbol{\theta}}_{N,p})  
\end{bmatrix}   \Bigg \} (\widehat{\boldsymbol{\theta}}^{(0)}_{N} - \boldsymbol{\theta}_N).
\end{align*}
Multiplying through by $\boldsymbol{W}_N$, we obtain
\begin{align*}
& \boldsymbol{W}_N (\widehat{\boldsymbol{\theta}}^{(1)}_{N} - \boldsymbol{\theta}_N) = -  \sqrt{N} \{ \boldsymbol{M}_N(\boldsymbol{\theta}_N)\}^{-1/2} \overbrace{\boldsymbol{H}_N(\boldsymbol{\theta}_N) \left ( \nabla_{\boldsymbol{\theta}} \boldsymbol{U}_{N}^{(N)}(\boldsymbol{\theta}) \bigg |_{\boldsymbol{\theta} = \widehat{\boldsymbol{\theta}}^{(0)}_N} \right )^{-1} }^{\overset{\mathbb{P}}{\to} I_p \text{ by (a)}} \boldsymbol{U}_{N}^{(N)}(\boldsymbol{\theta}_N) \\
& \quad \quad +\sqrt{N} \{ \boldsymbol{M}_N(\boldsymbol{\theta}_N)\}^{-1/2} \boldsymbol{H}_N(\boldsymbol{\theta}_N) \Bigg \{ \underbrace{I_p - \left ( \nabla_{\boldsymbol{\theta}} \boldsymbol{U}_{N}^{(N)}(\boldsymbol{\theta}) \bigg |_{\boldsymbol{\theta} = \widehat{\boldsymbol{\theta}}^{(0)}_N} \right )^{-1}  \nabla_{\boldsymbol{\theta}} \boldsymbol{U}_{N}^{(N)}(\boldsymbol{\theta}) \Bigg |_{\boldsymbol{\theta} = {\boldsymbol{\theta}}_N}}_{\overset{\mathbb{P}}{\to} \boldsymbol{0}_{p \times p} \text{ by (b)}} \\
    &\quad \quad \quad \quad -\frac{1}{2} \underbrace{ \left ( \nabla_{\boldsymbol{\theta}} \boldsymbol{U}_{N}^{(N)}(\boldsymbol{\theta}) \bigg |_{\boldsymbol{\theta} = \widehat{\boldsymbol{\theta}}^{(0)}_N} \right )^{-1} 
    \begin{bmatrix}
    (\widehat{\boldsymbol{\theta}}_N^{(0)} - \boldsymbol{\theta}_N)^T \boldsymbol{Q}_{N,1}(\widetilde{\boldsymbol{\theta}}_{N,1}) \\
    \vdots \\
   (\widehat{\boldsymbol{\theta}}_N^{(0)} - \boldsymbol{\theta}_N)^T \boldsymbol{Q}_{N,p}(\widetilde{\boldsymbol{\theta}}_{N,p})
\end{bmatrix} }_{=o_{\mathbb{P}}(1)\text{ by (c)}}  \Bigg \} (\widehat{\boldsymbol{\theta}}^{(0)}_{N} - \boldsymbol{\theta}_N),
\end{align*}
where we invoked facts (a), (b), and (c) verified below. The first summand converges to a normal distribution by Lemma~\ref{lemma_matrix_slut} and the central limit theorem, whose application is justified under condition (ii). The second summand converges to zero in probability, as is seen by combining condition (vi) and another application of~Lemma~\ref{lemma_matrix_slut}---note that $\left\{\boldsymbol{M}_N(\boldsymbol{\theta}_N)\right\}^{1/2}$, $\boldsymbol{H}_N(\boldsymbol{\theta}_N)$, $\left\{\boldsymbol{M}_N(\boldsymbol{\theta}_N)\right\}^{-1/2}$, and $\left\{\boldsymbol{H}_N(\boldsymbol{\theta}_N)\right\}^{-1}$ are all $O(1)$ under conditions (ii), (iii) and (iv).


It remains to verify the following facts:
\begin{enumerate}[(a)]
    \item $\boldsymbol{H}_N(\boldsymbol{\theta}_N) \left ( \nabla_{\boldsymbol{\theta}} \boldsymbol{U}_{N}^{(N)}(\boldsymbol{\theta}) \bigg |_{\boldsymbol{\theta} = \widehat{\boldsymbol{\theta}}^{(0)}_N} \right )^{-1}  \overset{\mathbb{P}}{\to} I_p$ 
    \item $\left ( \nabla_{\boldsymbol{\theta}} \boldsymbol{U}_{N}^{(N)}(\boldsymbol{\theta}) \bigg |_{\boldsymbol{\theta} = \widehat{\boldsymbol{\theta}}^{(0)}_N} \right )^{-1}  \nabla_{\boldsymbol{\theta}} \boldsymbol{U}_{N}^{(N)}(\boldsymbol{\theta}) \bigg |_{\boldsymbol{\theta} = {\boldsymbol{\theta}}_N} \overset{\mathbb{P}}{\to} I_p$
    \item $\left ( \nabla_{\boldsymbol{\theta}} \boldsymbol{U}_{N}^{(N)}(\boldsymbol{\theta}) \bigg |_{\boldsymbol{\theta} = \widehat{\boldsymbol{\theta}}^{(0)}_N} \right )^{-1}  \begin{bmatrix}
    (\widehat{\boldsymbol{\theta}}_N^{(0)} - \boldsymbol{\theta}_N)^T \boldsymbol{Q}_{N,1}(\widetilde{\boldsymbol{\theta}}_{N,1}) \\
    \vdots \\
   (\widehat{\boldsymbol{\theta}}_N^{(0)} - \boldsymbol{\theta}_N)^T \boldsymbol{Q}_{N,p}(\widetilde{\boldsymbol{\theta}}_{N,p})  
\end{bmatrix}  \overset{\mathbb{P}}{\to} \boldsymbol{0}_{p \times p}$
\end{enumerate}
Observe first that $\widehat{\boldsymbol{\theta}}_N^{(0)}  - \boldsymbol{\theta}_N= o_{\mathbb{P}}(1)$ under our assumptions: this follows from condition (vi), and the fact that the matrices $\left\{\boldsymbol{M}_N(\boldsymbol{\theta}_N)\right\}^{1/2}$ and $\left\{\boldsymbol{H}_N(\boldsymbol{\theta}_N)\right\}^{-1}$ are bounded under conditions (iii) and (iv). For fact (a), see that
\begin{align*}
& \left. \nabla_{\boldsymbol{\theta}}\boldsymbol{U}_N^{(N)}(\boldsymbol{\theta}) \right|_{\boldsymbol{\theta} = \widehat{\boldsymbol{\theta}}^{(0)}_N} - \boldsymbol{H}_N(\boldsymbol{\theta}_N) \\
&= \mathbb{P}_N\left[\left.\nabla_{\boldsymbol{\theta}}\boldsymbol{U}_N(Z; \boldsymbol{\theta}) \right|_{\boldsymbol{\theta} = \widehat{\boldsymbol{\theta}}^{(0)}_N}\right] - \mathbb{E}\left(\left.\nabla_{\boldsymbol{\theta}}\boldsymbol{U}_N(Z; \boldsymbol{\theta}) \right|_{\boldsymbol{\theta} = \boldsymbol{\theta}_N}\right) \\
&= \mathbb{P}_N\left[\left.\nabla_{\boldsymbol{\theta}}\boldsymbol{U}_N(Z; \boldsymbol{\theta}) \right|_{\boldsymbol{\theta} = \widehat{\boldsymbol{\theta}}^{(0)}_N} - \left.\nabla_{\boldsymbol{\theta}}\boldsymbol{U}_N(Z; \boldsymbol{\theta}) \right|_{\boldsymbol{\theta} = \boldsymbol{\theta}_N}\right] \\
& \quad \quad + \left\{\mathbb{P}_N\left[\left.\nabla_{\boldsymbol{\theta}}\boldsymbol{U}_N(Z; \boldsymbol{\theta}) \right|_{\boldsymbol{\theta} = \boldsymbol{\theta}_N}\right] - \mathbb{E}\left(\left.\nabla_{\boldsymbol{\theta}}\boldsymbol{U}_N(Z; \boldsymbol{\theta}) \right|_{\boldsymbol{\theta} = \boldsymbol{\theta}_N}\right)\right\}.
\end{align*}
The first summand is bounded above by $C_1 \lVert \widehat{\boldsymbol{\theta}}_N^{(0)} - \boldsymbol{\theta}_N\rVert = o_{\mathbb{P}}(1)$ by condition (i), and the second summand converges to zero by a weak law of large numbers, justified by condition (v). Thus, (a) holds by the continuous mapping theorem---note that all matrices involved are invertible and stochastically bounded under condition (iii). Fact (b) is shown using the same argument as for the first summand analyzed above for fact (a). Finally, fact (c) holds by condition (i) (i.e., the second derivative matrices are uniformly bounded), the fact that $\widehat{\boldsymbol{\theta}}_N^{(0)}  - \boldsymbol{\theta}_N= o_{\mathbb{P}}(1)$, and condition (iii) (i.e., the left multiplying matrix is $O_{\mathbb{P}}(1)$).
\end{proof}

\begin{proof}[Proof of Theorem \ref{thm:asymp} of the Main Text]
    It is straightforward to verify that defining $\boldsymbol{H}_N(\boldsymbol{\theta})$ as in Section~\ref{sec:onestep} (i.e., ignoring the component of the gradient of the penalized estimating equation corresponding to the derivative of $\mathbb{D}_i$) results in an asymptotically equivalent estimator under the conditions of this theorem. Thus, Theorem B.2 yields that 
    \[\boldsymbol{W}_N(\widehat{\boldsymbol{\theta}}_{\Lambda_N}^{(1)} - \boldsymbol{\theta}_N) \overset{d}{\to} \mathcal{N}(\mathbf{0}_p, I_p).\]
    Therefore, noting that \[ \boldsymbol{W}_N^T \boldsymbol{W}_N = N \cdot \boldsymbol{H}_N^T(\boldsymbol{\theta}_N)\{\boldsymbol{M}_N(\boldsymbol{\theta}_N)\}^{-1} \boldsymbol{H}_N(\boldsymbol{\theta}_N),\] it holds that $\widehat{\boldsymbol{\theta}}_{\Lambda_N}^{(1)}$ is approximately normally distributed with mean $\boldsymbol{\theta}_N$ and variance given by $\boldsymbol{W}_N^{-1} \{\boldsymbol{W}_N^{-1}\}^T = \{\boldsymbol{W}_N^T \boldsymbol{W}_N\}^{-1} = \frac{1}{N}\{\boldsymbol{H}_N(\boldsymbol{\theta}_N)\}^{-1} \boldsymbol{M}_N(\boldsymbol{\theta}_N)  \{\boldsymbol{H}_N^T(\boldsymbol{\theta}_N)\}^{-1}$. Consequently, $\widehat{\boldsymbol{\beta}}^{(1)}_{\Lambda_N} = \mathbb{B} \widehat{\boldsymbol{\theta}}_{\Lambda_N}^{(1)}$ is approximately normally distributed with mean $\boldsymbol{\beta}_N = \mathbb{B}\boldsymbol{\theta}_N$ and variance given by
    $\boldsymbol{V}_N = \mathbb{B} \{\boldsymbol{W}_N^T \boldsymbol{W}_N\}^{-1}\mathbb{B}^T$.
\end{proof}

\subsection{Condition~\eqref{app:condition_bias}} 
To provide intuition for condition~\eqref{app:condition_bias} in the statement of Theorem~\ref{thm:asymp_app}, we derive an interpretable set of conditions that imply condition~\eqref{app:condition_bias} in a special case of our general framework: a non-functional, univariate (i.e., $p=q=1$), non-clustered (i.e., $n_i = 1~\forall~i \in [N]$) ridge regression that is weighted by the inverse of the working covariance matrix. This setting provides insight for the more general case: (1) the fully-iterated unweighted ridge is analogous to the initial penalized GEE estimator that is also unweighted (i.e., adopts an independence working covariance structure) and fully-iterated; (2) the fully-iterated weighted ridge is analogous to the fully-iterated and weighted (i.e., adopts some non-independence working covariance structure) penalized GEE. In this special case, we show that if the scaled smoothing parameters for the initial estimate satisfy $\frac{1}{N}\lambda_{0,N} = O(N^{-1/2})$, and for the weighted estimator satisfy $\frac{1}{N}\lambda_{N} = O(N^{-1/2})$, then condition~\eqref{app:condition_bias} holds. That is, if each of these scaled smoothing parameters individually (i.e., no conditions are required jointly on these rates) go to zero fast enough, then $\sqrt{N} \{\boldsymbol{M}_N(\boldsymbol{\theta}_N) \}^{-1/2} \boldsymbol{H}_N(\boldsymbol{\theta}_N)(\widehat{\boldsymbol{\theta}}^{(0)}_{N} - \boldsymbol{\theta}_N) = O_{\mathbb{P}}(1)$. Importantly, the smoothing parameter rates we require are weaker than (i.e., implied by) the rates needed for the theoretical properties described in \cite{gee_functional}. For example, even for the small knot setting, \cite{gee_functional} require that $\lambda = O(N^\gamma)$ for $\gamma \leq (\tilde{p} + 2 - \tilde{q})/(2\tilde{p}+3)$. This implies, for instance, that for a $\tilde{p}^{th}$ order truncated polynomial using cubic B-splines, $\gamma \leq (\tilde{p} + 2 - (\tilde{p} + 1))/(2\tilde{p}+3) = 1/9$. Put onto the scale of a single cluster, \cite{gee_functional} requires the faster rate of $\frac{1}{N}\lambda = O(N^{-8/9})$ than the  $O(N^{-1/2})$ rate required by our theory in this special case. By this reasoning, condition~\eqref{app:condition_bias} is a weak assumption in this special case. Our conjecture is that this extends to more complicated settings although we omit such analysis in this work.

\paragraph{Univariate Weighted Ridge} In the special case we explore, we define the following population parameters and estimators. At the outset, we define quantities with matrix notation to be consistent with the notation used in the remainder of the paper. Later on, we restrict our analysis to the $p=1$ case for simplicity:
\begin{itemize}
    \item Fully-Iterated (Population) Parameter (Closed-Form): for $\lambda_N$ such that $\frac{1}{N}\lambda_N \to 0$,
    \begin{align*}
        \beta_N &= \mathbb{E}(\widetilde{\text{Var}}^{-1}_N(Y \mid X)[XX^T + \frac{1}{N}\lambda_N I_p])^{-1} \mathbb{E}(\widetilde{\text{Var}}^{-1}_N(Y \mid X)[XY]) \\
        & \overset{N \to \infty}{\to} \beta^* = \mathbb{E}({\text{Var}}^{-1}(Y \mid X)[XX^T])^{-1} \mathbb{E}({\text{Var}}^{-1}(Y \mid X)[XY]) \\
        & \equiv [\mathbb{E}(XX^T)]^{-1} \mathbb{E}(XY)
    \end{align*}
    \item Fully-Iterated Estimator (Closed-Form): \[ \widehat{\beta}^*_N = \mathbb{P}_N(\widetilde{\text{Var}}_N^{-1}(Y \mid X)[XX^T + \frac{1}{N}\lambda_N I_p])^{-1} \mathbb{P}_N(\widetilde{\text{Var}}_N^{-1}(Y \mid X)[XY]) \]
    \item Unweighted Penalized (Initial Estimator) Population Parameter: if $\frac{1}{N}\lambda_{0,N} \to 0$,\\ \[ \beta^{(0)}_N = \mathbb{E}(XX^T + \frac{1}{N}\lambda_{0,N} I_p)\mathbb{E}(XY) \overset{N \to \infty}{\to} \beta^* \]
    \item Unweighted Penalized (Initial) Estimator: $\widehat{\beta}^{(0)}_N = \mathbb{P}_N(XX^T + \frac{1}{N}\lambda_{0,N} I_p)\mathbb{P}_N(XY)$
    
    \item Coefficient Estimator Variance:
    \begin{align*}
        S_N(\beta) &= H^{-1}_N(\beta) M_N(\beta) H^{-1}_N (\beta) \\
        &= \frac{1}{N} \{\mathbb{E}(\widetilde{\text{Var}}^{-1}_N(Y \mid X)[XX^T + \frac{1}{N}\lambda_N I_p])\}^{-1} \Bigg[ \mathbb{E}\left ( \frac{\text{Var}(Y \mid X)}{\widetilde{\text{Var}}_N^2(Y \mid X)} XX^T \right ) + \\
        &~~~~~~~\text{Var}\{\widetilde{\text{Var}}_N^{-1}(Y \mid X) X X^T(\beta - \beta^*) \}  \Bigg ] \{\mathbb{E}[\widetilde{\text{Var}}_N^{-1}(Y \mid X) X X^T + \frac{1}{N}\lambda_N I_p] \}^{-1}\\
    & \asymp \frac{1}{N} {( \mathbb{E}[\text{Var}}^{-1}(Y \mid X)XX^T])^{-1}~~\text{as } N \to \infty ~(\text{if } \beta \to \beta^*)
    \end{align*}
    \item Variance of Initial Coefficient Estimator
    \begin{align*}
        & S^{(0)}_N(\beta) \\
        & = \frac{1}{N} (\mathbb{E}  [XX^T + \frac{1}{N}\lambda_{0,N} I_p])^{-1} \left\{\mathbb{E}[ \text{Var}(Y \mid X) X X^T ] + \text{Var}[X X^T (\beta - \beta^*)] \right \} (\mathbb{E}  [XX^T + \frac{1}{N}\lambda_{0,N} I_p])^{-1} \\
        & \asymp \frac{1}{N} (\mathbb{E}  [XX^T])^{-1} \left\{\mathbb{E}[ \text{Var}(Y \mid X) X X^T ] \right \} (\mathbb{E}  [XX^T])^{-1} 
        ~~\text{as } N \to \infty ~(\text{if } \beta \to \beta^*)
    \end{align*}
\end{itemize}

We now begin our exploration of this special case. First, observe that the above definitions imply, so long as $\frac{1}{N}\lambda_N \to 0$ and $\frac{1}{N}\lambda_{0,N} \to 0$,
\begin{align*}
    S^{-1}_N(\beta_N)  S^{(0)}_N(\beta_N^{(0)}) & \overset{N \to \infty}{\to} \mathbb{E}[\text{Var}^{-1}(Y \mid X)XX^T] (\mathbb{E}  [XX^T])^{-1} \mathbb{E}[ \text{Var}(Y \mid X) X X^T ]  (\mathbb{E}  [XX^T])^{-1},
\end{align*}
so that $S^{-1}_N(\beta_N)  S^{(0)}_N(\beta_N^{(0)}) = O(1)$.

Thus, when $p=1$, the Taylor expansion for the quantity in condition~\eqref{app:condition_bias} has the form:
\begin{align*}
S^{-1/2}_N(\beta_N) \left(\widehat{\beta}_{0,N} - \beta_N \right) &~\overset{p=1}{=}~ \frac{\widehat{\beta}^{(0)}_N  - {\beta}^{(0)}_N }{\sqrt{S_N(\beta_N})} + \frac{{\beta}^{(0)}_N  - {\beta}_N }{\sqrt{S_N(\beta_N})} ~~~~~~~~~~~~~~~~~~~~~~~~~~~~~~~~~~~~~ \\
&~~= O(1) * \underbrace{ \frac{\widehat{\beta}^{(0)}_N  - {\beta}^{(0)}_N }{\sqrt{S_N^{(0)}(\beta_N^{(0)}})} }_{\substack{ \overset{d}{\to} N(0,1) \\ \text{under weak} \\ 
\text{CLT conditions}}} ~+~ A
\end{align*}
where 
\begin{align*}
A=& \frac{\sqrt{N} \left \{ \mathbb{E} \left(\frac{X^2}{\tilde{\sigma}^2(X)} \right) + \frac{1}{N}\lambda_N \right \} }{\sqrt{ \mathbb{E}( \frac{\sigma^2(X)}{\tilde{\sigma}^4(X)}X^2)  + (\beta_N - \beta^*) \text{Var}(\frac{X^2}{\tilde{\sigma}^2(X))})}} * \frac{\frac{1}{N}\lambda_N \mathbb{E}(XY)\mathbb{E}(\frac{1}{\tilde{\sigma}^2(X)}) - \frac{1}{N}\lambda_{0,N} \mathbb{E}(XY / \tilde{\sigma}^2(X))}{ \{ \mathbb{E}(X^2) + \frac{1}{N}\lambda_{0,N} \} \{ \mathbb{E}(\frac{X^2}{\tilde{\sigma}^2(X)}) + \frac{1}{N}\lambda_N \mathbb{E} (\frac{1}{\tilde{\sigma}^2(X)})} \\
& = O(1) ~~ \text{~~if $\frac{1}{N}\lambda_N = O(n^{-1/2})$ and $\frac{1}{N}\lambda_{0,N}  = O(n^{-1/2})$}
\end{align*}

The above expansion shows that
if, individually, $\frac{1}{N}\lambda_N = O(N^{-1/2})$ and $\frac{1}{N}\lambda_{0,N}  = O(N^{-1/2})$, then the initial (unweighted) estimator converges at the same rate to $\beta_N$ as the fully-iterated (weighted) estimator. That is, in this special case, condition~\eqref{app:condition_bias} follows from the assumption that the scaled smoothing parameters of both estimators go to zero asymptotically at a rate of $1/\sqrt{n}$ or faster.


\section{Simulation Experiments}
\subsection{Additional Simulation Results from Main Paper} \label{sec:sims_app}
Here we provide additional results from the main text simulations presented in Section~\ref{sec:simulations}. The tables below include results from gamma-distributed data and fGEE fits that adopt a Kronecker product based working correlation where an FPCA was used to estimate correlation in the functional direction. 

The plot below is the average functional coefficient estimate for each covariate (columns) and method (rows): i) the initial pffr fit, ii) the one-step fGEE, and iii) a fully-iterated Full-1step fGEE. The fGEEs are fit with the correct working correlation (exchangeable in the longitudinal direction and AR1 in the functional direction). The data are generated as in the main text simulations presented in Section~\ref{sec:simulations} with Poisson outcomes, $N = 50$, and $n_i=5$. The methods are fit as described in Section~\ref{sec:simulations}. The plots below show pointwise confidence intervals of the average coefficient estimate value (see figure caption for details). 

\begin{figure*}[!h]
\centering
\includegraphics[width=0.7 \linewidth]{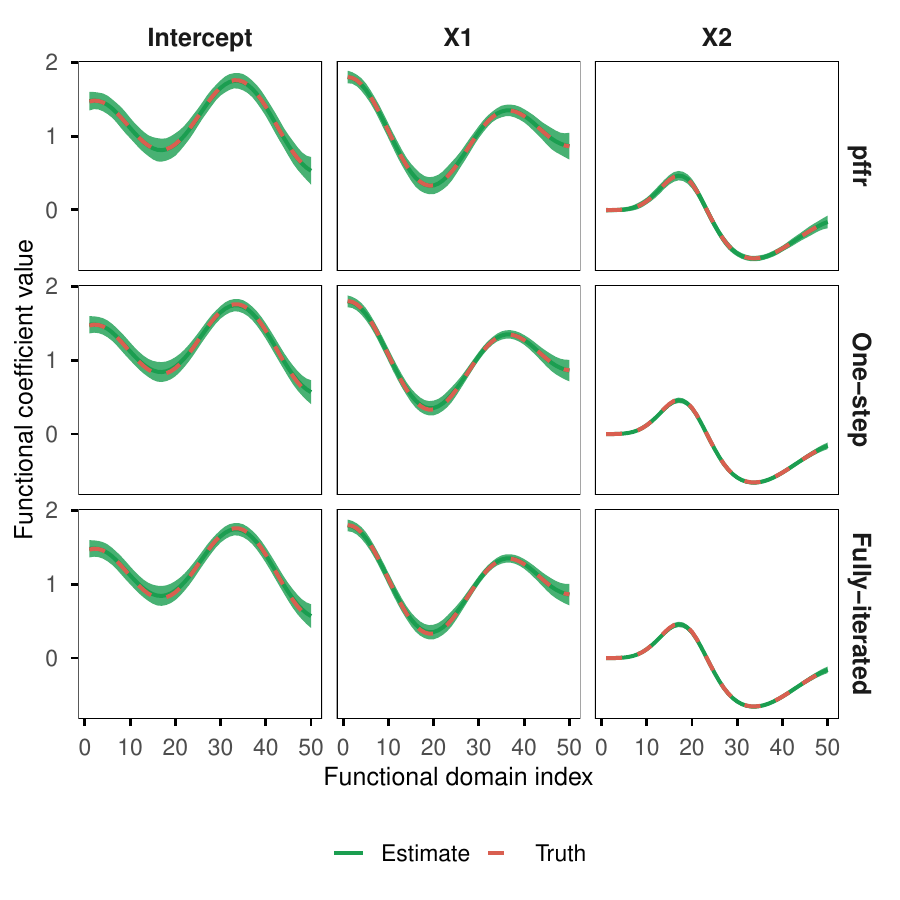}
\caption{\footnotesize \textbf{Average functional coefficient estimates from simulations.} Functional coefficient indices are shown in the columns. Functional regression methods are shown in the rows: i) the initial pffr fit, ii) the one-step fGEE, and iii) a fully-iterated Full-1step fGEE. The fGEEs are fit with the correct working correlation: exchangeable in the longitudinal direction and AR1 in the functional direction. The outcomes are simulated as Poisson and $N = 50$, $n_i=5$. The plots show pointwise confidence intervals of the average coefficient estimate value (calculated across $T=300$ simulation replicates). Specifically, for a given functional coefficient, we denote $\bar{\beta}(s) = \frac{1}{T} \sum_{t=1}^T \widehat{\beta}^{(t)}(s)$, and calculate 
$\mathrm{SE}(s) = \sqrt{\frac{1}{T(T-1)} \sum_{t=1}^T \left(\widehat{\beta}^{(t)}(s) - \bar{\beta}(s)\right)^2 }$. We then plot the pointwise CIs as $\bar{\beta}(s) \pm 1.96 \, \mathrm{SE}(s)$.}
\label{fig:coef_est_grid}
\end{figure*}

\newpage

\begin{table}[!h] \label{tab:rmse_rel_supp}
\centering
\caption{ \footnotesize Functional Coefficient Estimation Performance (RMSE) of each method relative to the $\texttt{pffr}$ fit ($\text{RMSE}/\text{RMSE}_{\text{pffr}}$).
    Cells contain the average of 300 replicates $\pm$ SE (SE$=0.00$ indicates a value $<0.01$). Outcomes were simulated with an $\mathbf{R}^* = \mathbf{R}^*_{\text{Lon}} \otimes \mathbf{R}^*_{\text{Fun}}$, where $\mathbf{R}^*_{\text{Fun}}$ had an AR1 structure and the table columns indicate results where $\mathbf{R}^*_{\text{Lon}}$ had exchangeable or AR1 correlation. The ``1-Step'' indicates a one-step was used for tuning and final coefficient estimation, and ``Full-1step'' indicates one-step tuning and a fully-iterated fGEE for final coefficient estimation, with the indicated working correlation. Values with $*$ indicate extreme outliers (from poor estimates) were removed from the average of that cell to avoid skewing results.}
\centering
\resizebox{\ifdim\width>\linewidth\linewidth\else\width\fi}{!}{
\fontsize{9}{11}\selectfont
\begin{tabular}[t]{>{\raggedright\arraybackslash}p{4.8cm}>{\raggedright\arraybackslash}p{2.1cm}>{\raggedright\arraybackslash}p{2.1cm}>{\raggedright\arraybackslash}p{2.1cm}>{\raggedright\arraybackslash}p{2.1cm}>{\raggedright\arraybackslash}p{2.1cm}>{\raggedright\arraybackslash}p{2.1cm}>{\raggedright\arraybackslash}p{2.1cm}>{\raggedright\arraybackslash}p{2.1cm}}
\toprule
\multicolumn{1}{c}{ } & \multicolumn{4}{c}{Exchangeable} & \multicolumn{4}{c}{AR(1)} \\
\cmidrule(l{3pt}r{3pt}){2-5} \cmidrule(l{3pt}r{3pt}){6-9}
Method & Gaussian & Poisson & Binomial & Gamma & Gaussian & Poisson & Binomial & Gamma\\
\midrule
\addlinespace[0em]
\multicolumn{9}{l}{\textbf{$N = 25,\; n_i = 5$}}\\
\hspace{1em}$\text{1-Step } {\mathbf{R}}_{\text{FPCA}} \otimes {\mathbf{R}}_{\text{Fun}}$ & 0.84 $\pm$ 0.01 & 0.87 $\pm$ 0.01 & 0.89 $\pm$ 0.02 & 0.95 $\pm$ 0.01 & 0.86 $\pm$ 0.01 & 0.89 $\pm$ 0.01 & 0.92 $\pm$ 0.01 & 0.93 $\pm$ 0.01\\
\hspace{1em}$\text{1-Step } {\mathbf{R}}_{\text{Lon}} \otimes {\mathbf{R}}_{\text{Fun}}$ & 0.82 $\pm$ 0.01 & 0.83 $\pm$ 0.01 & 0.88 $\pm$ 0.01 & 0.90 $\pm$ 0.00 & 0.85 $\pm$ 0.01 & 0.85 $\pm$ 0.01 & 0.92 $\pm$ 0.01 & 0.89 $\pm$ 0.01\\
\hspace{1em}$\text{1-Step } I_{n_i} \otimes {\mathbf{R}}_{\text{Fun}}$ & 0.86 $\pm$ 0.01 & 0.91 $\pm$ 0.01 & 0.93 $\pm$ 0.01 & 0.95 $\pm$ 0.00 & 0.90 $\pm$ 0.01 & 0.92 $\pm$ 0.01 & 0.95 $\pm$ 0.01 & 0.96 $\pm$ 0.00\\
\hspace{1em}$\text{1-Step } {\mathbf{R}}_{\text{Lon}} \otimes I_L$ & 0.82 $\pm$ 0.01 & 0.85 $\pm$ 0.01 & 0.88 $\pm$ 0.01 & 0.87 $\pm$ 0.00 & 0.85 $\pm$ 0.01 & 0.83 $\pm$ 0.01 & 0.92 $\pm$ 0.01 & 0.87 $\pm$ 0.00\\
\hspace{1em}$\text{1-Step } I_{n_i} \otimes I_L$ & 0.87 $\pm$ 0.01 & 0.88 $\pm$ 0.00 & 0.93 $\pm$ 0.01 & 0.92 $\pm$ 0.00 & 0.90 $\pm$ 0.01 & 0.90 $\pm$ 0.00 & 0.95 $\pm$ 0.01 & 0.93 $\pm$ 0.00\\
\hspace{1em}$\text{Full-1Step } {\mathbf{R}}_{\text{FPCA}} \otimes {\mathbf{R}}_{\text{Fun}}$ & 0.87 $\pm$ 0.01 & 0.90 $\pm$ 0.01 & 2.15 $\pm$ 0.29 & 1.02 $\pm$ 0.01 & 0.89 $\pm$ 0.01 & 0.91 $\pm$ 0.01 & 0.98 $\pm$ 0.03 & 1.88 $\pm$ 0.49$^*$\\
\hspace{1em}$\text{Full-1Step } {\mathbf{R}}_{\text{Lon}} \otimes {\mathbf{R}}_{\text{Fun}}$ & 0.82 $\pm$ 0.01 & 0.83 $\pm$ 0.01 & 1.00 $\pm$ 0.09$^*$ & 0.94 $\pm$ 0.04 & 0.86 $\pm$ 0.01 & 0.85 $\pm$ 0.01 & 0.94 $\pm$ 0.01 & 0.91 $\pm$ 0.01$^*$\\
\hspace{1em}$\text{Full-1Step } I_{n_i} \otimes {\mathbf{R}}_{\text{Fun}}$ & 0.87 $\pm$ 0.01 & 0.91 $\pm$ 0.01 & 0.96 $\pm$ 0.02 & 0.96 $\pm$ 0.01 & 0.90 $\pm$ 0.01 & 0.92 $\pm$ 0.01 & 0.97 $\pm$ 0.01 & 0.97 $\pm$ 0.01\\
\hspace{1em}$\text{Full-1Step } {\mathbf{R}}_{\text{Lon}} \otimes I_L$ & 0.81 $\pm$ 0.01 & 1.02 $\pm$ 0.09$^*$ & 1.29 $\pm$ 0.21$^*$ & 0.87 $\pm$ 0.00 & 0.85 $\pm$ 0.01 & 0.83 $\pm$ 0.01 & 0.93 $\pm$ 0.01 & 0.88 $\pm$ 0.01\\
\hspace{1em}$\text{Full-1Step } I_{n_i} \otimes I_L$ & 0.87 $\pm$ 0.01 & 0.89 $\pm$ 0.01 & 0.95 $\pm$ 0.01 & 0.92 $\pm$ 0.00 & 0.90 $\pm$ 0.01 & 0.90 $\pm$ 0.00 & 0.96 $\pm$ 0.01 & 0.93 $\pm$ 0.00\\
\addlinespace[0.6em]
\multicolumn{9}{l}{\textbf{$N = 25,\; n_i = 100$}}\\
\hspace{1em}$\text{1-Step } {\mathbf{R}}_{\text{FPCA}} \otimes {\mathbf{R}}_{\text{Fun}}$ & 0.81 $\pm$ 0.01 & 0.68 $\pm$ 0.01 & 0.88 $\pm$ 0.01 & 0.95 $\pm$ 0.01 & 0.88 $\pm$ 0.01 & 0.89 $\pm$ 0.01 & 0.93 $\pm$ 0.01 & 0.94 $\pm$ 0.00\\
\hspace{1em}$\text{1-Step } {\mathbf{R}}_{\text{Lon}} \otimes {\mathbf{R}}_{\text{Fun}}$ & 0.78 $\pm$ 0.01 & 0.65 $\pm$ 0.01 & 0.82 $\pm$ 0.01 & 0.90 $\pm$ 0.01 & 0.87 $\pm$ 0.01 & 0.87 $\pm$ 0.01 & 0.89 $\pm$ 0.01 & 0.92 $\pm$ 0.00\\
\hspace{1em}$\text{1-Step } I_{n_i} \otimes {\mathbf{R}}_{\text{Fun}}$ & 0.80 $\pm$ 0.01 & 0.87 $\pm$ 0.01 & 0.79 $\pm$ 0.01 & 0.91 $\pm$ 0.01 & 0.93 $\pm$ 0.01 & 0.95 $\pm$ 0.00 & 0.93 $\pm$ 0.01 & 0.98 $\pm$ 0.00\\
\hspace{1em}$\text{1-Step } {\mathbf{R}}_{\text{Lon}} \otimes I_L$ & 0.77 $\pm$ 0.01 & 0.65 $\pm$ 0.01 & 0.81 $\pm$ 0.01 & 0.87 $\pm$ 0.01 & 0.85 $\pm$ 0.01 & 0.86 $\pm$ 0.01 & 0.88 $\pm$ 0.01 & 0.89 $\pm$ 0.00\\
\hspace{1em}$\text{1-Step } I_{n_i} \otimes I_L$ & 0.79 $\pm$ 0.01 & 0.85 $\pm$ 0.01 & 0.79 $\pm$ 0.01 & 0.88 $\pm$ 0.01 & 0.90 $\pm$ 0.01 & 0.94 $\pm$ 0.00 & 0.92 $\pm$ 0.01 & 0.95 $\pm$ 0.00\\
\hspace{1em}$\text{Full-1Step } {\mathbf{R}}_{\text{FPCA}} \otimes {\mathbf{R}}_{\text{Fun}}$ & 0.83 $\pm$ 0.01 & 0.70 $\pm$ 0.01 & 6.70 $\pm$ 0.45$^*$ & 1.04 $\pm$ 0.02 & 0.88 $\pm$ 0.01 & 0.89 $\pm$ 0.01 & 0.93 $\pm$ 0.01 & 0.94 $\pm$ 0.00\\
\hspace{1em}$\text{Full-1Step } {\mathbf{R}}_{\text{Lon}} \otimes {\mathbf{R}}_{\text{Fun}}$ & 0.78 $\pm$ 0.01 & 0.66 $\pm$ 0.01 & 4.75 $\pm$ 0.69$^*$ & 0.91 $\pm$ 0.01 & 0.87 $\pm$ 0.01 & 0.87 $\pm$ 0.01 & 0.90 $\pm$ 0.01 & 0.92 $\pm$ 0.00\\
\hspace{1em}$\text{Full-1Step } I_{n_i} \otimes {\mathbf{R}}_{\text{Fun}}$ & 0.81 $\pm$ 0.01 & 0.87 $\pm$ 0.01 & 0.83 $\pm$ 0.01 & 0.92 $\pm$ 0.01 & 0.93 $\pm$ 0.01 & 0.95 $\pm$ 0.00 & 0.94 $\pm$ 0.01 & 0.98 $\pm$ 0.00\\
\hspace{1em}$\text{Full-1Step } {\mathbf{R}}_{\text{Lon}} \otimes I_L$ & 0.77 $\pm$ 0.01 & 0.86 $\pm$ 0.07$^*$ & 7.42 $\pm$ 0.95$^*$ & 0.88 $\pm$ 0.01 & 0.85 $\pm$ 0.01 & 0.86 $\pm$ 0.01 & 0.88 $\pm$ 0.01 & 0.90 $\pm$ 0.00\\
\hspace{1em}$\text{Full-1Step } I_{n_i} \otimes I_L$ & 0.79 $\pm$ 0.01 & 0.85 $\pm$ 0.01 & 0.90 $\pm$ 0.09 & 0.88 $\pm$ 0.01 & 0.90 $\pm$ 0.01 & 0.94 $\pm$ 0.00 & 0.92 $\pm$ 0.01 & 0.95 $\pm$ 0.00\\
\addlinespace[0.6em]
\multicolumn{9}{l}{\textbf{$N = 50,\; n_i = 5$}}\\
\hspace{1em}$\text{1-Step } {\mathbf{R}}_{\text{FPCA}} \otimes {\mathbf{R}}_{\text{Fun}}$ & 0.89 $\pm$ 0.01 & 0.87 $\pm$ 0.01 & 0.90 $\pm$ 0.01 & 0.95 $\pm$ 0.01 & 0.90 $\pm$ 0.01 & 0.90 $\pm$ 0.01 & 0.92 $\pm$ 0.01 & 0.94 $\pm$ 0.01\\
\hspace{1em}$\text{1-Step } {\mathbf{R}}_{\text{Lon}} \otimes {\mathbf{R}}_{\text{Fun}}$ & 0.88 $\pm$ 0.01 & 0.84 $\pm$ 0.01 & 0.88 $\pm$ 0.01 & 0.91 $\pm$ 0.00 & 0.89 $\pm$ 0.01 & 0.87 $\pm$ 0.01 & 0.89 $\pm$ 0.01 & 0.91 $\pm$ 0.00\\
\hspace{1em}$\text{1-Step } I_{n_i} \otimes {\mathbf{R}}_{\text{Fun}}$ & 0.93 $\pm$ 0.01 & 0.93 $\pm$ 0.00 & 0.92 $\pm$ 0.01 & 0.96 $\pm$ 0.00 & 0.94 $\pm$ 0.01 & 0.93 $\pm$ 0.00 & 0.93 $\pm$ 0.01 & 0.97 $\pm$ 0.00\\
\hspace{1em}$\text{1-Step } {\mathbf{R}}_{\text{Lon}} \otimes I_L$ & 0.88 $\pm$ 0.01 & 0.84 $\pm$ 0.01 & 0.87 $\pm$ 0.01 & 0.88 $\pm$ 0.00 & 0.88 $\pm$ 0.01 & 0.85 $\pm$ 0.01 & 0.89 $\pm$ 0.01 & 0.88 $\pm$ 0.01\\
\hspace{1em}$\text{1-Step } I_{n_i} \otimes I_L$ & 0.93 $\pm$ 0.01 & 0.91 $\pm$ 0.01 & 0.92 $\pm$ 0.01 & 0.93 $\pm$ 0.00 & 0.94 $\pm$ 0.01 & 0.91 $\pm$ 0.00 & 0.92 $\pm$ 0.01 & 0.93 $\pm$ 0.00\\
\hspace{1em}$\text{Full-1Step } {\mathbf{R}}_{\text{FPCA}} \otimes {\mathbf{R}}_{\text{Fun}}$ & 0.90 $\pm$ 0.01 & 0.88 $\pm$ 0.01 & 0.93 $\pm$ 0.01 & 0.97 $\pm$ 0.01 & 0.91 $\pm$ 0.01 & 0.90 $\pm$ 0.01 & 0.94 $\pm$ 0.01 & 0.96 $\pm$ 0.01\\
\hspace{1em}$\text{Full-1Step } {\mathbf{R}}_{\text{Lon}} \otimes {\mathbf{R}}_{\text{Fun}}$ & 0.88 $\pm$ 0.01 & 0.84 $\pm$ 0.01 & 0.90 $\pm$ 0.01 & 0.91 $\pm$ 0.00 & 0.89 $\pm$ 0.01 & 0.87 $\pm$ 0.01 & 0.91 $\pm$ 0.01 & 0.91 $\pm$ 0.01\\
\hspace{1em}$\text{Full-1Step } I_{n_i} \otimes {\mathbf{R}}_{\text{Fun}}$ & 0.94 $\pm$ 0.01 & 0.93 $\pm$ 0.01 & 0.94 $\pm$ 0.01 & 0.96 $\pm$ 0.00 & 0.95 $\pm$ 0.01 & 0.93 $\pm$ 0.00 & 0.95 $\pm$ 0.01 & 0.97 $\pm$ 0.00\\
\hspace{1em}$\text{Full-1Step } {\mathbf{R}}_{\text{Lon}} \otimes I_L$ & 0.88 $\pm$ 0.01 & 1.01 $\pm$ 0.07 & 0.96 $\pm$ 0.08$^*$ & 0.88 $\pm$ 0.00 & 0.88 $\pm$ 0.01 & 0.85 $\pm$ 0.01 & 0.90 $\pm$ 0.01 & 0.88 $\pm$ 0.01\\
\hspace{1em}$\text{Full-1Step } I_{n_i} \otimes I_L$ & 0.93 $\pm$ 0.01 & 0.91 $\pm$ 0.01 & 0.93 $\pm$ 0.01 & 0.93 $\pm$ 0.00 & 0.94 $\pm$ 0.01 & 0.91 $\pm$ 0.00 & 0.93 $\pm$ 0.01 & 0.94 $\pm$ 0.00\\
\addlinespace[0.6em]
\multicolumn{9}{l}{\textbf{$N = 50,\; n_i = 100$}}\\
\hspace{1em}$\text{1-Step } {\mathbf{R}}_{\text{FPCA}} \otimes {\mathbf{R}}_{\text{Fun}}$ & 0.87 $\pm$ 0.01 & 0.69 $\pm$ 0.01 & 0.92 $\pm$ 0.01 & 0.96 $\pm$ 0.01 & 0.90 $\pm$ 0.01 & 0.88 $\pm$ 0.01 & 0.92 $\pm$ 0.01 & 0.95 $\pm$ 0.00\\
\hspace{1em}$\text{1-Step } {\mathbf{R}}_{\text{Lon}} \otimes {\mathbf{R}}_{\text{Fun}}$ & 0.85 $\pm$ 0.01 & 0.65 $\pm$ 0.01 & 0.83 $\pm$ 0.01 & 0.92 $\pm$ 0.00 & 0.90 $\pm$ 0.01 & 0.87 $\pm$ 0.00 & 0.89 $\pm$ 0.01 & 0.94 $\pm$ 0.00\\
\hspace{1em}$\text{1-Step } I_{n_i} \otimes {\mathbf{R}}_{\text{Fun}}$ & 0.85 $\pm$ 0.01 & 0.89 $\pm$ 0.01 & 0.82 $\pm$ 0.01 & 0.93 $\pm$ 0.00 & 0.95 $\pm$ 0.01 & 0.97 $\pm$ 0.00 & 0.93 $\pm$ 0.01 & 0.99 $\pm$ 0.00\\
\hspace{1em}$\text{1-Step } {\mathbf{R}}_{\text{Lon}} \otimes I_L$ & 0.85 $\pm$ 0.01 & 0.65 $\pm$ 0.01 & 0.81 $\pm$ 0.01 & 0.88 $\pm$ 0.00 & 0.86 $\pm$ 0.01 & 0.87 $\pm$ 0.00 & 0.88 $\pm$ 0.01 & 0.91 $\pm$ 0.00\\
\hspace{1em}$\text{1-Step } I_{n_i} \otimes I_L$ & 0.84 $\pm$ 0.01 & 0.87 $\pm$ 0.01 & 0.81 $\pm$ 0.01 & 0.89 $\pm$ 0.00 & 0.92 $\pm$ 0.00 & 0.96 $\pm$ 0.00 & 0.92 $\pm$ 0.01 & 0.97 $\pm$ 0.00\\
\hspace{1em}$\text{Full-1Step } {\mathbf{R}}_{\text{FPCA}} \otimes {\mathbf{R}}_{\text{Fun}}$ & 0.88 $\pm$ 0.01 & 0.70 $\pm$ 0.01 & 5.33 $\pm$ 0.49 & 0.98 $\pm$ 0.01 & 0.90 $\pm$ 0.01 & 0.88 $\pm$ 0.01 & 0.93 $\pm$ 0.01 & 0.95 $\pm$ 0.00\\
\hspace{1em}$\text{Full-1Step } {\mathbf{R}}_{\text{Lon}} \otimes {\mathbf{R}}_{\text{Fun}}$ & 0.86 $\pm$ 0.01 & 0.66 $\pm$ 0.01 & 2.63 $\pm$ 0.52$^*$ & 0.93 $\pm$ 0.01 & 0.90 $\pm$ 0.01 & 0.87 $\pm$ 0.00 & 0.89 $\pm$ 0.01 & 0.94 $\pm$ 0.00\\
\hspace{1em}$\text{Full-1Step } I_{n_i} \otimes {\mathbf{R}}_{\text{Fun}}$ & 0.86 $\pm$ 0.01 & 0.90 $\pm$ 0.01 & 0.85 $\pm$ 0.01 & 0.93 $\pm$ 0.00 & 0.95 $\pm$ 0.01 & 0.97 $\pm$ 0.00 & 0.93 $\pm$ 0.01 & 0.99 $\pm$ 0.00\\
\hspace{1em}$\text{Full-1Step } {\mathbf{R}}_{\text{Lon}} \otimes I_L$ & 0.85 $\pm$ 0.01 & 0.76 $\pm$ 0.03 & 4.33 $\pm$ 0.85$^*$ & 0.89 $\pm$ 0.00 & 0.86 $\pm$ 0.01 & 0.87 $\pm$ 0.00 & 0.88 $\pm$ 0.01 & 0.91 $\pm$ 0.00\\
\hspace{1em}$\text{Full-1Step } I_{n_i} \otimes I_L$ & 0.84 $\pm$ 0.01 & 0.87 $\pm$ 0.01 & 0.82 $\pm$ 0.01 & 0.89 $\pm$ 0.00 & 0.92 $\pm$ 0.00 & 0.96 $\pm$ 0.00 & 0.92 $\pm$ 0.01 & 0.97 $\pm$ 0.00\\
\addlinespace[0.6em]
\multicolumn{9}{l}{\textbf{$N = 500,\; n_i = 5$}}\\
\hspace{1em}$\text{1-Step } {\mathbf{R}}_{\text{FPCA}} \otimes {\mathbf{R}}_{\text{Fun}}$ & 0.90 $\pm$ 0.01 & 0.89 $\pm$ 0.01 & 0.94 $\pm$ 0.01 & 0.94 $\pm$ 0.00 & 0.90 $\pm$ 0.01 & 0.91 $\pm$ 0.01 & 0.95 $\pm$ 0.01 & 0.93 $\pm$ 0.00\\
\hspace{1em}$\text{1-Step } {\mathbf{R}}_{\text{Lon}} \otimes {\mathbf{R}}_{\text{Fun}}$ & 0.89 $\pm$ 0.01 & 0.87 $\pm$ 0.01 & 0.91 $\pm$ 0.01 & 0.93 $\pm$ 0.00 & 0.89 $\pm$ 0.01 & 0.89 $\pm$ 0.01 & 0.92 $\pm$ 0.01 & 0.92 $\pm$ 0.00\\
\hspace{1em}$\text{1-Step } I_{n_i} \otimes {\mathbf{R}}_{\text{Fun}}$ & 0.95 $\pm$ 0.00 & 0.96 $\pm$ 0.00 & 0.95 $\pm$ 0.01 & 0.98 $\pm$ 0.00 & 0.96 $\pm$ 0.00 & 0.97 $\pm$ 0.00 & 0.96 $\pm$ 0.01 & 0.99 $\pm$ 0.00\\
\hspace{1em}$\text{1-Step } {\mathbf{R}}_{\text{Lon}} \otimes I_L$ & 0.85 $\pm$ 0.01 & 0.88 $\pm$ 0.01 & 0.89 $\pm$ 0.01 & 0.91 $\pm$ 0.00 & 0.86 $\pm$ 0.01 & 0.89 $\pm$ 0.01 & 0.90 $\pm$ 0.01 & 0.90 $\pm$ 0.00\\
\hspace{1em}$\text{1-Step } I_{n_i} \otimes I_L$ & 0.91 $\pm$ 0.00 & 0.95 $\pm$ 0.00 & 0.93 $\pm$ 0.01 & 0.96 $\pm$ 0.00 & 0.92 $\pm$ 0.00 & 0.96 $\pm$ 0.00 & 0.94 $\pm$ 0.01 & 0.96 $\pm$ 0.00\\
\hspace{1em}$\text{Full-1Step } {\mathbf{R}}_{\text{FPCA}} \otimes {\mathbf{R}}_{\text{Fun}}$ & 0.90 $\pm$ 0.01 & 0.89 $\pm$ 0.01 & 0.95 $\pm$ 0.01 & 0.94 $\pm$ 0.00 & 0.90 $\pm$ 0.01 & 0.91 $\pm$ 0.01 & 0.95 $\pm$ 0.01 & 0.93 $\pm$ 0.00\\
\hspace{1em}$\text{Full-1Step } {\mathbf{R}}_{\text{Lon}} \otimes {\mathbf{R}}_{\text{Fun}}$ & 0.89 $\pm$ 0.01 & 0.87 $\pm$ 0.01 & 0.92 $\pm$ 0.01 & 0.93 $\pm$ 0.00 & 0.89 $\pm$ 0.01 & 0.89 $\pm$ 0.01 & 0.92 $\pm$ 0.01 & 0.92 $\pm$ 0.00\\
\hspace{1em}$\text{Full-1Step } I_{n_i} \otimes {\mathbf{R}}_{\text{Fun}}$ & 0.95 $\pm$ 0.00 & 0.96 $\pm$ 0.00 & 0.95 $\pm$ 0.01 & 0.98 $\pm$ 0.00 & 0.96 $\pm$ 0.00 & 0.97 $\pm$ 0.00 & 0.96 $\pm$ 0.01 & 0.99 $\pm$ 0.00\\
\hspace{1em}$\text{Full-1Step } {\mathbf{R}}_{\text{Lon}} \otimes I_L$ & 0.85 $\pm$ 0.01 & 0.88 $\pm$ 0.01 & 0.89 $\pm$ 0.01 & 0.91 $\pm$ 0.00 & 0.86 $\pm$ 0.01 & 0.89 $\pm$ 0.01 & 0.90 $\pm$ 0.01 & 0.90 $\pm$ 0.00\\
\hspace{1em}$\text{Full-1Step } I_{n_i} \otimes I_L$ & 0.91 $\pm$ 0.00 & 0.95 $\pm$ 0.00 & 0.94 $\pm$ 0.01 & 0.96 $\pm$ 0.00 & 0.92 $\pm$ 0.00 & 0.96 $\pm$ 0.00 & 0.94 $\pm$ 0.01 & 0.96 $\pm$ 0.00\\
\addlinespace[0.6em]
\multicolumn{9}{l}{\textbf{$N = 500,\; n_i = 100$}}\\
\hspace{1em}$\text{1-Step } {\mathbf{R}}_{\text{FPCA}} \otimes {\mathbf{R}}_{\text{Fun}}$ & 0.91 $\pm$ 0.00 & 0.71 $\pm$ 0.01 & 1.05 $\pm$ 0.01 & 0.98 $\pm$ 0.00 & 0.94 $\pm$ 0.00 & 0.89 $\pm$ 0.01 & 0.96 $\pm$ 0.00 & 0.95 $\pm$ 0.00\\
\hspace{1em}$\text{1-Step } {\mathbf{R}}_{\text{Lon}} \otimes {\mathbf{R}}_{\text{Fun}}$ & 0.90 $\pm$ 0.00 & 0.67 $\pm$ 0.01 & 0.96 $\pm$ 0.01 & 0.97 $\pm$ 0.00 & 0.93 $\pm$ 0.00 & 0.88 $\pm$ 0.01 & 0.94 $\pm$ 0.01 & 0.94 $\pm$ 0.00\\
\hspace{1em}$\text{1-Step } I_{n_i} \otimes {\mathbf{R}}_{\text{Fun}}$ & 0.92 $\pm$ 0.00 & 0.96 $\pm$ 0.00 & 0.90 $\pm$ 0.01 & 0.98 $\pm$ 0.00 & 0.99 $\pm$ 0.00 & 0.98 $\pm$ 0.00 & 0.98 $\pm$ 0.00 & 1.00 $\pm$ 0.00\\
\hspace{1em}$\text{1-Step } {\mathbf{R}}_{\text{Lon}} \otimes I_L$ & 0.87 $\pm$ 0.00 & 0.67 $\pm$ 0.01 & 0.94 $\pm$ 0.01 & 0.94 $\pm$ 0.00 & 0.90 $\pm$ 0.00 & 0.89 $\pm$ 0.00 & 0.93 $\pm$ 0.01 & 0.95 $\pm$ 0.00\\
\hspace{1em}$\text{1-Step } I_{n_i} \otimes I_L$ & 0.88 $\pm$ 0.00 & 0.95 $\pm$ 0.00 & 0.89 $\pm$ 0.01 & 0.95 $\pm$ 0.00 & 0.96 $\pm$ 0.00 & 0.99 $\pm$ 0.00 & 0.97 $\pm$ 0.00 & 1.00 $\pm$ 0.00\\
\hspace{1em}$\text{Full-1Step } {\mathbf{R}}_{\text{FPCA}} \otimes {\mathbf{R}}_{\text{Fun}}$ & 0.91 $\pm$ 0.01 & 0.71 $\pm$ 0.01 & 1.07 $\pm$ 0.01 & 0.98 $\pm$ 0.00 & 0.94 $\pm$ 0.00 & 0.89 $\pm$ 0.01 & 0.96 $\pm$ 0.00 & 0.95 $\pm$ 0.00\\
\hspace{1em}$\text{Full-1Step } {\mathbf{R}}_{\text{Lon}} \otimes {\mathbf{R}}_{\text{Fun}}$ & 0.90 $\pm$ 0.00 & 0.67 $\pm$ 0.01 & 0.98 $\pm$ 0.01 & 0.97 $\pm$ 0.00 & 0.93 $\pm$ 0.00 & 0.88 $\pm$ 0.01 & 0.94 $\pm$ 0.01 & 0.94 $\pm$ 0.00\\
\hspace{1em}$\text{Full-1Step } I_{n_i} \otimes {\mathbf{R}}_{\text{Fun}}$ & 0.92 $\pm$ 0.00 & 0.96 $\pm$ 0.00 & 0.91 $\pm$ 0.01 & 0.98 $\pm$ 0.00 & 0.99 $\pm$ 0.00 & 0.98 $\pm$ 0.00 & 0.98 $\pm$ 0.00 & 1.00 $\pm$ 0.00\\
\hspace{1em}$\text{Full-1Step } {\mathbf{R}}_{\text{Lon}} \otimes I_L$ & 0.87 $\pm$ 0.00 & 0.67 $\pm$ 0.01 & 0.95 $\pm$ 0.01 & 0.94 $\pm$ 0.00 & 0.90 $\pm$ 0.00 & 0.89 $\pm$ 0.00 & 0.93 $\pm$ 0.01 & 0.95 $\pm$ 0.00\\
\hspace{1em}$\text{Full-1Step } I_{n_i} \otimes I_L$ & 0.88 $\pm$ 0.00 & 0.95 $\pm$ 0.00 & 0.89 $\pm$ 0.01 & 0.95 $\pm$ 0.00 & 0.96 $\pm$ 0.00 & 0.99 $\pm$ 0.00 & 0.97 $\pm$ 0.00 & 1.00 $\pm$ 0.00\\
\bottomrule
\end{tabular}}
\end{table}

\begin{table}[!h]
\centering
\caption{\label{tab:jci_supp} \footnotesize
Joint 95\% CI coverage from 300 replicates $\pm$ SE (SE$=0.00$ indicates a value $<0.01$).
Table columns indicate if $\mathbf{R}^*_{\text{Lon}}$ had exchangeable or AR1 correlation. 
pffr (Wild) indicates CIs constructed for pffr coefficient estimates with quantiles obtained from a wild cluster bootstrap. 
}
\centering
\resizebox{\ifdim\width>\linewidth\linewidth\else\width\fi}{!}{
\fontsize{9}{11}\selectfont
\begin{tabular}[t]{>{\raggedright\arraybackslash}p{4.8cm}>{\raggedright\arraybackslash}p{2.1cm}>{\raggedright\arraybackslash}p{2.1cm}>{\raggedright\arraybackslash}p{2.1cm}>{\raggedright\arraybackslash}p{2.1cm}>{\raggedright\arraybackslash}p{2.1cm}>{\raggedright\arraybackslash}p{2.1cm}>{\raggedright\arraybackslash}p{2.1cm}>{\raggedright\arraybackslash}p{2.1cm}}
\toprule
\multicolumn{1}{c}{ } & \multicolumn{4}{c}{Exchangeable} & \multicolumn{4}{c}{AR(1)} \\
\cmidrule(l{3pt}r{3pt}){2-5} \cmidrule(l{3pt}r{3pt}){6-9}
Method & Gaussian & Poisson & Binomial & Gamma & Gaussian & Poisson & Binomial & Gamma\\
\midrule
\addlinespace[0em]
\multicolumn{9}{l}{\textbf{$N = 25,\; n_i = 5$}}\\
\hspace{1em}$\text{1-Step } {\mathbf{R}}_{\text{FPCA}} \otimes {\mathbf{R}}_{\text{Fun}}$ & 0.85 $\pm$ 0.02 & 0.80 $\pm$ 0.02 & 0.93 $\pm$ 0.01 & 0.74 $\pm$ 0.03 & 0.84 $\pm$ 0.02 & 0.80 $\pm$ 0.02 & 0.92 $\pm$ 0.02 & 0.78 $\pm$ 0.02\\
\hspace{1em}$\text{1-Step } {\mathbf{R}}_{\text{Lon}} \otimes {\mathbf{R}}_{\text{Fun}}$ & 0.90 $\pm$ 0.02 & 0.90 $\pm$ 0.02 & 0.95 $\pm$ 0.01 & 0.83 $\pm$ 0.02 & 0.90 $\pm$ 0.02 & 0.88 $\pm$ 0.02 & 0.95 $\pm$ 0.01 & 0.86 $\pm$ 0.02\\
\hspace{1em}$\text{1-Step } I_{n_i} \otimes {\mathbf{R}}_{\text{Fun}}$ & 0.89 $\pm$ 0.02 & 0.84 $\pm$ 0.02 & 0.94 $\pm$ 0.01 & 0.80 $\pm$ 0.02 & 0.87 $\pm$ 0.02 & 0.85 $\pm$ 0.02 & 0.93 $\pm$ 0.01 & 0.81 $\pm$ 0.02\\
\hspace{1em}$\text{1-Step } {\mathbf{R}}_{\text{Lon}} \otimes I_L$ & 0.92 $\pm$ 0.02 & 0.92 $\pm$ 0.02 & 0.95 $\pm$ 0.01 & 0.86 $\pm$ 0.02 & 0.91 $\pm$ 0.02 & 0.90 $\pm$ 0.02 & 0.94 $\pm$ 0.01 & 0.88 $\pm$ 0.02\\
\hspace{1em}$\text{1-Step } I_{n_i} \otimes I_L$ & 0.91 $\pm$ 0.02 & 0.88 $\pm$ 0.02 & 0.95 $\pm$ 0.01 & 0.85 $\pm$ 0.02 & 0.90 $\pm$ 0.02 & 0.86 $\pm$ 0.02 & 0.94 $\pm$ 0.01 & 0.85 $\pm$ 0.02\\
\hspace{1em}$\text{Full-1Step } {\mathbf{R}}_{\text{FPCA}} \otimes {\mathbf{R}}_{\text{Fun}}$ & 0.83 $\pm$ 0.02 & 0.80 $\pm$ 0.02 & 0.86 $\pm$ 0.02 & 0.71 $\pm$ 0.03 & 0.82 $\pm$ 0.02 & 0.79 $\pm$ 0.02 & 0.91 $\pm$ 0.02 & 0.71 $\pm$ 0.03\\
\hspace{1em}$\text{Full-1Step } {\mathbf{R}}_{\text{Lon}} \otimes {\mathbf{R}}_{\text{Fun}}$ & 0.90 $\pm$ 0.02 & 0.90 $\pm$ 0.02 & 0.94 $\pm$ 0.01 & 0.84 $\pm$ 0.02 & 0.89 $\pm$ 0.02 & 0.88 $\pm$ 0.02 & 0.93 $\pm$ 0.01 & 0.85 $\pm$ 0.02\\
\hspace{1em}$\text{Full-1Step } I_{n_i} \otimes {\mathbf{R}}_{\text{Fun}}$ & 0.89 $\pm$ 0.02 & 0.84 $\pm$ 0.02 & 0.93 $\pm$ 0.01 & 0.80 $\pm$ 0.02 & 0.86 $\pm$ 0.02 & 0.84 $\pm$ 0.02 & 0.92 $\pm$ 0.02 & 0.80 $\pm$ 0.02\\
\hspace{1em}$\text{Full-1Step } {\mathbf{R}}_{\text{Lon}} \otimes I_L$ & 0.93 $\pm$ 0.02 & 0.91 $\pm$ 0.02 & 0.90 $\pm$ 0.02 & 0.87 $\pm$ 0.02 & 0.91 $\pm$ 0.02 & 0.90 $\pm$ 0.02 & 0.93 $\pm$ 0.01 & 0.88 $\pm$ 0.02\\
\hspace{1em}$\text{Full-1Step } I_{n_i} \otimes I_L$ & 0.91 $\pm$ 0.02 & 0.87 $\pm$ 0.02 & 0.94 $\pm$ 0.01 & 0.85 $\pm$ 0.02 & 0.90 $\pm$ 0.02 & 0.86 $\pm$ 0.02 & 0.93 $\pm$ 0.02 & 0.85 $\pm$ 0.02\\
\hspace{1em}$\text{pffr (Wild)}$ & 0.92 $\pm$ 0.02 & 0.85 $\pm$ 0.02 & 0.95 $\pm$ 0.01 & 0.84 $\pm$ 0.02 & 0.93 $\pm$ 0.01 & 0.86 $\pm$ 0.02 & 0.96 $\pm$ 0.01 & 0.86 $\pm$ 0.02\\
\addlinespace[0.6em]
\multicolumn{9}{l}{\textbf{$N = 25,\; n_i = 100$}}\\
\hspace{1em}$\text{1-Step } {\mathbf{R}}_{\text{FPCA}} \otimes {\mathbf{R}}_{\text{Fun}}$ & 0.92 $\pm$ 0.02 & 0.95 $\pm$ 0.01 & 0.98 $\pm$ 0.01 & 0.81 $\pm$ 0.02 & 0.89 $\pm$ 0.02 & 0.87 $\pm$ 0.02 & 0.92 $\pm$ 0.02 & 0.85 $\pm$ 0.02\\
\hspace{1em}$\text{1-Step } {\mathbf{R}}_{\text{Lon}} \otimes {\mathbf{R}}_{\text{Fun}}$ & 0.96 $\pm$ 0.01 & 0.98 $\pm$ 0.01 & 0.99 $\pm$ 0.01 & 0.89 $\pm$ 0.02 & 0.89 $\pm$ 0.02 & 0.88 $\pm$ 0.02 & 0.93 $\pm$ 0.02 & 0.87 $\pm$ 0.02\\
\hspace{1em}$\text{1-Step } I_{n_i} \otimes {\mathbf{R}}_{\text{Fun}}$ & 0.97 $\pm$ 0.01 & 0.91 $\pm$ 0.02 & 0.99 $\pm$ 0.01 & 0.89 $\pm$ 0.02 & 0.88 $\pm$ 0.02 & 0.82 $\pm$ 0.02 & 0.93 $\pm$ 0.01 & 0.85 $\pm$ 0.02\\
\hspace{1em}$\text{1-Step } {\mathbf{R}}_{\text{Lon}} \otimes I_L$ & 0.97 $\pm$ 0.01 & 0.98 $\pm$ 0.01 & 0.99 $\pm$ 0.01 & 0.91 $\pm$ 0.02 & 0.92 $\pm$ 0.02 & 0.88 $\pm$ 0.02 & 0.93 $\pm$ 0.01 & 0.89 $\pm$ 0.02\\
\hspace{1em}$\text{1-Step } I_{n_i} \otimes I_L$ & 0.97 $\pm$ 0.01 & 0.91 $\pm$ 0.02 & 0.98 $\pm$ 0.01 & 0.91 $\pm$ 0.02 & 0.92 $\pm$ 0.02 & 0.83 $\pm$ 0.02 & 0.93 $\pm$ 0.01 & 0.87 $\pm$ 0.02\\
\hspace{1em}$\text{Full-1Step } {\mathbf{R}}_{\text{FPCA}} \otimes {\mathbf{R}}_{\text{Fun}}$ & 0.91 $\pm$ 0.02 & 0.95 $\pm$ 0.01 & 0.53 $\pm$ 0.03 & 0.78 $\pm$ 0.02 & 0.89 $\pm$ 0.02 & 0.87 $\pm$ 0.02 & 0.91 $\pm$ 0.02 & 0.86 $\pm$ 0.02\\
\hspace{1em}$\text{Full-1Step } {\mathbf{R}}_{\text{Lon}} \otimes {\mathbf{R}}_{\text{Fun}}$ & 0.96 $\pm$ 0.01 & 0.97 $\pm$ 0.01 & 0.81 $\pm$ 0.02 & 0.88 $\pm$ 0.02 & 0.89 $\pm$ 0.02 & 0.88 $\pm$ 0.02 & 0.93 $\pm$ 0.01 & 0.86 $\pm$ 0.02\\
\hspace{1em}$\text{Full-1Step } I_{n_i} \otimes {\mathbf{R}}_{\text{Fun}}$ & 0.96 $\pm$ 0.01 & 0.91 $\pm$ 0.02 & 0.99 $\pm$ 0.01 & 0.87 $\pm$ 0.02 & 0.88 $\pm$ 0.02 & 0.82 $\pm$ 0.02 & 0.93 $\pm$ 0.01 & 0.85 $\pm$ 0.02\\
\hspace{1em}$\text{Full-1Step } {\mathbf{R}}_{\text{Lon}} \otimes I_L$ & 0.98 $\pm$ 0.01 & 0.98 $\pm$ 0.01 & 0.61 $\pm$ 0.03 & 0.90 $\pm$ 0.02 & 0.92 $\pm$ 0.02 & 0.88 $\pm$ 0.02 & 0.93 $\pm$ 0.01 & 0.88 $\pm$ 0.02\\
\hspace{1em}$\text{Full-1Step } I_{n_i} \otimes I_L$ & 0.97 $\pm$ 0.01 & 0.91 $\pm$ 0.02 & 0.98 $\pm$ 0.01 & 0.90 $\pm$ 0.02 & 0.92 $\pm$ 0.02 & 0.83 $\pm$ 0.02 & 0.93 $\pm$ 0.01 & 0.87 $\pm$ 0.02\\
\hspace{1em}$\text{pffr (Wild)}$ & 0.79 $\pm$ 0.02 & 0.78 $\pm$ 0.02 & 0.75 $\pm$ 0.03 & 0.76 $\pm$ 0.02 & 0.99 $\pm$ 0.01 & 0.78 $\pm$ 0.02 & 0.93 $\pm$ 0.01 & 0.79 $\pm$ 0.02\\
\addlinespace[0.6em]
\multicolumn{9}{l}{\textbf{$N = 50,\; n_i = 5$}}\\
\hspace{1em}$\text{1-Step } {\mathbf{R}}_{\text{FPCA}} \otimes {\mathbf{R}}_{\text{Fun}}$ & 0.88 $\pm$ 0.02 & 0.85 $\pm$ 0.02 & 0.93 $\pm$ 0.01 & 0.82 $\pm$ 0.02 & 0.86 $\pm$ 0.02 & 0.85 $\pm$ 0.02 & 0.93 $\pm$ 0.01 & 0.83 $\pm$ 0.02\\
\hspace{1em}$\text{1-Step } {\mathbf{R}}_{\text{Lon}} \otimes {\mathbf{R}}_{\text{Fun}}$ & 0.90 $\pm$ 0.02 & 0.90 $\pm$ 0.02 & 0.95 $\pm$ 0.01 & 0.87 $\pm$ 0.02 & 0.89 $\pm$ 0.02 & 0.87 $\pm$ 0.02 & 0.94 $\pm$ 0.01 & 0.87 $\pm$ 0.02\\
\hspace{1em}$\text{1-Step } I_{n_i} \otimes {\mathbf{R}}_{\text{Fun}}$ & 0.89 $\pm$ 0.02 & 0.85 $\pm$ 0.02 & 0.95 $\pm$ 0.01 & 0.85 $\pm$ 0.02 & 0.89 $\pm$ 0.02 & 0.85 $\pm$ 0.02 & 0.94 $\pm$ 0.01 & 0.84 $\pm$ 0.02\\
\hspace{1em}$\text{1-Step } {\mathbf{R}}_{\text{Lon}} \otimes I_L$ & 0.92 $\pm$ 0.02 & 0.90 $\pm$ 0.02 & 0.95 $\pm$ 0.01 & 0.89 $\pm$ 0.02 & 0.92 $\pm$ 0.02 & 0.90 $\pm$ 0.02 & 0.95 $\pm$ 0.01 & 0.90 $\pm$ 0.02\\
\hspace{1em}$\text{1-Step } I_{n_i} \otimes I_L$ & 0.92 $\pm$ 0.02 & 0.88 $\pm$ 0.02 & 0.95 $\pm$ 0.01 & 0.89 $\pm$ 0.02 & 0.90 $\pm$ 0.02 & 0.86 $\pm$ 0.02 & 0.94 $\pm$ 0.01 & 0.89 $\pm$ 0.02\\
\hspace{1em}$\text{Full-1Step } {\mathbf{R}}_{\text{FPCA}} \otimes {\mathbf{R}}_{\text{Fun}}$ & 0.87 $\pm$ 0.02 & 0.85 $\pm$ 0.02 & 0.92 $\pm$ 0.02 & 0.82 $\pm$ 0.02 & 0.87 $\pm$ 0.02 & 0.83 $\pm$ 0.02 & 0.92 $\pm$ 0.02 & 0.80 $\pm$ 0.02\\
\hspace{1em}$\text{Full-1Step } {\mathbf{R}}_{\text{Lon}} \otimes {\mathbf{R}}_{\text{Fun}}$ & 0.90 $\pm$ 0.02 & 0.90 $\pm$ 0.02 & 0.94 $\pm$ 0.01 & 0.88 $\pm$ 0.02 & 0.89 $\pm$ 0.02 & 0.87 $\pm$ 0.02 & 0.94 $\pm$ 0.01 & 0.86 $\pm$ 0.02\\
\hspace{1em}$\text{Full-1Step } I_{n_i} \otimes {\mathbf{R}}_{\text{Fun}}$ & 0.89 $\pm$ 0.02 & 0.85 $\pm$ 0.02 & 0.94 $\pm$ 0.01 & 0.85 $\pm$ 0.02 & 0.88 $\pm$ 0.02 & 0.85 $\pm$ 0.02 & 0.93 $\pm$ 0.01 & 0.84 $\pm$ 0.02\\
\hspace{1em}$\text{Full-1Step } {\mathbf{R}}_{\text{Lon}} \otimes I_L$ & 0.92 $\pm$ 0.02 & 0.90 $\pm$ 0.02 & 0.94 $\pm$ 0.01 & 0.90 $\pm$ 0.02 & 0.92 $\pm$ 0.02 & 0.90 $\pm$ 0.02 & 0.94 $\pm$ 0.01 & 0.89 $\pm$ 0.02\\
\hspace{1em}$\text{Full-1Step } I_{n_i} \otimes I_L$ & 0.92 $\pm$ 0.02 & 0.87 $\pm$ 0.02 & 0.94 $\pm$ 0.01 & 0.89 $\pm$ 0.02 & 0.90 $\pm$ 0.02 & 0.86 $\pm$ 0.02 & 0.94 $\pm$ 0.01 & 0.88 $\pm$ 0.02\\
\hspace{1em}$\text{pffr (Wild)}$ & 0.97 $\pm$ 0.01 & 0.91 $\pm$ 0.02 & 0.99 $\pm$ 0.01 & 0.94 $\pm$ 0.01 & 0.97 $\pm$ 0.01 & 0.91 $\pm$ 0.02 & 0.99 $\pm$ 0.01 & 0.94 $\pm$ 0.01\\
\addlinespace[0.6em]
\multicolumn{9}{l}{\textbf{$N = 50,\; n_i = 100$}}\\
\hspace{1em}$\text{1-Step } {\mathbf{R}}_{\text{FPCA}} \otimes {\mathbf{R}}_{\text{Fun}}$ & 0.94 $\pm$ 0.01 & 0.96 $\pm$ 0.01 & 0.97 $\pm$ 0.01 & 0.89 $\pm$ 0.02 & 0.92 $\pm$ 0.02 & 0.90 $\pm$ 0.02 & 0.94 $\pm$ 0.01 & 0.89 $\pm$ 0.02\\
\hspace{1em}$\text{1-Step } {\mathbf{R}}_{\text{Lon}} \otimes {\mathbf{R}}_{\text{Fun}}$ & 0.96 $\pm$ 0.01 & 0.97 $\pm$ 0.01 & 0.98 $\pm$ 0.01 & 0.92 $\pm$ 0.02 & 0.93 $\pm$ 0.02 & 0.90 $\pm$ 0.02 & 0.95 $\pm$ 0.01 & 0.90 $\pm$ 0.02\\
\hspace{1em}$\text{1-Step } I_{n_i} \otimes {\mathbf{R}}_{\text{Fun}}$ & 0.98 $\pm$ 0.01 & 0.91 $\pm$ 0.02 & 0.99 $\pm$ 0.01 & 0.91 $\pm$ 0.02 & 0.92 $\pm$ 0.02 & 0.84 $\pm$ 0.02 & 0.94 $\pm$ 0.01 & 0.89 $\pm$ 0.02\\
\hspace{1em}$\text{1-Step } {\mathbf{R}}_{\text{Lon}} \otimes I_L$ & 0.97 $\pm$ 0.01 & 0.98 $\pm$ 0.01 & 0.98 $\pm$ 0.01 & 0.94 $\pm$ 0.01 & 0.95 $\pm$ 0.01 & 0.91 $\pm$ 0.02 & 0.95 $\pm$ 0.01 & 0.92 $\pm$ 0.02\\
\hspace{1em}$\text{1-Step } I_{n_i} \otimes I_L$ & 0.98 $\pm$ 0.01 & 0.94 $\pm$ 0.01 & 0.99 $\pm$ 0.01 & 0.94 $\pm$ 0.01 & 0.95 $\pm$ 0.01 & 0.85 $\pm$ 0.02 & 0.95 $\pm$ 0.01 & 0.90 $\pm$ 0.02\\
\hspace{1em}$\text{Full-1Step } {\mathbf{R}}_{\text{FPCA}} \otimes {\mathbf{R}}_{\text{Fun}}$ & 0.94 $\pm$ 0.01 & 0.96 $\pm$ 0.01 & 0.71 $\pm$ 0.03 & 0.89 $\pm$ 0.02 & 0.92 $\pm$ 0.02 & 0.90 $\pm$ 0.02 & 0.94 $\pm$ 0.01 & 0.89 $\pm$ 0.02\\
\hspace{1em}$\text{Full-1Step } {\mathbf{R}}_{\text{Lon}} \otimes {\mathbf{R}}_{\text{Fun}}$ & 0.96 $\pm$ 0.01 & 0.98 $\pm$ 0.01 & 0.89 $\pm$ 0.02 & 0.91 $\pm$ 0.02 & 0.93 $\pm$ 0.02 & 0.90 $\pm$ 0.02 & 0.95 $\pm$ 0.01 & 0.90 $\pm$ 0.02\\
\hspace{1em}$\text{Full-1Step } I_{n_i} \otimes {\mathbf{R}}_{\text{Fun}}$ & 0.98 $\pm$ 0.01 & 0.91 $\pm$ 0.02 & 0.98 $\pm$ 0.01 & 0.92 $\pm$ 0.02 & 0.92 $\pm$ 0.02 & 0.84 $\pm$ 0.02 & 0.94 $\pm$ 0.01 & 0.89 $\pm$ 0.02\\
\hspace{1em}$\text{Full-1Step } {\mathbf{R}}_{\text{Lon}} \otimes I_L$ & 0.97 $\pm$ 0.01 & 0.98 $\pm$ 0.01 & 0.76 $\pm$ 0.03 & 0.93 $\pm$ 0.01 & 0.95 $\pm$ 0.01 & 0.91 $\pm$ 0.02 & 0.95 $\pm$ 0.01 & 0.92 $\pm$ 0.02\\
\hspace{1em}$\text{Full-1Step } I_{n_i} \otimes I_L$ & 0.98 $\pm$ 0.01 & 0.94 $\pm$ 0.01 & 0.98 $\pm$ 0.01 & 0.94 $\pm$ 0.01 & 0.95 $\pm$ 0.01 & 0.85 $\pm$ 0.02 & 0.95 $\pm$ 0.01 & 0.90 $\pm$ 0.02\\
\hspace{1em}$\text{pffr (Wild)}$ & 0.89 $\pm$ 0.02 & 0.82 $\pm$ 0.02 & 0.85 $\pm$ 0.02 & 0.85 $\pm$ 0.02 & 1.00 $\pm$ 0.00 & 0.85 $\pm$ 0.02 & 0.99 $\pm$ 0.01 & 0.87 $\pm$ 0.02\\
\addlinespace[0.6em]
\multicolumn{9}{l}{\textbf{$N = 500,\; n_i = 5$}}\\
\hspace{1em}$\text{1-Step } {\mathbf{R}}_{\text{FPCA}} \otimes {\mathbf{R}}_{\text{Fun}}$ & 0.95 $\pm$ 0.01 & 0.94 $\pm$ 0.01 & 0.95 $\pm$ 0.01 & 0.94 $\pm$ 0.01 & 0.95 $\pm$ 0.01 & 0.93 $\pm$ 0.02 & 0.95 $\pm$ 0.01 & 0.95 $\pm$ 0.01\\
\hspace{1em}$\text{1-Step } {\mathbf{R}}_{\text{Lon}} \otimes {\mathbf{R}}_{\text{Fun}}$ & 0.94 $\pm$ 0.01 & 0.93 $\pm$ 0.01 & 0.96 $\pm$ 0.01 & 0.93 $\pm$ 0.01 & 0.95 $\pm$ 0.01 & 0.93 $\pm$ 0.02 & 0.96 $\pm$ 0.01 & 0.94 $\pm$ 0.01\\
\hspace{1em}$\text{1-Step } I_{n_i} \otimes {\mathbf{R}}_{\text{Fun}}$ & 0.95 $\pm$ 0.01 & 0.93 $\pm$ 0.02 & 0.96 $\pm$ 0.01 & 0.94 $\pm$ 0.01 & 0.95 $\pm$ 0.01 & 0.93 $\pm$ 0.02 & 0.96 $\pm$ 0.01 & 0.94 $\pm$ 0.01\\
\hspace{1em}$\text{1-Step } {\mathbf{R}}_{\text{Lon}} \otimes I_L$ & 0.97 $\pm$ 0.01 & 0.94 $\pm$ 0.01 & 0.96 $\pm$ 0.01 & 0.95 $\pm$ 0.01 & 0.96 $\pm$ 0.01 & 0.93 $\pm$ 0.01 & 0.97 $\pm$ 0.01 & 0.95 $\pm$ 0.01\\
\hspace{1em}$\text{1-Step } I_{n_i} \otimes I_L$ & 0.97 $\pm$ 0.01 & 0.94 $\pm$ 0.01 & 0.96 $\pm$ 0.01 & 0.95 $\pm$ 0.01 & 0.96 $\pm$ 0.01 & 0.93 $\pm$ 0.01 & 0.97 $\pm$ 0.01 & 0.95 $\pm$ 0.01\\
\hspace{1em}$\text{Full-1Step } {\mathbf{R}}_{\text{FPCA}} \otimes {\mathbf{R}}_{\text{Fun}}$ & 0.95 $\pm$ 0.01 & 0.94 $\pm$ 0.01 & 0.95 $\pm$ 0.01 & 0.94 $\pm$ 0.01 & 0.95 $\pm$ 0.01 & 0.93 $\pm$ 0.01 & 0.95 $\pm$ 0.01 & 0.95 $\pm$ 0.01\\
\hspace{1em}$\text{Full-1Step } {\mathbf{R}}_{\text{Lon}} \otimes {\mathbf{R}}_{\text{Fun}}$ & 0.94 $\pm$ 0.01 & 0.93 $\pm$ 0.01 & 0.96 $\pm$ 0.01 & 0.94 $\pm$ 0.01 & 0.95 $\pm$ 0.01 & 0.92 $\pm$ 0.02 & 0.96 $\pm$ 0.01 & 0.94 $\pm$ 0.01\\
\hspace{1em}$\text{Full-1Step } I_{n_i} \otimes {\mathbf{R}}_{\text{Fun}}$ & 0.95 $\pm$ 0.01 & 0.93 $\pm$ 0.01 & 0.96 $\pm$ 0.01 & 0.94 $\pm$ 0.01 & 0.95 $\pm$ 0.01 & 0.93 $\pm$ 0.02 & 0.96 $\pm$ 0.01 & 0.94 $\pm$ 0.01\\
\hspace{1em}$\text{Full-1Step } {\mathbf{R}}_{\text{Lon}} \otimes I_L$ & 0.97 $\pm$ 0.01 & 0.94 $\pm$ 0.01 & 0.96 $\pm$ 0.01 & 0.95 $\pm$ 0.01 & 0.96 $\pm$ 0.01 & 0.93 $\pm$ 0.01 & 0.96 $\pm$ 0.01 & 0.95 $\pm$ 0.01\\
\hspace{1em}$\text{Full-1Step } I_{n_i} \otimes I_L$ & 0.97 $\pm$ 0.01 & 0.94 $\pm$ 0.01 & 0.96 $\pm$ 0.01 & 0.95 $\pm$ 0.01 & 0.96 $\pm$ 0.01 & 0.93 $\pm$ 0.02 & 0.97 $\pm$ 0.01 & 0.95 $\pm$ 0.01\\
\hspace{1em}$\text{pffr (Wild)}$ & 1.00 $\pm$ 0.00 & 1.00 $\pm$ 0.00 & 1.00 $\pm$ 0.00 & 1.00 $\pm$ 0.00 & 1.00 $\pm$ 0.00 & 1.00 $\pm$ 0.00 & 1.00 $\pm$ 0.00 & 1.00 $\pm$ 0.00\\
\addlinespace[0.6em]
\multicolumn{9}{l}{\textbf{$N = 500,\; n_i = 100$}}\\
\hspace{1em}$\text{1-Step } {\mathbf{R}}_{\text{FPCA}} \otimes {\mathbf{R}}_{\text{Fun}}$ & 0.93 $\pm$ 0.01 & 0.90 $\pm$ 0.02 & 0.97 $\pm$ 0.01 & 0.79 $\pm$ 0.02 & 0.95 $\pm$ 0.01 & 0.89 $\pm$ 0.02 & 0.95 $\pm$ 0.01 & 0.92 $\pm$ 0.02\\
\hspace{1em}$\text{1-Step } {\mathbf{R}}_{\text{Lon}} \otimes {\mathbf{R}}_{\text{Fun}}$ & 0.94 $\pm$ 0.01 & 0.92 $\pm$ 0.02 & 0.97 $\pm$ 0.01 & 0.81 $\pm$ 0.02 & 0.94 $\pm$ 0.01 & 0.89 $\pm$ 0.02 & 0.95 $\pm$ 0.01 & 0.93 $\pm$ 0.01\\
\hspace{1em}$\text{1-Step } I_{n_i} \otimes {\mathbf{R}}_{\text{Fun}}$ & 0.95 $\pm$ 0.01 & 0.92 $\pm$ 0.02 & 0.98 $\pm$ 0.01 & 0.93 $\pm$ 0.02 & 0.95 $\pm$ 0.01 & 0.91 $\pm$ 0.02 & 0.95 $\pm$ 0.01 & 0.94 $\pm$ 0.01\\
\hspace{1em}$\text{1-Step } {\mathbf{R}}_{\text{Lon}} \otimes I_L$ & 0.96 $\pm$ 0.01 & 0.92 $\pm$ 0.02 & 0.97 $\pm$ 0.01 & 0.84 $\pm$ 0.02 & 0.95 $\pm$ 0.01 & 0.91 $\pm$ 0.02 & 0.96 $\pm$ 0.01 & 0.92 $\pm$ 0.02\\
\hspace{1em}$\text{1-Step } I_{n_i} \otimes I_L$ & 0.97 $\pm$ 0.01 & 0.93 $\pm$ 0.01 & 0.98 $\pm$ 0.01 & 0.93 $\pm$ 0.01 & 0.96 $\pm$ 0.01 & 0.91 $\pm$ 0.02 & 0.95 $\pm$ 0.01 & 0.94 $\pm$ 0.01\\
\hspace{1em}$\text{Full-1Step } {\mathbf{R}}_{\text{FPCA}} \otimes {\mathbf{R}}_{\text{Fun}}$ & 0.93 $\pm$ 0.01 & 0.90 $\pm$ 0.02 & 0.96 $\pm$ 0.01 & 0.79 $\pm$ 0.02 & 0.95 $\pm$ 0.01 & 0.89 $\pm$ 0.02 & 0.95 $\pm$ 0.01 & 0.92 $\pm$ 0.02\\
\hspace{1em}$\text{Full-1Step } {\mathbf{R}}_{\text{Lon}} \otimes {\mathbf{R}}_{\text{Fun}}$ & 0.94 $\pm$ 0.01 & 0.92 $\pm$ 0.02 & 0.97 $\pm$ 0.01 & 0.81 $\pm$ 0.02 & 0.94 $\pm$ 0.01 & 0.89 $\pm$ 0.02 & 0.95 $\pm$ 0.01 & 0.93 $\pm$ 0.02\\
\hspace{1em}$\text{Full-1Step } I_{n_i} \otimes {\mathbf{R}}_{\text{Fun}}$ & 0.95 $\pm$ 0.01 & 0.92 $\pm$ 0.02 & 0.98 $\pm$ 0.01 & 0.92 $\pm$ 0.02 & 0.95 $\pm$ 0.01 & 0.91 $\pm$ 0.02 & 0.95 $\pm$ 0.01 & 0.94 $\pm$ 0.01\\
\hspace{1em}$\text{Full-1Step } {\mathbf{R}}_{\text{Lon}} \otimes I_L$ & 0.96 $\pm$ 0.01 & 0.92 $\pm$ 0.02 & 0.98 $\pm$ 0.01 & 0.84 $\pm$ 0.02 & 0.95 $\pm$ 0.01 & 0.90 $\pm$ 0.02 & 0.96 $\pm$ 0.01 & 0.92 $\pm$ 0.02\\
\hspace{1em}$\text{Full-1Step } I_{n_i} \otimes I_L$ & 0.97 $\pm$ 0.01 & 0.93 $\pm$ 0.01 & 0.98 $\pm$ 0.01 & 0.93 $\pm$ 0.01 & 0.96 $\pm$ 0.01 & 0.91 $\pm$ 0.02 & 0.95 $\pm$ 0.01 & 0.94 $\pm$ 0.01\\
\hspace{1em}$\text{pffr (Wild)}$ & 1.00 $\pm$ 0.00 & 0.93 $\pm$ 0.01 & 0.98 $\pm$ 0.01 & 0.93 $\pm$ 0.01 & 1.00 $\pm$ 0.00 & 0.96 $\pm$ 0.01 & 1.00 $\pm$ 0.00 & 0.99 $\pm$ 0.01\\
\bottomrule
\end{tabular}}
\end{table}

\begin{table}[!h]
\centering
\caption{\label{tab:pci_supp} \footnotesize
Pointwise 95\% CI coverage from 300 replicates $\pm$ SE (SE$=0.00$ indicates a value $<0.01$). 
Table columns indicate if $\mathbf{R}^*_{\text{Lon}}$ had exchangeable or AR1 correlation. 
pffr (Wild) indicates CIs constructed for pffr coefficient estimates with quantiles obtained from a wild cluster bootstrap. 
}
\centering
\resizebox{\ifdim\width>\linewidth\linewidth\else\width\fi}{!}{
\fontsize{9}{11}\selectfont
\begin{tabular}[t]{>{\raggedright\arraybackslash}p{4.8cm}>{\raggedright\arraybackslash}p{2.1cm}>{\raggedright\arraybackslash}p{2.1cm}>{\raggedright\arraybackslash}p{2.1cm}>{\raggedright\arraybackslash}p{2.1cm}>{\raggedright\arraybackslash}p{2.1cm}>{\raggedright\arraybackslash}p{2.1cm}>{\raggedright\arraybackslash}p{2.1cm}>{\raggedright\arraybackslash}p{2.1cm}}
\toprule
\multicolumn{1}{c}{ } & \multicolumn{4}{c}{Exchangeable} & \multicolumn{4}{c}{AR(1)} \\
\cmidrule(l{3pt}r{3pt}){2-5} \cmidrule(l{3pt}r{3pt}){6-9}
Method & Gaussian & Poisson & Binomial & Gamma & Gaussian & Poisson & Binomial & Gamma\\
\midrule
\addlinespace[0em]
\multicolumn{9}{l}{\textbf{$N = 25,\; n_i = 5$}}\\
\hspace{1em}$\text{1-Step } {\mathbf{R}}_{\text{FPCA}} \otimes {\mathbf{R}}_{\text{Fun}}$ & 0.92 $\pm$ 0.02 & 0.93 $\pm$ 0.01 & 0.96 $\pm$ 0.01 & 0.90 $\pm$ 0.02 & 0.92 $\pm$ 0.02 & 0.93 $\pm$ 0.02 & 0.96 $\pm$ 0.01 & 0.92 $\pm$ 0.02\\
\hspace{1em}$\text{1-Step } {\mathbf{R}}_{\text{Lon}} \otimes {\mathbf{R}}_{\text{Fun}}$ & 0.95 $\pm$ 0.01 & 0.95 $\pm$ 0.01 & 0.97 $\pm$ 0.01 & 0.93 $\pm$ 0.01 & 0.94 $\pm$ 0.01 & 0.95 $\pm$ 0.01 & 0.96 $\pm$ 0.01 & 0.93 $\pm$ 0.01\\
\hspace{1em}$\text{1-Step } I_{n_i} \otimes {\mathbf{R}}_{\text{Fun}}$ & 0.94 $\pm$ 0.01 & 0.93 $\pm$ 0.01 & 0.96 $\pm$ 0.01 & 0.93 $\pm$ 0.02 & 0.93 $\pm$ 0.01 & 0.93 $\pm$ 0.01 & 0.96 $\pm$ 0.01 & 0.93 $\pm$ 0.02\\
\hspace{1em}$\text{1-Step } {\mathbf{R}}_{\text{Lon}} \otimes I_L$ & 0.95 $\pm$ 0.01 & 0.95 $\pm$ 0.01 & 0.97 $\pm$ 0.01 & 0.94 $\pm$ 0.01 & 0.95 $\pm$ 0.01 & 0.95 $\pm$ 0.01 & 0.96 $\pm$ 0.01 & 0.94 $\pm$ 0.01\\
\hspace{1em}$\text{1-Step } I_{n_i} \otimes I_L$ & 0.95 $\pm$ 0.01 & 0.94 $\pm$ 0.01 & 0.96 $\pm$ 0.01 & 0.93 $\pm$ 0.01 & 0.94 $\pm$ 0.01 & 0.94 $\pm$ 0.01 & 0.96 $\pm$ 0.01 & 0.93 $\pm$ 0.01\\
\hspace{1em}$\text{Full-1Step } {\mathbf{R}}_{\text{FPCA}} \otimes {\mathbf{R}}_{\text{Fun}}$ & 0.92 $\pm$ 0.02 & 0.92 $\pm$ 0.02 & 0.90 $\pm$ 0.02 & 0.89 $\pm$ 0.02 & 0.92 $\pm$ 0.02 & 0.92 $\pm$ 0.02 & 0.95 $\pm$ 0.01 & 0.87 $\pm$ 0.02\\
\hspace{1em}$\text{Full-1Step } {\mathbf{R}}_{\text{Lon}} \otimes {\mathbf{R}}_{\text{Fun}}$ & 0.95 $\pm$ 0.01 & 0.95 $\pm$ 0.01 & 0.96 $\pm$ 0.01 & 0.93 $\pm$ 0.01 & 0.94 $\pm$ 0.01 & 0.95 $\pm$ 0.01 & 0.96 $\pm$ 0.01 & 0.93 $\pm$ 0.01\\
\hspace{1em}$\text{Full-1Step } I_{n_i} \otimes {\mathbf{R}}_{\text{Fun}}$ & 0.94 $\pm$ 0.01 & 0.93 $\pm$ 0.01 & 0.96 $\pm$ 0.01 & 0.92 $\pm$ 0.02 & 0.93 $\pm$ 0.01 & 0.93 $\pm$ 0.01 & 0.96 $\pm$ 0.01 & 0.92 $\pm$ 0.02\\
\hspace{1em}$\text{Full-1Step } {\mathbf{R}}_{\text{Lon}} \otimes I_L$ & 0.95 $\pm$ 0.01 & 0.95 $\pm$ 0.01 & 0.93 $\pm$ 0.02 & 0.94 $\pm$ 0.01 & 0.95 $\pm$ 0.01 & 0.95 $\pm$ 0.01 & 0.96 $\pm$ 0.01 & 0.94 $\pm$ 0.01\\
\hspace{1em}$\text{Full-1Step } I_{n_i} \otimes I_L$ & 0.95 $\pm$ 0.01 & 0.94 $\pm$ 0.01 & 0.96 $\pm$ 0.01 & 0.93 $\pm$ 0.01 & 0.94 $\pm$ 0.01 & 0.94 $\pm$ 0.01 & 0.96 $\pm$ 0.01 & 0.93 $\pm$ 0.01\\
\hspace{1em}$\text{pffr (Wild)}$ & 0.96 $\pm$ 0.01 & 0.93 $\pm$ 0.01 & 0.97 $\pm$ 0.01 & 0.93 $\pm$ 0.01 & 0.97 $\pm$ 0.01 & 0.94 $\pm$ 0.01 & 0.97 $\pm$ 0.01 & 0.94 $\pm$ 0.01\\
\addlinespace[0.6em]
\multicolumn{9}{l}{\textbf{$N = 25,\; n_i = 100$}}\\
\hspace{1em}$\text{1-Step } {\mathbf{R}}_{\text{FPCA}} \otimes {\mathbf{R}}_{\text{Fun}}$ & 0.95 $\pm$ 0.01 & 0.97 $\pm$ 0.01 & 0.98 $\pm$ 0.01 & 0.93 $\pm$ 0.02 & 0.95 $\pm$ 0.01 & 0.95 $\pm$ 0.01 & 0.95 $\pm$ 0.01 & 0.95 $\pm$ 0.01\\
\hspace{1em}$\text{1-Step } {\mathbf{R}}_{\text{Lon}} \otimes {\mathbf{R}}_{\text{Fun}}$ & 0.97 $\pm$ 0.01 & 0.98 $\pm$ 0.01 & 0.98 $\pm$ 0.01 & 0.95 $\pm$ 0.01 & 0.95 $\pm$ 0.01 & 0.95 $\pm$ 0.01 & 0.96 $\pm$ 0.01 & 0.95 $\pm$ 0.01\\
\hspace{1em}$\text{1-Step } I_{n_i} \otimes {\mathbf{R}}_{\text{Fun}}$ & 0.97 $\pm$ 0.01 & 0.95 $\pm$ 0.01 & 0.99 $\pm$ 0.01 & 0.94 $\pm$ 0.01 & 0.94 $\pm$ 0.01 & 0.94 $\pm$ 0.01 & 0.96 $\pm$ 0.01 & 0.94 $\pm$ 0.01\\
\hspace{1em}$\text{1-Step } {\mathbf{R}}_{\text{Lon}} \otimes I_L$ & 0.97 $\pm$ 0.01 & 0.98 $\pm$ 0.01 & 0.99 $\pm$ 0.01 & 0.95 $\pm$ 0.01 & 0.95 $\pm$ 0.01 & 0.95 $\pm$ 0.01 & 0.96 $\pm$ 0.01 & 0.95 $\pm$ 0.01\\
\hspace{1em}$\text{1-Step } I_{n_i} \otimes I_L$ & 0.97 $\pm$ 0.01 & 0.95 $\pm$ 0.01 & 0.98 $\pm$ 0.01 & 0.95 $\pm$ 0.01 & 0.95 $\pm$ 0.01 & 0.94 $\pm$ 0.01 & 0.96 $\pm$ 0.01 & 0.94 $\pm$ 0.01\\
\hspace{1em}$\text{Full-1Step } {\mathbf{R}}_{\text{FPCA}} \otimes {\mathbf{R}}_{\text{Fun}}$ & 0.95 $\pm$ 0.01 & 0.96 $\pm$ 0.01 & 0.60 $\pm$ 0.03 & 0.91 $\pm$ 0.02 & 0.95 $\pm$ 0.01 & 0.95 $\pm$ 0.01 & 0.95 $\pm$ 0.01 & 0.95 $\pm$ 0.01\\
\hspace{1em}$\text{Full-1Step } {\mathbf{R}}_{\text{Lon}} \otimes {\mathbf{R}}_{\text{Fun}}$ & 0.97 $\pm$ 0.01 & 0.98 $\pm$ 0.01 & 0.84 $\pm$ 0.02 & 0.95 $\pm$ 0.01 & 0.95 $\pm$ 0.01 & 0.95 $\pm$ 0.01 & 0.96 $\pm$ 0.01 & 0.95 $\pm$ 0.01\\
\hspace{1em}$\text{Full-1Step } I_{n_i} \otimes {\mathbf{R}}_{\text{Fun}}$ & 0.97 $\pm$ 0.01 & 0.95 $\pm$ 0.01 & 0.98 $\pm$ 0.01 & 0.94 $\pm$ 0.01 & 0.94 $\pm$ 0.01 & 0.94 $\pm$ 0.01 & 0.96 $\pm$ 0.01 & 0.94 $\pm$ 0.01\\
\hspace{1em}$\text{Full-1Step } {\mathbf{R}}_{\text{Lon}} \otimes I_L$ & 0.98 $\pm$ 0.01 & 0.98 $\pm$ 0.01 & 0.67 $\pm$ 0.03 & 0.95 $\pm$ 0.01 & 0.95 $\pm$ 0.01 & 0.95 $\pm$ 0.01 & 0.96 $\pm$ 0.01 & 0.95 $\pm$ 0.01\\
\hspace{1em}$\text{Full-1Step } I_{n_i} \otimes I_L$ & 0.97 $\pm$ 0.01 & 0.95 $\pm$ 0.01 & 0.98 $\pm$ 0.01 & 0.94 $\pm$ 0.01 & 0.95 $\pm$ 0.01 & 0.94 $\pm$ 0.01 & 0.96 $\pm$ 0.01 & 0.94 $\pm$ 0.01\\
\hspace{1em}$\text{pffr (Wild)}$ & 0.94 $\pm$ 0.01 & 0.93 $\pm$ 0.01 & 0.92 $\pm$ 0.02 & 0.93 $\pm$ 0.02 & 0.99 $\pm$ 0.00 & 0.93 $\pm$ 0.01 & 0.96 $\pm$ 0.01 & 0.93 $\pm$ 0.01\\
\addlinespace[0.6em]
\multicolumn{9}{l}{\textbf{$N = 50,\; n_i = 5$}}\\
\hspace{1em}$\text{1-Step } {\mathbf{R}}_{\text{FPCA}} \otimes {\mathbf{R}}_{\text{Fun}}$ & 0.93 $\pm$ 0.01 & 0.93 $\pm$ 0.01 & 0.95 $\pm$ 0.01 & 0.91 $\pm$ 0.02 & 0.93 $\pm$ 0.02 & 0.92 $\pm$ 0.02 & 0.95 $\pm$ 0.01 & 0.92 $\pm$ 0.02\\
\hspace{1em}$\text{1-Step } {\mathbf{R}}_{\text{Lon}} \otimes {\mathbf{R}}_{\text{Fun}}$ & 0.94 $\pm$ 0.01 & 0.94 $\pm$ 0.01 & 0.96 $\pm$ 0.01 & 0.93 $\pm$ 0.01 & 0.93 $\pm$ 0.01 & 0.94 $\pm$ 0.01 & 0.95 $\pm$ 0.01 & 0.94 $\pm$ 0.01\\
\hspace{1em}$\text{1-Step } I_{n_i} \otimes {\mathbf{R}}_{\text{Fun}}$ & 0.93 $\pm$ 0.01 & 0.93 $\pm$ 0.01 & 0.95 $\pm$ 0.01 & 0.93 $\pm$ 0.01 & 0.93 $\pm$ 0.01 & 0.93 $\pm$ 0.01 & 0.95 $\pm$ 0.01 & 0.93 $\pm$ 0.01\\
\hspace{1em}$\text{1-Step } {\mathbf{R}}_{\text{Lon}} \otimes I_L$ & 0.95 $\pm$ 0.01 & 0.95 $\pm$ 0.01 & 0.96 $\pm$ 0.01 & 0.94 $\pm$ 0.01 & 0.94 $\pm$ 0.01 & 0.94 $\pm$ 0.01 & 0.95 $\pm$ 0.01 & 0.94 $\pm$ 0.01\\
\hspace{1em}$\text{1-Step } I_{n_i} \otimes I_L$ & 0.94 $\pm$ 0.01 & 0.93 $\pm$ 0.01 & 0.95 $\pm$ 0.01 & 0.93 $\pm$ 0.01 & 0.94 $\pm$ 0.01 & 0.93 $\pm$ 0.01 & 0.95 $\pm$ 0.01 & 0.94 $\pm$ 0.01\\
\hspace{1em}$\text{Full-1Step } {\mathbf{R}}_{\text{FPCA}} \otimes {\mathbf{R}}_{\text{Fun}}$ & 0.93 $\pm$ 0.01 & 0.93 $\pm$ 0.02 & 0.95 $\pm$ 0.01 & 0.91 $\pm$ 0.02 & 0.92 $\pm$ 0.02 & 0.92 $\pm$ 0.02 & 0.94 $\pm$ 0.01 & 0.92 $\pm$ 0.02\\
\hspace{1em}$\text{Full-1Step } {\mathbf{R}}_{\text{Lon}} \otimes {\mathbf{R}}_{\text{Fun}}$ & 0.94 $\pm$ 0.01 & 0.94 $\pm$ 0.01 & 0.95 $\pm$ 0.01 & 0.93 $\pm$ 0.01 & 0.93 $\pm$ 0.01 & 0.94 $\pm$ 0.01 & 0.95 $\pm$ 0.01 & 0.93 $\pm$ 0.01\\
\hspace{1em}$\text{Full-1Step } I_{n_i} \otimes {\mathbf{R}}_{\text{Fun}}$ & 0.93 $\pm$ 0.01 & 0.93 $\pm$ 0.01 & 0.95 $\pm$ 0.01 & 0.93 $\pm$ 0.01 & 0.93 $\pm$ 0.01 & 0.93 $\pm$ 0.01 & 0.95 $\pm$ 0.01 & 0.93 $\pm$ 0.01\\
\hspace{1em}$\text{Full-1Step } {\mathbf{R}}_{\text{Lon}} \otimes I_L$ & 0.95 $\pm$ 0.01 & 0.94 $\pm$ 0.01 & 0.95 $\pm$ 0.01 & 0.94 $\pm$ 0.01 & 0.94 $\pm$ 0.01 & 0.94 $\pm$ 0.01 & 0.95 $\pm$ 0.01 & 0.94 $\pm$ 0.01\\
\hspace{1em}$\text{Full-1Step } I_{n_i} \otimes I_L$ & 0.94 $\pm$ 0.01 & 0.93 $\pm$ 0.01 & 0.95 $\pm$ 0.01 & 0.93 $\pm$ 0.01 & 0.94 $\pm$ 0.01 & 0.93 $\pm$ 0.01 & 0.95 $\pm$ 0.01 & 0.94 $\pm$ 0.01\\
\hspace{1em}$\text{pffr (Wild)}$ & 0.98 $\pm$ 0.01 & 0.95 $\pm$ 0.01 & 0.99 $\pm$ 0.01 & 0.96 $\pm$ 0.01 & 0.99 $\pm$ 0.01 & 0.96 $\pm$ 0.01 & 0.99 $\pm$ 0.01 & 0.97 $\pm$ 0.01\\
\addlinespace[0.6em]
\multicolumn{9}{l}{\textbf{$N = 50,\; n_i = 100$}}\\
\hspace{1em}$\text{1-Step } {\mathbf{R}}_{\text{FPCA}} \otimes {\mathbf{R}}_{\text{Fun}}$ & 0.96 $\pm$ 0.01 & 0.97 $\pm$ 0.01 & 0.97 $\pm$ 0.01 & 0.94 $\pm$ 0.01 & 0.95 $\pm$ 0.01 & 0.94 $\pm$ 0.01 & 0.95 $\pm$ 0.01 & 0.94 $\pm$ 0.01\\
\hspace{1em}$\text{1-Step } {\mathbf{R}}_{\text{Lon}} \otimes {\mathbf{R}}_{\text{Fun}}$ & 0.96 $\pm$ 0.01 & 0.98 $\pm$ 0.01 & 0.98 $\pm$ 0.01 & 0.95 $\pm$ 0.01 & 0.95 $\pm$ 0.01 & 0.95 $\pm$ 0.01 & 0.96 $\pm$ 0.01 & 0.95 $\pm$ 0.01\\
\hspace{1em}$\text{1-Step } I_{n_i} \otimes {\mathbf{R}}_{\text{Fun}}$ & 0.97 $\pm$ 0.01 & 0.95 $\pm$ 0.01 & 0.98 $\pm$ 0.01 & 0.95 $\pm$ 0.01 & 0.95 $\pm$ 0.01 & 0.94 $\pm$ 0.01 & 0.96 $\pm$ 0.01 & 0.94 $\pm$ 0.01\\
\hspace{1em}$\text{1-Step } {\mathbf{R}}_{\text{Lon}} \otimes I_L$ & 0.97 $\pm$ 0.01 & 0.98 $\pm$ 0.01 & 0.98 $\pm$ 0.01 & 0.96 $\pm$ 0.01 & 0.96 $\pm$ 0.01 & 0.95 $\pm$ 0.01 & 0.96 $\pm$ 0.01 & 0.95 $\pm$ 0.01\\
\hspace{1em}$\text{1-Step } I_{n_i} \otimes I_L$ & 0.97 $\pm$ 0.01 & 0.95 $\pm$ 0.01 & 0.98 $\pm$ 0.01 & 0.96 $\pm$ 0.01 & 0.95 $\pm$ 0.01 & 0.94 $\pm$ 0.01 & 0.96 $\pm$ 0.01 & 0.95 $\pm$ 0.01\\
\hspace{1em}$\text{Full-1Step } {\mathbf{R}}_{\text{FPCA}} \otimes {\mathbf{R}}_{\text{Fun}}$ & 0.96 $\pm$ 0.01 & 0.97 $\pm$ 0.01 & 0.75 $\pm$ 0.03 & 0.94 $\pm$ 0.01 & 0.95 $\pm$ 0.01 & 0.94 $\pm$ 0.01 & 0.95 $\pm$ 0.01 & 0.94 $\pm$ 0.01\\
\hspace{1em}$\text{Full-1Step } {\mathbf{R}}_{\text{Lon}} \otimes {\mathbf{R}}_{\text{Fun}}$ & 0.97 $\pm$ 0.01 & 0.97 $\pm$ 0.01 & 0.90 $\pm$ 0.02 & 0.95 $\pm$ 0.01 & 0.95 $\pm$ 0.01 & 0.95 $\pm$ 0.01 & 0.96 $\pm$ 0.01 & 0.95 $\pm$ 0.01\\
\hspace{1em}$\text{Full-1Step } I_{n_i} \otimes {\mathbf{R}}_{\text{Fun}}$ & 0.96 $\pm$ 0.01 & 0.95 $\pm$ 0.01 & 0.98 $\pm$ 0.01 & 0.95 $\pm$ 0.01 & 0.95 $\pm$ 0.01 & 0.94 $\pm$ 0.01 & 0.96 $\pm$ 0.01 & 0.94 $\pm$ 0.01\\
\hspace{1em}$\text{Full-1Step } {\mathbf{R}}_{\text{Lon}} \otimes I_L$ & 0.97 $\pm$ 0.01 & 0.98 $\pm$ 0.01 & 0.79 $\pm$ 0.02 & 0.96 $\pm$ 0.01 & 0.96 $\pm$ 0.01 & 0.95 $\pm$ 0.01 & 0.96 $\pm$ 0.01 & 0.95 $\pm$ 0.01\\
\hspace{1em}$\text{Full-1Step } I_{n_i} \otimes I_L$ & 0.97 $\pm$ 0.01 & 0.95 $\pm$ 0.01 & 0.98 $\pm$ 0.01 & 0.96 $\pm$ 0.01 & 0.95 $\pm$ 0.01 & 0.94 $\pm$ 0.01 & 0.96 $\pm$ 0.01 & 0.95 $\pm$ 0.01\\
\hspace{1em}$\text{pffr (Wild)}$ & 0.95 $\pm$ 0.01 & 0.93 $\pm$ 0.01 & 0.94 $\pm$ 0.01 & 0.93 $\pm$ 0.01 & 1.00 $\pm$ 0.00 & 0.93 $\pm$ 0.01 & 0.98 $\pm$ 0.01 & 0.94 $\pm$ 0.01\\
\addlinespace[0.6em]
\multicolumn{9}{l}{\textbf{$N = 500,\; n_i = 5$}}\\
\hspace{1em}$\text{1-Step } {\mathbf{R}}_{\text{FPCA}} \otimes {\mathbf{R}}_{\text{Fun}}$ & 0.95 $\pm$ 0.01 & 0.95 $\pm$ 0.01 & 0.95 $\pm$ 0.01 & 0.95 $\pm$ 0.01 & 0.95 $\pm$ 0.01 & 0.95 $\pm$ 0.01 & 0.95 $\pm$ 0.01 & 0.95 $\pm$ 0.01\\
\hspace{1em}$\text{1-Step } {\mathbf{R}}_{\text{Lon}} \otimes {\mathbf{R}}_{\text{Fun}}$ & 0.95 $\pm$ 0.01 & 0.95 $\pm$ 0.01 & 0.96 $\pm$ 0.01 & 0.95 $\pm$ 0.01 & 0.95 $\pm$ 0.01 & 0.95 $\pm$ 0.01 & 0.95 $\pm$ 0.01 & 0.95 $\pm$ 0.01\\
\hspace{1em}$\text{1-Step } I_{n_i} \otimes {\mathbf{R}}_{\text{Fun}}$ & 0.95 $\pm$ 0.01 & 0.95 $\pm$ 0.01 & 0.96 $\pm$ 0.01 & 0.95 $\pm$ 0.01 & 0.95 $\pm$ 0.01 & 0.95 $\pm$ 0.01 & 0.95 $\pm$ 0.01 & 0.95 $\pm$ 0.01\\
\hspace{1em}$\text{1-Step } {\mathbf{R}}_{\text{Lon}} \otimes I_L$ & 0.96 $\pm$ 0.01 & 0.95 $\pm$ 0.01 & 0.96 $\pm$ 0.01 & 0.95 $\pm$ 0.01 & 0.96 $\pm$ 0.01 & 0.95 $\pm$ 0.01 & 0.96 $\pm$ 0.01 & 0.95 $\pm$ 0.01\\
\hspace{1em}$\text{1-Step } I_{n_i} \otimes I_L$ & 0.95 $\pm$ 0.01 & 0.95 $\pm$ 0.01 & 0.96 $\pm$ 0.01 & 0.95 $\pm$ 0.01 & 0.96 $\pm$ 0.01 & 0.95 $\pm$ 0.01 & 0.96 $\pm$ 0.01 & 0.95 $\pm$ 0.01\\
\hspace{1em}$\text{Full-1Step } {\mathbf{R}}_{\text{FPCA}} \otimes {\mathbf{R}}_{\text{Fun}}$ & 0.95 $\pm$ 0.01 & 0.95 $\pm$ 0.01 & 0.95 $\pm$ 0.01 & 0.95 $\pm$ 0.01 & 0.95 $\pm$ 0.01 & 0.95 $\pm$ 0.01 & 0.95 $\pm$ 0.01 & 0.95 $\pm$ 0.01\\
\hspace{1em}$\text{Full-1Step } {\mathbf{R}}_{\text{Lon}} \otimes {\mathbf{R}}_{\text{Fun}}$ & 0.95 $\pm$ 0.01 & 0.95 $\pm$ 0.01 & 0.95 $\pm$ 0.01 & 0.95 $\pm$ 0.01 & 0.95 $\pm$ 0.01 & 0.95 $\pm$ 0.01 & 0.95 $\pm$ 0.01 & 0.95 $\pm$ 0.01\\
\hspace{1em}$\text{Full-1Step } I_{n_i} \otimes {\mathbf{R}}_{\text{Fun}}$ & 0.95 $\pm$ 0.01 & 0.95 $\pm$ 0.01 & 0.95 $\pm$ 0.01 & 0.95 $\pm$ 0.01 & 0.95 $\pm$ 0.01 & 0.95 $\pm$ 0.01 & 0.95 $\pm$ 0.01 & 0.95 $\pm$ 0.01\\
\hspace{1em}$\text{Full-1Step } {\mathbf{R}}_{\text{Lon}} \otimes I_L$ & 0.96 $\pm$ 0.01 & 0.95 $\pm$ 0.01 & 0.96 $\pm$ 0.01 & 0.95 $\pm$ 0.01 & 0.96 $\pm$ 0.01 & 0.95 $\pm$ 0.01 & 0.96 $\pm$ 0.01 & 0.95 $\pm$ 0.01\\
\hspace{1em}$\text{Full-1Step } I_{n_i} \otimes I_L$ & 0.95 $\pm$ 0.01 & 0.95 $\pm$ 0.01 & 0.96 $\pm$ 0.01 & 0.95 $\pm$ 0.01 & 0.96 $\pm$ 0.01 & 0.95 $\pm$ 0.01 & 0.96 $\pm$ 0.01 & 0.95 $\pm$ 0.01\\
\hspace{1em}$\text{pffr (Wild)}$ & 1.00 $\pm$ 0.00 & 1.00 $\pm$ 0.00 & 1.00 $\pm$ 0.00 & 1.00 $\pm$ 0.00 & 1.00 $\pm$ 0.00 & 1.00 $\pm$ 0.00 & 1.00 $\pm$ 0.00 & 1.00 $\pm$ 0.00\\
\addlinespace[0.6em]
\multicolumn{9}{l}{\textbf{$N = 500,\; n_i = 100$}}\\
\hspace{1em}$\text{1-Step } {\mathbf{R}}_{\text{FPCA}} \otimes {\mathbf{R}}_{\text{Fun}}$ & 0.95 $\pm$ 0.01 & 0.95 $\pm$ 0.01 & 0.96 $\pm$ 0.01 & 0.93 $\pm$ 0.01 & 0.95 $\pm$ 0.01 & 0.94 $\pm$ 0.01 & 0.95 $\pm$ 0.01 & 0.95 $\pm$ 0.01\\
\hspace{1em}$\text{1-Step } {\mathbf{R}}_{\text{Lon}} \otimes {\mathbf{R}}_{\text{Fun}}$ & 0.95 $\pm$ 0.01 & 0.95 $\pm$ 0.01 & 0.96 $\pm$ 0.01 & 0.94 $\pm$ 0.01 & 0.95 $\pm$ 0.01 & 0.94 $\pm$ 0.01 & 0.95 $\pm$ 0.01 & 0.95 $\pm$ 0.01\\
\hspace{1em}$\text{1-Step } I_{n_i} \otimes {\mathbf{R}}_{\text{Fun}}$ & 0.95 $\pm$ 0.01 & 0.94 $\pm$ 0.01 & 0.96 $\pm$ 0.01 & 0.95 $\pm$ 0.01 & 0.95 $\pm$ 0.01 & 0.94 $\pm$ 0.01 & 0.95 $\pm$ 0.01 & 0.95 $\pm$ 0.01\\
\hspace{1em}$\text{1-Step } {\mathbf{R}}_{\text{Lon}} \otimes I_L$ & 0.96 $\pm$ 0.01 & 0.95 $\pm$ 0.01 & 0.96 $\pm$ 0.01 & 0.94 $\pm$ 0.01 & 0.95 $\pm$ 0.01 & 0.94 $\pm$ 0.01 & 0.95 $\pm$ 0.01 & 0.95 $\pm$ 0.01\\
\hspace{1em}$\text{1-Step } I_{n_i} \otimes I_L$ & 0.96 $\pm$ 0.01 & 0.94 $\pm$ 0.01 & 0.96 $\pm$ 0.01 & 0.95 $\pm$ 0.01 & 0.95 $\pm$ 0.01 & 0.94 $\pm$ 0.01 & 0.96 $\pm$ 0.01 & 0.95 $\pm$ 0.01\\
\hspace{1em}$\text{Full-1Step } {\mathbf{R}}_{\text{FPCA}} \otimes {\mathbf{R}}_{\text{Fun}}$ & 0.95 $\pm$ 0.01 & 0.95 $\pm$ 0.01 & 0.96 $\pm$ 0.01 & 0.93 $\pm$ 0.01 & 0.95 $\pm$ 0.01 & 0.94 $\pm$ 0.01 & 0.95 $\pm$ 0.01 & 0.95 $\pm$ 0.01\\
\hspace{1em}$\text{Full-1Step } {\mathbf{R}}_{\text{Lon}} \otimes {\mathbf{R}}_{\text{Fun}}$ & 0.95 $\pm$ 0.01 & 0.95 $\pm$ 0.01 & 0.96 $\pm$ 0.01 & 0.94 $\pm$ 0.01 & 0.95 $\pm$ 0.01 & 0.94 $\pm$ 0.01 & 0.95 $\pm$ 0.01 & 0.95 $\pm$ 0.01\\
\hspace{1em}$\text{Full-1Step } I_{n_i} \otimes {\mathbf{R}}_{\text{Fun}}$ & 0.95 $\pm$ 0.01 & 0.94 $\pm$ 0.01 & 0.96 $\pm$ 0.01 & 0.95 $\pm$ 0.01 & 0.95 $\pm$ 0.01 & 0.94 $\pm$ 0.01 & 0.95 $\pm$ 0.01 & 0.95 $\pm$ 0.01\\
\hspace{1em}$\text{Full-1Step } {\mathbf{R}}_{\text{Lon}} \otimes I_L$ & 0.96 $\pm$ 0.01 & 0.95 $\pm$ 0.01 & 0.96 $\pm$ 0.01 & 0.94 $\pm$ 0.01 & 0.95 $\pm$ 0.01 & 0.94 $\pm$ 0.01 & 0.95 $\pm$ 0.01 & 0.95 $\pm$ 0.01\\
\hspace{1em}$\text{Full-1Step } I_{n_i} \otimes I_L$ & 0.96 $\pm$ 0.01 & 0.94 $\pm$ 0.01 & 0.96 $\pm$ 0.01 & 0.95 $\pm$ 0.01 & 0.95 $\pm$ 0.01 & 0.94 $\pm$ 0.01 & 0.95 $\pm$ 0.01 & 0.95 $\pm$ 0.01\\
\hspace{1em}$\text{pffr (Wild)}$ & 0.99 $\pm$ 0.00 & 0.94 $\pm$ 0.01 & 0.97 $\pm$ 0.01 & 0.96 $\pm$ 0.01 & 1.00 $\pm$ 0.00 & 0.97 $\pm$ 0.01 & 1.00 $\pm$ 0.00 & 0.98 $\pm$ 0.01\\
\bottomrule
\end{tabular}}
\end{table}

\begin{table}[!h]
\centering
\caption{\label{tab:pt_width_supp} \footnotesize Relative pointwise CI width (mean $\pm$ SE) vs.\ pffr (Wild). We denote $UB^{(r)}(s)$ and $LB^{(r)}(s)$ and $UB_\text{pffr}^{(r)}(s)/LB_\text{pffr}^{(r)}(s)$ as the upper/lower bounds of the CIs (at $s$) of the indicated method and pffr, respectively. Below we report the average ratio $\frac{1}{(q+1)|\mathcal{S}|}\sum_{r=0}^q\sum_{s \in \mathcal{S}}\frac{UB^{(r)}(s) -LB(s)^{(r)}}{UB^{(r)}_\text{pffr}(s) -LB^{(r)}_\text{pffr}(s)}$ across 300 simulation replicates. Values $<1$ indicate narrower 95\% CIs. Values with $*$ indicate extreme outliers (from poor estimates) were removed from the average of that cell to avoid skewing results.}
\centering
\resizebox{\ifdim\width>\linewidth\linewidth\else\width\fi}{!}{
\fontsize{9}{11}\selectfont
\begin{tabular}[t]{>{\raggedright\arraybackslash}p{4.8cm}>{\raggedright\arraybackslash}p{2.1cm}>{\raggedright\arraybackslash}p{2.1cm}>{\raggedright\arraybackslash}p{2.1cm}>{\raggedright\arraybackslash}p{2.1cm}>{\raggedright\arraybackslash}p{2.1cm}>{\raggedright\arraybackslash}p{2.1cm}>{\raggedright\arraybackslash}p{2.1cm}>{\raggedright\arraybackslash}p{2.1cm}}
\toprule
\multicolumn{1}{c}{ } & \multicolumn{4}{c}{Exchangeable} & \multicolumn{4}{c}{AR(1)} \\
\cmidrule(l{3pt}r{3pt}){2-5} \cmidrule(l{3pt}r{3pt}){6-9}
Method & Gaussian & Poisson & Binomial & Gamma & Gaussian & Poisson & Binomial & Gamma\\
\midrule
\addlinespace[0em]
\multicolumn{9}{l}{\textbf{$N = 25,\; n_i = 5$}}\\
\hspace{1em}$\text{1-Step } {\mathbf{R}}_{\text{FPCA}} \otimes {\mathbf{R}}_{\text{Fun}}$ & 0.70 $\pm$ 0.01 & 0.79 $\pm$ 0.00 & 0.83 $\pm$ 0.01 & 0.79 $\pm$ 0.00 & 0.70 $\pm$ 0.01 & 0.80 $\pm$ 0.00 & 0.80 $\pm$ 0.01 & 0.81 $\pm$ 0.00\\
\hspace{1em}$\text{1-Step } {\mathbf{R}}_{\text{Lon}} \otimes {\mathbf{R}}_{\text{Fun}}$ & 0.75 $\pm$ 0.01 & 0.83 $\pm$ 0.00 & 0.85 $\pm$ 0.01 & 0.84 $\pm$ 0.00 & 0.75 $\pm$ 0.01 & 0.83 $\pm$ 0.00 & 0.82 $\pm$ 0.01 & 0.84 $\pm$ 0.00\\
\hspace{1em}$\text{1-Step } I_{n_i} \otimes {\mathbf{R}}_{\text{Fun}}$ & 0.83 $\pm$ 0.01 & 0.92 $\pm$ 0.00 & 0.94 $\pm$ 0.01 & 0.93 $\pm$ 0.00 & 0.80 $\pm$ 0.01 & 0.91 $\pm$ 0.00 & 0.87 $\pm$ 0.01 & 0.91 $\pm$ 0.00\\
\hspace{1em}$\text{1-Step } {\mathbf{R}}_{\text{Lon}} \otimes I_L$ & 0.78 $\pm$ 0.01 & 0.87 $\pm$ 0.01 & 0.84 $\pm$ 0.01 & 0.83 $\pm$ 0.00 & 0.77 $\pm$ 0.01 & 0.83 $\pm$ 0.00 & 0.80 $\pm$ 0.01 & 0.82 $\pm$ 0.00\\
\hspace{1em}$\text{1-Step } I_{n_i} \otimes I_L$ & 0.86 $\pm$ 0.01 & 0.92 $\pm$ 0.00 & 0.92 $\pm$ 0.01 & 0.92 $\pm$ 0.00 & 0.83 $\pm$ 0.01 & 0.90 $\pm$ 0.00 & 0.86 $\pm$ 0.01 & 0.90 $\pm$ 0.00\\
\hspace{1em}$\text{Full-1Step } {\mathbf{R}}_{\text{FPCA}} \otimes {\mathbf{R}}_{\text{Fun}}$ & 0.70 $\pm$ 0.01 & 0.80 $\pm$ 0.00 & 0.83 $\pm$ 0.01 & 0.79 $\pm$ 0.00 & 0.70 $\pm$ 0.01 & 0.81 $\pm$ 0.00 & 0.80 $\pm$ 0.01 & 0.82 $\pm$ 0.01\\
\hspace{1em}$\text{Full-1Step } {\mathbf{R}}_{\text{Lon}} \otimes {\mathbf{R}}_{\text{Fun}}$ & 0.75 $\pm$ 0.01 & 0.83 $\pm$ 0.00 & 0.90 $\pm$ 0.04 & 0.84 $\pm$ 0.00 & 0.75 $\pm$ 0.01 & 0.83 $\pm$ 0.00 & 0.82 $\pm$ 0.01 & 0.84 $\pm$ 0.00\\
\hspace{1em}$\text{Full-1Step } I_{n_i} \otimes {\mathbf{R}}_{\text{Fun}}$ & 0.83 $\pm$ 0.01 & 0.92 $\pm$ 0.00 & 0.94 $\pm$ 0.01 & 0.93 $\pm$ 0.00 & 0.80 $\pm$ 0.01 & 0.91 $\pm$ 0.00 & 0.87 $\pm$ 0.01 & 0.91 $\pm$ 0.00\\
\hspace{1em}$\text{Full-1Step } {\mathbf{R}}_{\text{Lon}} \otimes I_L$ & 0.78 $\pm$ 0.01 & 0.90 $\pm$ 0.02 & 1.00 $\pm$ 0.05 & 0.83 $\pm$ 0.00 & 0.77 $\pm$ 0.01 & 0.83 $\pm$ 0.00 & 0.81 $\pm$ 0.01 & 0.82 $\pm$ 0.00\\
\hspace{1em}$\text{Full-1Step } I_{n_i} \otimes I_L$ & 0.86 $\pm$ 0.01 & 0.92 $\pm$ 0.00 & 0.93 $\pm$ 0.01 & 0.92 $\pm$ 0.00 & 0.83 $\pm$ 0.01 & 0.90 $\pm$ 0.00 & 0.86 $\pm$ 0.01 & 0.89 $\pm$ 0.00\\
\hspace{1em}$\text{pffr (Wild)}$ & 1.00 $\pm$ 0.00 & 1.00 $\pm$ 0.00 & 1.00 $\pm$ 0.00 & 1.00 $\pm$ 0.00 & 1.00 $\pm$ 0.00 & 1.00 $\pm$ 0.00 & 1.00 $\pm$ 0.00 & 1.00 $\pm$ \vphantom{5} 0.00\\
\addlinespace[0.6em]
\multicolumn{9}{l}{\textbf{$N = 25,\; n_i = 100$}}\\
\hspace{1em}$\text{1-Step } {\mathbf{R}}_{\text{FPCA}} \otimes {\mathbf{R}}_{\text{Fun}}$ & 0.73 $\pm$ 0.01 & 0.76 $\pm$ 0.01 & 1.27 $\pm$ 0.02 & 0.78 $\pm$ 0.00 & 0.51 $\pm$ 0.00 & 0.89 $\pm$ 0.00 & 0.84 $\pm$ 0.00 & 0.90 $\pm$ 0.00\\
\hspace{1em}$\text{1-Step } {\mathbf{R}}_{\text{Lon}} \otimes {\mathbf{R}}_{\text{Fun}}$ & 0.79 $\pm$ 0.01 & 0.78 $\pm$ 0.01 & 1.18 $\pm$ 0.01 & 0.82 $\pm$ 0.00 & 0.51 $\pm$ 0.00 & 0.87 $\pm$ 0.00 & 0.82 $\pm$ 0.00 & 0.89 $\pm$ 0.00\\
\hspace{1em}$\text{1-Step } I_{n_i} \otimes {\mathbf{R}}_{\text{Fun}}$ & 1.00 $\pm$ 0.01 & 1.06 $\pm$ 0.00 & 1.39 $\pm$ 0.01 & 1.09 $\pm$ 0.00 & 0.57 $\pm$ 0.00 & 1.01 $\pm$ 0.00 & 0.90 $\pm$ 0.00 & 1.03 $\pm$ 0.00\\
\hspace{1em}$\text{1-Step } {\mathbf{R}}_{\text{Lon}} \otimes I_L$ & 0.83 $\pm$ 0.01 & 0.82 $\pm$ 0.01 & 1.18 $\pm$ 0.01 & 0.81 $\pm$ 0.00 & 0.50 $\pm$ 0.00 & 0.87 $\pm$ 0.00 & 0.81 $\pm$ 0.00 & 0.87 $\pm$ 0.00\\
\hspace{1em}$\text{1-Step } I_{n_i} \otimes I_L$ & 1.02 $\pm$ 0.01 & 1.05 $\pm$ 0.00 & 1.35 $\pm$ 0.01 & 1.08 $\pm$ 0.00 & 0.56 $\pm$ 0.00 & 1.00 $\pm$ 0.00 & 0.89 $\pm$ 0.00 & 1.01 $\pm$ 0.00\\
\hspace{1em}$\text{Full-1Step } {\mathbf{R}}_{\text{FPCA}} \otimes {\mathbf{R}}_{\text{Fun}}$ & 0.73 $\pm$ 0.01 & 0.76 $\pm$ 0.01 & 1.58 $\pm$ 0.06 & 0.78 $\pm$ 0.00 & 0.51 $\pm$ 0.00 & 0.89 $\pm$ 0.00 & 0.84 $\pm$ 0.00 & 0.90 $\pm$ 0.00\\
\hspace{1em}$\text{Full-1Step } {\mathbf{R}}_{\text{Lon}} \otimes {\mathbf{R}}_{\text{Fun}}$ & 0.79 $\pm$ 0.01 & 0.78 $\pm$ 0.01 & 2.23 $\pm$ 0.21 & 0.82 $\pm$ 0.00 & 0.51 $\pm$ 0.00 & 0.87 $\pm$ 0.00 & 0.82 $\pm$ 0.00 & 0.89 $\pm$ 0.00\\
\hspace{1em}$\text{Full-1Step } I_{n_i} \otimes {\mathbf{R}}_{\text{Fun}}$ & 1.00 $\pm$ 0.01 & 1.06 $\pm$ 0.00 & 1.39 $\pm$ 0.01 & 1.09 $\pm$ 0.00 & 0.57 $\pm$ 0.00 & 1.01 $\pm$ 0.00 & 0.90 $\pm$ 0.00 & 1.03 $\pm$ 0.00\\
\hspace{1em}$\text{Full-1Step } {\mathbf{R}}_{\text{Lon}} \otimes I_L$ & 0.83 $\pm$ 0.01 & 0.96 $\pm$ 0.11 & 5.29 $\pm$ 0.58$^*$ & 0.81 $\pm$ 0.00 & 0.50 $\pm$ 0.00 & 0.87 $\pm$ 0.00 & 0.81 $\pm$ 0.00 & 0.87 $\pm$ 0.00\\
\hspace{1em}$\text{Full-1Step } I_{n_i} \otimes I_L$ & 1.02 $\pm$ 0.01 & 1.06 $\pm$ 0.00 & 1.37 $\pm$ 0.01 & 1.08 $\pm$ 0.00 & 0.56 $\pm$ 0.00 & 1.00 $\pm$ 0.00 & 0.89 $\pm$ 0.00 & 1.01 $\pm$ 0.00\\
\hspace{1em}$\text{pffr (Wild)}$ & 1.00 $\pm$ 0.00 & 1.00 $\pm$ 0.00 & 1.00 $\pm$ 0.00 & 1.00 $\pm$ 0.00 & 1.00 $\pm$ 0.00 & 1.00 $\pm$ 0.00 & 1.00 $\pm$ 0.00 & 1.00 $\pm$ \vphantom{4} 0.00\\
\addlinespace[0.6em]
\multicolumn{9}{l}{\textbf{$N = 50,\; n_i = 5$}}\\
\hspace{1em}$\text{1-Step } {\mathbf{R}}_{\text{FPCA}} \otimes {\mathbf{R}}_{\text{Fun}}$ & 0.60 $\pm$ 0.00 & 0.71 $\pm$ 0.00 & 0.65 $\pm$ 0.00 & 0.71 $\pm$ 0.00 & 0.58 $\pm$ 0.00 & 0.71 $\pm$ 0.00 & 0.62 $\pm$ 0.00 & 0.70 $\pm$ 0.00\\
\hspace{1em}$\text{1-Step } {\mathbf{R}}_{\text{Lon}} \otimes {\mathbf{R}}_{\text{Fun}}$ & 0.62 $\pm$ 0.00 & 0.72 $\pm$ 0.00 & 0.65 $\pm$ 0.00 & 0.73 $\pm$ 0.00 & 0.60 $\pm$ 0.00 & 0.72 $\pm$ 0.00 & 0.62 $\pm$ 0.00 & 0.71 $\pm$ 0.00\\
\hspace{1em}$\text{1-Step } I_{n_i} \otimes {\mathbf{R}}_{\text{Fun}}$ & 0.70 $\pm$ 0.00 & 0.83 $\pm$ 0.00 & 0.72 $\pm$ 0.00 & 0.81 $\pm$ 0.00 & 0.66 $\pm$ 0.00 & 0.80 $\pm$ 0.00 & 0.66 $\pm$ 0.00 & 0.78 $\pm$ 0.00\\
\hspace{1em}$\text{1-Step } {\mathbf{R}}_{\text{Lon}} \otimes I_L$ & 0.63 $\pm$ 0.01 & 0.74 $\pm$ 0.00 & 0.64 $\pm$ 0.00 & 0.72 $\pm$ 0.00 & 0.61 $\pm$ 0.01 & 0.72 $\pm$ 0.00 & 0.61 $\pm$ 0.00 & 0.70 $\pm$ 0.00\\
\hspace{1em}$\text{1-Step } I_{n_i} \otimes I_L$ & 0.71 $\pm$ 0.01 & 0.83 $\pm$ 0.00 & 0.71 $\pm$ 0.00 & 0.81 $\pm$ 0.00 & 0.68 $\pm$ 0.01 & 0.80 $\pm$ 0.00 & 0.65 $\pm$ 0.00 & 0.77 $\pm$ 0.00\\
\hspace{1em}$\text{Full-1Step } {\mathbf{R}}_{\text{FPCA}} \otimes {\mathbf{R}}_{\text{Fun}}$ & 0.60 $\pm$ 0.00 & 0.71 $\pm$ 0.00 & 0.65 $\pm$ 0.00 & 0.71 $\pm$ 0.00 & 0.58 $\pm$ 0.00 & 0.71 $\pm$ 0.00 & 0.62 $\pm$ 0.00 & 0.70 $\pm$ 0.00\\
\hspace{1em}$\text{Full-1Step } {\mathbf{R}}_{\text{Lon}} \otimes {\mathbf{R}}_{\text{Fun}}$ & 0.62 $\pm$ 0.00 & 0.72 $\pm$ 0.00 & 0.65 $\pm$ 0.00 & 0.73 $\pm$ 0.00 & 0.60 $\pm$ 0.00 & 0.72 $\pm$ 0.00 & 0.62 $\pm$ 0.00 & 0.71 $\pm$ 0.00\\
\hspace{1em}$\text{Full-1Step } I_{n_i} \otimes {\mathbf{R}}_{\text{Fun}}$ & 0.70 $\pm$ 0.00 & 0.83 $\pm$ 0.00 & 0.72 $\pm$ 0.00 & 0.81 $\pm$ 0.00 & 0.66 $\pm$ 0.00 & 0.80 $\pm$ 0.00 & 0.66 $\pm$ 0.00 & 0.78 $\pm$ 0.00\\
\hspace{1em}$\text{Full-1Step } {\mathbf{R}}_{\text{Lon}} \otimes I_L$ & 0.63 $\pm$ 0.01 & 0.76 $\pm$ 0.01 & 0.74 $\pm$ 0.07 & 0.72 $\pm$ 0.00 & 0.61 $\pm$ 0.01 & 0.72 $\pm$ 0.00 & 0.61 $\pm$ 0.00 & 0.70 $\pm$ 0.00\\
\hspace{1em}$\text{Full-1Step } I_{n_i} \otimes I_L$ & 0.71 $\pm$ 0.01 & 0.83 $\pm$ 0.00 & 0.71 $\pm$ 0.00 & 0.81 $\pm$ 0.00 & 0.68 $\pm$ 0.01 & 0.80 $\pm$ 0.00 & 0.66 $\pm$ 0.00 & 0.77 $\pm$ 0.00\\
\hspace{1em}$\text{pffr (Wild)}$ & 1.00 $\pm$ 0.00 & 1.00 $\pm$ 0.00 & 1.00 $\pm$ 0.00 & 1.00 $\pm$ 0.00 & 1.00 $\pm$ 0.00 & 1.00 $\pm$ 0.00 & 1.00 $\pm$ 0.00 & 1.00 $\pm$ \vphantom{3} 0.00\\
\addlinespace[0.6em]
\multicolumn{9}{l}{\textbf{$N = 50,\; n_i = 100$}}\\
\hspace{1em}$\text{1-Step } {\mathbf{R}}_{\text{FPCA}} \otimes {\mathbf{R}}_{\text{Fun}}$ & 0.71 $\pm$ 0.00 & 0.70 $\pm$ 0.01 & 1.16 $\pm$ 0.01 & 0.78 $\pm$ 0.00 & 0.37 $\pm$ 0.00 & 0.83 $\pm$ 0.00 & 0.70 $\pm$ 0.00 & 0.86 $\pm$ 0.00\\
\hspace{1em}$\text{1-Step } {\mathbf{R}}_{\text{Lon}} \otimes {\mathbf{R}}_{\text{Fun}}$ & 0.75 $\pm$ 0.00 & 0.70 $\pm$ 0.00 & 1.07 $\pm$ 0.01 & 0.80 $\pm$ 0.00 & 0.37 $\pm$ 0.00 & 0.83 $\pm$ 0.00 & 0.69 $\pm$ 0.00 & 0.86 $\pm$ 0.00\\
\hspace{1em}$\text{1-Step } I_{n_i} \otimes {\mathbf{R}}_{\text{Fun}}$ & 0.85 $\pm$ 0.00 & 1.02 $\pm$ 0.00 & 1.18 $\pm$ 0.01 & 1.05 $\pm$ 0.00 & 0.42 $\pm$ 0.00 & 0.99 $\pm$ 0.00 & 0.76 $\pm$ 0.00 & 1.00 $\pm$ 0.00\\
\hspace{1em}$\text{1-Step } {\mathbf{R}}_{\text{Lon}} \otimes I_L$ & 0.77 $\pm$ 0.00 & 0.73 $\pm$ 0.01 & 1.05 $\pm$ 0.01 & 0.79 $\pm$ 0.00 & 0.37 $\pm$ 0.00 & 0.83 $\pm$ 0.00 & 0.68 $\pm$ 0.00 & 0.85 $\pm$ 0.00\\
\hspace{1em}$\text{1-Step } I_{n_i} \otimes I_L$ & 0.86 $\pm$ 0.00 & 1.03 $\pm$ 0.00 & 1.16 $\pm$ 0.01 & 1.06 $\pm$ 0.00 & 0.42 $\pm$ 0.00 & 0.99 $\pm$ 0.00 & 0.75 $\pm$ 0.00 & 0.99 $\pm$ 0.00\\
\hspace{1em}$\text{Full-1Step } {\mathbf{R}}_{\text{FPCA}} \otimes {\mathbf{R}}_{\text{Fun}}$ & 0.71 $\pm$ 0.00 & 0.70 $\pm$ 0.01 & 1.34 $\pm$ 0.03 & 0.78 $\pm$ 0.00 & 0.37 $\pm$ 0.00 & 0.83 $\pm$ 0.00 & 0.70 $\pm$ 0.00 & 0.86 $\pm$ 0.00\\
\hspace{1em}$\text{Full-1Step } {\mathbf{R}}_{\text{Lon}} \otimes {\mathbf{R}}_{\text{Fun}}$ & 0.75 $\pm$ 0.00 & 0.70 $\pm$ 0.00 & 1.42 $\pm$ 0.10 & 0.80 $\pm$ 0.00 & 0.37 $\pm$ 0.00 & 0.83 $\pm$ 0.00 & 0.69 $\pm$ 0.00 & 0.86 $\pm$ 0.00\\
\hspace{1em}$\text{Full-1Step } I_{n_i} \otimes {\mathbf{R}}_{\text{Fun}}$ & 0.85 $\pm$ 0.00 & 1.02 $\pm$ 0.00 & 1.19 $\pm$ 0.01 & 1.05 $\pm$ 0.00 & 0.42 $\pm$ 0.00 & 0.99 $\pm$ 0.00 & 0.76 $\pm$ 0.00 & 1.00 $\pm$ 0.00\\
\hspace{1em}$\text{Full-1Step } {\mathbf{R}}_{\text{Lon}} \otimes I_L$ & 0.77 $\pm$ 0.00 & 0.74 $\pm$ 0.01 & 4.05 $\pm$ 0.55 & 0.79 $\pm$ 0.00 & 0.37 $\pm$ 0.00 & 0.83 $\pm$ 0.00 & 0.68 $\pm$ 0.00 & 0.85 $\pm$ 0.00\\
\hspace{1em}$\text{Full-1Step } I_{n_i} \otimes I_L$ & 0.86 $\pm$ 0.00 & 1.03 $\pm$ 0.00 & 1.16 $\pm$ 0.01 & 1.06 $\pm$ 0.00 & 0.42 $\pm$ 0.00 & 0.99 $\pm$ 0.00 & 0.75 $\pm$ 0.00 & 0.99 $\pm$ 0.00\\
\hspace{1em}$\text{pffr (Wild)}$ & 1.00 $\pm$ 0.00 & 1.00 $\pm$ 0.00 & 1.00 $\pm$ 0.00 & 1.00 $\pm$ 0.00 & 1.00 $\pm$ 0.00 & 1.00 $\pm$ 0.00 & 1.00 $\pm$ 0.00 & 1.00 $\pm$ \vphantom{2} 0.00\\
\addlinespace[0.6em]
\multicolumn{9}{l}{\textbf{$N = 500,\; n_i = 5$}}\\
\hspace{1em}$\text{1-Step } {\mathbf{R}}_{\text{FPCA}} \otimes {\mathbf{R}}_{\text{Fun}}$ & 0.25 $\pm$ 0.00 & 0.36 $\pm$ 0.00 & 0.23 $\pm$ 0.00 & 0.33 $\pm$ 0.00 & 0.24 $\pm$ 0.00 & 0.36 $\pm$ 0.00 & 0.22 $\pm$ 0.00 & 0.31 $\pm$ 0.00\\
\hspace{1em}$\text{1-Step } {\mathbf{R}}_{\text{Lon}} \otimes {\mathbf{R}}_{\text{Fun}}$ & 0.25 $\pm$ 0.00 & 0.36 $\pm$ 0.00 & 0.22 $\pm$ 0.00 & 0.33 $\pm$ 0.00 & 0.24 $\pm$ 0.00 & 0.35 $\pm$ 0.00 & 0.21 $\pm$ 0.00 & 0.31 $\pm$ 0.00\\
\hspace{1em}$\text{1-Step } I_{n_i} \otimes {\mathbf{R}}_{\text{Fun}}$ & 0.29 $\pm$ 0.00 & 0.42 $\pm$ 0.00 & 0.25 $\pm$ 0.00 & 0.36 $\pm$ 0.00 & 0.27 $\pm$ 0.00 & 0.40 $\pm$ 0.00 & 0.23 $\pm$ 0.00 & 0.34 $\pm$ 0.00\\
\hspace{1em}$\text{1-Step } {\mathbf{R}}_{\text{Lon}} \otimes I_L$ & 0.25 $\pm$ 0.00 & 0.36 $\pm$ 0.00 & 0.22 $\pm$ 0.00 & 0.32 $\pm$ 0.00 & 0.24 $\pm$ 0.00 & 0.36 $\pm$ 0.00 & 0.21 $\pm$ 0.00 & 0.30 $\pm$ 0.00\\
\hspace{1em}$\text{1-Step } I_{n_i} \otimes I_L$ & 0.28 $\pm$ 0.00 & 0.42 $\pm$ 0.00 & 0.25 $\pm$ 0.00 & 0.36 $\pm$ 0.00 & 0.26 $\pm$ 0.00 & 0.40 $\pm$ 0.00 & 0.23 $\pm$ 0.00 & 0.33 $\pm$ 0.00\\
\hspace{1em}$\text{Full-1Step } {\mathbf{R}}_{\text{FPCA}} \otimes {\mathbf{R}}_{\text{Fun}}$ & 0.25 $\pm$ 0.00 & 0.36 $\pm$ 0.00 & 0.23 $\pm$ 0.00 & 0.33 $\pm$ 0.00 & 0.24 $\pm$ 0.00 & 0.36 $\pm$ 0.00 & 0.22 $\pm$ 0.00 & 0.31 $\pm$ 0.00\\
\hspace{1em}$\text{Full-1Step } {\mathbf{R}}_{\text{Lon}} \otimes {\mathbf{R}}_{\text{Fun}}$ & 0.25 $\pm$ 0.00 & 0.36 $\pm$ 0.00 & 0.22 $\pm$ 0.00 & 0.33 $\pm$ 0.00 & 0.24 $\pm$ 0.00 & 0.35 $\pm$ 0.00 & 0.21 $\pm$ 0.00 & 0.31 $\pm$ 0.00\\
\hspace{1em}$\text{Full-1Step } I_{n_i} \otimes {\mathbf{R}}_{\text{Fun}}$ & 0.29 $\pm$ 0.00 & 0.42 $\pm$ 0.00 & 0.25 $\pm$ 0.00 & 0.36 $\pm$ 0.00 & 0.27 $\pm$ 0.00 & 0.40 $\pm$ 0.00 & 0.23 $\pm$ 0.00 & 0.34 $\pm$ 0.00\\
\hspace{1em}$\text{Full-1Step } {\mathbf{R}}_{\text{Lon}} \otimes I_L$ & 0.25 $\pm$ 0.00 & 0.36 $\pm$ 0.00 & 0.22 $\pm$ 0.00 & 0.32 $\pm$ 0.00 & 0.24 $\pm$ 0.00 & 0.36 $\pm$ 0.00 & 0.21 $\pm$ 0.00 & 0.30 $\pm$ 0.00\\
\hspace{1em}$\text{Full-1Step } I_{n_i} \otimes I_L$ & 0.28 $\pm$ 0.00 & 0.42 $\pm$ 0.00 & 0.25 $\pm$ 0.00 & 0.36 $\pm$ 0.00 & 0.26 $\pm$ 0.00 & 0.40 $\pm$ 0.00 & 0.23 $\pm$ 0.00 & 0.33 $\pm$ 0.00\\
\hspace{1em}$\text{pffr (Wild)}$ & 1.00 $\pm$ 0.00 & 1.00 $\pm$ 0.00 & 1.00 $\pm$ 0.00 & 1.00 $\pm$ 0.00 & 1.00 $\pm$ 0.00 & 1.00 $\pm$ 0.00 & 1.00 $\pm$ 0.00 & 1.00 $\pm$ \vphantom{1} 0.00\\
\addlinespace[0.6em]
\multicolumn{9}{l}{\textbf{$N = 500,\; n_i = 100$}}\\
\hspace{1em}$\text{1-Step } {\mathbf{R}}_{\text{FPCA}} \otimes {\mathbf{R}}_{\text{Fun}}$ & 0.43 $\pm$ 0.00 & 0.58 $\pm$ 0.00 & 0.82 $\pm$ 0.00 & 0.75 $\pm$ 0.00 & 0.12 $\pm$ 0.00 & 0.69 $\pm$ 0.00 & 0.27 $\pm$ 0.00 & 0.68 $\pm$ 0.00\\
\hspace{1em}$\text{1-Step } {\mathbf{R}}_{\text{Lon}} \otimes {\mathbf{R}}_{\text{Fun}}$ & 0.43 $\pm$ 0.00 & 0.57 $\pm$ 0.00 & 0.77 $\pm$ 0.00 & 0.75 $\pm$ 0.00 & 0.12 $\pm$ 0.00 & 0.69 $\pm$ 0.00 & 0.27 $\pm$ 0.00 & 0.68 $\pm$ 0.00\\
\hspace{1em}$\text{1-Step } I_{n_i} \otimes {\mathbf{R}}_{\text{Fun}}$ & 0.46 $\pm$ 0.00 & 0.96 $\pm$ 0.00 & 0.76 $\pm$ 0.00 & 0.92 $\pm$ 0.00 & 0.14 $\pm$ 0.00 & 0.84 $\pm$ 0.00 & 0.30 $\pm$ 0.00 & 0.78 $\pm$ 0.00\\
\hspace{1em}$\text{1-Step } {\mathbf{R}}_{\text{Lon}} \otimes I_L$ & 0.43 $\pm$ 0.00 & 0.58 $\pm$ 0.00 & 0.75 $\pm$ 0.00 & 0.74 $\pm$ 0.00 & 0.12 $\pm$ 0.00 & 0.69 $\pm$ 0.00 & 0.27 $\pm$ 0.00 & 0.68 $\pm$ 0.00\\
\hspace{1em}$\text{1-Step } I_{n_i} \otimes I_L$ & 0.46 $\pm$ 0.00 & 0.97 $\pm$ 0.00 & 0.75 $\pm$ 0.00 & 0.92 $\pm$ 0.00 & 0.13 $\pm$ 0.00 & 0.85 $\pm$ 0.00 & 0.30 $\pm$ 0.00 & 0.78 $\pm$ 0.00\\
\hspace{1em}$\text{Full-1Step } {\mathbf{R}}_{\text{FPCA}} \otimes {\mathbf{R}}_{\text{Fun}}$ & 0.43 $\pm$ 0.00 & 0.58 $\pm$ 0.00 & 0.82 $\pm$ 0.00 & 0.75 $\pm$ 0.00 & 0.12 $\pm$ 0.00 & 0.69 $\pm$ 0.00 & 0.27 $\pm$ 0.00 & 0.68 $\pm$ 0.00\\
\hspace{1em}$\text{Full-1Step } {\mathbf{R}}_{\text{Lon}} \otimes {\mathbf{R}}_{\text{Fun}}$ & 0.43 $\pm$ 0.00 & 0.57 $\pm$ 0.00 & 0.77 $\pm$ 0.00 & 0.75 $\pm$ 0.00 & 0.12 $\pm$ 0.00 & 0.69 $\pm$ 0.00 & 0.27 $\pm$ 0.00 & 0.68 $\pm$ 0.00\\
\hspace{1em}$\text{Full-1Step } I_{n_i} \otimes {\mathbf{R}}_{\text{Fun}}$ & 0.46 $\pm$ 0.00 & 0.96 $\pm$ 0.00 & 0.76 $\pm$ 0.00 & 0.92 $\pm$ 0.00 & 0.14 $\pm$ 0.00 & 0.84 $\pm$ 0.00 & 0.30 $\pm$ 0.00 & 0.78 $\pm$ 0.00\\
\hspace{1em}$\text{Full-1Step } {\mathbf{R}}_{\text{Lon}} \otimes I_L$ & 0.43 $\pm$ 0.00 & 0.58 $\pm$ 0.00 & 0.75 $\pm$ 0.00 & 0.74 $\pm$ 0.00 & 0.12 $\pm$ 0.00 & 0.69 $\pm$ 0.00 & 0.27 $\pm$ 0.00 & 0.68 $\pm$ 0.00\\
\hspace{1em}$\text{Full-1Step } I_{n_i} \otimes I_L$ & 0.46 $\pm$ 0.00 & 0.97 $\pm$ 0.00 & 0.75 $\pm$ 0.00 & 0.92 $\pm$ 0.00 & 0.13 $\pm$ 0.00 & 0.85 $\pm$ 0.00 & 0.30 $\pm$ 0.00 & 0.78 $\pm$ 0.00\\
\hspace{1em}$\text{pffr (Wild)}$ & 1.00 $\pm$ 0.00 & 1.00 $\pm$ 0.00 & 1.00 $\pm$ 0.00 & 1.00 $\pm$ 0.00 & 1.00 $\pm$ 0.00 & 1.00 $\pm$ 0.00 & 1.00 $\pm$ 0.00 & 1.00 $\pm$ 0.00\\
\bottomrule
\end{tabular}}
\end{table}

\begin{table}[!h]
\centering
\caption{\label{tab:time_supp} \footnotesize Computation time in seconds (mean $\pm$ SE) averaged across 300 simulation replicates (SE$=0.00$ indicates a value $<0.01$).}
\centering
\resizebox{\ifdim\width>\linewidth\linewidth\else\width\fi}{!}{
\fontsize{9}{11}\selectfont
\begin{tabular}[t]{>{\raggedright\arraybackslash}p{4.8cm}>{\raggedright\arraybackslash}p{2.1cm}>{\raggedright\arraybackslash}p{2.1cm}>{\raggedright\arraybackslash}p{2.1cm}>{\raggedright\arraybackslash}p{2.1cm}>{\raggedright\arraybackslash}p{2.1cm}>{\raggedright\arraybackslash}p{2.1cm}>{\raggedright\arraybackslash}p{2.1cm}>{\raggedright\arraybackslash}p{2.1cm}}
\toprule
\multicolumn{1}{c}{ } & \multicolumn{4}{c}{Exchangeable} & \multicolumn{4}{c}{AR(1)} \\
\cmidrule(l{3pt}r{3pt}){2-5} \cmidrule(l{3pt}r{3pt}){6-9}
Method & Gaussian & Poisson & Binomial & Gamma & Gaussian & Poisson & Binomial & Gamma\\
\midrule
\addlinespace[0em]
\multicolumn{9}{l}{\textbf{$N = 25,\; n_i = 5$}}\\
\hspace{1em}$\text{1-Step } {\mathbf{R}}_{\text{FPCA}} \otimes {\mathbf{R}}_{\text{Fun}}$ & 2.57 $\pm$ 0.06 & 2.86 $\pm$ 0.06 & 2.82 $\pm$ 0.05 & 2.84 $\pm$ 0.07 & 2.44 $\pm$ 0.01 & 2.67 $\pm$ 0.02 & 2.79 $\pm$ 0.03 & 3.09 $\pm$ 0.05\\
\hspace{1em}$\text{1-Step } {\mathbf{R}}_{\text{Lon}} \otimes {\mathbf{R}}_{\text{Fun}}$ & 2.22 $\pm$ 0.05 & 2.47 $\pm$ 0.05 & 2.49 $\pm$ 0.05 & 2.36 $\pm$ 0.05 & 2.12 $\pm$ 0.01 & 2.25 $\pm$ 0.01 & 2.40 $\pm$ 0.02 & 2.52 $\pm$ 0.04\\
\hspace{1em}$\text{1-Step } I_{n_i} \otimes {\mathbf{R}}_{\text{Fun}}$ & 1.16 $\pm$ 0.02 & 1.37 $\pm$ 0.03 & 1.38 $\pm$ 0.03 & 1.35 $\pm$ 0.03 & 1.08 $\pm$ 0.00 & 1.24 $\pm$ 0.01 & 1.33 $\pm$ 0.01 & 1.45 $\pm$ 0.02\\
\hspace{1em}$\text{1-Step } {\mathbf{R}}_{\text{Lon}} \otimes I_L$ & 4.57 $\pm$ 0.12 & 4.71 $\pm$ 0.08 & 4.75 $\pm$ 0.08 & 4.66 $\pm$ 0.12 & 4.32 $\pm$ 0.02 & 4.67 $\pm$ 0.04 & 5.01 $\pm$ 0.06 & 5.28 $\pm$ 0.10\\
\hspace{1em}$\text{1-Step } I_{n_i} \otimes I_L$ & 0.66 $\pm$ 0.01 & 0.83 $\pm$ 0.02 & 0.84 $\pm$ 0.02 & 0.84 $\pm$ 0.02 & 0.64 $\pm$ 0.00 & 0.74 $\pm$ 0.00 & 0.80 $\pm$ 0.01 & 0.86 $\pm$ 0.01\\
\hspace{1em}$\text{Full-1Step } {\mathbf{R}}_{\text{FPCA}} \otimes {\mathbf{R}}_{\text{Fun}}$ & 10.00 $\pm$ 0.35 & 10.83 $\pm$ 0.59 & 15.58 $\pm$ 0.95 & 17.40 $\pm$ 0.72 & 7.90 $\pm$ 0.10 & 9.32 $\pm$ 0.58 & 11.51 $\pm$ 0.66 & 21.49 $\pm$ 0.86\\
\hspace{1em}$\text{Full-1Step } {\mathbf{R}}_{\text{Lon}} \otimes {\mathbf{R}}_{\text{Fun}}$ & 7.74 $\pm$ 0.17 & 7.94 $\pm$ 0.18 & 12.56 $\pm$ 0.53 & 15.99 $\pm$ 0.59 & 6.96 $\pm$ 0.05 & 8.07 $\pm$ 0.28 & 11.95 $\pm$ 0.50 & 17.02 $\pm$ 0.45\\
\hspace{1em}$\text{Full-1Step } I_{n_i} \otimes {\mathbf{R}}_{\text{Fun}}$ & 3.37 $\pm$ 0.07 & 3.22 $\pm$ 0.07 & 4.60 $\pm$ 0.20 & 6.95 $\pm$ 0.25 & 3.08 $\pm$ 0.02 & 4.37 $\pm$ 0.37 & 6.17 $\pm$ 0.52 & 6.73 $\pm$ 0.22\\
\hspace{1em}$\text{Full-1Step } {\mathbf{R}}_{\text{Lon}} \otimes I_L$ & 12.98 $\pm$ 0.28 & 35.67 $\pm$ 2.85 & 32.97 $\pm$ 2.17 & 27.24 $\pm$ 1.05 & 11.99 $\pm$ 0.07 & 17.72 $\pm$ 0.72 & 24.82 $\pm$ 0.83 & 32.02 $\pm$ 1.09\\
\hspace{1em}$\text{Full-1Step } I_{n_i} \otimes I_L$ & 0.87 $\pm$ 0.02 & 1.08 $\pm$ 0.02 & 1.33 $\pm$ 0.04 & 1.43 $\pm$ 0.04 & 0.84 $\pm$ 0.00 & 4.08 $\pm$ 0.70 & 5.03 $\pm$ 0.90 & 1.46 $\pm$ 0.03\\
\hspace{1em}$\text{pffr (Wild)}$ & 2.59 $\pm$ 0.05 & 2.84 $\pm$ 0.06 & 2.78 $\pm$ 0.05 & 2.82 $\pm$ 0.06 & 2.42 $\pm$ 0.01 & 2.61 $\pm$ 0.02 & 2.63 $\pm$ 0.02 & 3.04 $\pm$ 0.05\\
\addlinespace[0.6em]
\multicolumn{9}{l}{\textbf{$N = 25,\; n_i = 100$}}\\
\hspace{1em}$\text{1-Step } {\mathbf{R}}_{\text{FPCA}} \otimes {\mathbf{R}}_{\text{Fun}}$ & 9.95 $\pm$ 0.02 & 11.63 $\pm$ 0.02 & 11.68 $\pm$ 0.11 & 11.70 $\pm$ 0.03 & 13.16 $\pm$ 0.25 & 11.64 $\pm$ 0.04 & 11.77 $\pm$ 0.07 & 11.46 $\pm$ 0.04\\
\hspace{1em}$\text{1-Step } {\mathbf{R}}_{\text{Lon}} \otimes {\mathbf{R}}_{\text{Fun}}$ & 6.94 $\pm$ 0.02 & 8.50 $\pm$ 0.02 & 8.65 $\pm$ 0.08 & 8.56 $\pm$ 0.03 & 8.96 $\pm$ 0.17 & 8.29 $\pm$ 0.04 & 8.51 $\pm$ 0.06 & 8.24 $\pm$ 0.03\\
\hspace{1em}$\text{1-Step } I_{n_i} \otimes {\mathbf{R}}_{\text{Fun}}$ & 10.71 $\pm$ 0.02 & 12.13 $\pm$ 0.02 & 12.50 $\pm$ 0.11 & 12.02 $\pm$ 0.03 & 14.09 $\pm$ 0.26 & 11.86 $\pm$ 0.04 & 12.27 $\pm$ 0.08 & 11.57 $\pm$ 0.04\\
\hspace{1em}$\text{1-Step } {\mathbf{R}}_{\text{Lon}} \otimes I_L$ & 7.43 $\pm$ 0.02 & 9.15 $\pm$ 0.02 & 9.38 $\pm$ 0.08 & 9.18 $\pm$ 0.03 & 9.91 $\pm$ 0.19 & 9.03 $\pm$ 0.04 & 9.27 $\pm$ 0.06 & 8.84 $\pm$ 0.03\\
\hspace{1em}$\text{1-Step } I_{n_i} \otimes I_L$ & 4.21 $\pm$ 0.02 & 5.40 $\pm$ 0.02 & 5.80 $\pm$ 0.05 & 5.40 $\pm$ 0.03 & 4.59 $\pm$ 0.08 & 5.15 $\pm$ 0.03 & 5.68 $\pm$ 0.04 & 5.08 $\pm$ 0.03\\
\hspace{1em}$\text{Full-1Step } {\mathbf{R}}_{\text{FPCA}} \otimes {\mathbf{R}}_{\text{Fun}}$ & 56.83 $\pm$ 2.01 & 45.68 $\pm$ 1.38 & 207.23 $\pm$ 7.14 & 96.14 $\pm$ 3.85 & 30.01 $\pm$ 1.10 & 26.63 $\pm$ 1.58 & 30.59 $\pm$ 1.90 & 32.52 $\pm$ 0.19\\
\hspace{1em}$\text{Full-1Step } {\mathbf{R}}_{\text{Lon}} \otimes {\mathbf{R}}_{\text{Fun}}$ & 20.78 $\pm$ 0.15 & 22.20 $\pm$ 0.20 & 71.49 $\pm$ 3.56 & 36.57 $\pm$ 0.96 & 17.17 $\pm$ 0.35 & 15.42 $\pm$ 0.08 & 18.26 $\pm$ 0.15 & 19.11 $\pm$ 0.10\\
\hspace{1em}$\text{Full-1Step } I_{n_i} \otimes {\mathbf{R}}_{\text{Fun}}$ & 37.38 $\pm$ 0.28 & 32.62 $\pm$ 0.21 & 46.32 $\pm$ 1.04 & 67.82 $\pm$ 2.22 & 29.05 $\pm$ 0.56 & 26.49 $\pm$ 1.37 & 30.41 $\pm$ 0.25 & 32.05 $\pm$ 0.93\\
\hspace{1em}$\text{Full-1Step } {\mathbf{R}}_{\text{Lon}} \otimes I_L$ & 18.22 $\pm$ 0.12 & 45.69 $\pm$ 3.03 & 117.68 $\pm$ 5.27 & 33.48 $\pm$ 0.94 & 19.08 $\pm$ 0.72 & 17.58 $\pm$ 0.62 & 20.11 $\pm$ 0.15 & 20.39 $\pm$ 0.10\\
\hspace{1em}$\text{Full-1Step } I_{n_i} \otimes I_L$ & 4.63 $\pm$ 0.03 & 6.02 $\pm$ 0.03 & 7.94 $\pm$ 0.14 & 8.11 $\pm$ 0.11 & 4.84 $\pm$ 0.09 & 5.53 $\pm$ 0.03 & 6.79 $\pm$ 0.05 & 5.95 $\pm$ 0.03\\
\hspace{1em}$\text{pffr (Wild)}$ & 6.69 $\pm$ 0.01 & 7.11 $\pm$ 0.02 & 6.90 $\pm$ 0.06 & 7.56 $\pm$ 0.03 & 8.24 $\pm$ 0.14 & 6.95 $\pm$ 0.03 & 6.74 $\pm$ 0.04 & 7.37 $\pm$ 0.04\\
\addlinespace[0.6em]
\multicolumn{9}{l}{\textbf{$N = 50,\; n_i = 5$}}\\
\hspace{1em}$\text{1-Step } {\mathbf{R}}_{\text{FPCA}} \otimes {\mathbf{R}}_{\text{Fun}}$ & 3.83 $\pm$ 0.07 & 3.75 $\pm$ 0.03 & 3.98 $\pm$ 0.01 & 3.88 $\pm$ 0.09 & 3.61 $\pm$ 0.01 & 3.92 $\pm$ 0.03 & 4.13 $\pm$ 0.02 & 3.99 $\pm$ 0.03\\
\hspace{1em}$\text{1-Step } {\mathbf{R}}_{\text{Lon}} \otimes {\mathbf{R}}_{\text{Fun}}$ & 3.57 $\pm$ 0.06 & 3.73 $\pm$ 0.03 & 3.50 $\pm$ 0.01 & 3.76 $\pm$ 0.07 & 3.40 $\pm$ 0.01 & 3.83 $\pm$ 0.02 & 3.83 $\pm$ 0.02 & 3.93 $\pm$ 0.02\\
\hspace{1em}$\text{1-Step } I_{n_i} \otimes {\mathbf{R}}_{\text{Fun}}$ & 1.80 $\pm$ 0.03 & 1.97 $\pm$ 0.02 & 2.03 $\pm$ 0.01 & 2.20 $\pm$ 0.05 & 1.71 $\pm$ 0.01 & 1.93 $\pm$ 0.02 & 1.89 $\pm$ 0.01 & 2.14 $\pm$ 0.02\\
\hspace{1em}$\text{1-Step } {\mathbf{R}}_{\text{Lon}} \otimes I_L$ & 7.71 $\pm$ 0.12 & 7.54 $\pm$ 0.10 & 7.55 $\pm$ 0.02 & 7.88 $\pm$ 0.27 & 7.51 $\pm$ 0.03 & 7.98 $\pm$ 0.06 & 7.87 $\pm$ 0.04 & 8.09 $\pm$ 0.06\\
\hspace{1em}$\text{1-Step } I_{n_i} \otimes I_L$ & 0.84 $\pm$ 0.01 & 0.94 $\pm$ 0.01 & 1.00 $\pm$ 0.00 & 1.00 $\pm$ 0.01 & 0.80 $\pm$ 0.00 & 0.96 $\pm$ 0.00 & 1.20 $\pm$ 0.01 & 1.01 $\pm$ 0.00\\
\hspace{1em}$\text{Full-1Step } {\mathbf{R}}_{\text{FPCA}} \otimes {\mathbf{R}}_{\text{Fun}}$ & 11.42 $\pm$ 0.21 & 11.95 $\pm$ 0.70 & 14.44 $\pm$ 0.53 & 19.91 $\pm$ 0.83 & 10.92 $\pm$ 0.53 & 11.43 $\pm$ 0.67 & 14.28 $\pm$ 0.71 & 21.82 $\pm$ 0.87\\
\hspace{1em}$\text{Full-1Step } {\mathbf{R}}_{\text{Lon}} \otimes {\mathbf{R}}_{\text{Fun}}$ & 11.94 $\pm$ 0.22 & 11.80 $\pm$ 0.16 & 16.00 $\pm$ 0.28 & 21.21 $\pm$ 0.68 & 10.61 $\pm$ 0.06 & 11.65 $\pm$ 0.09 & 16.18 $\pm$ 0.30 & 21.78 $\pm$ 0.34\\
\hspace{1em}$\text{Full-1Step } I_{n_i} \otimes {\mathbf{R}}_{\text{Fun}}$ & 5.23 $\pm$ 0.09 & 4.68 $\pm$ 0.06 & 7.16 $\pm$ 0.09 & 9.72 $\pm$ 0.31 & 4.53 $\pm$ 0.03 & 4.63 $\pm$ 0.03 & 6.99 $\pm$ 0.10 & 9.06 $\pm$ 0.16\\
\hspace{1em}$\text{Full-1Step } {\mathbf{R}}_{\text{Lon}} \otimes I_L$ & 20.36 $\pm$ 0.31 & 43.17 $\pm$ 3.27 & 37.27 $\pm$ 1.85 & 35.95 $\pm$ 0.90 & 19.67 $\pm$ 0.13 & 24.21 $\pm$ 0.86 & 31.17 $\pm$ 0.28 & 37.60 $\pm$ 0.39\\
\hspace{1em}$\text{Full-1Step } I_{n_i} \otimes I_L$ & 1.09 $\pm$ 0.02 & 1.20 $\pm$ 0.01 & 1.43 $\pm$ 0.01 & 1.57 $\pm$ 0.03 & 1.02 $\pm$ 0.00 & 1.17 $\pm$ 0.00 & 1.42 $\pm$ 0.01 & 1.51 $\pm$ 0.01\\
\hspace{1em}$\text{pffr (Wild)}$ & 2.98 $\pm$ 0.05 & 2.91 $\pm$ 0.03 & 2.85 $\pm$ 0.01 & 2.92 $\pm$ 0.06 & 2.79 $\pm$ 0.01 & 2.99 $\pm$ 0.03 & 2.84 $\pm$ 0.01 & 3.03 $\pm$ 0.03\\
\addlinespace[0.6em]
\multicolumn{9}{l}{\textbf{$N = 50,\; n_i = 100$}}\\
\hspace{1em}$\text{1-Step } {\mathbf{R}}_{\text{FPCA}} \otimes {\mathbf{R}}_{\text{Fun}}$ & 19.34 $\pm$ 0.07 & 22.17 $\pm$ 0.04 & 22.91 $\pm$ 0.32 & 22.30 $\pm$ 0.06 & 20.63 $\pm$ 0.28 & 26.92 $\pm$ 0.68 & 27.90 $\pm$ 0.53 & 21.29 $\pm$ 0.09\\
\hspace{1em}$\text{1-Step } {\mathbf{R}}_{\text{Lon}} \otimes {\mathbf{R}}_{\text{Fun}}$ & 13.07 $\pm$ 0.07 & 16.17 $\pm$ 0.03 & 17.13 $\pm$ 0.24 & 16.31 $\pm$ 0.06 & 13.65 $\pm$ 0.20 & 19.73 $\pm$ 0.55 & 20.65 $\pm$ 0.40 & 15.25 $\pm$ 0.07\\
\hspace{1em}$\text{1-Step } I_{n_i} \otimes {\mathbf{R}}_{\text{Fun}}$ & 19.82 $\pm$ 0.06 & 23.36 $\pm$ 0.04 & 24.25 $\pm$ 0.34 & 23.14 $\pm$ 0.06 & 20.96 $\pm$ 0.30 & 28.06 $\pm$ 0.65 & 29.24 $\pm$ 0.54 & 21.59 $\pm$ 0.09\\
\hspace{1em}$\text{1-Step } {\mathbf{R}}_{\text{Lon}} \otimes I_L$ & 13.93 $\pm$ 0.07 & 17.90 $\pm$ 0.04 & 18.27 $\pm$ 0.25 & 17.76 $\pm$ 0.06 & 14.44 $\pm$ 0.22 & 21.12 $\pm$ 0.50 & 21.52 $\pm$ 0.40 & 16.67 $\pm$ 0.07\\
\hspace{1em}$\text{1-Step } I_{n_i} \otimes I_L$ & 7.23 $\pm$ 0.06 & 10.67 $\pm$ 0.04 & 10.58 $\pm$ 0.14 & 10.20 $\pm$ 0.04 & 6.65 $\pm$ 0.09 & 11.41 $\pm$ 0.26 & 12.51 $\pm$ 0.21 & 9.35 $\pm$ 0.05\\
\hspace{1em}$\text{Full-1Step } {\mathbf{R}}_{\text{FPCA}} \otimes {\mathbf{R}}_{\text{Fun}}$ & 68.17 $\pm$ 1.87 & 67.70 $\pm$ 0.69 & 312.93 $\pm$ 14.10 & 113.21 $\pm$ 3.60 & 40.40 $\pm$ 2.06 & 57.31 $\pm$ 2.18 & 67.09 $\pm$ 2.14 & 55.01 $\pm$ 0.34\\
\hspace{1em}$\text{Full-1Step } {\mathbf{R}}_{\text{Lon}} \otimes {\mathbf{R}}_{\text{Fun}}$ & 31.44 $\pm$ 0.17 & 38.99 $\pm$ 0.29 & 94.21 $\pm$ 4.64 & 54.94 $\pm$ 0.73 & 24.47 $\pm$ 0.40 & 35.64 $\pm$ 0.86 & 41.81 $\pm$ 0.93 & 33.23 $\pm$ 0.20\\
\hspace{1em}$\text{Full-1Step } I_{n_i} \otimes {\mathbf{R}}_{\text{Fun}}$ & 56.71 $\pm$ 0.33 & 59.48 $\pm$ 0.32 & 81.54 $\pm$ 1.35 & 99.65 $\pm$ 1.45 & 44.86 $\pm$ 2.11 & 55.77 $\pm$ 1.25 & 70.31 $\pm$ 1.53 & 54.55 $\pm$ 0.34\\
\hspace{1em}$\text{Full-1Step } {\mathbf{R}}_{\text{Lon}} \otimes I_L$ & 30.35 $\pm$ 0.16 & 68.70 $\pm$ 4.12 & 167.87 $\pm$ 8.92 & 54.21 $\pm$ 1.20 & 26.07 $\pm$ 0.38 & 39.29 $\pm$ 0.87 & 51.02 $\pm$ 3.07 & 36.34 $\pm$ 0.22\\
\hspace{1em}$\text{Full-1Step } I_{n_i} \otimes I_L$ & 7.95 $\pm$ 0.05 & 10.88 $\pm$ 0.06 & 14.01 $\pm$ 0.18 & 12.72 $\pm$ 0.06 & 7.55 $\pm$ 0.09 & 12.54 $\pm$ 0.32 & 14.86 $\pm$ 0.33 & 10.51 $\pm$ 0.05\\
\hspace{1em}$\text{pffr (Wild)}$ & 9.67 $\pm$ 0.03 & 11.32 $\pm$ 0.02 & 10.10 $\pm$ 0.12 & 11.29 $\pm$ 0.04 & 9.97 $\pm$ 0.13 & 12.77 $\pm$ 0.36 & 11.94 $\pm$ 0.20 & 10.54 $\pm$ 0.07\\
\addlinespace[0.6em]
\multicolumn{9}{l}{\textbf{$N = 500,\; n_i = 5$}}\\
\hspace{1em}$\text{1-Step } {\mathbf{R}}_{\text{FPCA}} \otimes {\mathbf{R}}_{\text{Fun}}$ & 23.05 $\pm$ 0.20 & 25.20 $\pm$ 0.33 & 25.11 $\pm$ 0.20 & 27.96 $\pm$ 0.52 & 23.19 $\pm$ 0.23 & 29.90 $\pm$ 0.61 & 25.51 $\pm$ 0.05 & 29.89 $\pm$ 0.62\\
\hspace{1em}$\text{1-Step } {\mathbf{R}}_{\text{Lon}} \otimes {\mathbf{R}}_{\text{Fun}}$ & 30.28 $\pm$ 0.27 & 32.41 $\pm$ 0.43 & 31.90 $\pm$ 0.26 & 36.13 $\pm$ 0.69 & 29.85 $\pm$ 0.31 & 38.94 $\pm$ 0.79 & 32.75 $\pm$ 0.07 & 38.53 $\pm$ 0.80\\
\hspace{1em}$\text{1-Step } I_{n_i} \otimes {\mathbf{R}}_{\text{Fun}}$ & 12.59 $\pm$ 0.10 & 14.28 $\pm$ 0.17 & 14.55 $\pm$ 0.11 & 15.99 $\pm$ 0.27 & 12.47 $\pm$ 0.11 & 16.52 $\pm$ 0.30 & 14.49 $\pm$ 0.03 & 16.63 $\pm$ 0.31\\
\hspace{1em}$\text{1-Step } {\mathbf{R}}_{\text{Lon}} \otimes I_L$ & 56.86 $\pm$ 0.46 & 60.19 $\pm$ 0.71 & 59.45 $\pm$ 0.44 & 66.96 $\pm$ 1.09 & 57.88 $\pm$ 0.50 & 71.17 $\pm$ 1.27 & 61.01 $\pm$ 0.14 & 70.02 $\pm$ 1.30\\
\hspace{1em}$\text{1-Step } I_{n_i} \otimes I_L$ & 3.42 $\pm$ 0.03 & 5.11 $\pm$ 0.06 & 5.38 $\pm$ 0.04 & 5.72 $\pm$ 0.10 & 3.39 $\pm$ 0.03 & 5.79 $\pm$ 0.11 & 5.36 $\pm$ 0.01 & 5.88 $\pm$ 0.11\\
\hspace{1em}$\text{Full-1Step } {\mathbf{R}}_{\text{FPCA}} \otimes {\mathbf{R}}_{\text{Fun}}$ & 54.62 $\pm$ 3.02 & 66.19 $\pm$ 0.92 & 77.71 $\pm$ 0.81 & 96.22 $\pm$ 3.20 & 54.72 $\pm$ 3.02 & 77.05 $\pm$ 1.60 & 76.86 $\pm$ 0.35 & 115.07 $\pm$ 9.36\\
\hspace{1em}$\text{Full-1Step } {\mathbf{R}}_{\text{Lon}} \otimes {\mathbf{R}}_{\text{Fun}}$ & 70.57 $\pm$ 0.65 & 89.43 $\pm$ 1.20 & 104.10 $\pm$ 1.00 & 122.85 $\pm$ 2.35 & 70.65 $\pm$ 0.73 & 102.98 $\pm$ 2.20 & 104.37 $\pm$ 0.49 & 131.59 $\pm$ 2.72\\
\hspace{1em}$\text{Full-1Step } I_{n_i} \otimes {\mathbf{R}}_{\text{Fun}}$ & 25.21 $\pm$ 0.21 & 29.30 $\pm$ 0.38 & 37.51 $\pm$ 0.29 & 42.66 $\pm$ 0.76 & 24.35 $\pm$ 0.23 & 33.84 $\pm$ 0.63 & 36.19 $\pm$ 0.16 & 43.34 $\pm$ 0.83\\
\hspace{1em}$\text{Full-1Step } {\mathbf{R}}_{\text{Lon}} \otimes I_L$ & 122.42 $\pm$ 1.26 & 162.84 $\pm$ 2.16 & 200.21 $\pm$ 1.59 & 218.75 $\pm$ 3.65 & 131.24 $\pm$ 8.02 & 183.09 $\pm$ 3.39 & 193.70 $\pm$ 0.87 & 240.47 $\pm$ 9.26\\
\hspace{1em}$\text{Full-1Step } I_{n_i} \otimes I_L$ & 4.14 $\pm$ 0.03 & 6.16 $\pm$ 0.07 & 7.24 $\pm$ 0.05 & 7.82 $\pm$ 0.13 & 4.11 $\pm$ 0.03 & 6.94 $\pm$ 0.13 & 6.98 $\pm$ 0.02 & 7.90 $\pm$ 0.14\\
\hspace{1em}$\text{pffr (Wild)}$ & 6.58 $\pm$ 0.05 & 7.16 $\pm$ 0.07 & 6.90 $\pm$ 0.04 & 8.15 $\pm$ 0.12 & 6.31 $\pm$ 0.05 & 8.21 $\pm$ 0.13 & 6.95 $\pm$ 0.02 & 8.28 $\pm$ 0.13\\
\addlinespace[0.6em]
\multicolumn{9}{l}{\textbf{$N = 500,\; n_i = 100$}}\\
\hspace{1em}$\text{1-Step } {\mathbf{R}}_{\text{FPCA}} \otimes {\mathbf{R}}_{\text{Fun}}$ & 161.15 $\pm$ 0.42 & 195.50 $\pm$ 1.02 & 244.10 $\pm$ 4.67 & 184.74 $\pm$ 0.46 & 173.06 $\pm$ 2.12 & 193.85 $\pm$ 0.43 & 232.68 $\pm$ 4.23 & 189.82 $\pm$ 0.88\\
\hspace{1em}$\text{1-Step } {\mathbf{R}}_{\text{Lon}} \otimes {\mathbf{R}}_{\text{Fun}}$ & 111.85 $\pm$ 0.27 & 149.61 $\pm$ 0.84 & 198.96 $\pm$ 3.74 & 148.93 $\pm$ 0.34 & 113.00 $\pm$ 1.58 & 148.51 $\pm$ 0.36 & 189.01 $\pm$ 3.48 & 148.69 $\pm$ 0.64\\
\hspace{1em}$\text{1-Step } I_{n_i} \otimes {\mathbf{R}}_{\text{Fun}}$ & 178.16 $\pm$ 0.37 & 215.29 $\pm$ 1.16 & 268.15 $\pm$ 5.33 & 212.97 $\pm$ 0.54 & 189.46 $\pm$ 2.57 & 210.95 $\pm$ 0.49 & 256.99 $\pm$ 4.94 & 208.39 $\pm$ 0.83\\
\hspace{1em}$\text{1-Step } {\mathbf{R}}_{\text{Lon}} \otimes I_L$ & 118.14 $\pm$ 0.35 & 151.87 $\pm$ 0.98 & 213.95 $\pm$ 4.25 & 157.47 $\pm$ 0.40 & 123.12 $\pm$ 1.61 & 150.03 $\pm$ 0.34 & 203.80 $\pm$ 3.94 & 162.38 $\pm$ 0.77\\
\hspace{1em}$\text{1-Step } I_{n_i} \otimes I_L$ & 54.11 $\pm$ 0.25 & 81.74 $\pm$ 0.57 & 122.54 $\pm$ 2.24 & 82.57 $\pm$ 0.40 & 59.04 $\pm$ 0.74 & 84.63 $\pm$ 0.46 & 104.11 $\pm$ 1.97 & 93.08 $\pm$ 0.56\\
\hspace{1em}$\text{Full-1Step } {\mathbf{R}}_{\text{FPCA}} \otimes {\mathbf{R}}_{\text{Fun}}$ & 311.65 $\pm$ 1.21 & 419.43 $\pm$ 2.76 & 813.43 $\pm$ 18.60 & 511.77 $\pm$ 17.09 & 282.71 $\pm$ 3.95 & 342.37 $\pm$ 0.99 & 441.93 $\pm$ 8.97 & 390.50 $\pm$ 2.31\\
\hspace{1em}$\text{Full-1Step } {\mathbf{R}}_{\text{Lon}} \otimes {\mathbf{R}}_{\text{Fun}}$ & 194.44 $\pm$ 0.38 & 264.40 $\pm$ 1.82 & 449.09 $\pm$ 10.68 & 303.22 $\pm$ 1.33 & 179.31 $\pm$ 2.83 & 223.57 $\pm$ 0.77 & 295.21 $\pm$ 6.39 & 253.62 $\pm$ 1.38\\
\hspace{1em}$\text{Full-1Step } I_{n_i} \otimes {\mathbf{R}}_{\text{Fun}}$ & 371.03 $\pm$ 0.74 & 421.48 $\pm$ 3.08 & 661.56 $\pm$ 14.27 & 549.75 $\pm$ 2.84 & 324.73 $\pm$ 4.74 & 364.12 $\pm$ 1.93 & 500.39 $\pm$ 10.45 & 416.14 $\pm$ 1.90\\
\hspace{1em}$\text{Full-1Step } {\mathbf{R}}_{\text{Lon}} \otimes I_L$ & 210.18 $\pm$ 0.88 & 328.53 $\pm$ 12.12 & 538.08 $\pm$ 16.18 & 359.53 $\pm$ 11.97 & 213.07 $\pm$ 3.01 & 263.79 $\pm$ 1.25 & 336.11 $\pm$ 6.46 & 288.10 $\pm$ 1.40\\
\hspace{1em}$\text{Full-1Step } I_{n_i} \otimes I_L$ & 65.97 $\pm$ 0.28 & 92.03 $\pm$ 0.60 & 125.46 $\pm$ 2.23 & 97.32 $\pm$ 0.30 & 62.81 $\pm$ 1.08 & 88.82 $\pm$ 0.39 & 118.41 $\pm$ 1.93 & 91.37 $\pm$ 0.48\\
\hspace{1em}$\text{pffr (Wild)}$ & 49.35 $\pm$ 0.22 & 56.25 $\pm$ 0.32 & 61.14 $\pm$ 1.00 & 53.87 $\pm$ 0.13 & 51.59 $\pm$ 0.53 & 54.28 $\pm$ 0.18 & 59.47 $\pm$ 0.90 & 56.62 $\pm$ 0.25\\
\bottomrule
\end{tabular}}
\end{table}

\clearpage

\subsection{Fully-Iterated FGEE Comparisons}\label{app:full_sims}
\begin{table}[!h]
\centering
\caption{\footnotesize Functional Coefficient Estimation Performance (RMSE) of each method relative to the $\texttt{pffr}$ fit ($\text{RMSE}/\text{RMSE}_{\text{pffr}}$).
    Cells contain the average of 300 replicates $\pm$ SE (SE$=0.00$ indicates a value $<0.01$). Outcomes were simulated with an $\mathbf{R}^* = \mathbf{R}^*_{\text{Lon}} \otimes \mathbf{R}^*_{\text{Fun}}$, where $\mathbf{R}^*_{\text{Fun}}$ had an AR1 structure and the table columns indicate results where $\mathbf{R}^*_{\text{Lon}}$ had exchangeable or AR1 correlation. The ``1-Step'' indicates a one-step was used for tuning and final coefficient estimation, and ``Full-1step'' indicates one-step tuning and a fully-iterated fGEE for final coefficient estimation, with the indicated working correlation. $*$ indicates that values $\geq100$ were removed from that cell to avoid skewing the mean.}
\centering
\resizebox{\ifdim\width>\linewidth\linewidth\else\width\fi}{!}{
\fontsize{9}{11}\selectfont
\begin{tabular}[t]{>{\raggedright\arraybackslash}p{4.8cm}>{\raggedright\arraybackslash}p{2.1cm}>{\raggedright\arraybackslash}p{2.1cm}>{\raggedright\arraybackslash}p{2.1cm}>{\raggedright\arraybackslash}p{2.1cm}>{\raggedright\arraybackslash}p{2.1cm}>{\raggedright\arraybackslash}p{2.1cm}}
\toprule
\multicolumn{1}{c}{ } & \multicolumn{3}{c}{Exchangeable} & \multicolumn{3}{c}{AR(1)} \\
\cmidrule(l{3pt}r{3pt}){2-4} \cmidrule(l{3pt}r{3pt}){5-7}
Method & Gaussian & Poisson & Binomial & Gaussian & Poisson & Binomial\\
\midrule
\addlinespace[0em]
\multicolumn{7}{l}{\textbf{$N = 25,\; n_i = 5$}}\\
\hspace{1em}$\text{1-Step } {\mathbf{R}}_{\text{Lon}} \otimes {\mathbf{R}}_{\text{Fun}}$ & 0.82 $\pm$ 0.01 & 0.83 $\pm$ 0.01 & 0.88 $\pm$ 0.01 & 0.85 $\pm$ 0.01 & 0.85 $\pm$ 0.01 & 0.92 $\pm$ 0.01\\
\hspace{1em}$\text{Full-1Step } {\mathbf{R}}_{\text{Lon}} \otimes {\mathbf{R}}_{\text{Fun}}$ & 0.82 $\pm$ 0.01 & 0.83 $\pm$ 0.01 & 1.00 $\pm$ 0.09$^*$ & 0.86 $\pm$ 0.01 & 0.85 $\pm$ 0.01 & 0.94 $\pm$ 0.01\\
\hspace{1em}$\text{Full-Full } \mathbf{R}_{\text{Lon}} \otimes \mathbf{R}_{\text{Fun}}$ &0.84 $\pm$ 0.01 & 0.89 $\pm$ 0.01$^*$ & 8.71 $\pm$ 1.04$^*$ & 0.87 $\pm$ 0.01 & 0.88 $\pm$ 0.01 & 2.44 $\pm$ 0.44$^*$\\
\addlinespace[0.6em]
\multicolumn{7}{l}{\textbf{$N = 25,\; n_i = 100$}}\\
\hspace{1em}$\text{1-Step } {\mathbf{R}}_{\text{Lon}} \otimes {\mathbf{R}}_{\text{Fun}}$ & 0.79 $\pm$ 0.01 & 0.63 $\pm$ 0.01 & 0.81 $\pm$ 0.01 & 0.88 $\pm$ 0.01 & 0.86 $\pm$ 0.01 & 0.89 $\pm$ 0.01\\
\hspace{1em}$\text{Full-1Step } {\mathbf{R}}_{\text{Lon}} \otimes {\mathbf{R}}_{\text{Fun}}$ & 0.79 $\pm$ 0.01 & 0.64 $\pm$ 0.01 & 5.46 $\pm$ 0.77$^*$ & 0.88 $\pm$ 0.01 & 0.86 $\pm$ 0.01 & 0.90 $\pm$ 0.01\\
\hspace{1em}$\text{Full-Full } \mathbf{R}_{\text{Lon}} \otimes \mathbf{R}_{\text{Fun}}$ &0.81 $\pm$ 0.01 & 0.67 $\pm$ 0.01 & 22.80 $\pm$ 1.48$^*$ & 0.88 $\pm$ 0.01 & 0.87 $\pm$ 0.01 & 0.91 $\pm$ 0.01\\
\addlinespace[0.6em]
\multicolumn{7}{l}{\textbf{$N = 50,\; n_i = 5$}}\\
\hspace{1em}$\text{1-Step } {\mathbf{R}}_{\text{Lon}} \otimes {\mathbf{R}}_{\text{Fun}}$ & 0.88 $\pm$ 0.01 & 0.84 $\pm$ 0.01 & 0.88 $\pm$ 0.01 & 0.89 $\pm$ 0.01 & 0.87 $\pm$ 0.01 & 0.89 $\pm$ 0.01\\
\hspace{1em}$\text{Full-1Step } {\mathbf{R}}_{\text{Lon}} \otimes {\mathbf{R}}_{\text{Fun}}$ & 0.88 $\pm$ 0.01 & 0.84 $\pm$ 0.01 & 0.90 $\pm$ 0.01 & 0.89 $\pm$ 0.01 & 0.87 $\pm$ 0.01 & 0.91 $\pm$ 0.01\\
\hspace{1em}$\text{Full-Full } \mathbf{R}_{\text{Lon}} \otimes \mathbf{R}_{\text{Fun}}$ &0.89 $\pm$ 0.01 & 0.86 $\pm$ 0.01 & 1.51 $\pm$ 0.32$^*$ & 0.91 $\pm$ 0.01 & 0.88 $\pm$ 0.01 & 1.08 $\pm$ 0.09\\
\addlinespace[0.6em]
\multicolumn{7}{l}{\textbf{$N = 50,\; n_i = 100$}}\\
\hspace{1em}$\text{1-Step } {\mathbf{R}}_{\text{Lon}} \otimes {\mathbf{R}}_{\text{Fun}}$ & 0.85 $\pm$ 0.01 & 0.66 $\pm$ 0.01 & 0.84 $\pm$ 0.01 & 0.90 $\pm$ 0.01 & 0.87 $\pm$ 0.01 & 0.90 $\pm$ 0.01\\
\hspace{1em}$\text{Full-1Step } {\mathbf{R}}_{\text{Lon}} \otimes {\mathbf{R}}_{\text{Fun}}$ & 0.85 $\pm$ 0.01 & 0.66 $\pm$ 0.01 & 2.58 $\pm$ 0.49$^*$ & 0.90 $\pm$ 0.01 & 0.87 $\pm$ 0.01 & 0.91 $\pm$ 0.01\\
\hspace{1em}$\text{Full-Full } \mathbf{R}_{\text{Lon}} \otimes \mathbf{R}_{\text{Fun}}$ &0.86 $\pm$ 0.01 & 0.68 $\pm$ 0.01 & 13.59 $\pm$ 1.24$^*$ & 0.91 $\pm$ 0.01 & 0.88 $\pm$ 0.01 & 0.91 $\pm$ 0.01\\
\bottomrule
\end{tabular}}
\end{table}

\begin{table}[!h]
\centering
\caption{\footnotesize
Joint 95\% CI coverage from 300 replicates $\pm$ SE. 
Table columns indicate if $\mathbf{R}^*_{\text{Lon}}$ had exchangeable or AR1 correlation. 
pffr ($z_{1-\alpha/2}$) are standard Wald CIs constructed with Gaussian quantiles.}
\centering
\resizebox{\ifdim\width>\linewidth\linewidth\else\width\fi}{!}{
\fontsize{9}{11}\selectfont
\begin{tabular}[t]{>{\raggedright\arraybackslash}p{4.8cm}>{\raggedright\arraybackslash}p{2.1cm}>{\raggedright\arraybackslash}p{2.1cm}>{\raggedright\arraybackslash}p{2.1cm}>{\raggedright\arraybackslash}p{2.1cm}>{\raggedright\arraybackslash}p{2.1cm}>{\raggedright\arraybackslash}p{2.1cm}}
\toprule
\multicolumn{1}{c}{ } & \multicolumn{3}{c}{Exchangeable} & \multicolumn{3}{c}{AR(1)} \\
\cmidrule(l{3pt}r{3pt}){2-4} \cmidrule(l{3pt}r{3pt}){5-7}
Method & Gaussian & Poisson & Binomial & Gaussian & Poisson & Binomial\\
\midrule
\addlinespace[0em]
\multicolumn{7}{l}{\textbf{$N = 25,\; n_i = 5$}}\\
\hspace{1em}$\text{1-Step } {\mathbf{R}}_{\text{Lon}} \otimes {\mathbf{R}}_{\text{Fun}}$ & 0.90 $\pm$ 0.02 & 0.90 $\pm$ 0.02 & 0.95 $\pm$ 0.01 & 0.90 $\pm$ 0.02 & 0.88 $\pm$ 0.02 & 0.95 $\pm$ 0.01\\
\hspace{1em}$\text{Full-1Step } {\mathbf{R}}_{\text{Lon}} \otimes {\mathbf{R}}_{\text{Fun}}$ & 0.90 $\pm$ 0.02 & 0.90 $\pm$ 0.02 & 0.94 $\pm$ 0.01 & 0.89 $\pm$ 0.02 & 0.88 $\pm$ 0.02 & 0.93 $\pm$ 0.01\\
\hspace{1em}$\text{Full-Full } \mathbf{R}_{\text{Lon}} \otimes \mathbf{R}_{\text{Fun}}$ &0.91 $\pm$ 0.02 & 0.90 $\pm$ 0.02 & 0.62 $\pm$ 0.03 & 0.91 $\pm$ 0.02 & 0.89 $\pm$ 0.02 & 0.82 $\pm$ 0.02\\
\hspace{1em}$\text{pffr ($z_{1-\alpha/2}$)}$ & 0.07 $\pm$ 0.01 & 0.56 $\pm$ 0.03 & 0.33 $\pm$ 0.03 & 0.06 $\pm$ 0.01 & 0.56 $\pm$ 0.03 & 0.32 $\pm$ 0.03\\
\addlinespace[0.6em]
\multicolumn{7}{l}{\textbf{$N = 25,\; n_i = 100$}}\\
\hspace{1em}$\text{1-Step } {\mathbf{R}}_{\text{Lon}} \otimes {\mathbf{R}}_{\text{Fun}}$ & 0.97 $\pm$ 0.01 & 0.98 $\pm$ 0.01 & 0.99 $\pm$ 0.01 & 0.87 $\pm$ 0.02 & 0.89 $\pm$ 0.02 & 0.91 $\pm$ 0.02\\
\hspace{1em}$\text{Full-1Step } {\mathbf{R}}_{\text{Lon}} \otimes {\mathbf{R}}_{\text{Fun}}$ & 0.97 $\pm$ 0.01 & 0.98 $\pm$ 0.01 & 0.78 $\pm$ 0.02 & 0.87 $\pm$ 0.02 & 0.88 $\pm$ 0.02 & 0.91 $\pm$ 0.02\\
\hspace{1em}$\text{Full-Full } \mathbf{R}_{\text{Lon}} \otimes \mathbf{R}_{\text{Fun}}$ &0.97 $\pm$ 0.01 & 0.97 $\pm$ 0.01 & 0.21 $\pm$ 0.02 & 0.89 $\pm$ 0.02 & 0.90 $\pm$ 0.02 & 0.92 $\pm$ 0.02\\
\hspace{1em}$\text{pffr ($z_{1-\alpha/2}$)}$ & 0.50 $\pm$ 0.03 & 0.81 $\pm$ 0.02 & 0.65 $\pm$ 0.03 & 0.25 $\pm$ 0.03 & 0.77 $\pm$ 0.02 & 0.61 $\pm$ 0.03\\
\addlinespace[0.6em]
\multicolumn{7}{l}{\textbf{$N = 50,\; n_i = 5$}}\\
\hspace{1em}$\text{1-Step } {\mathbf{R}}_{\text{Lon}} \otimes {\mathbf{R}}_{\text{Fun}}$ & 0.90 $\pm$ 0.02 & 0.90 $\pm$ 0.02 & 0.95 $\pm$ 0.01 & 0.89 $\pm$ 0.02 & 0.87 $\pm$ 0.02 & 0.94 $\pm$ 0.01\\
\hspace{1em}$\text{Full-1Step } {\mathbf{R}}_{\text{Lon}} \otimes {\mathbf{R}}_{\text{Fun}}$ & 0.90 $\pm$ 0.02 & 0.90 $\pm$ 0.02 & 0.94 $\pm$ 0.01 & 0.89 $\pm$ 0.02 & 0.87 $\pm$ 0.02 & 0.94 $\pm$ 0.01\\
\hspace{1em}$\text{Full-Full } \mathbf{R}_{\text{Lon}} \otimes \mathbf{R}_{\text{Fun}}$ &0.90 $\pm$ 0.02 & 0.90 $\pm$ 0.02 & 0.88 $\pm$ 0.02 & 0.89 $\pm$ 0.02 & 0.89 $\pm$ 0.02 & 0.92 $\pm$ 0.02\\
\hspace{1em}$\text{pffr ($z_{1-\alpha/2}$)}$ & 0.04 $\pm$ 0.01 & 0.59 $\pm$ 0.03 & 0.23 $\pm$ 0.02 & 0.04 $\pm$ 0.01 & 0.57 $\pm$ 0.03 & 0.26 $\pm$ 0.03\\
\addlinespace[0.6em]
\multicolumn{7}{l}{\textbf{$N = 50,\; n_i = 100$}}\\
\hspace{1em}$\text{1-Step } {\mathbf{R}}_{\text{Lon}} \otimes {\mathbf{R}}_{\text{Fun}}$ & 0.97 $\pm$ 0.01 & 0.96 $\pm$ 0.01 & 0.98 $\pm$ 0.01 & 0.92 $\pm$ 0.02 & 0.87 $\pm$ 0.02 & 0.94 $\pm$ 0.01\\
\hspace{1em}$\text{Full-1Step } {\mathbf{R}}_{\text{Lon}} \otimes {\mathbf{R}}_{\text{Fun}}$ & 0.97 $\pm$ 0.01 & 0.96 $\pm$ 0.01 & 0.91 $\pm$ 0.02 & 0.92 $\pm$ 0.02 & 0.87 $\pm$ 0.02 & 0.94 $\pm$ 0.01\\
\hspace{1em}$\text{Full-Full } \mathbf{R}_{\text{Lon}} \otimes \mathbf{R}_{\text{Fun}}$ &0.97 $\pm$ 0.01 & 0.96 $\pm$ 0.01 & 0.48 $\pm$ 0.03 & 0.92 $\pm$ 0.02 & 0.88 $\pm$ 0.02 & 0.94 $\pm$ 0.01\\
\hspace{1em}$\text{pffr ($z_{1-\alpha/2}$)}$ & 0.49 $\pm$ 0.03 & 0.78 $\pm$ 0.02 & 0.68 $\pm$ 0.03 & 0.19 $\pm$ 0.02 & 0.80 $\pm$ 0.02 & 0.63 $\pm$ 0.03\\
\bottomrule
\end{tabular}}
\end{table}

\begin{table}[!h]
\centering
\caption{\footnotesize
Pointwise 95\% CI coverage from 300 replicates $\pm$ SE. 
Table columns indicate if $\mathbf{R}^*_{\text{Lon}}$ had exchangeable or AR1 correlation. 
pffr ($z_{1-\alpha/2}$) are standard Wald CIs constructed with Gaussian quantiles.}
\centering
\resizebox{\ifdim\width>\linewidth\linewidth\else\width\fi}{!}{
\fontsize{9}{11}\selectfont
\begin{tabular}[t]{>{\raggedright\arraybackslash}p{4.8cm}>{\raggedright\arraybackslash}p{2.1cm}>{\raggedright\arraybackslash}p{2.1cm}>{\raggedright\arraybackslash}p{2.1cm}>{\raggedright\arraybackslash}p{2.1cm}>{\raggedright\arraybackslash}p{2.1cm}>{\raggedright\arraybackslash}p{2.1cm}}
\toprule
\multicolumn{1}{c}{ } & \multicolumn{3}{c}{Exchangeable} & \multicolumn{3}{c}{AR(1)} \\
\cmidrule(l{3pt}r{3pt}){2-4} \cmidrule(l{3pt}r{3pt}){5-7}
Method & Gaussian & Poisson & Binomial & Gaussian & Poisson & Binomial\\
\midrule
\addlinespace[0em]
\multicolumn{7}{l}{\textbf{$N = 25,\; n_i = 5$}}\\
\hspace{1em}$\text{1-Step } {\mathbf{R}}_{\text{Lon}} \otimes {\mathbf{R}}_{\text{Fun}}$ & 0.95 $\pm$ 0.01 & 0.95 $\pm$ 0.01 & 0.97 $\pm$ 0.01 & 0.94 $\pm$ 0.01 & 0.95 $\pm$ 0.01 & 0.96 $\pm$ 0.01\\
\hspace{1em}$\text{Full-1Step } {\mathbf{R}}_{\text{Lon}} \otimes {\mathbf{R}}_{\text{Fun}}$ & 0.95 $\pm$ 0.01 & 0.95 $\pm$ 0.01 & 0.96 $\pm$ 0.01 & 0.94 $\pm$ 0.01 & 0.95 $\pm$ 0.01 & 0.96 $\pm$ 0.01\\
\hspace{1em}$\text{Full-Full } \mathbf{R}_{\text{Lon}} \otimes \mathbf{R}_{\text{Fun}}$ &0.95 $\pm$ 0.01 & 0.95 $\pm$ 0.01 & 0.72 $\pm$ 0.03 & 0.95 $\pm$ 0.01 & 0.95 $\pm$ 0.01 & 0.91 $\pm$ 0.02\\
\hspace{1em}$\text{pffr ($z_{1-\alpha/2}$)}$ & 0.67 $\pm$ 0.03 & 0.87 $\pm$ 0.02 & 0.80 $\pm$ 0.02 & 0.66 $\pm$ 0.03 & 0.87 $\pm$ 0.02 & 0.80 $\pm$ 0.02\\
\addlinespace[0.6em]
\multicolumn{7}{l}{\textbf{$N = 25,\; n_i = 100$}}\\
\hspace{1em}$\text{1-Step } {\mathbf{R}}_{\text{Lon}} \otimes {\mathbf{R}}_{\text{Fun}}$ & 0.97 $\pm$ 0.01 & 0.98 $\pm$ 0.01 & 0.98 $\pm$ 0.01 & 0.95 $\pm$ 0.01 & 0.95 $\pm$ 0.01 & 0.95 $\pm$ 0.01\\
\hspace{1em}$\text{Full-1Step } {\mathbf{R}}_{\text{Lon}} \otimes {\mathbf{R}}_{\text{Fun}}$ & 0.97 $\pm$ 0.01 & 0.98 $\pm$ 0.01 & 0.82 $\pm$ 0.02 & 0.95 $\pm$ 0.01 & 0.95 $\pm$ 0.01 & 0.95 $\pm$ 0.01\\
\hspace{1em}$\text{Full-Full } \mathbf{R}_{\text{Lon}} \otimes \mathbf{R}_{\text{Fun}}$ &0.98 $\pm$ 0.01 & 0.98 $\pm$ 0.01 & 0.39 $\pm$ 0.03 & 0.95 $\pm$ 0.01 & 0.95 $\pm$ 0.01 & 0.95 $\pm$ 0.01\\
\hspace{1em}$\text{pffr ($z_{1-\alpha/2}$)}$ & 0.86 $\pm$ 0.02 & 0.94 $\pm$ 0.01 & 0.90 $\pm$ 0.02 & 0.79 $\pm$ 0.02 & 0.92 $\pm$ 0.02 & 0.88 $\pm$ 0.02\\
\addlinespace[0.6em]
\multicolumn{7}{l}{\textbf{$N = 50,\; n_i = 5$}}\\
\hspace{1em}$\text{1-Step } {\mathbf{R}}_{\text{Lon}} \otimes {\mathbf{R}}_{\text{Fun}}$ & 0.94 $\pm$ 0.01 & 0.94 $\pm$ 0.01 & 0.96 $\pm$ 0.01 & 0.93 $\pm$ 0.01 & 0.94 $\pm$ 0.01 & 0.95 $\pm$ 0.01\\
\hspace{1em}$\text{Full-1Step } {\mathbf{R}}_{\text{Lon}} \otimes {\mathbf{R}}_{\text{Fun}}$ & 0.94 $\pm$ 0.01 & 0.94 $\pm$ 0.01 & 0.95 $\pm$ 0.01 & 0.93 $\pm$ 0.01 & 0.94 $\pm$ 0.01 & 0.95 $\pm$ 0.01\\
\hspace{1em}$\text{Full-Full } \mathbf{R}_{\text{Lon}} \otimes \mathbf{R}_{\text{Fun}}$ &0.94 $\pm$ 0.01 & 0.95 $\pm$ 0.01 & 0.92 $\pm$ 0.02 & 0.94 $\pm$ 0.01 & 0.94 $\pm$ 0.01 & 0.94 $\pm$ 0.01\\
\hspace{1em}$\text{pffr ($z_{1-\alpha/2}$)}$ & 0.65 $\pm$ 0.03 & 0.87 $\pm$ 0.02 & 0.77 $\pm$ 0.02 & 0.65 $\pm$ 0.03 & 0.86 $\pm$ 0.02 & 0.77 $\pm$ 0.02\\
\addlinespace[0.6em]
\multicolumn{7}{l}{\textbf{$N = 50,\; n_i = 100$}}\\
\hspace{1em}$\text{1-Step } {\mathbf{R}}_{\text{Lon}} \otimes {\mathbf{R}}_{\text{Fun}}$ & 0.96 $\pm$ 0.01 & 0.97 $\pm$ 0.01 & 0.97 $\pm$ 0.01 & 0.94 $\pm$ 0.01 & 0.94 $\pm$ 0.01 & 0.95 $\pm$ 0.01\\
\hspace{1em}$\text{Full-1Step } {\mathbf{R}}_{\text{Lon}} \otimes {\mathbf{R}}_{\text{Fun}}$ & 0.96 $\pm$ 0.01 & 0.97 $\pm$ 0.01 & 0.91 $\pm$ 0.02 & 0.94 $\pm$ 0.01 & 0.94 $\pm$ 0.01 & 0.95 $\pm$ 0.01\\
\hspace{1em}$\text{Full-Full } \mathbf{R}_{\text{Lon}} \otimes \mathbf{R}_{\text{Fun}}$ &0.97 $\pm$ 0.01 & 0.97 $\pm$ 0.01 & 0.61 $\pm$ 0.03 & 0.94 $\pm$ 0.01 & 0.94 $\pm$ 0.01 & 0.95 $\pm$ 0.01\\
\hspace{1em}$\text{pffr ($z_{1-\alpha/2}$)}$ & 0.84 $\pm$ 0.02 & 0.92 $\pm$ 0.02 & 0.89 $\pm$ 0.02 & 0.77 $\pm$ 0.02 & 0.92 $\pm$ 0.02 & 0.88 $\pm$ 0.02\\
\bottomrule
\end{tabular}}
\end{table}

\begin{table}[!h]
\centering
\caption{\footnotesize Relative pointwise CI width (mean $\pm$ SE) vs.\ pffr (Wild). SE$=0.00$ indicates a value $<0.01$. We denote $UB^{(r)}(s)$ and $LB^{(r)}(s)$ and $UB_\text{pffr}^{(r)}(s)/LB_\text{pffr}^{(r)}(s)$ as the upper/lower bounds of the CIs (at $s$) of the indicated method and pffr, respectively. Below we report the average ratio $\frac{\frac{1}{(q+1)|\mathcal{S}|}\sum_{r=0}^q\sum_{s \in \mathcal{S}} [UB^{(r)}(s) -LB(s)^{(r)}]}{\frac{1}{(q+1)|\mathcal{S}|}\sum_{r=0}^q\sum_{s \in \mathcal{S}}[UB^{(r)}_\text{pffr}(s) -LB^{(r)}_\text{pffr}(s)]}$ across 300 simulation replicates. Values $<1$ indicate narrower 95\% CIs. Values with $*$ indicate extreme outliers (from poor estimates) were removed from the average of that cell to avoid skewing results.}
\centering
\resizebox{\ifdim\width>\linewidth\linewidth\else\width\fi}{!}{
\fontsize{9}{11}\selectfont
\begin{tabular}[t]{>{\raggedright\arraybackslash}p{4.8cm}>{\raggedright\arraybackslash}p{2.1cm}>{\raggedright\arraybackslash}p{2.1cm}>{\raggedright\arraybackslash}p{2.1cm}>{\raggedright\arraybackslash}p{2.1cm}>{\raggedright\arraybackslash}p{2.1cm}>{\raggedright\arraybackslash}p{2.1cm}}
\toprule
\multicolumn{1}{c}{ } & \multicolumn{3}{c}{Exchangeable} & \multicolumn{3}{c}{AR(1)} \\
\cmidrule(l{3pt}r{3pt}){2-4} \cmidrule(l{3pt}r{3pt}){5-7}
Method & Gaussian & Poisson & Binomial & Gaussian & Poisson & Binomial\\
\midrule
\addlinespace[0em]
\multicolumn{7}{l}{\textbf{$N = 25,\; n_i = 5$}}\\
\hspace{1em}$\text{1-Step } {\mathbf{R}}_{\text{Lon}} \otimes {\mathbf{R}}_{\text{Fun}}$ & 0.77 $\pm$ 0.01 & 0.85 $\pm$ 0.00 & 0.88 $\pm$ 0.01 & 0.75 $\pm$ 0.01 & 0.85 $\pm$ 0.00 & 0.83 $\pm$ 0.01\\
\hspace{1em}$\text{Full-1Step } {\mathbf{R}}_{\text{Lon}} \otimes {\mathbf{R}}_{\text{Fun}}$ & 0.77 $\pm$ 0.01 & 0.86 $\pm$ 0.00 & 0.91 $\pm$ 0.03 & 0.75 $\pm$ 0.01 & 0.85 $\pm$ 0.00 & 0.83 $\pm$ 0.01\\
\hspace{1em}$\text{Full-Full } \mathbf{R}_{\text{Lon}} \otimes \mathbf{R}_{\text{Fun}}$ &0.81 $\pm$ 0.01 & 0.90 $\pm$ 0.01 & 3.44 $\pm$ 0.43$^*$ & 0.79 $\pm$ 0.01 & 0.88 $\pm$ 0.01 & 1.86 $\pm$ 0.34\\
\addlinespace[0.6em]
\multicolumn{7}{l}{\textbf{$N = 25,\; n_i = 100$}}\\
\hspace{1em}$\text{1-Step } {\mathbf{R}}_{\text{Lon}} \otimes {\mathbf{R}}_{\text{Fun}}$ & 0.93 $\pm$ 0.01 & 0.77 $\pm$ 0.01 & 1.26 $\pm$ 0.01 & 0.52 $\pm$ 0.00 & 0.91 $\pm$ 0.00 & 0.85 $\pm$ 0.00\\
\hspace{1em}$\text{Full-1Step } {\mathbf{R}}_{\text{Lon}} \otimes {\mathbf{R}}_{\text{Fun}}$ & 0.93 $\pm$ 0.01 & 0.77 $\pm$ 0.01 & 2.39 $\pm$ 0.21 & 0.52 $\pm$ 0.00 & 0.91 $\pm$ 0.00 & 0.85 $\pm$ 0.00\\
\hspace{1em}$\text{Full-Full } \mathbf{R}_{\text{Lon}} \otimes \mathbf{R}_{\text{Fun}}$ &0.99 $\pm$ 0.01 & 0.79 $\pm$ 0.01 & 13.80 $\pm$ 1.11$^*$ & 0.53 $\pm$ 0.00 & 0.92 $\pm$ 0.00 & 0.87 $\pm$ 0.00\\
\addlinespace[0.6em]
\multicolumn{7}{l}{\textbf{$N = 50,\; n_i = 5$}}\\
\hspace{1em}$\text{1-Step } {\mathbf{R}}_{\text{Lon}} \otimes {\mathbf{R}}_{\text{Fun}}$ & 0.63 $\pm$ 0.00 & 0.75 $\pm$ 0.00 & 0.67 $\pm$ 0.00 & 0.60 $\pm$ 0.00 & 0.74 $\pm$ 0.00 & 0.63 $\pm$ 0.00\\
\hspace{1em}$\text{Full-1Step } {\mathbf{R}}_{\text{Lon}} \otimes {\mathbf{R}}_{\text{Fun}}$ & 0.63 $\pm$ 0.00 & 0.75 $\pm$ 0.00 & 0.67 $\pm$ 0.00 & 0.60 $\pm$ 0.00 & 0.74 $\pm$ 0.00 & 0.63 $\pm$ 0.00\\
\hspace{1em}$\text{Full-Full } \mathbf{R}_{\text{Lon}} \otimes \mathbf{R}_{\text{Fun}}$ &0.65 $\pm$ 0.01 & 0.77 $\pm$ 0.00 & 1.32 $\pm$ 0.24 & 0.62 $\pm$ 0.00 & 0.76 $\pm$ 0.00 & 0.66 $\pm$ 0.01\\
\addlinespace[0.6em]
\multicolumn{7}{l}{\textbf{$N = 50,\; n_i = 100$}}\\
\hspace{1em}$\text{1-Step } {\mathbf{R}}_{\text{Lon}} \otimes {\mathbf{R}}_{\text{Fun}}$ & 0.85 $\pm$ 0.00 & 0.73 $\pm$ 0.01 & 1.16 $\pm$ 0.01 & 0.38 $\pm$ 0.00 & 0.88 $\pm$ 0.00 & 0.72 $\pm$ 0.00\\
\hspace{1em}$\text{Full-1Step } {\mathbf{R}}_{\text{Lon}} \otimes {\mathbf{R}}_{\text{Fun}}$ & 0.85 $\pm$ 0.00 & 0.73 $\pm$ 0.01 & 1.53 $\pm$ 0.11 & 0.38 $\pm$ 0.00 & 0.88 $\pm$ 0.00 & 0.72 $\pm$ 0.00\\
\hspace{1em}$\text{Full-Full } \mathbf{R}_{\text{Lon}} \otimes \mathbf{R}_{\text{Fun}}$ &0.88 $\pm$ 0.01 & 0.74 $\pm$ 0.01 & 5.77 $\pm$ 0.60$^*$ & 0.38 $\pm$ 0.00 & 0.89 $\pm$ 0.00 & 0.72 $\pm$ 0.00\\
\bottomrule
\end{tabular}}
\end{table}

\begin{table}[!h]
\centering
\caption{ \footnotesize Computation time in seconds (mean $\pm$ SE) averaged across 300 simulation replicates (SE$=0.00$ indicates a value $<0.01$).}
\centering
\resizebox{\ifdim\width>\linewidth\linewidth\else\width\fi}{!}{
\fontsize{9}{11}\selectfont
\begin{tabular}[t]{>{\raggedright\arraybackslash}p{4.8cm}>{\raggedright\arraybackslash}p{2.1cm}>{\raggedright\arraybackslash}p{2.1cm}>{\raggedright\arraybackslash}p{2.1cm}>{\raggedright\arraybackslash}p{2.1cm}>{\raggedright\arraybackslash}p{2.1cm}>{\raggedright\arraybackslash}p{2.1cm}}
\toprule
\multicolumn{1}{c}{ } & \multicolumn{3}{c}{Exchangeable} & \multicolumn{3}{c}{AR(1)} \\
\cmidrule(l{3pt}r{3pt}){2-4} \cmidrule(l{3pt}r{3pt}){5-7}
Method & Gaussian & Poisson & Binomial & Gaussian & Poisson & Binomial\\
\midrule
\addlinespace[0em]
\multicolumn{7}{l}{\textbf{$N = 25,\; n_i = 5$}}\\
\hspace{1em}$\text{1-Step } {\mathbf{R}}_{\text{Lon}} \otimes {\mathbf{R}}_{\text{Fun}}$ & 3.34 $\pm$ 0.05 & 3.16 $\pm$ 0.05 & 3.07 $\pm$ 0.06 & 3.78 $\pm$ 0.04 & 4.33 $\pm$ 0.02 & 4.27 $\pm$ 0.06\\
\hspace{1em}$\text{Full-1Step } {\mathbf{R}}_{\text{Lon}} \otimes {\mathbf{R}}_{\text{Fun}}$ & 11.65 $\pm$ 0.19 & 9.56 $\pm$ 0.18 & 15.54 $\pm$ 0.74 & 12.32 $\pm$ 0.14 & 13.80 $\pm$ 0.13 & 18.79 $\pm$ 0.53\\
\hspace{1em}$\text{Full-Full } \mathbf{R}_{\text{Lon}} \otimes \mathbf{R}_{\text{Fun}}$ &11505.60 $\pm$ 190.20 & 14720.25 $\pm$ 244.80 & 14424.23 $\pm$ 330.02 & 12036.27 $\pm$ 143.77 & 19750.26 $\pm$ 133.91 & 19226.47 $\pm$ 326.92\\
\hspace{1em}$\text{pffr ($z_{1-\alpha/2}$)}$ & 0.49 $\pm$ 0.01 & 0.46 $\pm$ 0.01 & 0.54 $\pm$ 0.01 & 0.54 $\pm$ 0.01 & 0.64 $\pm$ 0.00 & 0.74 $\pm$ 0.01\\
\addlinespace[0.6em]
\multicolumn{7}{l}{\textbf{$N = 25,\; n_i = 100$}}\\
\hspace{1em}$\text{1-Step } {\mathbf{R}}_{\text{Lon}} \otimes {\mathbf{R}}_{\text{Fun}}$ & 9.42 $\pm$ 0.13 & 11.43 $\pm$ 0.18 & 12.71 $\pm$ 0.17 & 10.38 $\pm$ 0.12 & 9.66 $\pm$ 0.18 & 10.80 $\pm$ 0.20\\
\hspace{1em}$\text{Full-1Step } {\mathbf{R}}_{\text{Lon}} \otimes {\mathbf{R}}_{\text{Fun}}$ & 28.01 $\pm$ 0.43 & 29.92 $\pm$ 0.57 & 101.13 $\pm$ 5.06 & 19.15 $\pm$ 0.23 & 18.09 $\pm$ 0.34 & 22.17 $\pm$ 0.40\\
\hspace{1em}$\text{Full-Full } \mathbf{R}_{\text{Lon}} \otimes \mathbf{R}_{\text{Fun}}$ &23367.43 $\pm$ 381.19 & 30301.22 $\pm$ 552.99 & 40854.39 $\pm$ 645.95 & 21826.08 $\pm$ 267.42 & 20059.42 $\pm$ 341.07 & 21899.23 $\pm$ 374.88\\
\hspace{1em}$\text{pffr ($z_{1-\alpha/2}$)}$ & 3.37 $\pm$ 0.04 & 3.54 $\pm$ 0.05 & 3.66 $\pm$ 0.05 & 3.54 $\pm$ 0.04 & 3.07 $\pm$ 0.05 & 3.11 $\pm$ 0.05\\
\addlinespace[0.6em]
\multicolumn{7}{l}{\textbf{$N = 50,\; n_i = 5$}}\\
\hspace{1em}$\text{1-Step } {\mathbf{R}}_{\text{Lon}} \otimes {\mathbf{R}}_{\text{Fun}}$ & 4.86 $\pm$ 0.09 & 4.57 $\pm$ 0.10 & 5.35 $\pm$ 0.09 & 7.11 $\pm$ 0.04 & 6.39 $\pm$ 0.07 & 7.66 $\pm$ 0.21\\
\hspace{1em}$\text{Full-1Step } {\mathbf{R}}_{\text{Lon}} \otimes {\mathbf{R}}_{\text{Fun}}$ & 15.58 $\pm$ 0.31 & 14.14 $\pm$ 0.33 & 23.86 $\pm$ 0.66 & 21.48 $\pm$ 0.17 & 20.29 $\pm$ 0.26 & 33.20 $\pm$ 1.00\\
\hspace{1em}$\text{Full-Full } \mathbf{R}_{\text{Lon}} \otimes \mathbf{R}_{\text{Fun}}$ &17428.70 $\pm$ 311.78 & 21569.52 $\pm$ 451.47 & 25811.09 $\pm$ 500.80 & 23409.21 $\pm$ 162.22 & 28956.80 $\pm$ 323.95 & 37379.66 $\pm$ 999.06\\
\hspace{1em}$\text{pffr ($z_{1-\alpha/2}$)}$ & 0.56 $\pm$ 0.01 & 0.52 $\pm$ 0.01 & 0.61 $\pm$ 0.01 & 0.77 $\pm$ 0.01 & 0.73 $\pm$ 0.01 & 0.84 $\pm$ 0.02\\
\addlinespace[0.6em]
\multicolumn{7}{l}{\textbf{$N = 50,\; n_i = 100$}}\\
\hspace{1em}$\text{1-Step } {\mathbf{R}}_{\text{Lon}} \otimes {\mathbf{R}}_{\text{Fun}}$ & 20.67 $\pm$ 0.12 & 19.37 $\pm$ 0.29 & 25.36 $\pm$ 0.30 & 17.50 $\pm$ 0.26 & 19.53 $\pm$ 0.33 & 27.31 $\pm$ 0.46\\
\hspace{1em}$\text{Full-1Step } {\mathbf{R}}_{\text{Lon}} \otimes {\mathbf{R}}_{\text{Fun}}$ & 52.15 $\pm$ 0.39 & 43.86 $\pm$ 0.81 & 136.11 $\pm$ 6.98 & 31.77 $\pm$ 0.49 & 33.51 $\pm$ 0.58 & 51.30 $\pm$ 0.98\\
\hspace{1em}$\text{Full-Full } \mathbf{R}_{\text{Lon}} \otimes \mathbf{R}_{\text{Fun}}$ &49281.50 $\pm$ 390.53 & 46607.16 $\pm$ 789.80 & 67701.35 $\pm$ 1248.51 & 35692.11 $\pm$ 526.50 & 37210.09 $\pm$ 660.31 & 51646.17 $\pm$ 971.96\\
\hspace{1em}$\text{pffr ($z_{1-\alpha/2}$)}$ & 6.90 $\pm$ 0.03 & 5.74 $\pm$ 0.09 & 6.41 $\pm$ 0.08 & 5.89 $\pm$ 0.09 & 5.71 $\pm$ 0.10 & 6.62 $\pm$ 0.10\\
\bottomrule
\end{tabular}}
\end{table}

\clearpage

\subsection{Additional Simulations} \label{app:Li}
We tested one-step performance in a setting where the outcome was simulated to be correlated in both longitudinal and functional directions (i.e. $\text{Cov}(Y_{i,j}(s_1), Y_{i,j'}(s_2) \mid \mathbf{X}_i) \neq 0$ for $s_1,s_2 \in \mathcal{S}$ and $j, j' \in [n_i]$) with an underlying correlation structure other than the Kronecker product correlation used to generate outcomes in the main text. Namely, we simulated the outcome to have an exchangeable correlation structure in the longitudinal direction, allowing for comparison with the marginal decomposition (``Marginal'') approach proposed in \cite{li_2022}, which models both within- and between-functional observation correlation. For fair comparison, 
we simulated data with their marginal decomposition scheme and code, using the model 
\[Y_{i,j}(s) = \beta_0(s) + X_{1,i}\beta_1(s) + X_{2,i,j} \beta_2(s) + W_{i,j}(s) + \epsilon_{i,j}(s)\]
where $\beta_0(s) = 3 + \text{sin}(\pi s) + \sqrt{2} \text{cos}(3\pi s)$, $\beta_1(s) = 3 + \text{cos}(2\pi s) + \sqrt{2} \text{cos}(3\pi s)$, and $\beta_2(s) = \frac{1}{60} \left[ \phi(\frac{s-0.2}{0.1^2}) + \phi(\frac{s-0.1}{0.07^2})\right ] -\frac{1}{200}\phi(\frac{s-0.35}{0.1^2}) - \frac{1}{250}\phi(\frac{s-0.65}{0.06^2})$.
Based on simulations in \cite{li_2022}, we drew $X_{1,i} \sim N(0,1)$, and $X_{2,i,j} = j + e_{i,j}$, where $e_{i,j} \sim N(\alpha e_{i, j-1}, 1)$, with $e_{i,0}=0$, $\alpha=0.7$. 
We set parameters as in \cite{li_2022}: $W_{i,j}(s) = \sum_{k=1}^2 ( \xi_{i,k} + \zeta_{i,j,k}) \psi_k(s)$ where the orthonormal functions $\psi_1(s) = 1 ~\forall s \in \mathcal{S}$ and $\psi_2(s) = \sqrt{2}\text{sin}(2 \pi s)$, $\xi_{i,1} \overset{\text{iid}}{\sim} N(0,3)$, $\xi_{i,2} \overset{\text{iid}}{\sim} N(0,2)$, $\zeta_{i,j,1} \overset{\text{iid}}{\sim} N(0,1.5)$, $\xi_{i,j,2} \overset{\text{iid}}{\sim} N(0, 1)$, and $\epsilon_{i,j}(s) \overset{\text{iid}}{\sim} N(0, 10)$. 

We describe the Marginal approach from \cite{li_2022} in Appendix~\ref{app:li_method}. We compared its performance with that of the one-step and fully-iterated fGEE under varying degrees of working correlation misspecification (see Appendix~\ref{app:Li_results} for results). 


\subsubsection{Marginal Approach of Li et al., (2022)} \label{app:li_method}
We briefly describe the method proposed in \cite{li_2022}. Specifically, they model the mean as \[ \mu_{i,j}(s) \equiv \mathbb{E}(Y_{i,j}(s) \mid \mathbf{X}_{i,j}) = \beta_0(s) + \sum_{r=1}^q X_{i,j,r} \beta_r(s),\]
and define residuals ${\epsilon}_{i,j}(s)=Y_{i,j}(s) - {\mu}_{i,j}(s)$. The residuals are further decomposed as $\epsilon_{i,j}(s)=w_{i,j}(s)+e_{i,j}(s)$, where $e_{i,j}(s)$ denotes white noise with variance $\sigma^2$. They expand $w_{i,j}(s)$ as $w_{i,j}(s)\approx \sum_{k=1}^{K}\xi_{i,j,k}\phi_k(s)$, where $\phi_k(s)$ forms an orthonormal basis and $\xi_{i,j,k}$ are the associated random coefficients. In \cite{li_2022}, the longitudinal correlation is modeled by imposing different correlation structure on the random scores $\xi_{i,j,k}$, including independence, exchangeable, and unspecified smooth. 

For the model estimation, \cite{li_2022} begin by fitting an initial estimator under a working-independence correlation. In their simulations, they use $\texttt{refund::pffr}$ to fit this model with a second order difference penalty. We used the same initial estimator for our one-step in these simulations to ensure comparability. Next, they estimate marginal eigenfunctions $\phi_k$ (pooled over observations within- and across-clusters) based on the residual curves with the fast FACE algorithm (\texttt{refund::fpca.face}) \citep{fpca.face}, and predict the random scores $\xi_{i,j,k}$. Given the predicted scores, the final estimated marginal covariance can be constructed accordingly. 

Specifically, for an exchangeable longitudinal working correlation, \cite{li_2022} use a one-way nested ANOVA on the scores to estimate the necessary variance components. Plugging these estimated components into a finite-rank mixed-model representation, they refit the model using mixed-effects software $\texttt{mgcv::gam}$/$\texttt{bam}$ to obtain a refined estimate of the functional coefficients. The pointwise CIs for the estimated functional coefficients are constructed through $\texttt{mgcv}$'s approximation of the posterior of the spline coefficients, which is based on a Bayesian interpretation of spline smoothing penalties as priors for the spline regression coefficients.

\subsubsection{Simulation Results} \label{app:Li_results}
Table~\ref{tab:Li_rmse_rel} shows that the one-step yields comparable functional coefficient estimate RMSE as \cite{li_2022}. Similar to the simulations in the main text, the one-step that models correlation only in the longitudinal direction ($\text{1-Step }\widehat{\mathbf{R}}_{\text{Lon}} \otimes I_L$), tends to exhibit the best RMSE. This mirrors the same finding as the simulation results in the main text: modeling $\text{Cov}(Y_{i,j}(s), Y_{i,j'}(s) \mid \mathbf{X}_i)$ at each point $s$ across values of $j, j' \in [n_i]$, is enough to capture efficiency gains, even though the data were simulated such that $\text{Cov}(Y_{i,j}(s), Y_{i,j'}(s') \mid \mathbf{X}_i) \neq 0$, for $s \neq s'$. These simulations were based on the code from \cite{li_2022} which generates outcomes to be correlated in both longitudinal and functional directions with a covariance decomposition that differs from the main text simulations (i.e., these simulations do not use the Kronecker product-based correlation structure used in the main text). 

Table~\ref{tab:Li_pci} shows that the one-step generally exhibits superior empirical pointwise 95\% CI coverage compared to \cite{li_2022}. Their Marginal approach exhibits slightly better pointwise CI coverage compared to the $\text{1-Step }\widehat{\mathbf{R}}_{\text{Lon}} \otimes I_L$ in the smallest sample sizes tested of $N=25$, $n_i=5$. For larger cluster numbers $N$, the pointwise CI coverage of the one-step and Marginal approach were comparable when $n_i=5$. However, the Marginal grows highly anti-conservative for large $n_i$, a feature acknowledged in \cite{li_2022}. For $n_i=100$, the pointwise CI coverage of \cite{li_2022} was often around 0.80. In contrast, the one-step tended to exhibit pointwise coverage close to the nominal 0.95. Table~\ref{tab:Li_pt_width} shows that the pointwise CI coverage width of the one-step was also substantially narrower than that of \cite{li_2022}, showing that the one-step improves statistical efficiency. We are not aware of code to calculate joint CIs with the Marginal approach. 

These simulations also illustrate the scalability of the one-step. Appendix Table~\ref{tab:Li_time} shows that the one-step scales well with both $N$ and $n_i$. In contrast, the Marginal approach is too memory-intensive to fit larger datasets and fit-times scale super-linearly as a function of $N$ and $n_i$. For example, when $N=500$ and $n_i = 100$, their method ran out of memory on a high performance computing cluster (35Gb). The one-step method was dramatically faster than \cite{li_2022} for larger cluster sizes and often comparable in timing for smaller cluster sizes $n_i$. For larger $N$, however, the one-step was often over 10 times faster than \cite{li_2022} (depending on the working correlation structure) even for small cluster sizes ($n_i=5$). 

Together these simulations provide an example of how the scalable correlation structure adopted here still yields gains in statistical and computational efficiency, even when compared to other methods that model correlation in both longitudinal and functional directions. Finally, the one-step and fully-iterated fGEE exhibit comparable results, further supporting our asymptotic theory.

\begin{table}[!h]
\centering
\caption{\label{tab:Li_rmse_rel} \footnotesize Functional Coefficient Estimation Performance (RMSE) of each method relative to the $\texttt{pffr}$ fit ($\text{RMSE}/\text{RMSE}_{\text{pffr}}$).
    Cells contain the average of 300 replicates $\pm$ SE (SE$=0.00$ indicates a value $<0.01$). The ``1-Step'' indicates a one-step was used for tuning and final coefficient estimation, and ``Full-1step'' indicates one-step tuning and a fully-iterated fGEE for final coefficient estimation, with the indicated working correlation. We indicate out-of-memory (35Gb) with symbol --- . }
\centering
\resizebox{\ifdim\width>\linewidth\linewidth\else\width\fi}{!}{
\fontsize{5}{7}\selectfont
\begin{tabular}[t]{>{\raggedright\arraybackslash}p{4.8cm}>{\raggedright\arraybackslash}p{2.1cm}}
\toprule
\multicolumn{1}{c}{ } & \multicolumn{1}{c}{Exchangeable} \\
\cmidrule(l{3pt}r{3pt}){2-2}
Method & Gaussian\\
\midrule
\addlinespace[0em]
\multicolumn{2}{l}{\textbf{$N = 25,\; n_i = 5$}}\\
\hspace{1em}\text{Li et. al } & 0.97 $\pm$ 0.01\\
\hspace{1em}$\text{1-Step } {\mathbf{R}}_{\text{Lon}} \otimes {\mathbf{R}}_{\text{Fun}}$ & 1.00 $\pm$ 0.01\\
\hspace{1em}$\text{1-Step } {\mathbf{R}}_{\text{Lon}} \otimes I_L$ & 0.96 $\pm$ 0.00\\
\hspace{1em}$\text{1-Step } I_{n_i} \otimes {\mathbf{R}}_{\text{Fun}}$ & 1.03 $\pm$ \vphantom{1} 0.01\\
\hspace{1em}$\text{1-Step } I_{n_i} \otimes I_L$ & 1.00 $\pm$ \vphantom{4} 0.00\\
\hspace{1em}$\text{Full-1Step } {\mathbf{R}}_{\text{Lon}} \otimes {\mathbf{R}}_{\text{Fun}}$ & 1.01 $\pm$ 0.01\\
\hspace{1em}$\text{Full-1Step } {\mathbf{R}}_{\text{Lon}} \otimes I_L$ & 0.96 $\pm$ 0.00\\
\hspace{1em}$\text{Full-1Step } I_{n_i} \otimes {\mathbf{R}}_{\text{Fun}}$ & 1.04 $\pm$ 0.01\\
\hspace{1em}$\text{Full-1Step } I_{n_i} \otimes I_L$ & 1.00 $\pm$ \vphantom{4} 0.00\\
\addlinespace[0.6em]
\multicolumn{2}{l}{\textbf{$N = 25,\; n_i = 100$}}\\
\hspace{1em}\text{Li et. al } & 1.05 $\pm$ 0.01\\
\hspace{1em}$\text{1-Step } {\mathbf{R}}_{\text{Lon}} \otimes {\mathbf{R}}_{\text{Fun}}$ & 1.02 $\pm$ 0.01\\
\hspace{1em}$\text{1-Step } {\mathbf{R}}_{\text{Lon}} \otimes I_L$ & 0.99 $\pm$ \vphantom{1} 0.00\\
\hspace{1em}$\text{1-Step } I_{n_i} \otimes {\mathbf{R}}_{\text{Fun}}$ & 1.03 $\pm$ 0.01\\
\hspace{1em}$\text{1-Step } I_{n_i} \otimes I_L$ & 1.00 $\pm$ \vphantom{3} 0.00\\
\hspace{1em}$\text{Full-1Step } {\mathbf{R}}_{\text{Lon}} \otimes {\mathbf{R}}_{\text{Fun}}$ & 1.02 $\pm$ 0.01\\
\hspace{1em}$\text{Full-1Step } {\mathbf{R}}_{\text{Lon}} \otimes I_L$ & 0.99 $\pm$ \vphantom{1} 0.00\\
\hspace{1em}$\text{Full-1Step } I_{n_i} \otimes {\mathbf{R}}_{\text{Fun}}$ & 1.03 $\pm$ 0.01\\
\hspace{1em}$\text{Full-1Step } I_{n_i} \otimes I_L$ & 1.00 $\pm$ \vphantom{3} 0.00\\
\addlinespace[0.6em]
\multicolumn{2}{l}{\textbf{$N = 50,\; n_i = 5$}}\\
\hspace{1em}\text{Li et. al } & 0.96 $\pm$ 0.01\\
\hspace{1em}$\text{1-Step } {\mathbf{R}}_{\text{Lon}} \otimes {\mathbf{R}}_{\text{Fun}}$ & 0.99 $\pm$ 0.01\\
\hspace{1em}$\text{1-Step } {\mathbf{R}}_{\text{Lon}} \otimes I_L$ & 0.95 $\pm$ 0.01\\
\hspace{1em}$\text{1-Step } I_{n_i} \otimes {\mathbf{R}}_{\text{Fun}}$ & 1.02 $\pm$ 0.01\\
\hspace{1em}$\text{1-Step } I_{n_i} \otimes I_L$ & 1.00 $\pm$ \vphantom{2} 0.00\\
\hspace{1em}$\text{Full-1Step } {\mathbf{R}}_{\text{Lon}} \otimes {\mathbf{R}}_{\text{Fun}}$ & 0.99 $\pm$ 0.01\\
\hspace{1em}$\text{Full-1Step } {\mathbf{R}}_{\text{Lon}} \otimes I_L$ & 0.95 $\pm$ 0.01\\
\hspace{1em}$\text{Full-1Step } I_{n_i} \otimes {\mathbf{R}}_{\text{Fun}}$ & 1.02 $\pm$ 0.01\\
\hspace{1em}$\text{Full-1Step } I_{n_i} \otimes I_L$ & 1.00 $\pm$ \vphantom{2} 0.00\\
\addlinespace[0.6em]
\multicolumn{2}{l}{\textbf{$N = 50,\; n_i = 100$}}\\
\hspace{1em}\text{Li et. al } & 1.01 $\pm$ 0.01\\
\hspace{1em}$\text{1-Step } {\mathbf{R}}_{\text{Lon}} \otimes {\mathbf{R}}_{\text{Fun}}$ & 1.00 $\pm$ \vphantom{1} 0.00\\
\hspace{1em}$\text{1-Step } {\mathbf{R}}_{\text{Lon}} \otimes I_L$ & 0.99 $\pm$ 0.00\\
\hspace{1em}$\text{1-Step } I_{n_i} \otimes {\mathbf{R}}_{\text{Fun}}$ & 1.01 $\pm$ 0.00\\
\hspace{1em}$\text{1-Step } I_{n_i} \otimes I_L$ & 1.00 $\pm$ \vphantom{1} 0.00\\
\hspace{1em}$\text{Full-1Step } {\mathbf{R}}_{\text{Lon}} \otimes {\mathbf{R}}_{\text{Fun}}$ & 1.00 $\pm$ \vphantom{1} 0.00\\
\hspace{1em}$\text{Full-1Step } {\mathbf{R}}_{\text{Lon}} \otimes I_L$ & 0.99 $\pm$ 0.00\\
\hspace{1em}$\text{Full-1Step } I_{n_i} \otimes {\mathbf{R}}_{\text{Fun}}$ & 1.01 $\pm$ 0.00\\
\hspace{1em}$\text{Full-1Step } I_{n_i} \otimes I_L$ & 1.00 $\pm$ \vphantom{1} 0.00\\
\addlinespace[0.6em]
\multicolumn{2}{l}{\textbf{$N = 500,\; n_i = 5$}}\\
\hspace{1em}\text{Li et. al } & 0.94 $\pm$ 0.01\\
\hspace{1em}$\text{1-Step } {\mathbf{R}}_{\text{Lon}} \otimes {\mathbf{R}}_{\text{Fun}}$ & 0.96 $\pm$ 0.00\\
\hspace{1em}$\text{1-Step } {\mathbf{R}}_{\text{Lon}} \otimes I_L$ & 0.95 $\pm$ 0.00\\
\hspace{1em}$\text{1-Step } I_{n_i} \otimes {\mathbf{R}}_{\text{Fun}}$ & 1.00 $\pm$ \vphantom{1} 0.00\\
\hspace{1em}$\text{1-Step } I_{n_i} \otimes I_L$ & 0.99 $\pm$ 0.00\\
\hspace{1em}$\text{Full-1Step } {\mathbf{R}}_{\text{Lon}} \otimes {\mathbf{R}}_{\text{Fun}}$ & 0.96 $\pm$ 0.01\\
\hspace{1em}$\text{Full-1Step } {\mathbf{R}}_{\text{Lon}} \otimes I_L$ & 0.95 $\pm$ 0.00\\
\hspace{1em}$\text{Full-1Step } I_{n_i} \otimes {\mathbf{R}}_{\text{Fun}}$ & 1.00 $\pm$ \vphantom{1} 0.00\\
\hspace{1em}$\text{Full-1Step } I_{n_i} \otimes I_L$ & 0.99 $\pm$ 0.00\\
\addlinespace[0.6em]
\multicolumn{2}{l}{\textbf{$N = 500,\; n_i = 100$}}\\
\hspace{1em}\text{Li et. al } & ~~~~~~---~~~~~~\\
\hspace{1em}$\text{1-Step } {\mathbf{R}}_{\text{Lon}} \otimes {\mathbf{R}}_{\text{Fun}}$ & 1.00 $\pm$ 0.00\\
\hspace{1em}$\text{1-Step } {\mathbf{R}}_{\text{Lon}} \otimes I_L$ & 1.00 $\pm$ 0.01\\
\hspace{1em}$\text{1-Step } I_{n_i} \otimes {\mathbf{R}}_{\text{Fun}}$ & 1.00 $\pm$ 0.00\\
\hspace{1em}$\text{1-Step } I_{n_i} \otimes I_L$ & 1.00 $\pm$ 0.00\\
\hspace{1em}$\text{Full-1Step } {\mathbf{R}}_{\text{Lon}} \otimes {\mathbf{R}}_{\text{Fun}}$ & 1.00 $\pm$ 0.00\\
\hspace{1em}$\text{Full-1Step } {\mathbf{R}}_{\text{Lon}} \otimes I_L$ & 1.00 $\pm$ 0.01\\
\hspace{1em}$\text{Full-1Step } I_{n_i} \otimes {\mathbf{R}}_{\text{Fun}}$ & 1.00 $\pm$ 0.00\\
\hspace{1em}$\text{Full-1Step } I_{n_i} \otimes I_L$ & 1.00 $\pm$ 0.00\\
\bottomrule
\end{tabular}}
\end{table}

\begin{table}[!h]
\centering
\caption{\label{tab:Li_pci} \footnotesize
Pointwise 95\% CI coverage from 300 replicates $\pm$ SE (SE$=0.00$ indicates a value $<0.01$).
Table columns indicate if $\mathbf{R}^*_{\text{Lon}}$ had exchangeable or AR1 correlation. pffr (Wild) indicates CIs constructed for pffr coefficient estimates with quantiles obtained from a wild cluster bootstrap. pffr ($z_{1-\alpha/2}$) are standard Wald CIs constructed with Gaussian quantiles. We indicate out-of-memory (35Gb) with symbol --- .}
\centering
\resizebox{\ifdim\width>\linewidth\linewidth\else\width\fi}{!}{
\fontsize{5}{7}\selectfont
\begin{tabular}[t]{>{\raggedright\arraybackslash}p{4.8cm}>{\raggedright\arraybackslash}p{2.1cm}}
\toprule
\multicolumn{1}{c}{ } & \multicolumn{1}{c}{Exchangeable} \\
\cmidrule(l{3pt}r{3pt}){2-2}
Method & Gaussian\\
\midrule
\addlinespace[0em]
\multicolumn{2}{l}{\textbf{$N = 25,\; n_i = 5$}}\\
\hspace{1em}\text{Li et. al } & 0.95 $\pm$ \vphantom{1} 0.01\\
\hspace{1em}$\text{1-Step } {\mathbf{R}}_{\text{Lon}} \otimes {\mathbf{R}}_{\text{Fun}}$ & 0.88 $\pm$ 0.02\\
\hspace{1em}$\text{1-Step } {\mathbf{R}}_{\text{Lon}} \otimes I_L$ & 0.92 $\pm$ 0.02\\
\hspace{1em}$\text{1-Step } I_{n_i} \otimes {\mathbf{R}}_{\text{Fun}}$ & 0.87 $\pm$ 0.02\\
\hspace{1em}$\text{1-Step } I_{n_i} \otimes I_L$ & 0.91 $\pm$ 0.02\\
\hspace{1em}$\text{Full-1Step } {\mathbf{R}}_{\text{Lon}} \otimes {\mathbf{R}}_{\text{Fun}}$ & 0.88 $\pm$ 0.02\\
\hspace{1em}$\text{Full-1Step } {\mathbf{R}}_{\text{Lon}} \otimes I_L$ & 0.92 $\pm$ 0.02\\
\hspace{1em}$\text{Full-1Step } I_{n_i} \otimes {\mathbf{R}}_{\text{Fun}}$ & 0.87 $\pm$ 0.02\\
\hspace{1em}$\text{Full-1Step } I_{n_i} \otimes I_L$ & 0.91 $\pm$ 0.02\\
\hspace{1em}$\text{pffr (Wild)}$ & 0.97 $\pm$ \vphantom{1} 0.01\\
\hspace{1em}$\text{pffr ($z_{1-\alpha/2}$)}$ & 0.88 $\pm$ 0.02\\
\addlinespace[0.6em]
\multicolumn{2}{l}{\textbf{$N = 25,\; n_i = 100$}}\\
\hspace{1em}\text{Li et. al } & 0.82 $\pm$ 0.02\\
\hspace{1em}$\text{1-Step } {\mathbf{R}}_{\text{Lon}} \otimes {\mathbf{R}}_{\text{Fun}}$ & 0.94 $\pm$ 0.01\\
\hspace{1em}$\text{1-Step } {\mathbf{R}}_{\text{Lon}} \otimes I_L$ & 0.94 $\pm$ \vphantom{1} 0.01\\
\hspace{1em}$\text{1-Step } I_{n_i} \otimes {\mathbf{R}}_{\text{Fun}}$ & 0.93 $\pm$ 0.02\\
\hspace{1em}$\text{1-Step } I_{n_i} \otimes I_L$ & 0.93 $\pm$ 0.01\\
\hspace{1em}$\text{Full-1Step } {\mathbf{R}}_{\text{Lon}} \otimes {\mathbf{R}}_{\text{Fun}}$ & 0.94 $\pm$ 0.01\\
\hspace{1em}$\text{Full-1Step } {\mathbf{R}}_{\text{Lon}} \otimes I_L$ & 0.94 $\pm$ \vphantom{1} 0.01\\
\hspace{1em}$\text{Full-1Step } I_{n_i} \otimes {\mathbf{R}}_{\text{Fun}}$ & 0.93 $\pm$ 0.02\\
\hspace{1em}$\text{Full-1Step } I_{n_i} \otimes I_L$ & 0.93 $\pm$ 0.01\\
\hspace{1em}$\text{pffr (Wild)}$ & 0.97 $\pm$ 0.01\\
\hspace{1em}$\text{pffr ($z_{1-\alpha/2}$)}$ & 0.92 $\pm$ \vphantom{2} 0.02\\
\addlinespace[0.6em]
\multicolumn{2}{l}{\textbf{$N = 50,\; n_i = 5$}}\\
\hspace{1em}\text{Li et. al } & 0.95 $\pm$ 0.01\\
\hspace{1em}$\text{1-Step } {\mathbf{R}}_{\text{Lon}} \otimes {\mathbf{R}}_{\text{Fun}}$ & 0.93 $\pm$ 0.02\\
\hspace{1em}$\text{1-Step } {\mathbf{R}}_{\text{Lon}} \otimes I_L$ & 0.94 $\pm$ 0.01\\
\hspace{1em}$\text{1-Step } I_{n_i} \otimes {\mathbf{R}}_{\text{Fun}}$ & 0.93 $\pm$ \vphantom{1} 0.01\\
\hspace{1em}$\text{1-Step } I_{n_i} \otimes I_L$ & 0.94 $\pm$ \vphantom{1} 0.01\\
\hspace{1em}$\text{Full-1Step } {\mathbf{R}}_{\text{Lon}} \otimes {\mathbf{R}}_{\text{Fun}}$ & 0.93 $\pm$ 0.02\\
\hspace{1em}$\text{Full-1Step } {\mathbf{R}}_{\text{Lon}} \otimes I_L$ & 0.94 $\pm$ 0.01\\
\hspace{1em}$\text{Full-1Step } I_{n_i} \otimes {\mathbf{R}}_{\text{Fun}}$ & 0.93 $\pm$ \vphantom{1} 0.01\\
\hspace{1em}$\text{Full-1Step } I_{n_i} \otimes I_L$ & 0.94 $\pm$ \vphantom{1} 0.01\\
\hspace{1em}$\text{pffr (Wild)}$ & 0.98 $\pm$ \vphantom{1} 0.01\\
\hspace{1em}$\text{pffr ($z_{1-\alpha/2}$)}$ & 0.92 $\pm$ \vphantom{1} 0.02\\
\addlinespace[0.6em]
\multicolumn{2}{l}{\textbf{$N = 50,\; n_i = 100$}}\\
\hspace{1em}\text{Li et. al } & 0.80 $\pm$ 0.02\\
\hspace{1em}$\text{1-Step } {\mathbf{R}}_{\text{Lon}} \otimes {\mathbf{R}}_{\text{Fun}}$ & 0.93 $\pm$ 0.01\\
\hspace{1em}$\text{1-Step } {\mathbf{R}}_{\text{Lon}} \otimes I_L$ & 0.93 $\pm$ 0.01\\
\hspace{1em}$\text{1-Step } I_{n_i} \otimes {\mathbf{R}}_{\text{Fun}}$ & 0.93 $\pm$ 0.01\\
\hspace{1em}$\text{1-Step } I_{n_i} \otimes I_L$ & 0.94 $\pm$ 0.01\\
\hspace{1em}$\text{Full-1Step } {\mathbf{R}}_{\text{Lon}} \otimes {\mathbf{R}}_{\text{Fun}}$ & 0.93 $\pm$ 0.01\\
\hspace{1em}$\text{Full-1Step } {\mathbf{R}}_{\text{Lon}} \otimes I_L$ & 0.93 $\pm$ 0.01\\
\hspace{1em}$\text{Full-1Step } I_{n_i} \otimes {\mathbf{R}}_{\text{Fun}}$ & 0.93 $\pm$ 0.01\\
\hspace{1em}$\text{Full-1Step } I_{n_i} \otimes I_L$ & 0.94 $\pm$ 0.01\\
\hspace{1em}$\text{pffr (Wild)}$ & 0.98 $\pm$ 0.01\\
\hspace{1em}$\text{pffr ($z_{1-\alpha/2}$)}$ & 0.92 $\pm$ 0.02\\
\addlinespace[0.6em]
\multicolumn{2}{l}{\textbf{$N = 500,\; n_i = 5$}}\\
\hspace{1em}\text{Li et. al } & 0.96 $\pm$ 0.01\\
\hspace{1em}$\text{1-Step } {\mathbf{R}}_{\text{Lon}} \otimes {\mathbf{R}}_{\text{Fun}}$ & 0.95 $\pm$ 0.01\\
\hspace{1em}$\text{1-Step } {\mathbf{R}}_{\text{Lon}} \otimes I_L$ & 0.95 $\pm$ 0.01\\
\hspace{1em}$\text{1-Step } I_{n_i} \otimes {\mathbf{R}}_{\text{Fun}}$ & 0.94 $\pm$ 0.01\\
\hspace{1em}$\text{1-Step } I_{n_i} \otimes I_L$ & 0.95 $\pm$ 0.01\\
\hspace{1em}$\text{Full-1Step } {\mathbf{R}}_{\text{Lon}} \otimes {\mathbf{R}}_{\text{Fun}}$ & 0.95 $\pm$ 0.01\\
\hspace{1em}$\text{Full-1Step } {\mathbf{R}}_{\text{Lon}} \otimes I_L$ & 0.95 $\pm$ 0.01\\
\hspace{1em}$\text{Full-1Step } I_{n_i} \otimes {\mathbf{R}}_{\text{Fun}}$ & 0.94 $\pm$ 0.01\\
\hspace{1em}$\text{Full-1Step } I_{n_i} \otimes I_L$ & 0.95 $\pm$ 0.01\\
\hspace{1em}$\text{pffr (Wild)}$ & 0.99 $\pm$ 0.01\\
\hspace{1em}$\text{pffr ($z_{1-\alpha/2}$)}$ & 0.91 $\pm$ 0.02\\
\addlinespace[0.6em]
\multicolumn{2}{l}{\textbf{$N = 500,\; n_i = 100$}}\\
\hspace{1em}\text{Li et. al } & ~~~~~~---~~~~~~\\
\hspace{1em}$\text{1-Step } {\mathbf{R}}_{\text{Lon}} \otimes {\mathbf{R}}_{\text{Fun}}$ & 0.95 $\pm$ 0.04\\
\hspace{1em}$\text{1-Step } {\mathbf{R}}_{\text{Lon}} \otimes I_L$ & 0.95 $\pm$ 0.04\\
\hspace{1em}$\text{1-Step } I_{n_i} \otimes {\mathbf{R}}_{\text{Fun}}$ & 0.96 $\pm$ 0.04\\
\hspace{1em}$\text{1-Step } I_{n_i} \otimes I_L$ & 0.95 $\pm$ 0.04\\
\hspace{1em}$\text{Full-1Step } {\mathbf{R}}_{\text{Lon}} \otimes {\mathbf{R}}_{\text{Fun}}$ & 0.95 $\pm$ 0.04\\
\hspace{1em}$\text{Full-1Step } {\mathbf{R}}_{\text{Lon}} \otimes I_L$ & 0.95 $\pm$ 0.04\\
\hspace{1em}$\text{Full-1Step } I_{n_i} \otimes {\mathbf{R}}_{\text{Fun}}$ & 0.96 $\pm$ 0.04\\
\hspace{1em}$\text{Full-1Step } I_{n_i} \otimes I_L$ & 0.95 $\pm$ 0.04\\
\hspace{1em}$\text{pffr (Wild)}$ & 1.00 $\pm$ 0.00\\
\hspace{1em}$\text{pffr ($z_{1-\alpha/2}$)}$ & 0.94 $\pm$ 0.04\\
\bottomrule
\end{tabular}}
\end{table}

\begin{table}[!h]
\centering
\caption{\label{tab:Li_pt_width} \footnotesize Relative pointwise CI width (mean $\pm$ SE) vs.\ pffr (Wild). SE$=0.00$ indicates a value $<0.01$. We denote $UB^{(r)}(s)$ and $LB^{(r)}(s)$ and $UB_\text{pffr}^{(r)}(s)/LB_\text{pffr}^{(r)}(s)$ as the upper/lower bounds of the CIs (at $s$) of the indicated method and pffr, respectively. Below we report the average ratio $\frac{1}{(q+1)|\mathcal{S}|}\sum_{r=0}^q\sum_{s \in \mathcal{S}}\frac{UB^{(r)}(s) -LB(s)^{(r)}}{UB^{(r)}_\text{pffr}(s) -LB^{(r)}_\text{pffr}(s)}$ across 300 simulation replicates. Values $<1$ indicate narrower 95\% CIs. Values with $*$ indicate extreme outliers (from poor estimates) were removed from the average of that cell to avoid skewing results. We indicate out-of-memory (35Gb) with symbol --- .}
\centering
\resizebox{\ifdim\width>\linewidth\linewidth\else\width\fi}{!}{
\fontsize{5}{7}\selectfont
\begin{tabular}[t]{>{\raggedright\arraybackslash}p{4.8cm}>{\raggedright\arraybackslash}p{2.1cm}}
\toprule
\multicolumn{1}{c}{ } & \multicolumn{1}{c}{Exchangeable} \\
\cmidrule(l{3pt}r{3pt}){2-2}
Method & Gaussian\\
\midrule
\addlinespace[0em]
\multicolumn{2}{l}{\textbf{$N = 25,\; n_i = 5$}}\\
\hspace{1em}\text{Li et. al } & 0.69 $\pm$ 0.01\\
\hspace{1em}$\text{1-Step } {\mathbf{R}}_{\text{Lon}} \otimes {\mathbf{R}}_{\text{Fun}}$ & 0.57 $\pm$ 0.00\\
\hspace{1em}$\text{1-Step } {\mathbf{R}}_{\text{Lon}} \otimes I_L$ & 0.60 $\pm$ 0.00\\
\hspace{1em}$\text{1-Step } I_{n_i} \otimes {\mathbf{R}}_{\text{Fun}}$ & 0.62 $\pm$ 0.00\\
\hspace{1em}$\text{1-Step } I_{n_i} \otimes I_L$ & 0.69 $\pm$ 0.00\\
\hspace{1em}$\text{Full-1Step } {\mathbf{R}}_{\text{Lon}} \otimes {\mathbf{R}}_{\text{Fun}}$ & 0.57 $\pm$ 0.00\\
\hspace{1em}$\text{Full-1Step } {\mathbf{R}}_{\text{Lon}} \otimes I_L$ & 0.60 $\pm$ 0.00\\
\hspace{1em}$\text{Full-1Step } I_{n_i} \otimes {\mathbf{R}}_{\text{Fun}}$ & 0.62 $\pm$ 0.00\\
\hspace{1em}$\text{Full-1Step } I_{n_i} \otimes I_L$ & 0.69 $\pm$ 0.00\\
\addlinespace[0.6em]
\multicolumn{2}{l}{\textbf{$N = 25,\; n_i = 100$}}\\
\hspace{1em}\text{Li et. al } & 0.75 $\pm$ 0.01\\
\hspace{1em}$\text{1-Step } {\mathbf{R}}_{\text{Lon}} \otimes {\mathbf{R}}_{\text{Fun}}$ & 0.68 $\pm$ 0.00\\
\hspace{1em}$\text{1-Step } {\mathbf{R}}_{\text{Lon}} \otimes I_L$ & 0.67 $\pm$ 0.00\\
\hspace{1em}$\text{1-Step } I_{n_i} \otimes {\mathbf{R}}_{\text{Fun}}$ & 0.82 $\pm$ 0.00\\
\hspace{1em}$\text{1-Step } I_{n_i} \otimes I_L$ & 0.84 $\pm$ 0.00\\
\hspace{1em}$\text{Full-1Step } {\mathbf{R}}_{\text{Lon}} \otimes {\mathbf{R}}_{\text{Fun}}$ & 0.68 $\pm$ 0.00\\
\hspace{1em}$\text{Full-1Step } {\mathbf{R}}_{\text{Lon}} \otimes I_L$ & 0.67 $\pm$ 0.00\\
\hspace{1em}$\text{Full-1Step } I_{n_i} \otimes {\mathbf{R}}_{\text{Fun}}$ & 0.82 $\pm$ 0.00\\
\hspace{1em}$\text{Full-1Step } I_{n_i} \otimes I_L$ & 0.84 $\pm$ 0.00\\
\addlinespace[0.6em]
\multicolumn{2}{l}{\textbf{$N = 50,\; n_i = 5$}}\\
\hspace{1em}\text{Li et. al } & 0.57 $\pm$ 0.00\\
\hspace{1em}$\text{1-Step } {\mathbf{R}}_{\text{Lon}} \otimes {\mathbf{R}}_{\text{Fun}}$ & 0.50 $\pm$ 0.00\\
\hspace{1em}$\text{1-Step } {\mathbf{R}}_{\text{Lon}} \otimes I_L$ & 0.51 $\pm$ 0.00\\
\hspace{1em}$\text{1-Step } I_{n_i} \otimes {\mathbf{R}}_{\text{Fun}}$ & 0.57 $\pm$ 0.00\\
\hspace{1em}$\text{1-Step } I_{n_i} \otimes I_L$ & 0.61 $\pm$ 0.00\\
\hspace{1em}$\text{Full-1Step } {\mathbf{R}}_{\text{Lon}} \otimes {\mathbf{R}}_{\text{Fun}}$ & 0.50 $\pm$ 0.00\\
\hspace{1em}$\text{Full-1Step } {\mathbf{R}}_{\text{Lon}} \otimes I_L$ & 0.51 $\pm$ 0.00\\
\hspace{1em}$\text{Full-1Step } I_{n_i} \otimes {\mathbf{R}}_{\text{Fun}}$ & 0.57 $\pm$ 0.00\\
\hspace{1em}$\text{Full-1Step } I_{n_i} \otimes I_L$ & 0.61 $\pm$ 0.00\\
\addlinespace[0.6em]
\multicolumn{2}{l}{\textbf{$N = 50,\; n_i = 100$}}\\
\hspace{1em}\text{Li et. al } & 0.63 $\pm$ 0.01\\
\hspace{1em}$\text{1-Step } {\mathbf{R}}_{\text{Lon}} \otimes {\mathbf{R}}_{\text{Fun}}$ & 0.58 $\pm$ 0.00\\
\hspace{1em}$\text{1-Step } {\mathbf{R}}_{\text{Lon}} \otimes I_L$ & 0.57 $\pm$ 0.00\\
\hspace{1em}$\text{1-Step } I_{n_i} \otimes {\mathbf{R}}_{\text{Fun}}$ & 0.74 $\pm$ 0.00\\
\hspace{1em}$\text{1-Step } I_{n_i} \otimes I_L$ & 0.75 $\pm$ 0.00\\
\hspace{1em}$\text{Full-1Step } {\mathbf{R}}_{\text{Lon}} \otimes {\mathbf{R}}_{\text{Fun}}$ & 0.58 $\pm$ 0.00\\
\hspace{1em}$\text{Full-1Step } {\mathbf{R}}_{\text{Lon}} \otimes I_L$ & 0.57 $\pm$ 0.00\\
\hspace{1em}$\text{Full-1Step } I_{n_i} \otimes {\mathbf{R}}_{\text{Fun}}$ & 0.74 $\pm$ 0.00\\
\hspace{1em}$\text{Full-1Step } I_{n_i} \otimes I_L$ & 0.75 $\pm$ 0.00\\
\addlinespace[0.6em]
\multicolumn{2}{l}{\textbf{$N = 500,\; n_i = 5$}}\\
\hspace{1em}\text{Li et. al } & 0.30 $\pm$ 0.00\\
\hspace{1em}$\text{1-Step } {\mathbf{R}}_{\text{Lon}} \otimes {\mathbf{R}}_{\text{Fun}}$ & 0.30 $\pm$ 0.00\\
\hspace{1em}$\text{1-Step } {\mathbf{R}}_{\text{Lon}} \otimes I_L$ & 0.30 $\pm$ 0.00\\
\hspace{1em}$\text{1-Step } I_{n_i} \otimes {\mathbf{R}}_{\text{Fun}}$ & 0.38 $\pm$ 0.00\\
\hspace{1em}$\text{1-Step } I_{n_i} \otimes I_L$ & 0.39 $\pm$ 0.00\\
\hspace{1em}$\text{Full-1Step } {\mathbf{R}}_{\text{Lon}} \otimes {\mathbf{R}}_{\text{Fun}}$ & 0.30 $\pm$ 0.00\\
\hspace{1em}$\text{Full-1Step } {\mathbf{R}}_{\text{Lon}} \otimes I_L$ & 0.30 $\pm$ 0.00\\
\hspace{1em}$\text{Full-1Step } I_{n_i} \otimes {\mathbf{R}}_{\text{Fun}}$ & 0.38 $\pm$ 0.00\\
\hspace{1em}$\text{Full-1Step } I_{n_i} \otimes I_L$ & 0.39 $\pm$ 0.00\\
\addlinespace[0.6em]
\multicolumn{2}{l}{\textbf{$N = 500,\; n_i = 100$}}\\
\hspace{1em}\text{Li et. al } & ~~~~~~---~~~~~~\\
\hspace{1em}$\text{1-Step } {\mathbf{R}}_{\text{Lon}} \otimes {\mathbf{R}}_{\text{Fun}}$ & 0.29 $\pm$ 0.00\\
\hspace{1em}$\text{1-Step } {\mathbf{R}}_{\text{Lon}} \otimes I_L$ & 0.29 $\pm$ 0.00\\
\hspace{1em}$\text{1-Step } I_{n_i} \otimes {\mathbf{R}}_{\text{Fun}}$ & 0.43 $\pm$ 0.00\\
\hspace{1em}$\text{1-Step } I_{n_i} \otimes I_L$ & 0.45 $\pm$ 0.00\\
\hspace{1em}$\text{Full-1Step } {\mathbf{R}}_{\text{Lon}} \otimes {\mathbf{R}}_{\text{Fun}}$ & 0.29 $\pm$ 0.00\\
\hspace{1em}$\text{Full-1Step } {\mathbf{R}}_{\text{Lon}} \otimes I_L$ & 0.29 $\pm$ 0.00\\
\hspace{1em}$\text{Full-1Step } I_{n_i} \otimes {\mathbf{R}}_{\text{Fun}}$ & 0.43 $\pm$ 0.00\\
\hspace{1em}$\text{Full-1Step } I_{n_i} \otimes I_L$ & 0.45 $\pm$ 0.00\\
\bottomrule
\end{tabular}}
\end{table}

\begin{table}[!h]
\centering
\caption{\label{tab:Li_time}  \footnotesize Computation time in seconds (mean $\pm$ SE) averaged across 300 simulation replicates. SE$=0.00$ indicates a value $<0.01$. We indicate out-of-memory (35Gb) with symbol --- . pffr (Wild) indicates CIs constructed for pffr coefficient estimates with quantiles obtained from a wild cluster bootstrap. pffr ($z_{1-\alpha/2}$) are standard Wald CIs constructed with Gaussian quantiles.}
\centering
\resizebox{\ifdim\width>\linewidth\linewidth\else\width\fi}{!}{
\fontsize{5}{7}\selectfont
\begin{tabular}[t]{>{\raggedright\arraybackslash}p{4.8cm}>{\raggedright\arraybackslash}p{2.1cm}}
\toprule
\multicolumn{1}{c}{ } & \multicolumn{1}{c}{Exchangeable} \\
\cmidrule(l{3pt}r{3pt}){2-2}
Method & Gaussian\\
\midrule
\addlinespace[0em]
\multicolumn{2}{l}{\textbf{$N = 25,\; n_i = 5$}}\\
\hspace{1em}\text{Li et. al } & 0.48 $\pm$ 0.01\\
\hspace{1em}$\text{1-Step } {\mathbf{R}}_{\text{Lon}} \otimes {\mathbf{R}}_{\text{Fun}}$ & 1.75 $\pm$ 0.01\\
\hspace{1em}$\text{1-Step } {\mathbf{R}}_{\text{Lon}} \otimes I_L$ & 4.75 $\pm$ 0.03\\
\hspace{1em}$\text{1-Step } I_{n_i} \otimes {\mathbf{R}}_{\text{Fun}}$ & 1.36 $\pm$ 0.01\\
\hspace{1em}$\text{1-Step } I_{n_i} \otimes I_L$ & 0.70 $\pm$ 0.00\\
\hspace{1em}$\text{Full-1Step } {\mathbf{R}}_{\text{Lon}} \otimes {\mathbf{R}}_{\text{Fun}}$ & 5.39 $\pm$ 0.16\\
\hspace{1em}$\text{Full-1Step } {\mathbf{R}}_{\text{Lon}} \otimes I_L$ & 11.40 $\pm$ 0.11\\
\hspace{1em}$\text{Full-1Step } I_{n_i} \otimes {\mathbf{R}}_{\text{Fun}}$ & 3.25 $\pm$ 0.15\\
\hspace{1em}$\text{Full-1Step } I_{n_i} \otimes I_L$ & 0.84 $\pm$ 0.01\\
\hspace{1em}$\text{pffr (Wild)}$ & 2.64 $\pm$ 0.01\\
\hspace{1em}$\text{pffr ($z_{1-\alpha/2}$)}$ & 0.40 $\pm$ 0.00\\
\addlinespace[0.6em]
\multicolumn{2}{l}{\textbf{$N = 25,\; n_i = 100$}}\\
\hspace{1em}\text{Li et. al } & 36.28 $\pm$ 0.38\\
\hspace{1em}$\text{1-Step } {\mathbf{R}}_{\text{Lon}} \otimes {\mathbf{R}}_{\text{Fun}}$ & 12.29 $\pm$ 0.20\\
\hspace{1em}$\text{1-Step } {\mathbf{R}}_{\text{Lon}} \otimes I_L$ & 9.38 $\pm$ 0.15\\
\hspace{1em}$\text{1-Step } I_{n_i} \otimes {\mathbf{R}}_{\text{Fun}}$ & 13.12 $\pm$ 0.21\\
\hspace{1em}$\text{1-Step } I_{n_i} \otimes I_L$ & 5.06 $\pm$ 0.06\\
\hspace{1em}$\text{Full-1Step } {\mathbf{R}}_{\text{Lon}} \otimes {\mathbf{R}}_{\text{Fun}}$ & 23.71 $\pm$ 0.43\\
\hspace{1em}$\text{Full-1Step } {\mathbf{R}}_{\text{Lon}} \otimes I_L$ & 16.62 $\pm$ 0.29\\
\hspace{1em}$\text{Full-1Step } I_{n_i} \otimes {\mathbf{R}}_{\text{Fun}}$ & 25.94 $\pm$ 0.45\\
\hspace{1em}$\text{Full-1Step } I_{n_i} \otimes I_L$ & 5.42 $\pm$ 0.08\\
\hspace{1em}$\text{pffr (Wild)}$ & 9.30 $\pm$ 0.13\\
\hspace{1em}$\text{pffr ($z_{1-\alpha/2}$)}$ & 5.17 $\pm$ 0.06\\
\addlinespace[0.6em]
\multicolumn{2}{l}{\textbf{$N = 50,\; n_i = 5$}}\\
\hspace{1em}\text{Li et. al } & 0.94 $\pm$ 0.01\\
\hspace{1em}$\text{1-Step } {\mathbf{R}}_{\text{Lon}} \otimes {\mathbf{R}}_{\text{Fun}}$ & 3.18 $\pm$ 0.01\\
\hspace{1em}$\text{1-Step } {\mathbf{R}}_{\text{Lon}} \otimes I_L$ & 7.66 $\pm$ 0.01\\
\hspace{1em}$\text{1-Step } I_{n_i} \otimes {\mathbf{R}}_{\text{Fun}}$ & 2.10 $\pm$ 0.01\\
\hspace{1em}$\text{1-Step } I_{n_i} \otimes I_L$ & 0.87 $\pm$ 0.00\\
\hspace{1em}$\text{Full-1Step } {\mathbf{R}}_{\text{Lon}} \otimes {\mathbf{R}}_{\text{Fun}}$ & 7.47 $\pm$ 0.04\\
\hspace{1em}$\text{Full-1Step } {\mathbf{R}}_{\text{Lon}} \otimes I_L$ & 17.50 $\pm$ 0.11\\
\hspace{1em}$\text{Full-1Step } I_{n_i} \otimes {\mathbf{R}}_{\text{Fun}}$ & 4.65 $\pm$ 0.02\\
\hspace{1em}$\text{Full-1Step } I_{n_i} \otimes I_L$ & 1.05 $\pm$ 0.00\\
\hspace{1em}$\text{pffr (Wild)}$ & 3.27 $\pm$ 0.01\\
\hspace{1em}$\text{pffr ($z_{1-\alpha/2}$)}$ & 0.93 $\pm$ 0.01\\
\addlinespace[0.6em]
\multicolumn{2}{l}{\textbf{$N = 50,\; n_i = 100$}}\\
\hspace{1em}\text{Li et. al } & 206.20 $\pm$ 1.89\\
\hspace{1em}$\text{1-Step } {\mathbf{R}}_{\text{Lon}} \otimes {\mathbf{R}}_{\text{Fun}}$ & 22.94 $\pm$ 0.39\\
\hspace{1em}$\text{1-Step } {\mathbf{R}}_{\text{Lon}} \otimes I_L$ & 17.91 $\pm$ 0.29\\
\hspace{1em}$\text{1-Step } I_{n_i} \otimes {\mathbf{R}}_{\text{Fun}}$ & 25.71 $\pm$ 0.40\\
\hspace{1em}$\text{1-Step } I_{n_i} \otimes I_L$ & 9.32 $\pm$ 0.12\\
\hspace{1em}$\text{Full-1Step } {\mathbf{R}}_{\text{Lon}} \otimes {\mathbf{R}}_{\text{Fun}}$ & 40.45 $\pm$ 0.75\\
\hspace{1em}$\text{Full-1Step } {\mathbf{R}}_{\text{Lon}} \otimes I_L$ & 29.65 $\pm$ 0.51\\
\hspace{1em}$\text{Full-1Step } I_{n_i} \otimes {\mathbf{R}}_{\text{Fun}}$ & 44.29 $\pm$ 0.72\\
\hspace{1em}$\text{Full-1Step } I_{n_i} \otimes I_L$ & 10.06 $\pm$ 0.14\\
\hspace{1em}$\text{pffr (Wild)}$ & 12.95 $\pm$ 0.19\\
\hspace{1em}$\text{pffr ($z_{1-\alpha/2}$)}$ & 7.46 $\pm$ 0.10\\
\addlinespace[0.6em]
\multicolumn{2}{l}{\textbf{$N = 500,\; n_i = 5$}}\\
\hspace{1em}\text{Li et. al } & 214.57 $\pm$ 1.71\\
\hspace{1em}$\text{1-Step } {\mathbf{R}}_{\text{Lon}} \otimes {\mathbf{R}}_{\text{Fun}}$ & 23.91 $\pm$ 0.03\\
\hspace{1em}$\text{1-Step } {\mathbf{R}}_{\text{Lon}} \otimes I_L$ & 62.74 $\pm$ 0.11\\
\hspace{1em}$\text{1-Step } I_{n_i} \otimes {\mathbf{R}}_{\text{Fun}}$ & 13.69 $\pm$ 0.02\\
\hspace{1em}$\text{1-Step } I_{n_i} \otimes I_L$ & 4.35 $\pm$ 0.01\\
\hspace{1em}$\text{Full-1Step } {\mathbf{R}}_{\text{Lon}} \otimes {\mathbf{R}}_{\text{Fun}}$ & 48.95 $\pm$ 0.29\\
\hspace{1em}$\text{Full-1Step } {\mathbf{R}}_{\text{Lon}} \otimes I_L$ & 125.31 $\pm$ 0.67\\
\hspace{1em}$\text{Full-1Step } I_{n_i} \otimes {\mathbf{R}}_{\text{Fun}}$ & 23.08 $\pm$ 0.12\\
\hspace{1em}$\text{Full-1Step } I_{n_i} \otimes I_L$ & 4.88 $\pm$ 0.01\\
\hspace{1em}$\text{pffr (Wild)}$ & 7.96 $\pm$ 0.02\\
\hspace{1em}$\text{pffr ($z_{1-\alpha/2}$)}$ & 4.48 $\pm$ 0.01\\
\addlinespace[0.6em]
\multicolumn{2}{l}{\textbf{$N = 500,\; n_i = 100$}}\\
\hspace{1em}\text{Li et. al } & ~~~~~~---~~~~~~\\
\hspace{1em}$\text{1-Step } {\mathbf{R}}_{\text{Lon}} \otimes {\mathbf{R}}_{\text{Fun}}$ & 163.34 $\pm$ 0.89\\
\hspace{1em}$\text{1-Step } {\mathbf{R}}_{\text{Lon}} \otimes I_L$ & 130.82 $\pm$ 0.52\\
\hspace{1em}$\text{1-Step } I_{n_i} \otimes {\mathbf{R}}_{\text{Fun}}$ & 208.25 $\pm$ 1.54\\
\hspace{1em}$\text{1-Step } I_{n_i} \otimes I_L$ & 55.85 $\pm$ 0.37\\
\hspace{1em}$\text{Full-1Step } {\mathbf{R}}_{\text{Lon}} \otimes {\mathbf{R}}_{\text{Fun}}$ & 263.90 $\pm$ 3.09\\
\hspace{1em}$\text{Full-1Step } {\mathbf{R}}_{\text{Lon}} \otimes I_L$ & 207.05 $\pm$ 1.16\\
\hspace{1em}$\text{Full-1Step } I_{n_i} \otimes {\mathbf{R}}_{\text{Fun}}$ & 281.65 $\pm$ 5.56\\
\hspace{1em}$\text{Full-1Step } I_{n_i} \otimes I_L$ & 59.26 $\pm$ 0.47\\
\hspace{1em}$\text{pffr (Wild)}$ & 60.79 $\pm$ 0.80\\
\hspace{1em}$\text{pffr ($z_{1-\alpha/2}$)}$ & 48.13 $\pm$ 0.35\\
\bottomrule
\end{tabular}}
\end{table}

\clearpage

\section{Additional Application Analyses} \label{app:whisker_analysis}
\subsection{Background for Calcium Imaging Data Analysis} \label{app:background_calcimImaging}

Since calcium imaging and electrophysiology record the activity of many neurons, and recordings are collected in several animals, analyses differ in how the target population is defined and the nesting of neurons within animal is modeled. For example, the \textit{neural pseudo-population} strategy, as we refer to it, fits a single model to a dataset that pools neurons across animals (e.g. see Figures 1, 3, and 3 of \cite{WILLMORE20233541, calImag, roesch2009ventral}, respectively). This conceptualizes neurons, both within and across animals, as exchangeable given covariates and model parameters. The \textit{animal-specific neural population} strategy, as we refer to it, summarizes the collection of neurons separately in each animal, and then summarizes the animal-specific statistics with a secondary pooled test statistic (e.g. see Figs 2H in \citep{cal_imag_photometry}, Fig 1G, 1I in \cite{nature_sl}). The animal-level summary is usually a model fit to, or an average of,  the activity of all neurons recorded from that animal (e.g. see Figs 2G in \citep{cal_imag_photometry}). This 
ignores uncertainty in the animal-level statistics when estimating a pooled test statistic. A third approach estimates a test statistic on data from each neuron separately and then fits a model to those statistics (e.g. see Figures 1K, 2E-H of \cite{nature_sl}). 
The pooled test ignores uncertainty in the neuron-level statistics, and models the neuron-level statistics estimated on data from neurons in the same animal as independent. 

Analysis strategies differ in how the longitudinal structure of experiments are modeled. One strategy is to treat the neural responses of cluster $i$ –– however  defined –– as exchangeable across trials given model parameters 
(e.g. see Figure 3E of \cite{science_da} for an example from photometry).
A second strategy averages the response across trials and analyzes those trial-averaged measures (e.g. see Figure 2 of \cite{coddington}).
This discards longitudinal information. 
A third strategy accounts for the longitudinal structure with random effects \cite{loewinger2025}, yielding conditional estimates in non-Gaussian outcome settings. 
Analyses vary in how the densely-sampled neural time-series of each trial is conceptualized. 
Arguably, the most common strategy analyzes univariate summaries of the time-series (e.g. a trial firing rate for each neuron $i$ and trial $j$: $\bar{Y}_{i,j} = \frac{1}{| \mathcal{S}|} \sum_{s \in \mathcal{S}} Y_{i,j}(s)$) pooled across animals and/or trials
(e.g. see Figures 1, 3, and 3 of \cite{WILLMORE20233541, calImag, roesch2009ventral}, respectively). This strategy can obscure behavior–brain associations and substantially change scientific conclusions because it discards timing information about how covariate-outcome relationships evolve across trial timepoints \citep{loewinger2025}.
A second strategy is to retain the time-series structure, but model covariate-neural activity associations as constant across trial timepoints for each neuron. For example, \cite{nature_sl} has the goal of identifying neurons associated with a particular behavior over time (e.g. see Figures 1-2). To that effect, they 
fit a Pearson correlation between behavior and neural activity in each cell separately. This is comparable 
to the linear regression $\mathbb{E}[\boldsymbol{Y}_{i,j}(s) \mid {x}_{i,j}(s)] = \gamma_0^{(i)} +\gamma_1^{(i)} x_{i,j}(s)$ (e.g. see figure 1K in \cite{nature_sl}). This models the covariate–outcome relationship as constant across trial timepoints. 
 A third strategy proposed in \citep{loewinger2025} addresses this by modeling each trial as a functional outcome. They demonstrate this  strategy on fiber photometry data, a recording technique that yields a single neural signal per animal. It would be desirable to apply longitudinal FDA strategies to other recording techniques, such as calcium imaging or electrophysiology data. This has not been done, to our knowledge, and we believe this is primarily due to the fact that those modalities record potentially tens of thousands of signals per animal, and existing approaches may not scale.
 
\subsection{Main Text Analyses with Different Correlation Structures} \label{app:neuro_analysis_supp}
{Here we show fGEE estimates fitted with a range of correlation structures to illustrate that the conclusions from the analysis in main text Figures~\ref{fig:neuro_fig}-\ref{fig:stim_fig} are not sensitive to the working correlation structure.}
\begin{figure*}[!t]
	\begin{subfigure}[t]{0.49\textwidth}
\includegraphics[width=0.95 \linewidth]{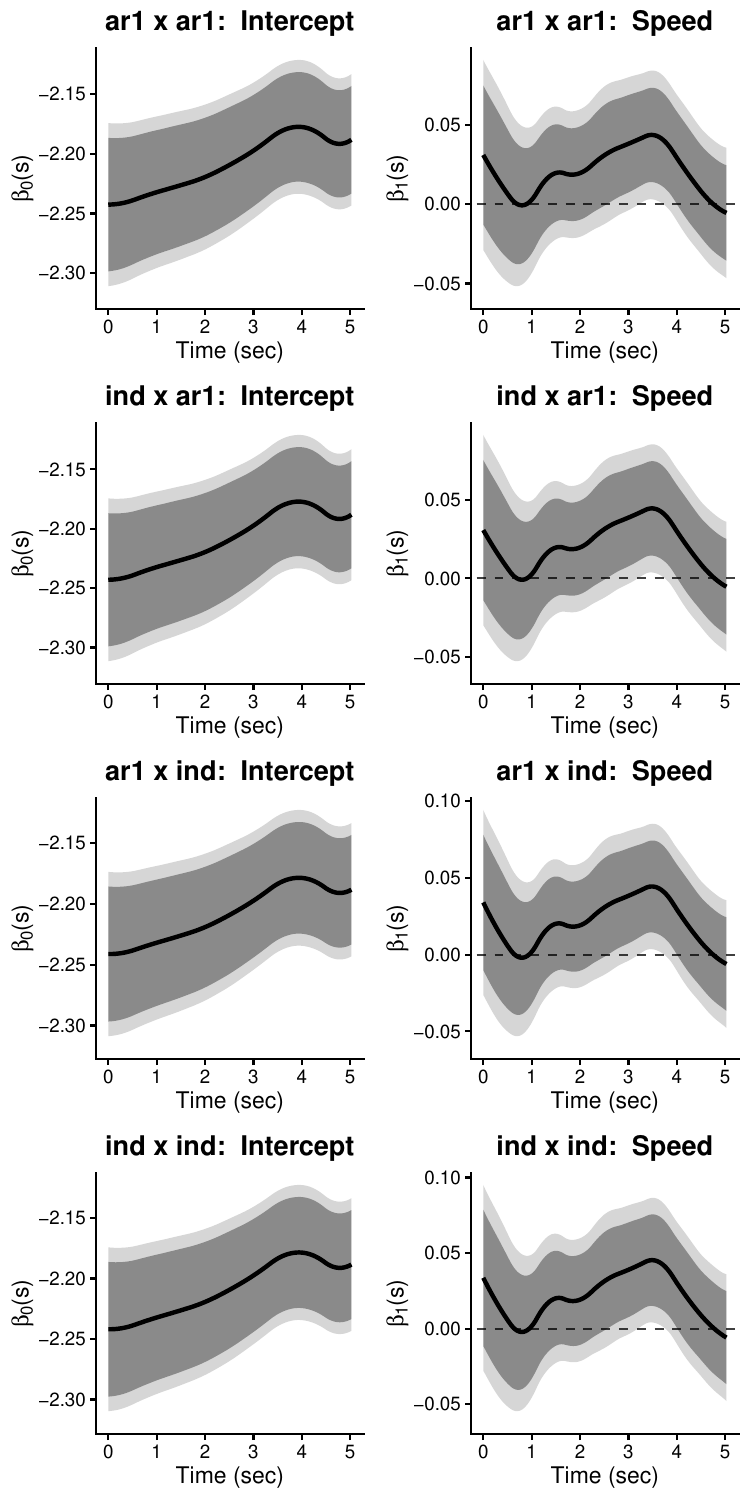}
\caption{\footnotesize \textbf{Speed–neural activity association.}}
	\end{subfigure}
	\begin{subfigure}[t]{0.49\textwidth}
\includegraphics[width=0.95 \linewidth]{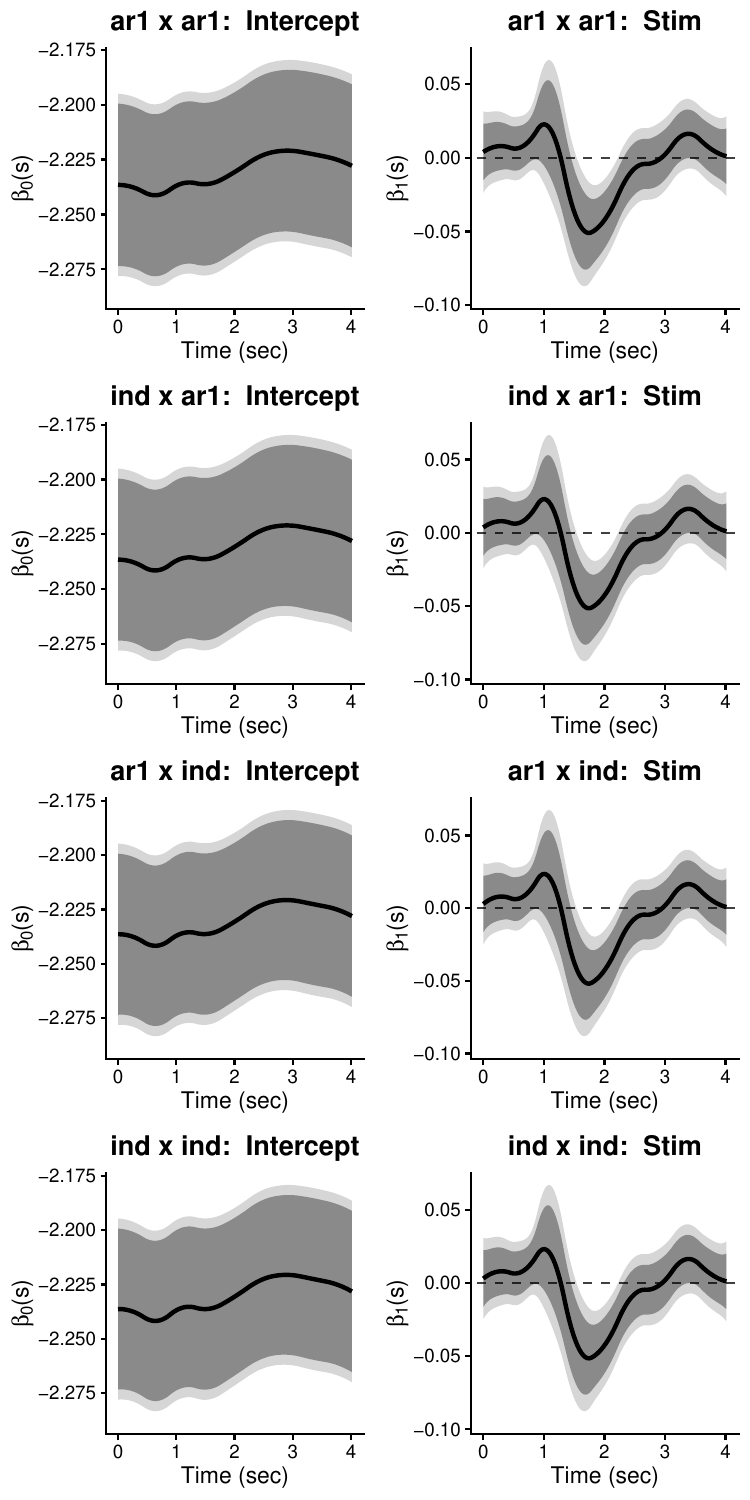}
\caption{\footnotesize 
\textbf{Whisker stimulation effect.}
}
	\end{subfigure}
\caption{\footnotesize
Functional coefficient one-step estimates with a working covariance $\mathbb{V} = \mathbb{V}_{Lon} \otimes \mathbb{V}_{Fun}$. ``ar1 $\times$ ind'' indicates that $\mathbb{V}_{Lon}$ has an AR1 structure, and $\mathbb{V}_{Fun}$ has a working independence correlation structure. (a) Fit times (sec) in order from top to bottom: 349.84, 394.15, 394.61, 188.86. (b) Fit times (sec) in order from top to bottom: 41.48, 36.64, 91.95, 16.65. } 
\end{figure*}

\end{appendices}

\setstretch{1.75} 
\bibliographystyle{biom}
\bibliography{refs}
\end{document}